\documentclass[reqno]{article}

\usepackage{fullpage}
\usepackage{slashed}

\usepackage[english]{babel}
\usepackage[utf8]{inputenc}
\usepackage[LGR,T1]{fontenc}

\newcommand\Koppa{\text{\begingroup\fontencoding{LGR}\selectfont\char21\endgroup}}

\usepackage{amsmath,amssymb,amsthm, comment, mathtools, cite, tikz, imakeidx, musicography, authblk}
\usepackage{hyperref, enumerate}
\makeindex[title=Symbol Index, columns=2 ]

\DeclareFontFamily{U}{mathx}{}
\DeclareFontShape{U}{mathx}{m}{n}{<-> mathx10}{}
\DeclareSymbolFont{mathx}{U}{mathx}{m}{n}
\DeclareMathAccent{\widehat}{0}{mathx}{"70}
\DeclareMathAccent{\widecheck}{0}{mathx}{"71}

\usepackage[margin=1in,dvips]{geometry}
\usepackage{amsmath, amsfonts, amsthm, amssymb,verbatim, marvosym, slashed}

\usepackage{graphicx,color,bbm}
\usepackage{mathrsfs,cancel}

\newcommand{\ph}{{\widehat{u}}}

\newcommand{\Rs}{R^\star}

\newcommand{\Wsv}{W^\natural}
\newcommand{\Ws}{W^\star}
\newcommand{\Wsc}{W{}}
\newcommand{\Wsl}{W^\sharp}

\newcommand{\f}{f_\natural}

\newcommand{\h}{h_\natural}
\newcommand{\re}{r_\natural}

\newcommand{\Gz}{\Gamma_\natural}
\newcommand{\Dz}{\Delta_\natural}
\newcommand{\Gp}{\Gamma_+}

\newcommand{\PH}{\mathcal{P}_H}
\newcommand{\HB}{H_{\BfB}}
\newcommand{\Hchi}{H_\chi}

\newcommand{\PHB}{\mathcal{P}_{\HB}}
\newcommand{\PHchi}{\mathcal{P}_{\Hchi}}

\newcommand{\Pam}{\mathcal{P}_{a, M}}

\newcommand{\Pl}{\mathcal{P}_\ell}
\newcommand{\Ph}{\mathcal{P}_h}

\newcommand{\rs}{r^\star\!}
\newcommand{\BfB}{{\bf B}}

\newcommand{\PX}{\mathcal{P}_\chi}

\newcommand{\Ltheta}{\mathcal{L}_{\theta}}
\newcommand{\ve}{\gamma_\natural}

\newcommand{\wchi}{\widetilde{\chi}}
\newcommand{\rweight}{{\langle r \rangle}_{a}}
\newcommand{\unochi}{{\mathbf{u}_{\mathcal{T}}}}
\newcommand{\bfuT}{\mathbf{u}_{\mathcal{T}}}

\newcommand{\wRs}{\widetilde{R}^\star}
\newcommand{\Ks}{K^\star}
\newcommand{\wpa}{\widetilde{\partial}}

\newcommand{\mcA}{\mathcal{A}}
\newcommand{\mcAs}{\mathcal{A}^\star}
\newcommand{\FM}{{\mathcal{FM}}}
\newcommand{\FMs}{{\mathcal{FM}^\star}}
\newcommand{\WM}{{\mathcal{WM}}}
\newcommand{\WMs}{{\mathcal{WM}^\star}}

\newcommand{\Wdot}[1][-\!q]{{\mathbb{W}^{1\!,1\!;{#1}}_{\!\omega\!, m}}}
\newcommand{\Wbb}[1][s]{{\mathbb{W}[{#1}]}}
\newcommand{\Hdot}[1][q]{{\mathbb{H}^{0;{#1}}_{\!\omega\!, m}}}
\newcommand{\BT}{\mathbf{B}_{\mathcal{T}}}
\newcommand{\BThat}{\widehat{\mathbf{B}}_{\mathcal{T}}}

\newcommand{\Ssharp}{{\mathcal{S}^\sharp_1}}

\newcommand{\ximin}{\xi_\ell}
\newcommand{\ximax}{\xi_r}
\newcommand{\etamin}{\eta_{min}}
\newcommand{\etamax}{\eta_{max}}

\newcommand{\chisharp}{\chi_\sharp}
\newcommand{\rsharp}{r^\sharp}
\newcommand{\rmax}{\overline{r}^\sharp}
\newcommand{\rmin}{\underline{r}^\sharp}

\newcommand{\wtau}{{\widetilde{\tau}}}

\newcommand{\wSigtau}{{\widetilde{\Sigma}_{\widetilde{\tau}}}}
\newcommand{\wphi}{{\widetilde{\phi}}}

\newcommand{\oW}{W_0}
\newcommand{\wfe}{\widetilde{f}_{\natural}}

\newcommand{\Wm}{W_-}
\newcommand{\Wp}{W_+}
\newcommand{\Wpm}{W_{\pm}}
\newcommand{\Uplus}{\mathbf{U}_+}
\newcommand{\Uminus}{\mathbf{U}_-}

\newcommand{\fsharp}{f_\sharp}
\newcommand{\hsharp}{h_\sharp}

\newcommand{\uchi}{\widehat{\mathbf{u}}_\chi}

\newcommand \la \langle
\newcommand \ra \rangle

\newcommand \del	\partial

\newcommand \Ical 	{\mathcal I}

\newcommand \eps 	\epsilon 

\newcommand \be 		{\begin{equation}}
\newcommand \ee 		{\end{equation}}

 \def\pa{\partial}

\newcommand{\Hchicheck}{\check{H}_\chi}

\newtheorem{theorem}{Theorem}[section]
\newtheorem{proposition}[theorem]{Proposition} 
\newtheorem{lemma}[theorem]{Lemma}
\newtheorem{corollary}[theorem]{Corollary} 
\newtheorem{definition}[theorem]{Definition} 
\newtheorem{remark}[theorem]{Remark}

\numberwithin{equation}{section}

\mathtoolsset{showonlyrefs=false}

\DeclareSymbolFont{bbold}{U}{bbold}{m}{n}
\DeclareSymbolFontAlphabet{\mathbbold}{bbold}

\let\oldmarginpar\marginpar
\renewcommand\marginpar[1]{\-\oldmarginpar[\raggedleft\footnotesize #1]%
{\raggedright\footnotesize #1}}
\title{The wave equation on subextremal Kerr spacetimes \\ with small non-decaying first order terms}
\author[1,2]{Gustav Holzegel\thanks{gholzegel@uni-muenster.de}}
\author[1]{Christopher Kauffman\thanks{ckauffma@uni-muenster.de}}
\affil[1]{\small Westf\"alische Wilhelms-Universit\"at M\"unster,
Mathematisches~Institut,~Einsteinstrasse~62,~48149~M\"unster,~Bundesrepublik~Deutschland \vskip.2pc \ }
\affil[2]{\small Imperial College London,
Department of Mathematics,
South~Kensington~Campus,~London~SW7~2AZ,~United~Kingdom}
\begin{document}
\maketitle
\begin{abstract}
We consider the perturbed covariant wave equation
$\Box_{g_{M,a}} \Psi = \varepsilon \BfB \Psi$ on the exterior of a fixed subextremal Kerr spacetime $\left(\mathcal{M},g_{M,a}\right)$. Here $\BfB$ is a suitably regular first order differential operator respecting the symmetries of Kerr whose coefficients are assumed to decay in space but not in time. We establish integrated decay estimates for solutions of the associated Cauchy problem. The proof adapts the framework introduced by Dafermos--Rodnianski--Shlapentokh-Rothman \cite{DRSR} in the $\varepsilon=0$ case. We combine their estimates with a new global pseudodifferential commutator estimate, which generalises our previous work in the Schwarzschild case. The construction of the commutator exploits the central observation of \cite{DRSR} that superradiant frequencies are not trapped. A further major technical ingredient of the proof consists in establishing appropriate convolution estimates for expressions arising from the interplay of the $\BfB$-term and the time cutoffs in the microlocal framework.
\end{abstract}
\setcounter{tocdepth}{2}
\tableofcontents
\section{Introduction}
		This paper investigates the global in time behaviour of smooth solutions to scalar linear wave equations of the form  
\begin{equation}  \label{weintro}
\Box_{g_{a,M}} \Psi  = \varepsilon \BfB\Psi,   \ \ \textrm{with data} \ \ \ \Psi|_{\widetilde\Sigma_0} = \Psi_0 \ \ \ , \ \ \ n_{\widetilde\Sigma} \Psi |_{\widetilde\Sigma_0}= \Psi_1,
\end{equation}
on the domain of outer communications $\mathcal{R}$ of the Kerr manifold $(\mathcal{M},g_{a,M})$ with $|a|<M$ (i.e.~the sub-extremal case), where $\BfB$ is a suitably regular first order differential operator on $\mathcal{M}$ which commutes with the Killing vector fields $T = \partial_t, \Phi = \partial_\phi$ and $\Psi_0, \Psi_1$ are suitably regular functions defined on a spacelike hypersurface $\widetilde\Sigma_0$ which terminates at $\Ical^+$, as indicated in Figure \ref{fig:Penrose} on page \pageref{fig:Penrose} below.

For this introductory section we will assume familiarity with Boyer-Lindquist coordinates $(t,r,\theta, \phi)$ on $(\mathcal{R},g_{a,M})$, as well as regular coordinates $(\widetilde{\tau},r,\theta,\tilde{\phi})$ arising from the above slicing.  See Section \ref{sec:PhysicalSpaceDefs} for geometric background and Section \ref{sec:hyperboloidal} for the precise relation of $(t,r,\theta, \phi)$ and $(\widetilde{\tau},r,\theta,\tilde{\phi})$ coordinates.

\subsection{The main theorem}
We shall prove the following theorem:
\begin{theorem}\label{thm:Main}
Let $\BfB$ be a first order differential operator on $(\mathcal{R},g_{a,M})$ which is $C^1$ (including at the future event horizon $\mathcal{H}^+$), satisfies $\mathcal{L}_T \BfB=\mathcal{L}_\Phi \BfB=0$ and when expressed in the regular $(\tilde{\tau},r,\theta,\tilde{\phi})$-coordinates,
\begin{equation}
\BfB\Psi = \BfB^{\tilde{\tau}}T \Psi+\BfB^r\widetilde{\partial}_r\Psi+\BfB^\theta\widetilde{\partial}_\theta\Psi+\BfB^{\tilde{\phi}}\Phi \Psi+\BfB^0\Psi \, ,
\end{equation}
also the bounds
\begin{subequations}
\begin{align}
|\BfB^{\tilde{\tau}}| + |r\wpa_r\BfB^{\tilde{\tau}}| +|\BfB^r| + |r\wpa_r\BfB^r|+|r\BfB^0| + |r^2\wpa_r\BfB^0| &\leq C r^{-1-\alpha},  \label{Bassumption1} \\
|\BfB^{\tilde{\phi}}| + |r\wpa_r\BfB^{\tilde{\phi}}|+|\BfB^\theta| + |r\wpa_r\BfB^\theta| &\leq C r^{-2-\alpha}, \label{Bassumption2}
\end{align}
\end{subequations}
for some $\alpha>0$.
Then there exists an $\varepsilon_0>0$ (depending only on $\BfB$,$M$,$a$) such that for all $|\varepsilon|<\varepsilon_0$
 the following statement is true. Smooth solutions of (\ref{weintro}) satisfy, for any $\mathsf{T}>0$, the estimate 
 \begin{align} \label{pse}
 \int_0^{\mathsf{T}} d\widetilde{\tau}  \int_{\widetilde{\Sigma}_{\widetilde{\tau}}} r^2 dr \sin \theta d\theta d\widetilde{\phi}  \left\{ \frac{|T \Psi|^2 + |\widetilde{\partial}_r \Psi|^2}{r^{1+\alpha}}+ \frac{|\widetilde{\partial}_\theta \Psi|^2  + \frac{1}{\sin^2 \theta} |\Phi \Psi|^2}{r^3} + \frac{|\Psi|^2}{r^{3+\alpha}}  \right\}  \leq C E_0^1[\Psi] \, 
 \end{align}
 for a ${C}>0$ depending only on $\BfB$, $M$, $a$. Here the initial data energy on the right is defined by 
 \begin{align} \label{e10def}
E_0^1[\Psi] := E_0[\Psi] + E_0[W_0 \Psi] + E_0[T\Psi]+ E_0[\Phi\Psi] \, , 
 \end{align}
where
 \begin{align} \label{basicdataenergy}
E_\tau [\Psi] =  \int_{\widetilde{\Sigma}_\tau} r^2 dr \sin \theta d\theta d\widetilde{\phi} \left\{ \frac{1}{\Delta} |W_0 \Psi|^2 + \frac{1}{r^2}| T\Psi|^2 + \frac{1}{r^2} \left( |\widetilde{\partial}_\theta \Psi|^2  + \frac{1}{\sin^2 \theta} |\Phi \Psi|^2 + |\Psi|^2 \right) \right\} \, ,
 \end{align}
 $W_0$ is a spacetime vector field defined in (\ref{W0def}) and $\Delta=r^2-2Mr+a^2$.
\end{theorem}
 \begin{remark}
In the $\varepsilon=0$ case, the estimate \eqref{pse} follows from the celebrated \cite{DRSR}, which forms the endpoint of numerous previous works on the problem -- see the references in \cite{DRSR}. Both \cite{DRSR} and our proof of Theorem \ref{thm:Main} establish more refined statements than \eqref{pse}, including control on certain second order derivatives on the left. However, the physical space integrated local energy decay statement of \eqref{pse} lies at the heart of the matter and makes manifest that solutions indeed decay on the black hole exterior provided one allows for a loss of derivatives. The necessity of such a loss is well-known \cite{Sb15}.
\end{remark}
\begin{remark}
For $\varepsilon$ large one cannot expect the estimate \eqref{pse} to hold without exploiting further structure in the $\BfB$-term. See \cite{MoschidisSR} for very general results. The Teukolsky equation \cite{Teuko} is an example where such further structure is indeed present allowing one to prove (higher order) analogues of \eqref{pse}. See \cite{DHR, DHRteuk, Ma, RitaShlap}.
\end{remark}

\begin{remark}
The assumption that $\mathcal{L}_T \BfB=\mathcal{L}_\Phi \BfB=0$ is mainly for technical reasons (as it allows for simpler convolution estimates in the proof) and consistent with what one expects to see in applications. Our actual assumptions on $\BfB$ used in Theorem \ref{thm:MainFreqSpace}, which is a stronger frequency space version of Theorem \ref{thm:Main}, are significantly weaker and expressed entirely in Fourier space via uniform boundedness of the norm (\ref{Bnorm}) for a cut off version of $\BfB$.
\end{remark}

\subsection{Review of the Schwarzschild case} \label{sec:reviews}

The reason that Theorem \ref{thm:Main} does not directly follow perturbatively from  the $\varepsilon=0$ case, even for $a=0$, is rooted in the existence of trapped null geodesics on the black hole exterior. The $\BfB$-term cannot simply be absorbed in the estimates because general first derivatives of $\Psi$ are not controlled without degeneration in the region where trapped null geodesics exist, even for $\varepsilon=0$. In fact, it follows from the results of \cite{Sb15} that, for general $\BfB$ and $|\varepsilon|>0$, uniform boundedness (in the form $ E_{\mathsf{T}}[\Psi] \leq C E_0[\Psi]$ for all $\mathsf{T}>0$ with $C$ independent of $\mathsf{T}$) cannot hold without a loss of derivatives.\footnote{We note that $E_T[\Psi]$ could be added to the left hand side of (\ref{pse}) as the right hand side already loses derivatives.} Therefore, one cannot follow the well-established strategy \cite{DRNotes}  from the $\varepsilon=0$ case to prove uniform boundedness and integrated decay.
  
 In the $a=0$ case, the resolution offered in  \cite{HK20} was to commute (\ref{weintro}) with the spacetime vector field (expressed in the familiar Schwarzschild $(t,r,\theta,\phi)$-coordinates)
 \begin{align} \label{Wsdef}
 W_{S} =  \frac{r}{\sqrt{1-\frac{2M}{r}}} \left(\left(1-\frac{2M}{r}\right)\partial_{r} + h(r) \partial_t \right)  \ \ \ \textrm{with} \ \ \ h(r)=\left(1-\frac{3M}{r}\right)\sqrt{1+\frac{6M}{r}} \, ,
 \end{align}
which has the property that the commutator generates a term on the right hand side of the equation for $\Box_g W_S\Psi$, which, when applying the $\partial_t$-energy estimate, produces a non-degenerate spacetime term of the right sign of the form
 \begin{align} \label{wui}
 \int_0^T d\tau \int_{\Sigma_{\tau}} \frac{1}{r} \frac{1}{1-\frac{2M}{r}} |\partial_t W_S\Psi|^2
 \end{align}
 at the cost of only first derivative terms on the right hand side of $\Box_g W_S\Psi$. A Lagrangian estimate then allowed us to use (\ref{wui}) to control \emph{all} second derivatives of the form $DW_S\Psi$ non-degenerately (in spacetime $L^2$) in terms of lower order $D\Psi$-terms. In particular, the $\BfB W_S\Psi$-term on the right hand side could then be absorbed. Finally, a Hardy-type-estimate controlling non-degenerately $|D\Psi|^2$ using only a small portion of $|D W_S\Psi|^2$ and a large portion of the degenerate (at $r=3M$) energy allowed to absorb the first derivative error-terms using a large portion of the uncommuted energy and a small portion of the $W_S$ commuted energy. See \cite{HK20} for the details.

In the Kerr case, the above strategy does not work for two main reasons: (1) The $\partial_t$-estimate does not produce coercive boundary terms because of the well-known phenomenon of superradiance.\ (2) There does not exist a global vector field which inherits the good commutation properties of the Schwarzschild case. 

\subsection{Physical space estimates}
Before we discuss the resolution of the above difficulties in the next subsection we make two remarks about physical space estimates in the Kerr setting. First of all, for \emph{axisymmetric} solutions $\Psi$ of (\ref{weintro}), the proof in the Schwarzschild case generalises as superradiance is absent and trapping happens at a single radial value of $r$. One may replace $W_S$ by the vector field $W_0$ defined in (\ref{W0def}) and  the proof outlined in the previous section carries through. The second observation is that, in general, the good properties of $W_S$ can be salvaged in physical space in the Kerr case at least \emph{near the event horizon} and \emph{near infinity} in the following sense\footnote{Note that $\varepsilon$ does not have to be small in Proposition \ref{prop:ps}. However, as $\varepsilon$ becomes large, the parameters $R$ and $\tilde{\delta}$ would need to depend on $\varepsilon$. Finally, one may reasonably expect Proposition \ref{prop:ps} to hold for a large class of asymptotically flat spacetimes with non-degenerate event horizon but we will not pursue this direction here.}:
\begin{proposition} \label{prop:ps}
Consider solutions to (\ref{weintro}) with $\BfB$ satisfying the assumptions of Theorem \ref{thm:Main}. Let
\begin{align} \label{WKerr}
W_{\mathcal{I}^+} := \frac{{r^2+a^2}}{\sqrt{\Delta}} \left(R^\star + T + \frac{a}{r^2+a^2}\Phi\right) \ \ \ , \ \ \ W_{\mathcal{H}^+} := \frac{{r^2+a^2}}{\sqrt{\Delta}} \left(R^\star - T - \frac{a}{r^2+a^2}\Phi\right) \, 
\end{align} 
and define $|D\psi|^2= |L\psi|^2 + |\underline{L} \psi|^2 + |\slashed{\nabla}\psi|^2$ and $|\tilde{D}\psi|^2= |L\psi|^2 + |\frac{r^2+a^2}{\Delta} \underline{L} \psi|^2 + |\slashed{\nabla}\psi|^2$ with $L$ and $\underline{L}$ defined in (\ref{def:LLbar}) and $|\slashed{\nabla}\psi|^2$ denoting the norm of the gradient induced on the spheres of constant $(\tilde{\tau},r)$.  \\
Then
there exists an $R>0$ and a $\tilde{\delta}>0$ (both depending on $\BfB$, $a$, $M$) such that the following estimates hold for any $\mathsf{T}>0$:
\begin{align} \label{largeRps}
\int_{0}^{\mathsf{T}} d\tau \int_{\widetilde{\Sigma}_{\tau} \cap \{r \geq R\}} r^2 dr d\sigma\ \frac{|DW_{\mathcal{I}^+} \Psi|^2 }{r^{1+\alpha}} \lesssim &
 \int_{0}^{\mathsf{T}} d\tau \int_{\widetilde{\Sigma}_{\tau} \cap \{R-M \leq r \leq R\}} r^2 dr d\sigma \Big\{ |DW_{\mathcal{I}^+} \Psi|^2 \Big\}+ E_0[\chi_R W_{\mathcal{I}^+}\Psi]  \nonumber \\
&+ E_0[\Psi]  + \int_{0}^{\mathsf{T}} d\tau \int_{\widetilde{\Sigma}_{\tau} \cap \{r \geq R-M\}} r^2 dr d\sigma  \frac{ |D\Psi|^2 + r^{-2}|\Psi|^2}{r^{1+\alpha}} \, , 
\end{align}
where $d\sigma=\sin \theta d\theta d\widetilde{\phi} $, $\chi_R$ is a smooth radial cut-off equal to $1$ for $r\geq R$ and zero for $r \leq R-M$, and
\begin{align} \label{hozi}
\int_{0}^{\mathsf{T}} d\tau \int_{\widetilde{\Sigma}_{\tau} \cap \{r \leq r_++\tilde{\delta}\}} r^2 dr d\sigma  |DW_{\mathcal{H}^+}\Psi|^2 \lesssim &\int_{0}^{\mathsf{T}} d\tau \int_{\widetilde{\Sigma}_{\tau} \cap \{r_++\tilde{\delta}\leq r \leq r_++2\tilde{\delta}\}} r^2 dr d\sigma |DW_{\mathcal{H}^+}\Psi|^2 + E_0[\chi_{\tilde{\delta}}W_{\mathcal{H}^+}\Psi] \nonumber \\
&+ E_0[\Psi]  + \int_{0}^{\mathsf{T}} d\tau \int_{\widetilde{\Sigma}_{\tau} \cap \{ r \leq r_++2\tilde{\delta}\}} r^2 dr d\sigma \left(  |\tilde{D}\Psi|^2 + |\Psi|^2\right) \, ,
\end{align}
where $\chi_{\tilde{\delta}}$ is a smooth radial cut-off equal to $1$ for $r\leq r_++\tilde{\delta}$ and equal to zero for $r \geq r_++2\tilde{\delta}$. \\ Finally, the $\lesssim$ depends on $R, \BfB, a, M$ in the first estimate and on $\tilde{\delta}, \BfB,a,M$ in the second.
\end{proposition}
In other words, all derivatives of $W_{\mathcal{I}^+}\Psi$ can be controlled $L^2$ in spacetime near infinity if we can control them in a region of bounded $r$ and in addition control lower oder terms. Similarly near the horizon for $W_{\mathcal{H}^+}\Psi$ in (\ref{hozi}). In fact, the above estimates can be viewed as versions of the well-known $r^p$-estimates and redshift estimates respectively, which exploit the fact that asymptotically $W_{\mathcal{I}^+} \sim r L$ and $W_{\mathcal{H}^+} \sim \Delta^{-\frac{1}{2}} \underline{L}$.\footnote{The proof of (\ref{hozi}) actually shows that one can also replace $|DW_{\mathcal{H}^+}\Psi|^2$ by $\frac{1}{\Delta}|KW_{\mathcal{H}^+} \Psi |^2$  on the left hand side of (\ref{hozi}) (where $K=T+\frac{a}{r_+^2+a^2}\Phi$ is the Hawking Killing field), thereby making the redshift property more manifest.} The ``miracle” is that these vector fields have good commutation properties in the sense that a suitably localised energy estimate (based on the cut-off vector fields $\chi_R \, T$ for (\ref{largeRps}) and $\chi_{\tilde{\delta}} (T+\frac{a}{r_+^2+a^2}\Phi)$ for (\ref{hozi})) combined with a straightforward Lagrangian estimate for the commuted equation lead to the above estimates. Proposition \ref{prop:ps} will be proven in Appendix \ref{sec:Prop121Proof}. It is not going to be applied anywhere in this paper.

\subsection{Overview of the proof}
We now turn to some of the details in the proof of Theorem \ref{thm:Main}. In short, the resolution of the obstructions mentioned at the end of Section \ref{sec:reviews} is to construct a \emph{pseudodifferential} operator that mimics the properties of $W_S$ for the part of the solution which is supported on certain (non-superradiant) frequencies, i.e.~we carry out the commutation at the frequency decomposed level with an operator that depends on the frequencies.

\subsubsection{The class of future integrable solutions}
To employ the required frequency analysis, we need to ensure we can take the Fourier transform of the solution of (\ref{weintro}). However, solutions may a priori grow exponentially in time. To resolve this, we follow a variant of \cite{DRSR} to suitably cut-off the solution in time leading to the notion of \emph{future integrable ($L^2$ in time) solutions} from \cite{DRSR}. More specifically, let $\chi : \mathbb{R} \rightarrow \mathbb{R}$ be a smooth cut-off function identically zero for $x \leq 0$ and equal to $1$ for $x \geq 1$. For fixed $\mathcal{T}>2$ we define $\chi_{\mathcal{T}} (\widetilde{\tau})=\chi\big(\tfrac{\wtau}{\mathcal{T}} + 1\big)\chi\big(2-\tfrac{\wtau}{\mathcal{T}}\big)$ so that in particular $\chi_{\mathcal{T}} = 1$ for $\widetilde{\tau} \in [0, \mathcal{T}]$, and $\chi_{\mathcal{T}}$ is compactly supported in $\wtau \in [-\mathcal{T},2\mathcal{T}]$. 

 Let $\Psi$ satisfy  (\ref{weintro}) on $\mathcal{M} \cap \{\widetilde{\tau} \geq 0\}$. We now define $\Psi_{\mathcal{T}}$ to be the solution of
\begin{equation}
\Box_{g_{a, M}}\Psi_{\mathcal{T}} = \varepsilon\chi_{\mathcal{T}}\BfB \Psi_{\mathcal{T}} \, , 
\end{equation}
with initial data $\Psi_{\mathcal{T}}|_{\widetilde\Sigma_0} = \Psi_0$ and $n_{\widetilde\Sigma} \Psi_{\mathcal{T}} |_{\widetilde\Sigma_0}= \Psi_1$ as in (\ref{weintro}). By domain of dependence, $\Psi_{\mathcal{T}}=\Psi$ for $\widetilde{\tau} \in [0,{\mathcal{T}}]$.
Finally, we define
\begin{equation}
 \Psi_{\chi , {\mathcal{T}}}= \chi\Psi_{\mathcal{T}} \, ,
\end{equation}
which is well-defined on all of $\mathcal{R}$ since $\chi$ vanishes identically for all $\widetilde{\tau}\leq 0$. Moreover, $\Psi_{\chi, {\mathcal{T}}}$ satisfies
\begin{equation}\label{fui}
\Box_{g_{a, M}}( \Psi_{\chi , {\mathcal{T}}}) =  F+ \varepsilon\chi_{\mathcal{T}}(\BfB \Psi_{\chi , {\mathcal{T}}})  \ \ \ \textrm{where \ \ \ $F:= 2\nabla^\alpha \chi \nabla_\alpha \Psi_{\mathcal{T}} + (\Box_{g_{a, M}}\chi)\Psi_{\mathcal{T}} - \varepsilon \chi_{\mathcal{T}}\Psi_{\mathcal{T}} (\BfB \chi - \BfB^0 \chi)$,}
\end{equation}
with trivial initial data at $\widetilde{\tau}=0$. Since $\Psi_{\chi,{\mathcal{T}}}$ is equal to the original $\Psi$ for $\widetilde{\tau} \in [1,{\mathcal{T}}]$ proving (\ref{pse}) for ${\Psi}_{\chi,{\mathcal{T}}}$ with the $C$ on the right hand side independent of $\mathcal{T}$, will allow us to infer (\ref{pse}) for $\Psi$. The reason for introducing the cutoff factor $\chi_{\mathcal{T}}$ is that ${\Psi}_{\chi,{\mathcal{T}}}$ is then manifestly future integrable in the sense of \cite{DRSR} since ${\Psi}_{\chi,{\mathcal{T}}}$ satisfies the \emph{free} wave equation for $\widetilde{\tau} \geq 2{\mathcal{T}}$ (whose future integrability is a corollary of \cite{DRSR}) and vanishes identically with all derivatives at $\widetilde{\tau} =0$. Note that the term $F$ on the right hand side of (\ref{fui}) is supported only for $\mathcal{M} \cap \{ \widetilde{\tau} \in [0,1]\}$.
The reader may appreciate here the following technical point: Cutting off in the way described (instead of cutting off $\Psi$ itself to the past and to the future) avoids a future error-term in $F$ coming from the cutoff, which in our setting is problematic as we will not be able to obtain a uniform boundedness statement for the top order energy. The drawback is the cutoff on the first order term in (\ref{fui}), which at the Fourier level requires
some convolution estimates which we discuss below.

\subsubsection{The frequency localised equation}
To convey the main ideas, we restrict in the following to $\BfB= b(r,\theta)\partial_t$ in Boyer-Lindquist coordinates. 
Taking the Fourier transform in $t$ and in the axisymmetric direction $\phi$ (indicated by a hat and the subscript $m$) yields from (\ref{fui}) the separated equation
\begin{align} \label{omegam}
\left(\widehat{u}_{m}\right)^{\prime \prime} + (\omega-\omega_r)^2- \frac{\Delta}{(r^2+a^2)^2} \mathcal{L}_\theta[a\omega,m] \widehat{u}_m -V_{m}(r) \widehat{u}_{m }= \varepsilon  \cdot b(r,\theta) (\widehat{\chi_{\mathcal{T}}} * i\omega \widehat{u})_{m} + H_{m} \, ,
\end{align}
where $\widehat{u}_m=\sqrt{r^2+a^2} (\widehat{\Psi_{\chi,\mathcal{T}}})_m$, $H_m=\frac{\Delta \rho^2}{(r^2+a^2)^\frac{3}{2}} \widehat{F}_m$ and $\mathcal{L}_\theta[a\omega,m]=-\frac{1}{\sin \theta} \partial_{\theta} \left(\sin \theta \partial_\theta \cdot \right) + \left(a\omega \sin \theta -\frac{m}{\sin \theta}\right)^2$ is a second order angular operator that reduces to the spherical Laplacian if $a=0$.

For completeness, we also collect the fully separated form of (\ref{fui}) in the notation of \cite{DRSR}:
\begin{align} \label{fullsep}
\left(\widehat{u}^{(a\omega)}_{m \ell}\right)^{\prime \prime} + (\omega^2-V^{(a\omega)}_{m \ell}(r)) \widehat{u}^{(a\omega)}_{m \ell}=  \varepsilon  \cdot b(r,\theta) (\widehat{\chi_{\mathcal{T}}} * i\omega \widehat{u})^{(a\omega)}_{m \ell} + H^{(a\omega)}_{m \ell} \, ,
\end{align}
where $\widehat{u}^{(a\omega)}_{m \ell}=\sqrt{r^2+a^2}  (\widehat{\Psi_{\chi,\mathcal{T}}})^{(a\omega)}_{m \ell}$, $H^{(a\omega)}_{m \ell}=\frac{\Delta}{(r^2+a^2)^\frac{3}{2}} (\rho^2\widehat{F})^{(a\omega)}_{m\ell}$ and $V^{(a\omega)}_{m \ell}(r)=\frac{4Mr a m\omega -a^2 \omega^2 + \Delta \Lambda}{(r^2 +a^2)^2}+ V_1(r)$ with $\Lambda$ denoting the eigenvalues of the spheroidal Laplacian (see Section \ref{sep:KTensor}) enumerated by $\ell$ and $m$.

The main novel analytic aspects of this paper, including the commutation, will appear at the level of (\ref{omegam}).

\subsubsection{The separated estimate at the lowest order}
From the proof of Theorem 8.1 in \cite{DRSR} we deduce for some $r_{trap}=r_{trap}(M,a,\omega,m,\Lambda)$ with $r_+ < r_1 \leq r_{trap} \leq r_2<\infty$ (and $r_1,r_2$ depending only on the parameters $M$ and $a$)
the estimate
\begin{align} \label{fitt}
\int_{-\infty}^{\infty} dr^\star \frac{\Delta}{r^2} \left[\frac{\big|\widehat{u}^{(a\omega)}_{m \ell}{}^\prime\big|^2}{r^{1+\alpha}} + \left(\left(\frac{\Lambda}{r^3} +\frac{\omega^2}{r^{1+\alpha}}\right)\left(1-\frac{r_{trap}}{r}\right)^2+\frac{1}{r^{3+\alpha}}\right)\big| \widehat{u}^{(a\omega)}_{m \ell}\big|^2\right]  \nonumber \\
\lesssim \varepsilon \int_{-\infty}^{\infty} dr^\star  \frac{\Delta}{r^2}  \frac{\omega^2}{r^{1+\alpha}} \big| \widehat{u}^{(a\omega)}_{m \ell}\big|^2 + \mathbbm{1}_{bfr} |\widehat{u}^{(a\omega)}_{m \ell}(-\infty)|^2 + \textrm{terms involving $H^{(a\omega)}_{m \ell}$}\, .
\end{align}
The second and third term on the right hand side are present already for $\varepsilon=0$. The second term\footnote{For $|a| \ll M$ this term can in fact be dropped, as it can be absorbed in the proof. See \cite{DafRodsmalla}.} is supported on a range of bounded frequencies (see (\ref{mosta1}), (\ref{mosta2}) below and \cite{DRSR} for further discussion) and is controlled (for $\varepsilon=0$) by the initial energy on $\Sigma_0$ using the quantitative mode stability result of \cite{SR14}. As this estimate generalises straightforwardly to incorporate the first order $\BfB$-term, we will not comment further on dealing with this term. See Proposition \ref{prop:estimateh} and Appendix \ref{appendix:yakov}.
 The term involving $H^{(a\omega)}_{m \ell}$ on the right hand side contains products of the various multipliers with the inhomogeneous terms that arise when establishing (\ref{fitt}). When transforming (\ref{fitt}) back to physical space, this term will be either supported for $\widetilde{\tau} \in [0,1]$ only (namely for the $T$-multiplier) and easily controlled from local energy estimates near data, or controlled by borrowing a small amount of the non-degenerate term on the left. We will not comment on it further here. See Section \ref{sec:InitialDataBounds} and also Section 9.6 of \cite{DRSR} for further details.

Most important is the first term on the right hand side of (\ref{fitt}), which arises from the $\BfB$-term. As familiar already from the Schwarzschild case, it cannot be absorbed on the left for all frequencies because of the degeneration due to trapping on the left hand side. It can, however, be immediately absorbed for all frequency triples $(\omega,m,\Lambda)$, for which $r_{trap}=0$. This range includes the range of superradiant frequencies, more specifically the slightly larger range $\mathcal{G} = \Big\{ (\omega, m, \Lambda) \ \ | \ \ m\omega \in \left(-m^2\xi_\ell, \frac{a m^2}{2Mr_+}+m^2\xi_r\right) \Big\} $, where $\xi_\ell, \xi_r>0$ depend only on $M$ and $a$ and are such that $\xi_r \rightarrow 0$ as $|a| \rightarrow M$. Note that the above uses the fundamental insight from \cite{DRSR} that the superradiant frequencies are (quantitatively) not trapped!

\subsubsection{The trapped frequencies} \label{sec:trafeintro}
The strategy to obtain improved control for the solution supported on frequencies in the complement of $\mathcal{G}$ (which includes trapped frequencies) is to commute the separated equation (\ref{omegam}) with a pseudodifferential operator, which mimicks the vector field $W_S$ from (\ref{Wsdef}) in the Schwarzschild case. More specifically, we construct a $W^\natural=W^\natural(m,\omega,r)$ which is a priori defined in $\mathcal{G}^\prime \supset \mathcal{G}^c$ and has good commutation properties ``up to lower order terms”, in the sense that the $T$-estimate for the commuted equation provides control of the pseudo-differential derivative $|\omega W^\natural \psi|^2$ \emph{non-degenerately} in frequency space.\footnote{For concreteness, the set $\mathcal{G}^\prime$ will be defined in (\ref{def:GCalPrimeDef}) as the complement of $\Big\{ (\omega, m, \Lambda) \ | \ m\omega \in \left(-\frac{\xi_\ell}{2}, \frac{a m^2}{2Mr_+}+\frac{\xi_r}{2}\right) \Big\}$ and therefore not contain any superradiant frequencies.} Controlling  $|\omega W^\natural \psi|^2$ provides in turn non-degenerate control for \emph{all} first derivatives of $W^\natural \psi$ (indeed a form of the Lagrangian estimate applies in $\mathcal{G}^\prime$, where frequencies are non-superradiant) allowing to absorb the first term in (\ref{fitt}). In summary, this leads to an estimate of the form
\begin{align} \label{fitt2}
&\int_{-\infty}^{\infty} dr^\star \int \sin \theta d\theta \left[\frac{1}{r^{1+\alpha}} \big|\left(W^\natural\widehat{u}_{m}\right)^\prime\big|^2 +\frac{1}{r^3}\big| \partial_\theta W^\natural\widehat{u}_{m}\big|^2+ \left(\frac{1}{r} \omega^2 +\frac{m^2}{r^3}+\frac{1}{r^{3+\alpha}}\right)\big| W^\natural\widehat{u}_{m}\big|^2\right] \nonumber \\
\lesssim &\int_{-\infty}^{\infty} dr^\star \int \sin \theta d\theta \frac{\Delta}{r^2} \left[\frac{\big|\widehat{u}_{m}{}^\prime\big|^2}{r^{1+\alpha}} +\frac{1}{r^3}\big| \partial_\theta \widehat{u}_{m}\big|^2 + \left(\frac{m^2}{r^3} +\frac{\omega^2}{r^{1+\alpha}}+\frac{1}{r^{3+\alpha}}\right)\big| \widehat{u}_{m}\big|^2\right]  \nonumber \\
+ &\Big|\int_{-\infty}^{\infty} dr^\star  \int \sin \theta d\theta \frac{\Delta}{r^2} \left(\varepsilon  \cdot b(r,\theta) W^\natural(\widehat{\chi_{\mathcal{T}}} \star i\omega \widehat{u})_{m} + W^\natural H_{m}\right)  \left(\frac{1}{r^{1+\alpha}} W^\natural\widehat{u}_{m},\frac{1}{r} W^\natural\widehat{u}_{m}, \omega (W^\natural\widehat{u}_{m})\right) \Big| \, .
\end{align}
For the term involving $W^\natural H_m$ we refer directly to Section \ref{sec:InitialDataBounds}.\footnote{Intuitively, one might expect that after transforming back to physical space, this term will be  supported for $\widetilde{\tau} \in [0,1]$ only and hence easily controlled by initial data. However, there is an important caveat here as $W^\natural$ is not a physical space operator (not even near infinity!) preventing a direct application of Parseval's identity. See also footnote \ref{footnote:DRSRanalogue}.}
In Section \ref{intro:convolve} we will discuss how the term involving the convolution can be either absorbed on left hand side (using that $\varepsilon$ is small) or incorporated into the first term on the right hand side. Assuming this is done, we can easily convert the first line of (\ref{fitt2}) into an expression for the fully separated pieces $\widehat{u}^{(a\omega)}_{m \ell}$ (using that $W^\natural$ only depends on $\omega$ and $m$, not on $\Lambda$). The resulting expression can then be combined with the lowest order estimate (\ref{fitt}) through the means of a Hardy inequality of the form (see Lemma \ref{lem:DBoundtoNDBound} in Section \ref{sec:NDLOT})
\begin{align} \label{fitt3}
\int_{R^\star_1}^{R^\star_2} dr^\star  \omega^2 |\widehat{u}^{(a\omega)}_{m\ell}|^2 \lesssim C_{M,a} \left( \frac{1}{\hat{\delta}} \int_{R^\star_1}^{R^\star_2} dr^\star \left(1-\frac{r_{trap}}{r}\right)^2\omega^2 |\widehat{u}^{(a\omega)}_{m \ell}|^2 + \hat{\delta} \int_{R^\star_1}^{R^\star_2}dr^\star \omega^2 |W^\natural \widehat{u}^{(a\omega)}_{m \ell}|^2 \right)
\end{align}
for $\hat{\delta}>0$ and $-\infty<R^\star_1<r^\star(r_{trap})<R^\star_2<\infty$ with $R^\star_1,R^\star_2$ uniform in the frequency parameters, i.e.~also depending only on $M$ and $a$. Combining (\ref{fitt}), (\ref{fitt2}) and (\ref{fitt3}) allows to finally absorb the first term on the right in (\ref{fitt}) choosing $\hat{\delta}$ sufficiently small and then $\varepsilon \ll \hat{\delta}$, in complete analogy with the Schwarzschild case. Transforming the resulting estimate back to physical space then in principle leads to (\ref{pse}) for the trapped frequencies, provided we can control the terms in the last line of (\ref{fitt2}) as mentioned after (\ref{fitt2}).

We remark that the operator $W^\natural$ used above will be frequency dependent \emph{globally} in $r$. This seems necessary as cutting it off to a physical space multiplier near infinity would introduce general second derivatives which we generally do not know how to control (even away from trapping), not even in the Schwarzschild case! However, near infinity $W^\natural$ approximates the physical space vector field $W_{\mathcal{I}^+}$ up to terms that can be bounded by an integrated decay estimate. Similarly, near the horizon $W_{\mathcal{H}^+}$ approximates $W^\natural$.

\subsubsection{The superradiant frequencies}
Before we come to the important convolution estimates in Section \ref{intro:convolve}, there is a further technical obstacle to overcome that has to do with the weights in $r$. It turns out that if we simply cut off the vector field $W^\natural$ in $\mathcal{G} \cap \mathcal{G}^\prime$ so as to vanish in $(\mathcal{G}^\prime)^c$, say, then 
the resulting convolution estimates discussed below (as well as the error estimates on the cut-off term involving $W^\natural H$ in (\ref{fitt3})) do not close because they would exhibit bad decay in $r$.\footnote{This issue is somewhat present already in \cite{DRSR} and resolved there by employing only multipliers which are frequency {independent up to exponentially decaying terms} near infinity. We do not have this option here in view of the remarks at the end of Section \ref{sec:trafeintro}. \label{footnote:DRSRanalogue}} This is caused by the fact that the microlocal commutator now takes a very different form near $r=\infty$ for different frequencies  (namely vanishing identically for superradiant frequencies!).  We can remedy this by commuting the wave equation with an operator that -- while still frequency dependent -- looks \emph{the same asymptotically near infinity for all frequencies}. In this spirit we commute the equation for frequencies in $\mathcal{G}$ with an operator that for large $r$ agrees with $W_{\mathcal{I}^+}$. It turns out that in $\mathcal{G}$ one can commute directly with the operator $\chi_\sharp \cdot W_{\mathcal{I}^+}$, where $\chi_\sharp$ is a smooth cut-off in space depending on the frequency but the cut-off region being uniformly away from the horizon and infinity. This exploits a globally good sign of the commutator term produced by $W_{\mathcal{I}^+}$ for superradiant frequencies, coupled with a suitable Lagrangian estimate in the cut-off region. Details can be found in Proposition \ref{prop:LE} and the reader should also compare with the $\Koppa$-current in \cite{DRSR}. Finally, for technical reasons, we actually first define an interpolated $W^\sharp$ which is equal to $W_{\mathcal{I}^+}$ in $(\mathcal{G}^\prime)^c$
and equal to $W^\natural$ in $\mathcal{G} \setminus (\mathcal{G}^{\prime \prime})^c$ with $(\mathcal{G}^{\prime \prime})^c = \Big\{ (\omega, m) \ | \ m\omega \in (-\frac{3\xi_\ell}{4}m^2, -\frac{\xi_\ell}{2}m^2) \cup (\frac{am^2}{2Mr_+}+\frac{\xi_r}{2}m^2, \frac{am^2}{2Mr_+}+\frac{3\xi_r}{4}m^2 )\Big\}$ and then define the commuting operator to be $\chi_\sharp W^\sharp$. See Section \ref{sec:constructionsummary}. In summary, we obtain an estimate analogous to (\ref{fitt2}) for $\chi_\sharp W^\sharp$ instead of $W^\natural$, with $\chi_\sharp W^\sharp$ being defined a priori only for frequencies in $\mathcal{G}$.

\subsubsection{The convolution estimates} \label{intro:convolve}
Combining the results of the previous section, we can finally define the following global operator:
\begin{align}
W :=  \widetilde{\chi} W^\natural + (1-\widetilde{\chi}) \chi_\sharp W^\sharp \, ,
\end{align}
where $\widetilde{\chi}=\widetilde{\chi}\left(\frac{\omega}{m}\right)$ is a cut-off in frequency equal to $1$ in $\mathcal{G}^c$ and equal to zero in $(\mathcal{G}^{\prime \prime})^c$. Commuting with $W$ we produce an estimate analogous to (\ref{fitt2}) (now for $W$ instead of $W^\natural$) and our final task is 
to estimate the term involving the convolution in (\ref{fitt2}). Recall that we are going to integrate (\ref{fitt2}) over $\int_{-\infty}^\infty d\omega$ (and sum in $m$) and eventually transform the resulting integrals back to honest spacetime integrals using the Parseval identity. It is then clear from (\ref{fitt2}) that we will need to control (in $L^2_\omega \ell^2_m$) an error term of size $\varepsilon$ of the form (note $\chi_{\mathcal{T}}$ is axisymmetric, hence $\widehat{\chi_{\mathcal{T}}}$ is supported on $m=0$ only)
\begin{align} \label{errorconvolution}
 b(r,\theta)  W (\widehat{\chi_{\mathcal{T}}} * i \omega  \widehat{u})_{m} =  b(r,\theta)   (\widehat{\chi_{\mathcal{T}}} * i \omega  W \widehat{u}_{m} )+ (b(r,\theta)   W \widehat{\chi_{\mathcal{T}}} * i \omega   \widehat{u}_m) + \textrm{error} \, ,
\end{align}
where the error can be interpreted as a measure of the failure of the ``product rule” to hold. Note in particular that if $W$ arose from a physical space multiplier, then the error in (\ref{errorconvolution}) would vanish as $W$ would be linear in $\omega$ and $m$.\footnote{This is actually the case when $m=0$ (axisymmetry). Here superradiance is not present (hence $W^\sharp=0$) and $W^\natural$ can be chosen to be the global, physical space vector field $W_0$.}
The first term on the right hand side of (\ref{errorconvolution}) can be absorbed on the left in the estimate (\ref{fitt2}) and the second term incorporated in the first term on the right, as it is lower order in derivatives. This follows directly from Parseval’s identity (or by Young’s inequality noting that $\widehat{\chi_{\mathcal{T}}}$ and $W \widehat{\chi_{\mathcal{T}}}$ are bounded in $L^1_\omega$ uniformly in ${\mathcal{T}}$). This leaves analysing the error in (\ref{errorconvolution}), which is carried out in detail in Section \ref{sec:CMConvolution}. A weak Coifman-Meyer-type estimate allows us to bound the error in \eqref{errorconvolution} in a suitable weighted $L^2$ norm by energies which contain only lower order (i.e.~up to first) derivatives of $\ph$. In order for this error to have sufficient decay in space, we must tailor our estimate to take advantage of the notion that $\Wsc$ approaches the physical space operator $\oW$ (see (\ref{W0def})) at spatial infinity. The required bounds are stated already in Section \ref{sec:proofofconvolution} (see Proposition \ref{lem:CMBound}) and proven in Section \ref{sec:CMConvolution} for a larger class of products than required here.

\subsection{Related work and final comments}
We note the paper \cite{Mavro1}, which generalises the commutator vector field $W_S$ introduced in \cite{HK20} in the Schwarzschild case to the setting of the Schwarzschild-de Sitter black hole. In the latter case, because of the compactness of the spacelike slices, closing the resulting energy estimates at the level of the $W_S$-commuted energy allows one to prove  \emph{exponential decay} for the $W_S$-commuted energy without loss of derivatives. As a corollary, \cite{Mavro1} establishes exponential decay of the energy (with loss of derivatives) by physical space methods. Previously, exponential decay had been proven through spectral techniques \cite{haefner, dyatlov2011quasi, vasy}, while physical space methods had established decay ``faster than any polynomial” \cite{DRdS}. For quasi-linear applications of the commutator vector field in the context of a positive cosmological constant, see \cite{Mavro2}. Finally, it would be interesting to generalise the result presented here to the Kerr-de Sitter case.

While of independent interest, we were mainly led to consider the problem (\ref{weintro}) by studying
 generalisations from Schwarzschild to Kerr of the physical space Chandrasekhar transformation theory for the Teukolsky equation \cite{Teuko, Chandraschw, DHR, DHRteuk, Ma}, which appears in the context of the black hole stability problem. The Teukolsky equation on Kerr can be viewed as the covariant wave equation with additional (not necessarily small) first order terms. These first order terms can be removed by applying suitable differential operators  to the Teukolsky equation thereby transforming it into so-called Regge-Wheeler type equations. One interpretation of the results of the present paper is that this transformation theory is sufficiently robust in that one can actually allow for (small) first order terms to appear in the Regge-Wheeler type equations and still prove essentially the same decay results for the solution. In a similar direction, and more interestingly perhaps, the techniques introduced here may suggest a new approach to study the stability of slowly rotating Kerr black holes in harmonic gauge, see \cite{HaefnerHintzVasy}. As it is known that a form of decoupling happens in the Schwarzschild case \cite{Johnson18},  the generalisation to $|a| \ll M$ may lead to equations modelled by the toy problem studied here.
 
We finally note that our techniques should also be applicable to a large class of non-linear problems. A natural starting point would be to prove small data global existence for semi-linear toy problems of the form
\[
\Box_{g_{a,M}} \Psi  = \varepsilon \BfB( \Psi) + \mathcal{N}(\partial \psi)  \ \ \textrm{with data} \ \ \ \Psi|_{\Sigma} = \Psi_0 \ \ \ , \ \ \ n_{\Sigma} \Psi |_{\Sigma}= \Psi_1
\]
where $\varepsilon<\varepsilon_0$ (the $\varepsilon_0$ being that of Theorem \ref{thm:Main}) and $\mathcal{N}(\partial \Psi)$ is an at least quadratic non-linearity satisfying Klainerman’s classical null condition \cite{KlNull}. See \cite{LukKerr} for the slowly rotating case and $\varepsilon=0$. We note that -- at least in the Schwarzschild case -- using the vector field $W$ may lead to some immediate simplifications of the known proofs for $\varepsilon=0$: Since $W$ behaves like $r \partial_v$ near infinity and like $\frac{1}{\sqrt{1-\frac{2M}{r}}} \partial_u$ near the horizon it already incorporates a version of the redshift and of the $r^p$-weighted multiplier in one global vector field, which may help to reduce the number of commutator vector fields required to close the non-linear estimates.

 \subsection{Structure of the paper}
 In Section \ref{sec:prelim} we recall the Kerr metric, the hyperboloidal foliation of the exterior and define some of the vector fields to be used in the sequel. We also provide the reduction to future integrable solutions. In Section \ref{sec:choices}, a number of parameters are fixed (which if chosen in the right order all depend on $M$, $a$ only) that are important later in order to establish coercivity properties of appropriate multipliers and commutators. In Section \ref{sec:comc}, we construct the microlocal  commutator using the parameters fixed in the previous section and establish some of its main properties. Section \ref{sec:CEE} contains the proof of the main theorem following the outline given above up to proving a certain error estimate stated as Theorem \ref{thm:RHSBoundMain} whose proof has been moved to its own section, Section \ref{sec:proofofconvolution}. The error estimate in turn relies on two (more general) convolution estimates stated in Section \ref{sec:convo} whose proof has again been moved to its own section, Section \ref{sec:CMConvolution}. We believe that this structure makes the paper more readable. Finally, for completeness we have included two appendices, one containing some elementary physical space estimates (including a proof of Proposition \ref{prop:ps} above) and the other containing a mode stability result that is invoked for bounded frequencies in the proof of the main theorem.

\subsection{Acknowledgements}
G.H.~acknowledges support by the Alexander von Humboldt Foundation in the framework of the Alexander von Humboldt Professorship endowed by the Federal Ministry of Education and Research. Both authors acknowledge funding through ERC Consolidator Grant 772249 as well as Germany’s Excellence Strategy EXC 2044 390685587, Mathematics M\"unster: Dynamics–Geometry–Structure.

\section{Preliminaries} \label{sec:prelim}
		\subsection{The Kerr spacetime and the vector field $W_0$}\label{sec:PhysicalSpaceDefs}
For real numbers $a, M$ satisfying $M > 0$, $|a| < M$, we consider the domain of outer communications $\mathcal{R}$ of the Kerr metric $g_{a, M}$. In Boyer-Lindquist coordinates $\left(t,r,\theta,\phi\right)$ we have $\mathcal{R}=(-\infty,\infty)_t \times (r_+, \infty)_r \times \mathbb{S}^2_{\theta,\phi}$ with
\begin{equation} \label{Kerrmetric}
g_{a, M} = -\frac{\Delta}{\rho^2}\big(dt - a\sin^2\theta d\phi\big)^2 + \frac{\rho^2}{\Delta}dr^2 + \rho^2 d\theta^2 + \frac{\sin^2\theta}{\rho^2}\big(a\, dt - (r^2+a^2)d\phi)\big)^2,
\end{equation}
where $r_+ = M + \sqrt{M^2-a^2}$ and
\begin{equation}
\Delta = r^2-2Mr+a^2, \qquad \rho^2 = r^2+a^2\cos^2\theta.
\end{equation}
Using the invariance of (\ref{Kerrmetric}) under $\phi \rightarrow -\phi$ and $a \rightarrow -a$, we will assume $a>0$ without loss of generality.
 \begin{figure}
  \begin{center}
\begin{tikzpicture}[scale=0.7,
  mydot/.style={
    circle,
    fill=white,
    draw,
    outer sep=0pt,
    inner sep=1.5pt
  }
]
\filldraw[fill=lightgray, draw=black, dashed] (6,0)--node[midway, above, sloped] {$\mathcal{I}^+$}(3,3)--(0,0)--(3,-3)--node[midway, below, sloped] {$\mathcal{I}^-$}(6,0);
\draw (3,1) node[below]{$\widetilde{\Sigma}_{\widetilde\tau_2}$} arc (270:315:2) (3,1) arc (270:225:2);
\draw (3,0) node[below]{$\widetilde{\Sigma}_{\widetilde\tau_1}$} arc (270:315:3) (3,0) arc (270:225:3);
\draw (0,0) -- node[midway,above, sloped] {$\mathcal{H}^+$} (3,3) (0,0) --node[midway,below, sloped] {$\mathcal{H}^-$} (3,-3);
\draw [solid] (6,0) node[mydot] {}node[right]{$i^0$} (3,3)node[mydot]{}node[above]{$i^+$}(3,-3)node[mydot]{}node[below]{$i^-$};

\end{tikzpicture}
\end{center}
\caption{The Penrose diagram for $\mathcal{R}$}
\label{fig:Penrose}
\end{figure}
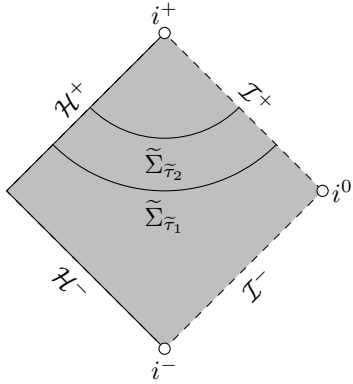


As indicated in the Penrose diagram, one may attach a null boundary $\mathcal{H}^+ \cup \mathcal{H}^-$ (known as the future and past event horizon respectively) to $\mathcal{R}$, to which the Kerr metric extends smoothly. See for instance \cite{DafRodsmalla} for a detailed construction starting from the manifold with boundary in regular coordinates. 

We next define the quantities 
\begin{align} \label{vdef}
\rweight = (r^2+a^2)^{1/2}\index{ra@$\rweight$}, \qquad v = \frac{\Delta^{1/2}}{r^2+a^2}\index{v@$v$}.
\end{align}
We define $v_0(a, M), r_0(a, M)$ such that\footnote{One easily computes that $r_0(a,M)$ is the unique real zero of the cubic $r^3-3Mr^2+ra^2+Ma^2$ in the interval $[r_+,\infty)$ and that the supremum of $v$ is attained in $(r_+,\infty)$ for $|a|<M$.}
\begin{equation}
v_0(a, M) = \sup_{r\in(r_+, \infty)}v\index{v0@$v_0$}, \qquad v  = v_0(a, M) \text{ at } r = r_0(a, M) \, .\index{r0@$r_0$}
\end{equation}
We will in general suppress the dependence of $v_0$, $r_0$ on $a, M$. We also define $h_0$ to be the increasing function satisfying
\begin{equation}\label{def:h0}
h_0^2[a, M](r) = 1-v_0^{-2}v^2, \index{h0@$h_0$}
\end{equation}
and consequently
\begin{equation}\label{def:f0}
f_0 = \frac{h_0}{v}\index{f0@$f_0$} \, .
\end{equation}
The fact that $h_0$ is increasing follows from the fact that $h_0$ vanishes at $r_0$ and that $\partial_r(v^2)$ has only one root in $(r_+, \infty)$. 
\begin{remark}
In the Schwarzschild spacetime, one may directly calculate
\[
r_0(0, M) = 3M, \qquad v_0(0, M) = (27M^2)^{-1/2}, \qquad h_0[0, M](r) = \big(1-\tfrac{3M}{r}\big)\sqrt{1+\tfrac{6M}{r}}.
\]
\end{remark}
We additionally define the radial tortoise coordinate $\rs$, and the corresponding vector field $\Rs$, such that
\begin{equation}
\Rs\index{R@$\Rs$} = \frac{\Delta}{r^2+a^2}\partial_r, \qquad \rs(r_0)\index{r@$\rs$} = 0, \qquad \Rs(\rs) = 1. \label{rrstar}
\end{equation}
We note that $\Rs = \partial_{\rs}$ in the $(t, \rs, \theta, \phi)$ coordinates. We will typically denote $R^\star f = f^\prime$ on functions. This gives a bijective map between $r\in (r_+, \infty)$ and $\rs \in(-\infty, \infty)$. We recall that the vector field $T:=\partial_t$ and $\Phi:=\partial_\phi$ are Killing for (\ref{Kerrmetric}). We recall also the Hawking field $K=T + \frac{a}{r_+^2+a^2} \Phi$ and define the \emph{modified Hawking field}
\begin{equation}\label{def:HField}
\Ks\index{Ks@$\Ks$} = T + \frac{a}{r^2+a^2}\Phi \, , 
\end{equation}
which is null at the horizon and timelike in the exterior (but not Killing). Finally, we define the vector field
\begin{equation} \label{W0def}
\oW\index{Wbar@$\oW$} = v^{-1}\Rs + v^{-1}h_0 \Ks \, ,
\end{equation}
which reduces to (\ref{Wsdef}) in the Schwarzschild case $a=0$. For future reference we will also define the (algebraically special) null vector fields
\begin{align} \label{def:LLbar}
L=T+R^\star+\frac{a}{r^2+a^2}\Phi \textrm{ \ \ \ and \ \ \ } \underline{L}=T-R^\star+\frac{a}{r^2+a^2}\Phi \, .
\end{align} 

\subsection{The hyperboloidal foliation} \label{sec:hyperboloidal}
We define the hyperboloidal foliation variable $\wtau = \wtau(t, r)$ such that
\begin{equation}\label{def:timefoliation}
\wtau(t, r_0)\index{tautilde@$\wtau$} = t, \qquad \Rs\wtau = -h_0, \qquad\partial_\theta\wtau = \partial_\phi\wtau = 0.
\end{equation}
We remark that $\wtau$ is smooth, and that surfaces of constant $\wtau$ connect the future event horizon and future null infinity (as indicated in Figure \ref{fig:Penrose}).
Thus, $\wSigtau$ determines a hyperboloidal foliation of hypersurfaces. We may extend this to a regular coordinate system $(\wtau, r, \theta, \widetilde{\phi})$ by defining $\widetilde{\phi}$ satisfying
\begin{equation}
\widetilde{\phi} = \phi \text{ for } r = r_0, \qquad \widetilde{\phi} = \phi + \overline{\phi}(r) \ \ \textrm{mod $2\pi$}, \qquad \Rs\overline{\phi} = -\frac{a h_0}{r^2+a^2}.
\end{equation}
In the $(\wtau, r, \theta, \widetilde{\phi})$ coordinates, one has the derivative
\begin{equation}\label{def:wpar}
\widetilde{\pa}_{r} = \partial_r + h_0\left(\frac{r^2+a^2}{\Delta}\right)\Ks = \Delta^{-1/2}\oW,
\end{equation}
as well as
\begin{equation}
\widetilde{\pa}_{\wtau} = \partial_t, \qquad \widetilde{\pa}_\theta = \partial_\theta, \qquad \widetilde{\pa}_{\widetilde{\phi}} = \partial_\phi.
\end{equation}
Then,
\begin{align}
g(\widetilde{\pa}_{r}, \widetilde{\pa}_{r}) = (1-h_0^2)\left(\tfrac{\rho^2}{\Delta}\right),  \ \ \ \
g(\widetilde{\pa}_{r}, \widetilde{\pa}_{\wtau}) = -h_0, \ \ \ \
g(\widetilde{\pa}_{r}, \widetilde{\pa}_{\widetilde{\phi}}) = ah_0\sin^2\theta.
\end{align}
All other components of $g$ in these coordinates are equal to their analogous quantities in Boyer-Lindquist coordinates. It follows from global bounds on $h_0$ as well as the definition \eqref{def:h0} that $g_{a, M}$ is regular with respect to the $(\wtau, r,\theta,  \widetilde{\phi})$ coordinates. In particular, in these coordinates the metric extends to the set $\mathcal{H}^+=(-\infty,\infty) \times \{r=r_+\} \times \mathbb{S}^2_{{\theta}, \widetilde{\phi}}$, which is part of the aforementioned null boundary of $\mathcal{R}$. See Appendix \ref{sec:FoliationCalculations} for further computations (volume elements, normal of $\widetilde{\Sigma}_{\tilde{\tau}}$, etc.) in these coordinates.

Finally, for future reference, we denote by $\slashed{g}:=\slashed{g}({\wtau},r)$ and $\slashed{\nabla}$ the induced metric and covariant derivative on the spheres $\mathbb{S}^2_{{\theta}, \widetilde{\phi}}$ of constant $({\wtau},r)$ in $\mathcal{R}$.

\subsection{Reduction to future integrable outgoing solutions} 

To establish the estimate (\ref{pse}) of Theorem \ref{thm:Main} we will make two simplifying assumptions, which are easily removed once (\ref{pse}) has been proven under these additional assumptions. See the very end of Section \ref{sec:MainTheoremProof}. 
\begin{align}
\textrm{(Auxiliary Assumption 1): The initial data for (\ref{weintro}) are smooth and compactly supported.} \label{aux1} \\
\textrm{(Auxiliary Assumption 2):  The operator $\BfB$ in (\ref{weintro}) is uniformly compactly supported in $r$.} \label{aux2}
\end{align}

Assuming now (\ref{aux1}) and (\ref{aux2}), we first construct a modified wave equation whose solutions are a priori future integrable in the sense of \cite{DRSR} allowing us to work with the Fourier transform and separation of variables. Estimates on solutions to the modified equation are easily converted to solutions to the original equation. 

\subsubsection{Constructing the modified wave equation} \label{sec:comof}
%
%
We begin by defining a smooth nondecreasing function
\begin{equation}\label{def:chiOneD}
\chi(y) = \begin{cases}
0 & y \leq 0, \\
1 & y \geq 1.
\end{cases}
\end{equation}
%
For given $\mathcal{T}\index{Taucal@$\mathcal{T}$} >2$ we also define the large-time cutoff
\begin{equation} \label{defchiT}
\chi_{{\mathcal{T}}}\index{chiT@$\chi_{\mathcal{T}}$} ({\wtau})=\chi\big(\tfrac{\wtau}{{\mathcal{T}}} + 1\big)\chi\big(2-\tfrac{\wtau}{{\mathcal{T}}}\big),
\end{equation}
so that in particular $\chi_{\mathcal{T}} = 1$ on $[0, {\mathcal{T}}]$, and $\chi_{\mathcal{T}}$ is compactly supported in $\wtau \in [-{\mathcal{T}},2{\mathcal{T}}]$ for all ${\mathcal{T}}$. We now define $\Psi_{\mathcal{T}}\index{psiT@$\Psi_{\mathcal{T}}$}$ to be the solution (defined in $\mathcal{R} \cap D^+(\widetilde{\Sigma}_0)$) of
\begin{equation}\label{def:PhysCutoffEq}
\Box_{g_{a, M}}\Psi_{\mathcal{T}} = \varepsilon\chi_{\mathcal{T}}\BfB \Psi_{\mathcal{T}}.
\end{equation}
with initial data $\Psi_{\mathcal{T}}|_{\widetilde\Sigma_0} = \Psi_0$ and $n_{\widetilde\Sigma_\tau} \Psi_\mathcal{T} |_{\widetilde\Sigma_0}= \Psi_1$ as in (\ref{weintro}). By domain of dependence, $\Psi_\mathcal{T}=\Psi$ for $\wtau \in [0,\mathcal{T}]$.
Finally, we define
\begin{equation}\label{def:Psichi}
 \Psi_{\chi , \mathcal{T}}\index{psichiT@$ \Psi_{\chi , \mathcal{T}}$}= \chi\Psi_\mathcal{T} \, ,
\end{equation}
which is well-defined and smooth on all of $\mathcal{R}$ since $\chi$ vanishes identically for $\wtau\leq 0$. Moreover, $ \Psi_{\chi , \mathcal{T}}$ satisfies
\begin{equation}\label{def:WaveEqCutoff}
\Box_{g_{a, M}}( \Psi_{\chi , \mathcal{T}}) =  \varepsilon\chi_T(\BfB \Psi_{\chi , \mathcal{T}}) + 2\nabla^\alpha \chi \nabla_\alpha \Psi_\mathcal{T} + (\Box_{g_{a, M}}\chi)\Psi_\mathcal{T} - \varepsilon \chi_\mathcal{T}\Psi_\mathcal{T} (\BfB \chi -\BfB^0 \chi).
\end{equation}
In summary, $ \Psi_{\chi , \mathcal{T}}$ is a smooth solution to the equation \eqref{def:WaveEqCutoff} on $\mathcal{R}$, which
\begin{itemize}
\item vanishes identically on $\mathcal{R} \cap \{ \wtau \leq 0\}$,
\item  is equal to $\Psi_\mathcal{T}$ and hence the original $\Psi$ for $\wtau \in [1,\mathcal{T}]$,
\item satisfies the free homogeneous wave equation for  $\wtau \geq 2\mathcal{T}$. 
\end{itemize}
 Note also that the last three terms on the right hand side of (\ref{def:WaveEqCutoff}) are supported only for $\wtau \in [0,1]$.  \\ 
 
 By the properties of $ \Psi_{\chi , \mathcal{T}}$ is clear that to prove Theorem \ref{thm:Main} it is sufficient to establish (\ref{pse}) for $\Psi_{\chi,\mathcal{T}}$ with constants independent of $\mathcal{T}$ and this is what we will do.

\subsubsection{Establishing future integrability}

The following estimates for $ \Psi_{\chi , \mathcal{T}}$ are standard (since the constant $C_{\mathsf{T}, \BfB}$ below is allowed to depend on $\mathsf{T}$) and follow from the properties of $\BfB$ and Gronwall's inequality. See Appendix \ref{sec:FoliationCalculations} for the necessary background on doing non-degenerate energy estimates for the slices $\tilde{\Sigma}_\tau$. 

\begin{proposition} \label{prop:basicgronwall}
Recall from (\ref{basicdataenergy}) the energy $E_\tau[\Psi]$ on constant $\tau=\wtau$ slices for $\wtau \geq 0$ and consider (\ref{weintro}) with the assumptions (\ref{aux1}) and (\ref{aux2}).
Then, solutions to (\ref{def:WaveEqCutoff}) satisfy for all $\mathsf{T}>0$ the estimate
\begin{align} \label{basiclocal}
E_\mathsf{T} [\Psi_{\chi,\mathcal{T}}] \leq C_{\mathsf{T}, \BfB} E_0  [\Psi_\mathcal{T}]  \, .
\end{align}
We also have for any $N\geq 0$ the higher order estimate
\begin{align} \label{higherlocal}
 E_\mathsf{T} [\mathfrak{D}^N \Psi_{\chi,\mathcal{T}}] \leq C_{\mathsf{T},\BfB, N}  E_0 [\mathfrak{D}^N \Psi_\mathcal{T}] \, ,
\end{align}
where $E_\tau [\mathfrak{D}^n \psi]$ is defined as 
\begin{align}
E_\tau [\mathfrak{D}^n \psi] =\sum_{0 \leq i_1+i_2+i_3\leq n} E_\tau \left[ (\slashed{\nabla})^{i_1} T^{i_2} (\widetilde{\partial}_r)^{i_3} \psi\right] \, .
\end{align}
\end{proposition}
Note that in view of $\Psi_\mathcal{T}|_{\widetilde\Sigma_0} = \Psi_0$ and $n_{\widetilde\Sigma_\tau} \Psi_\mathcal{T} |_{\widetilde\Sigma_0}= \Psi_1$ we have $E_0 [\mathfrak{D}^N \Psi_\mathcal{T}]=E_0 [\mathfrak{D}^N \Psi]$ for all $N\geq 0$ and hence the right hand sides of (\ref{basiclocal}) and (\ref{higherlocal}) are expressible purely in terms of the data $\Psi_0$ and $\Psi_1$ appearing in Theorem \ref{thm:Main}.
%
%
From Proposition \ref{prop:basicgronwall} and the fact that solutions to  the homogeneous wave equation have been shown to be square integrable in time for any fixed $r$ we conclude the following statement whose first part will in particular justify taking the Fourier transform in time and whose second part will allow us to express appropriate boundary conditions at null infinity at the microlocal level.

\begin{proposition} \label{prop:suffi}
Assume $\Psi_{\chi,\mathcal{T}}$ arises from a $\Psi$ satisfying (\ref{weintro}) with (\ref{aux1}) and (\ref{aux2}). Then, for any fixed $\mathcal{T}>2$, the solution to (\ref{def:WaveEqCutoff}) is square integrable in Boyer-Lindquist time  in that (with $d\sigma=\sin \theta d\theta d\tilde{\phi}$)
\begin{align} \label{sqi}
 \sup_{r \in [r_+,A]} \int_{-\infty}^\infty dt \int_{\mathbb{S}^2} d\sigma \sum_{0 \leq i_1+i_2+i_3 \leq N} \Big| (\slashed{\nabla})^{i_1} (T)^{i_2} \left(\widetilde{\partial}_r \right)^{i_3} \Psi_{\chi,\mathcal{T}}\Big|^2_{\slashed{g}} < \infty
\end{align}
holds for any $A< \infty$ and any $N$. 
Moreover, the solution to (\ref{def:WaveEqCutoff}) satisfies for any fixed $\mathcal{T}>2$ the weighted bound
\begin{align} \label{sqi2}
\sup_{r \geq 10M} \int_{-\infty}^\infty dt \int_{\mathbb{S}^2} d\sigma  \left[ \sum_{X \in \{R^\star,T,id\}}  r^{\frac{1}{2}-\delta} \big| X W_0 (r\Psi_{\chi,\mathcal{T}})\big|^2  +\sum_{X,Y \in \{T, \Phi,id\}}  r^{-1-\delta} |XY (r\Psi_{\chi,\mathcal{T}})|^2\right] < \infty \, .
\end{align}
\end{proposition}

\begin{proof}
Recall that $\psi_{\chi,\mathcal{T}}$ is identically zero to the past of $\widetilde{\Sigma}_0$ and satisfies the \emph{homogeneous} wave equation for $\tau \geq 2\mathcal{T}$. Therefore, by Theorems 3.1 and 3.2 of \cite{DRSR} (and the estimates (\ref{basiclocal}) and (\ref{higherlocal}) for up to time $2\mathcal{T}$) we have for any $\mathsf{T}>0$ the boundedness estimate
\begin{align}
E_\mathsf{T} [\mathfrak{D}^N \Psi_{\chi,\mathcal{T}}] \leq C_{\mathcal{T},\BfB, N} E_0 [\mathfrak{D}^N \Psi_\mathcal{T}] 
\end{align}
and the integrated decay estimate (note the degenerating factor of $r^{-1}$)
\begin{align}
\int_0^\infty  d\tau  E_\mathcal{\tau} [r^{-1} \mathfrak{D}^{N-1} \Psi_{\chi,\mathcal{T}}] \leq C^\prime_{\mathcal{T},\BfB, N} E_0 [\mathfrak{D}^N \Psi_\mathcal{T}] \, .
\end{align}
Using the mean value theorem and the fundamental theorem of calculus it is straightforward to establish (\ref{sqi}), in fact a quantitative statement thereof. 

We sketch the argument for the weighted statement (\ref{sqi2}), which is fairly standard, see for instance Theorem 6.1 of \cite{M16}. We first note that in view of (\ref{def:PhysCutoffEq}) and (\ref{def:Psichi}) and the fact that $W_0 \tilde{\tau}=0$ it suffices to prove (\ref{sqi2}) for $\Psi_\mathcal{T}$ instead of $\Psi_{\chi,\mathcal{T}}$ and with the integrand restricted to $\tilde{\tau}\geq0$. To prove the latter, one applies the $r^p$ method to (\ref{def:PhysCutoffEq}) and establishes for any $\delta>0$ (and any $N$) the spacetime bound 
\begin{align} \label{rph}
 \int_{-\infty}^\infty dt \int_{10M}^\infty dr  \int_{\mathbb{S}^2} d\sigma \sum_{0 \leq i_1+i_2+i_3 \leq N} \mathbf{1}_{\tilde{\tau}\geq 0}  \cdot r^{-1-\delta} \Big| (r \slashed{\nabla})^{i_1} (r L)^{i_2} (\underline{L})^{i_3} (r\Psi_\mathcal{T})\Big|^2_{\slashed{g}} < C \, , 
\end{align}
where the $C$ depends on the (compactly supported on the slice $\widetilde{\Sigma}_0$) data and potentially also on the cut-off time $\mathcal{T}$.
Note here the fact that the right hand side of (\ref{def:PhysCutoffEq})  is supported for $0 \leq \wtau \leq 2\mathcal{T}$ only and that for this compact time interval it can be controlled using $2\mathcal{T}$ times the weighted boundary terms generated by the $r^p$-method and the standard energy on constant $\wtau$-slices. Now (\ref{rph}) holds for $N-1$ without the integration in $r$ by a standard dyadic argument, from which the bounds in the second sum in (\ref{sqi2}) follow directly. For the other bounds one writes $R^\star W_0 (r \Psi_{\mathcal{T}}) = -\frac{R^\star v}{v} W_0 (r \Psi_{\mathcal{T}}) + \frac{1}{2} v^{-1} (L - \underline{L})\left(L+(h_0-1)K^\star\right) (r \Psi_{\mathcal{T}})$ and inserts the wave equation (\ref{def:PhysCutoffEq}) on the right. Using $|h_0-1|\leq Cr^{-2}$ and the bounds (\ref{rph}) for the terms appearing on the right, the desired estimate follows for $ R^\star W_0 (r\Psi_{\mathcal{T}})$. The argument for $TW_0 (r\Psi_{\mathcal{T}})$ is similar. The improved bound for $W_0$ now follows by another dyadic argument using the bound for $R^\star W_0 (r\Psi_{\mathcal{T}})$.
\end{proof}


We finally note that being identically zero to the past of $\widetilde{\Sigma}_0$, the solution $\Psi_{\chi,\mathcal{T}}$ is also manifestly ``outgoing'' in the sense of \cite{DRSR}.

\subsection{Separating the (modified) wave equation}
 We recall that the Kerr metric (\ref{Kerrmetric}) admits two global Killing vectors, $T = \partial_t$ and $\Phi = \partial_\phi$ as well as a Killing tensor, which in combination permit full separation of the wave operator $\Box_{g_{a,M}}$. We would like to apply these facts to the (solution $\Psi_{\chi,\mathcal{T}}$ of) equation (\ref{def:WaveEqCutoff}). In Proposition \ref{prop:suffi} we established that the solution is in $L^2_\phi L^2_t$ for all $\mathcal{T}$ justifying the definition of the Fourier transform of $\Psi_{\chi,\mathcal{T}}$ in $t$ (and $\phi$), thereby separating the solution into Fourier variables $\widehat{\Psi_{\chi,\mathcal{T}}}(\omega,r,\theta,m)$. We provide the details in Section \ref{sec:sepKilling} and note that this (partial) separation is already sufficient to construct our microlocal commutator and to carry out the main part of the analysis. While not needed for the core of our analysis (but for interpolating with the results of \cite{DRSR} in the $\varepsilon=0$ case) we provide the details of the full separation of the wave equation in Section \ref{sep:KTensor}.
 
For notational convenience, we consider for the rest of this section a generic inhomogeneous wave equation
\begin{equation}\label{def:WaveKerrInhom}
\Box_{g_{a, M}}\psi = F \, ,
\end{equation}
where $F$ is also allowed to depend linearly on $\psi$, and $\psi$ is assumed square integrable in the sense of Proposition \ref{prop:suffi}. In Section \ref{sec:RJHSDecomp} we reduce to the actual case of interest, equation (\ref{def:WaveEqCutoff}).
\subsubsection{The Fourier transform in the Kerr spacetime}
 Given a function $\psi$ on $\mathcal{R}$ which is square integrable in the sense of (\ref{sqi}), and $(\omega, m)\in \mathbb{R}\times\mathbb{Z}$, we define the Fourier transform 
\begin{equation}
\widehat{\psi}(\rs, \theta) = \widehat{\psi}_m^\omega(\rs, \theta) = (\mathcal{F}\psi)(\omega, \rs, \theta, m)
\end{equation}
 satisfying
\begin{align}
(\mathcal{F}\psi)(\omega, r, \theta, m) &= (2\pi)^{-1}\int_{-\infty}^\infty \int_0^{2\pi}e^{i(\omega t  - m\phi)}\psi(t, \rs, \theta, \phi)\,d\omega, \\
\psi(t, r, \theta, \phi) &= (2\pi)^{-1}\sum_{m\in\mathbb{Z}}\int_{-\infty}^\infty e^{-i(\omega t  - m\phi)}(\mathcal{F}\psi)(\omega, \rs, \theta, m)\,d\omega.
\end{align}
This gives the identities
\begin{equation}
\widehat{\partial_t u} = -i\omega\ph, \qquad \widehat{\partial_\phi u} = im\ph.
\end{equation}
Additionally, for square integrable functions $\psi, \zeta$, and for all $\rs, \theta$, Plancherel's theorem implies
\begin{subequations}\label{id:Plancherel}
\begin{align}
\int_{-\infty}^\infty \int_0^{2\pi} \psi\overline{\zeta}\, d\phi \, dt &= \sum_m\int_{-\infty}^\infty \widehat{\psi}\overline{\widehat{\zeta}}\, d\omega, \label{id:Plancherel1}\\
\int_{-\infty}^\infty \int_0^{2\pi} |\psi|^2\, d\phi \, dt &= \sum_m\int_{-\infty}^\infty |\widehat{\psi}|^2\, d\omega.\label{id:Plancherel2}
\end{align}
\end{subequations}
Later on it will be useful to define the convolution in frequency
\begin{equation}\label{def:Convolution}\index{*@$*$}
(\widehat{f}*\widehat{g})(\omega, m, \rs, \theta) = \sum_{m'}\int_{-\infty}^\infty \widehat{f}(\omega', m', \rs, \theta) \widehat{g}(\omega-\omega', \rs, \theta, m-m')\, d\omega'.
\end{equation}
\subsubsection{The separation with respect to the Killing vectors} \label{sec:sepKilling}
If $\psi$ is a solution of \eqref{def:WaveKerrInhom}, we define
\begin{equation}\label{def:u}
u = \rweight\psi,
\end{equation}
and similarly
\begin{equation}
\ph(\omega, \rs, \theta, m) = \ph_{m}^{\omega}(\rs, \theta) = \rweight\widehat{\psi}_{m}^\omega \, .\label{def:ph}
\end{equation}
Then, defining
\begin{equation}
H = H_{m}^{\omega} = \frac{\Delta\rho^2 \widehat{F}_m^\omega}{\rweight^3},\label{def:Hpart}
\end{equation}
the equation \eqref{def:WaveKerrInhom} in frequency space takes the form
\begin{equation}\label{def:SepWaveOperator}
\Pam[\omega, m]\ph := \Rs\Rs\ph + (\omega-\omega_r)^2\ph - v^2\mathcal{L}_\theta\ph - V_1\ph = H_m^\omega,
\end{equation}
where
\begin{equation}
\omega_r = \frac{am}{r^2+a^2}, \qquad  V_1 = \frac{\Delta}{(r^2+a^2)^2}\frac{a^2\Delta + 2Mr(r^2-a^2)}{(r^2+a^2)^2},
\end{equation}
and $\Ltheta\ph$ is the modified angular operator
\begin{equation}\label{def:SphLaplacian}\index{Ltheta@$\Ltheta$}
\mathcal{L}_\theta\ph = \mathcal{L}_\theta[m, a\omega](\ph) = -\tfrac{1}{\sin\theta}\partial_\theta(\sin\theta\partial_\theta\ph) + \big(a\omega\sin\theta - \tfrac{m}{\sin\theta}\big)^2\ph.
\end{equation}
We may equivalently write
\begin{equation}\label{def:SepWaveOperator2}
\Box_{g_{a, M}}\psi = \frac{\rweight^3}{\Delta\rho^2}\mathcal{F}^{-1}\left(\Pam\left(\rweight\widehat{\psi}\right)\right) \, .
\end{equation}
For fixed $(\omega, m, \rs)$, we additionally define the restricted angular component of the energy
\begin{align}
|\widehat{\slashed\nabla}\ph|^2 &=|\partial_\theta\ph|^2 +\big(a\omega\sin\theta - \tfrac{m}{\sin\theta}\big)^2|\ph|^2, \label{def:angEnergyPointwise}\\ 
\mathcal{E}_\theta[\ph] &= \int_0^\pi |\widehat{\slashed\nabla}\ph|^2\sin\theta d\theta.\label{def:angEnergy}
\end{align}
Integrating by parts in $\theta$ gives the identity
\begin{equation}\label{def:angEnergy}
\int_0^\pi\mathcal{L}_\theta\ph \overline{\ph} \,\sin\theta\, d\theta = \mathcal{E}_\theta[\ph],
\end{equation}
so consequently all eigenvalues of $\mathcal{L}_\theta$ are non-negative.

Defining 
\begin{equation}
\omega_+ = \frac{am}{r_+^2+a^2} = \frac{am}{2Mr_+},
\end{equation}
we note that $(\omega-\omega_r)^2$ has a unique root in $(r_+, \infty)$ if and only if
\begin{equation}
0 < am\omega < am\omega_+.
\end{equation}

\subsubsection{Additional separation using the Killing tensor}\label{sep:KTensor}
When using results from previous works, in particular in the estimates on lower derivatives of Section \ref{sec:LOT}, we will make use of \emph{Carter's separation} which was developed in \cite{Ca68} and used to great effect by Whiting in \cite{Wh89} to establish mode stability on the Kerr spacetime. This uses the \emph{oblate spheroidal harmonics} $S_{m\ell}(a\omega, \cos\theta)e^{im\phi}$ satisfying the eigenvalue problem
\begin{equation}\label{def:spheroidalharmonics}\index{Sml@$S_{m\ell}$}
L(\nu)(S_{m\ell}(\nu, \cos\theta)e^{im\phi}) = \lambda_{m\ell}^{(\nu)}S_{m\ell}(\nu, \cos\theta)e^{im\phi},
\end{equation}
where
\begin{equation}
L(\nu) u := -\frac{1}{\sin\theta}\big(\sin\theta\partial_\theta u\big) - \frac{1}{\sin^2\theta}\partial_\phi^2 u - \nu^2\cos^2\theta u.
\end{equation}
It follows that
\begin{equation}
(\mathcal{L}_\theta S_{m\ell}(a\omega, \cos\theta))e^{im\phi} = (\lambda_{m\ell}^{(a\omega)} -2ma\omega + a^2\omega^2)(S_{m\ell}(a\omega, \cos\theta)e^{im\phi}).
\end{equation}
We additionally have the well-known inequalities
\begin{equation}\label{est:Lambdalowerbounds}
\Lambda := \lambda_{m\ell}^{(a\omega)}+ a^2\omega^2 \geq2|ma\omega|, \qquad \Lambda \geq m(m+1).
\end{equation}
The former follows from positivity of \eqref{def:angEnergy} and the latter follows from decomposition in $m$ along with the inequality $uL(\nu)u + \nu^2  \geq -u\Delta_\omega u \geq m(m+1)u^2$, where $\Delta_\omega$ is the spherical Laplacian on the round sphere.

In this decomposition \eqref{def:SepWaveOperator} becomes
\begin{equation}\label{eq:WaveGeneric}
(\Pam\ph)^{(a\omega)}_{m\ell} := \Rs\Rs\ph^{(a\omega)}_{m\ell} + (\omega-\omega_r)^2\ph^{(a\omega)}_{m\ell} - v^2(\Lambda - 2am\omega)\ph^{(a\omega)}_{m\ell} - V_1\ph^{(a\omega)}_{m\ell} = H^{(a\omega)}_{m\ell} \, .
\end{equation}
We will generally explicitly state dependence on $\ell$ when in the fully separated regime. 
\begin{remark}\label{rem:DRSR}
We note that in \cite{DRSR} and \cite{SR14} the decomposition \eqref{eq:WaveGeneric} is written as
\begin{equation}
(\Pam\ph)^{(a\omega)}_{m\ell} := \Rs\Rs\ph^{(a\omega)}_{m\ell} + \omega^2\ph^{(a\omega)}_{m\ell} - V_0\ph^{(a\omega)}_{m\ell} - V_1\ph^{(a\omega)}_{m\ell} = H^{(a\omega)}_{m\ell},
\end{equation}
where 
\begin{equation}\index{V0@$V_0$}
V_0 = \frac{\Delta}{(r^2+a^2)^2}(\Lambda - 2am\omega) + 2\omega\omega_r - \omega_r^2 = \frac{\Delta}{(r^2+a^2)^2}(\Lambda - 2am\omega) - (\omega-\omega_r)^2 + \omega^2.
\end{equation}
\end{remark}

\subsubsection{Decomposition of the right hand side}\label{sec:RJHSDecomp} 

We now turn our attention back to the specific equation \eqref{def:WaveEqCutoff}, first taking the separation in $(\omega, m)$.
Define
\begin{equation}\label{def:uchi}\index{uchi@$\uchi$}
\uchi = (r^2+a^2)^{1/2}\mathcal{F}(\Psi_{\chi,\mathcal{T}}), \text{ so } ({\uchi})_m^\omega = (r^2+a^2)^{1/2}\big(\mathcal{F}(\Psi_{\chi,\mathcal{T}})\big)_m^\omega.
\end{equation}
 Plugging the right hand side of \eqref{def:WaveEqCutoff} into \eqref{def:Hpart} gives
\begin{equation}\label{def:WavePerturbedFreq}
\Pam{\uchi} = H_\chi + H_{\BfB},
\end{equation}
with
\begin{subequations}\label{def:HDecomp}
\begin{align}
\Hchi & = \frac{\Delta\rho^2}{\rweight^3}\mathcal{F}\big(2\nabla^\alpha \chi \nabla_\alpha \Psi_T+ (\Box_{g_{a, M}}\chi)\Psi_\mathcal{T} - \chi_\mathcal{T}\Psi_\mathcal{T} \varepsilon \left(\BfB \chi-\BfB^0 \chi\right)\big), \label{def:HDecompchi}\index{HB@$\HB$}\\
\HB &= \frac{\Delta\rho^2}{\rweight^4}\varepsilon(\widehat\chi_\mathcal{T} *(\mathcal{F}(\BfB(\mathbf{u}_\chi)-\rweight^{-1}\Rs\rweight\BfB^r\mathbf{u}_\chi))) =  \frac{\Delta\rho^2}{\rweight^4}\varepsilon\mathcal{F}(\BT(\mathbf{u}_\chi)) \label{def:HDecompB}\index{Hchi@$\Hchi$} ,
\end{align}
\end{subequations}
where ${\bf u}_{\chi}=\sqrt{r^2+a^2} \Psi_{\chi,\mathcal{T}}$ and we have introduced the operator $\BfB_\mathcal{T}$ defined by
\begin{equation}\label{def:BT}\index{BT@$\BT$}
\BT = \chi_\mathcal{T}\BfB -\chi_\mathcal{T}\frac{\Rs\rweight}{\rweight}\BfB^r .
\end{equation}


\subsection{Integration conventions}\label{sec:IntConventions}
As we will be integrating often in frequency space, we briefly detail some of the shorthand we will be using. Given frequency coefficients $(\omega, m)$, we define for $r_+ \leq r_1<r_2\leq \infty$ the restricted annular regions
\begin{subequations}
\begin{align}
\mcA_m^\omega (r_1,r_2) \index{A@$\mcA$}&= \{\omega\}\times (r_1, r_2) \times [0, \pi] \times \{m\} \subset \mathbb{R} \times (r_+, \infty) \times [0, \pi] \times \mathbb{Z}, \\
(\mcAs)_m^\omega  (r^\star_1,r^\star_2)  \index{Astar@$\mcAs$}&=\{\omega\}\times(r^\star_1,r^\star_2)\times [0, \pi] \times \{m\} \subset \mathbb{R}^2 \times [0, \pi] \times \mathbb{Z}.
\end{align}
\end{subequations}
We make the convention $\mcA_m^\omega= \mcA_m^\omega (r_+,\infty)$ and $(\mcAs)_m^\omega = (\mcAs)_m^\omega (-\infty, \infty)$ and we will also often suppress the dependence of $\mcA, \mcAs$ on $(\omega, m)$. For instance, we parametrise the integrals
\begin{subequations}
\begin{align}
\iint_{\mcA(r_1,r_2)}\ph = \int_0^\pi \int_{r_1}^{r_2} \ph(\omega, r, \theta, m) \sin\theta \, dr \, d\theta  \ \ \textrm{and} \ \ 
\iint_{\mcAs}\widehat{w} = \int_0^\pi \int_{-\infty}^\infty\widehat{w}(\omega, \rs, \theta, m)\sin\theta\, d\rs\, d\theta.
\end{align}
\end{subequations}
Note that standard substitution implies the change of variables (recall (\ref{rrstar}))
\begin{equation}
\iint_{\mcA(r_1,r_2)}\frac{r^2+a^2}{\Delta}\ph(\omega, \rs(r), \theta, m)\,  = \iint_{\mcAs (r_1^\star,r_2^\star)}\ph(\omega, \rs, \theta, m).
\end{equation}
We also define the weighted frequency integral over (subsets of) the exterior region of the Kerr spacetime,
\begin{subequations}
\begin{align}
\iiiint_{\FM (r_1,r_2)}\ph \index{FM@$\FM$}&= \sum_{m \in\mathbb{Z}}\int_{-\infty}^\infty\int_0^\pi \int_{r_1}^{r_2} \ph(\omega, r, \theta, m) \sin\theta \, dr \, d\theta\, d\omega, \\
\iiiint_{\FMs (r_1^\star,r_2^\star)}\widehat{w}\index{FMstar@$\FMs$} &= \sum_{m \in\mathbb{Z}}\int_{-\infty}^\infty\int_0^\pi \int_{r_1^\star}^{r_2^\star} \widehat{w}(\omega, \rs, \theta, m) \sin\theta \, d\rs \, d\theta\, d\omega.
\end{align}
\end{subequations}
as well as the weighted spacetime integrals in this region,
\begin{subequations}
\begin{align}
\iiiint_{\WM(r_1,r_2) }\ph  \index{WM@$\WM$}&= \int_{-\infty}^\infty\int_{r_1}^{r_2}\int_0^{2\pi}\int_0^\pi  \ph(\omega, r, \theta, m) \sin\theta \, d\theta\, d\phi\, dr\, dt, \\
\iiiint_{\WMs(r_1^\star,r_2^\star)}\widehat{w}  \index{WMstar@$\WMs$}&= \int_{-\infty}^\infty\int_{r_1^\star}^{r_2^\star} \int_0^{2\pi}\int_0^\pi  \widehat{w}(\omega, \rs, \theta, m) \sin\theta \, d\theta\, d\phi\, d\rs\, dt.
\end{align}
\end{subequations}
We will use $\FM=\FM (r_+,\infty)$, $\FMs =\FMs (-\infty, \infty)$, $\WM=\WM(r_+,\infty)$, $\WMs =\WMs(-\infty, \infty)$. 
\begin{remark}
We do not include the $\rho^2$ which appears in the volume form, as from a practical perspective it will generally appear implicitly in our integrand.
\end{remark}
The Plancherel identities \eqref{id:Plancherel1}, \eqref{id:Plancherel2} imply that, for any (square integrable in the sense of (\ref{sqi})) functions $\psi, \zeta$,
\begin{subequations}
\begin{alignat}{4}
\iiiint_{\FM(r_1,r_2)}\widehat\psi\overline{\widehat\zeta} &= \iiiint_{\WM(r_1,r_2)}\psi\overline{\zeta}, \qquad& \iiiint_{\FMs(r_1^\star,r_2^\star)}\widehat\psi\overline{\widehat\zeta} &= \iiiint_{\WMs(r_1^\star,r_2^\star)}\psi\overline{\zeta}\label{id:Plancherel3} \, , \\
\iiiint_{\FM(r_1,r_2)}|\widehat\psi|^2 &= \iiiint_{\WM(r_1,r_2)}|\psi|^2, \qquad& \iiiint_{\FMs(r_1^\star,r_2^\star)}|\widehat\psi|^2 &= \iiiint_{\WMs(r_1^\star,r_2^\star)}|\psi|^2.\label{id:Plancherel4}
\end{alignat}
\end{subequations}

\subsection{The microlocal currents and energies}
In the $(\omega, m)$ decomposition, we define the \emph{energy current} and \emph{angular energy current}
\begin{equation}\label{def:EnergyCurrent}
Q^T[\ph] = \mathfrak{I}(\ph'\overline{\omega\ph}), \qquad A^T[\ph] = v^2\mathfrak{I}\left((\partial_\theta\ph) \overline{\omega\ph}\right) \, ,
\end{equation}
as well as the \emph{Lagrangian current} and \emph{angular Lagrangian current} for a given (real) function $g(\rs)$
\begin{equation}\label{def:LagrangianCurrent}
Q^{g}[\ph] = g\mathfrak{R}(\ph'\overline{\ph}) - \tfrac12 g'|\ph|^2, \qquad A^g[\ph] = gv^2\mathfrak{R}(\partial_\theta\ph \overline{\ph}) \, .
\end{equation}
With these definitions we have (recalling the separated wave operator (\ref{def:SepWaveOperator})) the identities
\begin{align}
\label{id:StandardMLDiv}\Rs Q^T[\ph] + \frac{\partial_\theta(\sin\theta A^T[\ph])}{\sin\theta} &= \mathfrak{I}(\Pam\ph\overline{\omega\ph}) \, , \\
\Rs Q^g[\ph] + \frac{\partial_\theta(\sin\theta A^g[\ph]}{\sin\theta}) &= g\left(- |(\omega-\omega_r)\ph|^2 + v^2|\widehat{\slashed\nabla}\ph|^2 + |\ph'|^2 + V_1|\ph|^2 + \mathfrak{R}(\Pam\ph\overline{\ph})\right) - \tfrac12 g'' |\ph|^2.\label{Div:Lagrangian}
\end{align}
The fundamental theorem of calculus implies that for sufficiently regular $\ph$,
\begin{align} \label{microlocalE}
\iint_{\mcAs(r_1^\star,r_2^\star)}\mathfrak{I}(\Pam\ph\overline{\omega\ph}) = \int_0^\pi \left(\mathfrak{I}(\ph'\overline{\omega\ph})(r_2^\star, \theta) - \mathfrak{I}(\ph'\overline{\omega\ph})(r_1^\star, \theta)\right) \sin\theta\, d\theta \, , 
\end{align}
\begin{align} \label{microlocalg}
\iint_{\mcAs(r_1^\star,r_2^\star)} g\left(- |(\omega-\omega_r)\ph|^2 + v^2|\widehat{\slashed\nabla}\ph|^2 + |\ph'|^2 + V_1|\ph|^2 \right) - \tfrac12 g'' |\ph|^2  =   -\iint_{\mcAs(r_1^\star,r_2^\star)} g \mathfrak{R}(\Pam\ph\overline{\ph}) \nonumber \\
+ \int_0^\pi \left(
\left(g\mathfrak{R}(\ph'\overline{\ph}) - \tfrac12 g'|\ph|^2\right)(r_2^\star, \theta) - \left(g\mathfrak{R}(\ph'\overline{\ph}) - \tfrac12 g'|\ph|^2\right)(r_1^\star, \theta)\right) \sin\theta\, d\theta \, .
\end{align}

\subsection{A word about constants} \label{remark:CDependence}
Unless stated otherwise, all constants $c$, $C$ will depend only on the Kerr parameters $(M, a)$ and the notation $\lesssim$ will be used to express inequalities holding up to such multiplicative constants. In Section \ref{sec:choices} we will determine seven parameters ($\xi_\ell$, $\xi_r$, $\rmax$, $\rmin$, $c_\mathcal{G}$, $c_{\mathcal{G}^\prime}$, $\delta^\sharp$), which when chosen in the right order eventually also depend only on $(M,a)$. After Section \ref{sec:choices}, all these parameters are fixed for the rest of the paper. For the reader's convenience and potential future generalisations of our results we sometimes explicitly remark the dependence of an estimate on the parameters in the later sections.



\section{The frequency ranges $\mathcal{G}$ and $\mathcal{G}'$ and basic coercivity estimates} \label{sec:choices}
		In this section we decompose the frequency space and establish various coercivity estimates for expressions that will appear in the microlocal energy estimates later. As mentioned before, we can and will assume without loss of generality that $a > 0$. To motivate the analysis we recall a key observation of \cite{DRSR}, namely the notion that, for the separated wave equation, trapping and superradiance cause obstructions in disjoint frequency regimes. 
With this in mind, we define the frequency regimes
\begin{align}
\mathcal{G}  &:=\{(\omega, m) | m\omega \in (-\ximin m^2, m\omega_+ + \ximax m^2) \} \label{def:GCalDef}\index{G@$\mathcal{G}, \mathcal{G}'$},\\
\mathcal{G}' &:= \{(\omega, m) | m\omega \notin (-\tfrac12\ximin m^2, m\omega_+ + \tfrac12\ximax m^2) \}\label{def:GCalPrimeDef},
\end{align}
for some small constants $\ximin, \ximax$ which will be fixed in this section (and we recall $\omega_+=\frac{am}{r_+^2+a^2}$). Here we explicitly include the case $m=0$ in $\mathcal{G}'$.
It follows from the analysis of \cite{DRSR} that for sufficiently small $\xi_\ell$, $\xi_r$ depending only on $M$ and $a$, the set of trapped frequencies is contained in $\mathcal{G}^c$ (hence in particular, a neighbourhood of the superradiant regime $\{(\omega, m) | m\omega \in (0, m\omega_+]\}$ is not trapped). A precise version of this statement is given in Proposition \ref{est:TrappingInclusion}.

\subsection{The frequency dependent quantities $r^\sharp$, $R^\sharp$ and the constants $\rmax$, $\rmin$}

Besides $\xi_\ell$ and $\xi_r$, we will have to fix constants $\overline{r}^\sharp$ (large) and $\underline{r}^\sharp$ (close to $r_+$) defined abstractly by  
\begin{equation}
r_+ < \rmin <\rmax < \infty.
\end{equation}
To understand their relevance, we define for $m\omega \in (0, m\omega_+)$ the quantity $\rsharp$ to be the positive value satisfying
\begin{equation}
\omega-\frac{am}{{\rsharp}^2+a^2} = 0 \, , 
\end{equation}
so that $\omega - \omega_r = 0$ when $r=\rsharp$.\footnote{We notice already that $(\omega-\omega_r)^2$ is the only term ``with the wrong sign” in (\ref{def:SepWaveOperator}) suggesting that for fixed frequency $(\omega,m)$ one should be able to prove good estimates near $r^\sharp$ multiplying (\ref{def:SepWaveOperator}) by $\widehat{u}$ and integrating. Near the horizon and infinity one needs to be careful with weights, which is why we define $R^\sharp$ in (\ref{def:Rsharp}).} We immediately note the limits
\begin{equation}
\lim_{m\omega \to 0}\rsharp = \infty, \qquad \lim_{m\omega\to m\omega_+} \rsharp = r_+.
\end{equation}
and that $\rsharp$ is homogeneous of degree 0 in $(\omega, m)$, i.e.~$\rsharp(\omega,m)=\rsharp(k\omega, k m)$ for $k \in \mathbb{Z} \setminus \{0\}$.

Now given $\ximin, \ximax, \rmin, \rmax$, we define a cut-off version of $\rsharp$ as follows:
\begin{equation}\label{def:Rsharp}\index{Rsharp@$R^\sharp$}
R^\sharp =
\begin{cases}
\rmax & -\ximin m^2 < m\omega\leq\frac{am^2}{{\rmax}^2+a^2}, \\
\rsharp, & \frac{am^2}{{\rmax}^2+a^2}\leq m\omega\leq \frac{am^2}{{\rmin}^2+a^2}, \\
\rmin & \frac{am^2}{{\rmin}^2+a^2}\leq m\omega < m\omega_+ + \ximax m^2.
\end{cases}
\end{equation}
$R^\sharp$ is homogeneous of degree 0 in $(\omega, m)$ and -- unlike $r^\sharp$ -- defined and bounded on all of $\mathcal{G}$. The behaviour of $R^\sharp$ is shown in Figure \ref{fig:Rsharp}. Loosely speaking, therefore, $\rmin, \rmax$ provide quantitative bounds on how close we go to the boundaries of the superradiant regime from the interior of that regime. We remark that the function $R^\sharp$ is only continuous. A smooth version of $R^\sharp$, denoted $R^{\musDoubleSharp}$, is defined in (\ref{def:RDoubleSharp}) after having fixed a further constant $\delta^\sharp$ in Section \ref{sec:LocalLagrangian}.
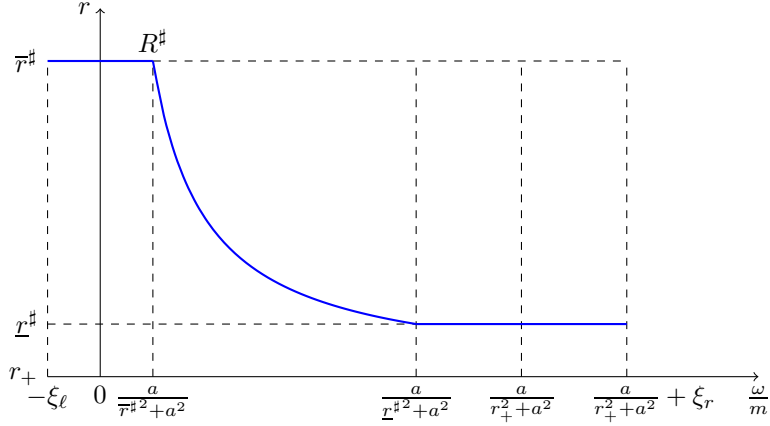
\begin{figure}
 \begin{center}
\begin{tikzpicture}[scale=0.7]
\draw [->](-1,0)node[left]{$r_+$}--(0,0)node[below]{$0$}--(12.5,0);
\draw (12.5, 0)node[below, black]{$\tfrac{\omega}{m}$};
\draw [->](0,0)--(0,7)node[left]{$r$};
\draw [black, dashed](-1,1)node[left, black]{$\rmin$}--(6,1);
\draw [black, dashed](1,6)--(10,6);
\draw [black, dashed](-1,0)node[below, black]{$-\ximin$}--(-1,6);
\draw [black, dashed](1,0)node[below, black]{$\frac{a}{{\rmax}^2+a^2}$}--(1,6)node[above, black]{$R^\sharp$};
\draw [black, dashed](6,0)node[below, black]{$\frac{a}{{\rmin}^2+a^2}$}--(6,6);
\draw [black, dashed](8,0)--(8,6);
\draw (8, 0)node[below, black]{$\scriptstyle \tfrac{a}{r_+^2+a^2}$};
\draw [black, dashed](10,0)node[below, black]{$\quad\quad \tfrac{a}{r_+^2+a^2}+ \ximax $}--(10,6);
\draw [blue, thick](-1,6)node[left, black]{$\rmax$}--(1,6);
\draw [domain=1:6, smooth, variable=\x, blue, thick] plot({\x}, {6/\x});
\draw [blue, thick](6,1)--(10,1);
\end{tikzpicture}
\caption{$R^\sharp$ as a function of $\omega/m$}
\label{fig:Rsharp}
\end{center}
\end{figure}

\subsection{The potential trapping radius $\re$}
We finally define the radius at which trapping (potentially) happens:
\begin{definition}
For all frequency pairs $(\omega, m)$ such that both $\omega \neq 0$ and $m\omega \notin (0, m\omega_+]$ (hence including the $m=0$ frequencies), we define the \underline{potential trapping radius} $\re = \re(\omega, m)$\index{re@$\re$} to be the (unique) value $r \in(r_+, \infty)$ such that
\begin{equation} \label{def:rnatural}
\Rs\left(\frac{(\omega-\omega_r)^2}{v^2}\right) = 0 \text{ at } r = \re.
\end{equation}
\end{definition}

The existence and uniqueness claim of the definition is shown immediately in Proposition \ref{est:rnaturalbounds} below.
Note that $\re$ is in particular defined for all $(\omega, m)\in\mathcal{G}' \setminus \{\omega=0\}$. 

\begin{proposition}\label{est:rnaturalbounds}
Let $(\omega, m)$ be such that both $\omega \neq 0$ and $m\omega \notin (0, m\omega_+]$ hold. Then the potential trapping radius $\re$ exists and is unique. Moreover, $\re=\re(\omega,m)$ is a homogeneous function of $(\omega,m)$ of degree zero. Finally, $\re$ satisfies the bound
\begin{align} \label{basicineq} 
\left(r_+^2+ \left(r_+^2 + a^2 -\eta \right)\right)^{1/2} \leq \re \leq r_+ + \left(r_+^2 + a^2 - \eta\right)^{1/2}  \, ,
\end{align}
where we have defined
\begin{align} \label{def:eta}
\eta\index{eta@$\eta$} = \frac{am}{\omega} \, .
 \end{align}
\end{proposition}

\begin{proof}
We recall that $a>0$ and $v = \frac{\Delta^{1/2}}{r^2+a^2}$. We define for $\omega \neq 0$ the shorthand notation
\begin{equation} \label{def:gamplus}
\Gamma = \Gamma(r, \eta)\index{Gamma@$\Gamma, \Gamma_+$} := (r^2+a^2)\frac{\omega-\omega_r}{\omega} = r^2 + a^2 - \eta,  \qquad \Gamma_+ := \Gamma(r_+, \eta) \, .
\end{equation}
We are concerned with frequencies $(\omega, m)$ such that both $\omega \neq 0$ and $m\omega \notin (0, m\omega_+]$. In particular, this set includes $\mathcal{G}^\prime \setminus \{(0,0)\}$ and we have
\begin{equation}\label{lim:GammaPlus}
0 < \Gamma_+ < \infty, \qquad \lim_{\tfrac{\omega}{m} \to 0^-}\Gamma_+ = \infty, \qquad \lim_{\tfrac{\omega}{m} \to \tfrac{a}{r_+^2+a^2}^+}\Gamma_+ = 0 \, .
\end{equation}
To prove the existence and uniqueness claim, it suffices to prove that the function
\begin{equation}\label{id:fsquared}
\frac{(\omega-\omega_r)^2}{\omega^2 v^2} = \frac{\Gamma^2}{\Delta}
\end{equation}
has a unique critical point in $(r_+,\infty)$, which by the asymptotic behaviour must be a global minimum. To prove this, we first expand
\begin{equation} \label{nifo}
\left(\frac{\Gamma^2}{\Delta}\right)' = \frac{2\Delta\Gamma\Gamma' - \Gamma^2\Delta'}{\Delta^2} = \frac{2\Gamma P_\eta(r) }{(r^2+a^2)\Delta} \, , 
\end{equation}
where 
\begin{equation}\label{def:Peta}\index{Peta@$P_\eta$}
P_\eta(r) := 2r\Delta - (r-M)\Gamma =  r^3 - 3Mr^2 + (a^2+\eta)r + (a^2-\eta)M \, .
\end{equation}
For the frequency range considered we have $\Gamma > 0$ for all $r \in (r_+, \infty)$.
It follows that for $r\in (r_+, \infty)$, $(\Delta^{-1}\Gamma^2)’ = 0$ if and only if $P_\eta(r) = 0$. From \eqref{def:Peta} and \eqref{lim:GammaPlus} one easily sees that  $P_\eta(r_+) < 0$, and $\partial_r^2 P_\eta > 6(r_+-M) > 0$ on $(r_+, \infty)$, $P_\eta$ has precisely one root in this region, hence $\re$ is uniquely defined. The homogeneity claim easily follows from the fact that everything depends only on $\eta$.


The estimate (\ref{basicineq}) is proven by showing that $P_\eta(r)$ changes sign in $\left[(r_+^2+ \Gamma_+)^{1/2},r_+ + \Gamma_+^{1/2}\right]$.
%
First,
\begin{alignat}{4}
\Gamma((r_+^2+ \Gamma_+)^{1/2}) &= 2\Gamma_+, \qquad & \Gamma(r_+ + \Gamma_+^{1/2}) &= 2\Gamma_+^{1/2}(r_+ + \Gamma_+^{1/2}), \\
\Delta((r_+^2+ \Gamma_+)^{1/2})&=\Gamma_+ + 2Mr_+ - 2M(r_+^2 + \Gamma_+)^{1/2} , \qquad & \Delta(r_+ + \Gamma_+^{1/2}) &= \Gamma_+^{1/2}(\Gamma_+^{1/2} + 2r_+ - 2M).
\end{alignat}
Consequently, applying the first part of \eqref{def:Peta},
\begin{align*}
P_\eta((r_+^2+ \Gamma_+)^{1/2} ) &= -2M((r_+^2 + \Gamma_+)^{1/2}- r_+)^2,\\
P_\eta(r_+ + \Gamma_+^{1/2} ) &= 2(r_+ + \Gamma_+^{1/2})\Gamma_+^{1/2}(\Gamma_+^{1/2} + 2r_+ - 2M) - 2(r_+ + \Gamma_+^{1/2})\Gamma_+^{1/2}(\Gamma_+^{1/2}+r_+-M).
\end{align*}
One notes that $P_\eta((r_+^2+ \Gamma_+)^{1/2} ) < 0 < P_\eta(r_+ + \Gamma_+^{1/2} )$, and \eqref{basicineq} follows.
\end{proof}
We have the following Corollary of Proposition \ref{est:rnaturalbounds}:
\begin{corollary}\label{cor:rnaturalasymptotic}
For sufficiently small $\widetilde{\ximin} > 0$ depending only on $a, M$, the inequality
\begin{equation}\label{est:rnaturalasymptotic}
\frac12 \left(-\frac{am}{\omega}\right)^{1/2} <\re < 2\left(-\frac{am}{\omega}\right)^{1/2}
\end{equation}
holds whenever $m\omega \in (-\widetilde{\ximin}  m^2, 0)$. Additionally, for any $\delta > 0$, there exists an $\widetilde{\ximax}$ depending only on $a, M, \delta$ such that
\begin{equation}\label{est:rnaturallimit}
\re < r_+ + \delta
\end{equation}
holds whenever $m\omega \in (m\omega_+, m\omega_+ + \widetilde{\ximax} m^2)$.
\end{corollary}
\begin{proof}
The inequality $m\omega \in (-\widetilde{\ximin}  m^2, 0)$ directly implies
$
\eta < -\frac{a}{\widetilde{\ximin}}.
$
The estimate \eqref{est:rnaturalasymptotic} follows from applying the bounds \eqref{basicineq} and noting the limits
\begin{equation}
\lim_{\eta \to -\infty} (-\eta)^{-1/2}(r_+^2+ \Gamma_+)^{1/2} = \lim_{\eta \to -\infty} (-\eta)^{-1/2} (r_+ + \Gamma_+^{1/2}) = 1.
\end{equation}
The statement \eqref{est:rnaturallimit} follows from Proposition \eqref{est:rnaturalbounds} along with the limit
$
\lim_{\eta \to (r_+^2+a^2)^-}\Gamma_+ = 0.
$
\end{proof}


\subsection{Outline of the selection of frequency constants}\label{sec:freqranges}
We now first outline how we will determine $\rmin$, $\rmax$, $\ximin$, $\ximax$, as well as two auxiliary constants, $c_{\mathcal{G}}$ and $c_{\mathcal{G}'}$. 
\begin{enumerate}
\item \textbf{Absolute bounds on $\ximin, \ximax, \rmin, \rmax$:} 
We show that there exist constants $\widetilde{\xi}_\ell$, $\widetilde{\xi}_{r}$, $\widetilde{\overline{r}}^\sharp$, $\widetilde{\underline{r}}^\sharp$ depending only on $M$ and $a$ such that (\ref{est:rnaturalasymptotic}) as well as the estimates of Proposition \ref{prop:AngularTermsLowerBound}, Corollary \ref{cor:rminmaxbounds} and
 Proposition \ref{est:TrappingInclusion}
hold for all
\begin{align}
\xi_\ell \leq \widetilde{\xi}_\ell \ \ \ , \ \ \ \xi_r \leq \widetilde{\xi}_r \ \ \ , \ \ \ {\overline{r}}^\sharp \geq \widetilde{\overline{r}}^\sharp \ \ \ , \ \ \ {\underline{r}}^\sharp \leq \widetilde{\underline{r}}^\sharp \, .
\end{align}
\item\textbf{Fixing $\ximax$ and $\rmin$:} We next fix once and for all $\rmin<\widetilde{\underline{r}}^\sharp$, and consequently choose $\ximax$ (also once and for all) such that Proposition \ref{prop:LowerBoundSuperrad} holds.
\item\textbf{We fix $\ximin, \rmax, c_{\mathcal{G}}, c_{\mathcal{G}'}$ once and for all such that the following relations hold:} 
\begin{subequations}
\begin{align}
c_{\mathcal{G}} &= c_0 c_{\mathcal{G}'}, \label{id:csharpRelation}\\
c_{\mathcal{G}'} &< \frac{1}{8}\min\left(\frac{r_+^2+a^2}{a}\ximin, 1\right),  \label{se2} \\
\ximin &= \frac{1}{8\rmax}, \label{se3} \\
\frac{1}{\rmax} &< \frac{c_{\mathcal{G}}^{1/2}}{2C_1}. \label{se4}
\end{align}
\end{subequations}
where $c_0 = \frac{r_+^2}{r_+^2+a^2}\left(\tfrac{r_+-M}{r_+}\right)^3$ and $C_1$ are explicit constants depending only on $M$ and $a$.\footnote{More specifically, $C_1$ is determined in Proposition \ref{Cor:PositiveEELagrange} below.}
 
We do this as follows. Given $\overline{r}^\sharp > \frac{r_+^2+a^2}{8a}$ we define $\xi_\ell$ by (\ref{se3}) and $c_{\mathcal{G}^\prime}=\frac{1}{80} \frac{r_+^2+a^2}{a} \frac{1}{\overline{r}^\sharp}$ so that (\ref{se2}) holds. We also define $c_\mathcal{G}$ by (\ref{id:csharpRelation}). We claim that we can choose $\overline{r}^\sharp$ large enough such that also (\ref{se4}) holds which — after the above choices — boils down to choosing $\overline{r}^\sharp$ to satisfy
\[
\frac{1}{\overline{r}^\sharp} < \frac{1}{2C_1} \sqrt{\frac{c_0}{80} \frac{r_+^2+a^2}{a} \frac{1}{\overline{r}^\sharp}} 
\]
which is indeed easily achieved. 

\end{enumerate}
With these choices of constants, the main results of this section are the coercivity estimate (\ref{est:ELCombined}) for frequencies in $\mathcal{G}$ localised in $ r\in [R^\sharp,\infty)$ and the spatially global estimate (\ref{est:ELCombined2}) for frequencies in $\mathcal{G}^\prime$.

\subsection{Absolute bounds on the frequency constants}\label{sec:FrequencyConstantsAbsoluteBounds}
In this section we set some preliminary bounds on $\ximin, \ximax, \rmin, \rmax$.




%

\begin{proposition}\label{prop:AngularTermsLowerBound}
Let $\widetilde{\xi}_r$ satisfy 
\begin{equation}\label{est:ximaxabsolutebound}
\frac{a^2}{r_+^2+a^2}+\widetilde{\xi}_r < \frac12.
\end{equation}
Then the following estimate holds whenever $m \omega  \leq m\omega_+ + \widetilde{\xi}_r m^2$: 
\begin{equation}\label{est:AngularTermsLB}
\left(\frac{m}{\sin\theta} - a\omega\sin\theta\right)^2 \geq \frac{m^2}{4\sin^2\theta} \, .
\end{equation}
If in addition $m\omega < 0$, this bound may be improved to
\begin{equation}\label{est:AngularTermsBetter}
\left(\frac{m}{\sin\theta} - a\omega\sin\theta\right)^2 \geq \frac{m^2}{\sin^2\theta}.
\end{equation}
\end{proposition}
\begin{proof}
The bound \eqref{est:AngularTermsBetter} is trivial (recalling that $a>0$). In order to prove \eqref{est:AngularTermsLB}, we write
\begin{equation}\label{est:PositivePotential}
\frac{m^2}{\sin\theta} - am\omega\sin\theta = \frac{m^2}{2\sin\theta} + \left(\frac{m^2}{2\sin\theta}-am\omega\sin\theta\right) \geq \frac{m^2}{2\sin\theta},
\end{equation}
which comes from the identity
\begin{equation}
\frac{m^2}{2\sin\theta}-am\omega\sin\theta \geq \frac{m^2}{2\sin\theta} - \left(\frac{a^2}{r_+^2+a^2} m^2+\widetilde{\xi}_r m^2\right)\sin\theta.
\end{equation}
Using our choice of $\widetilde{\xi}_r$ and squaring \eqref{est:PositivePotential} gives our result.
\end{proof}

\begin{corollary}\label{cor:rminmaxbounds}
Suppose $\widetilde{\xi}_r$ satisfies \eqref{est:ximaxabsolutebound} and $m \omega  \leq m\omega_+ + \widetilde{\xi}_r m^2$.
There exist quantities $\widetilde{\underline{r}}^\sharp$ and $\widetilde{\overline{r}}^\sharp$ depending only on $a$ and $M$, such that
\begin{equation}
v^2\left(\frac{m}{\sin\theta} - a\omega\sin\theta\right)^2 \geq \frac{v^2 m^2}{8\sin^2\theta} + \left(\frac{2am(r^2-r_+^2)}{(r^2+a^2)(r_+^2+a^2)}\right)^2 \text{ for } r < \widetilde{\underline{r}}^\sharp\label{est:rminbound}
\end{equation}
and
\begin{equation}\label{est:rmaxbound}
v^2\left(\frac{m}{\sin\theta} - a\omega\sin\theta\right)^2 \geq \frac{v^2 m^2}{8\sin^2\theta} + \left(\frac{m}{8r} + \frac{am}{r^2+a^2}\right)^2 \text{ for } r> \widetilde{\overline{r}}^\sharp
\end{equation}
hold.
\end{corollary}
\begin{proof}
This follows from (\ref{est:PositivePotential}), taking into account asymptotics as $r\to r_+$, $r\to\infty$ respectively.
\end{proof}

Finally, we would like to additionally choose our parameters $\ximin, \ximax$ small enough (depending only on $M$, $a$) such that the set of trapping frequencies, as defined in \cite{DRSR}, is wholly contained within $\mathcal{G}^c$.
\begin{proposition}\label{est:TrappingInclusion}
Given $\mathcal{G}_\natural$ as defined in Section 8.1 of \cite{DRSR}, then, there exists $\widetilde{\xi}_\ell$, $\widetilde{\xi}_r$ (depending only on $M$ and $a$) such that for $\ximin \leq \widetilde{\xi}_\ell$ and $\ximax \leq \widetilde{\xi}_r$, the inclusion relation $\mathcal{G}_\natural \subset\mathcal{G}^c$ holds.
\end{proposition}
\begin{proof}
By the characterization of $\mathcal{G}_\natural$ given in \cite{DRSR}, there exist fixed positive constants $\epsilon_{\text{width}}$ and $\alpha' > 0$ (depending only on $M$ and $a$) such that, for $(\omega, m, \Lambda) \in \mathcal{G}_\natural$,
\begin{equation}
\epsilon_{\text{width}}\Lambda < \omega^2, \qquad m\omega\notin (0, m\omega_++\alpha'\Lambda].
\end{equation}
It follows from the inequality $\Lambda > m^2$ that
\begin{equation}
m\omega \geq m\omega_+ + \alpha'm^2 \text{ or } m\omega \leq -\epsilon_{\text{width}}^{1/2}m^2.
\end{equation}
Then, for 
\begin{equation}\label{est:xiboundsDRSR}
\widetilde{\xi}_\ell <2\epsilon_{\text{width}}^{1/2}, \qquad \widetilde{\xi}_r  < 2\alpha',
\end{equation}
the desired inclusion holds.
\end{proof}
We recall already here the definition of the trapping radius $r_{trap} = r_{trap} (a, M, \omega, m, \Lambda)$ from \cite{DRSR}, which will appear in Proposition \ref{thm:DRSRMain}. In particular, we have $r_{trap}=0$ for $(\omega, m) \in \mathcal{G}$ by Proposition \ref{est:TrappingInclusion}.
\begin{remark}\label{rem:Trapping}
The trapping radius $r_{trap}$, as defined in \cite{DRSR} (and appearing in the estimate \eqref{est:DRSRMain} below), does not necessarily agree with the $\re$ defined above. We relate the quantities through two observations. First, the estimate \eqref{est:DRSRMain} holds if $r_{trap}$ is replaced by $r^0_{max}$, the point at which $V_0$ achieves its maximum (following from Lemma 8.6.1 in \cite{DRSR}). Second, at honest trapped frequencies $V_0 - \omega^2$ and $\left(V_0-\omega^2\right)^\prime$ vanish. Since $(v^{-2}(V_0 - \omega^2))'$ then vanishes at $\re$, it follows that $\re$ and $r^0_{max}$ coincide for these frequencies.
\end{remark}

\subsection{Fixing the quantities $\rmin$ and $\ximax$}\label{sec:SuperradiantEstimates}

We now fix  the quantities $\rmin$ and $\ximax$. The next proposition determines how to fix $\ximax$ from $\rmin$.


\begin{proposition}\label{prop:LowerBoundSuperrad}
Given $\rmin \in (r_+, \widetilde{\underline{r}}^\sharp)$, there exists a $\ximax > 0$ and a constant $c_{aux}>0$ depending on $\rmin$, $M$, $a$ such that the inequality
\begin{equation}\label{est:LowerBoundSuperrad}
\omega v\Rs\left(\frac{\omega-\omega_r}{v}\right) \geq c_{aux}\left( \frac{\omega^2}{r} + \frac{a m\omega}{r^3}\right)
\end{equation}
holds  for $r\in(\rmin, \infty)$ and $m\omega \in [0, m\omega_+ + \ximax m^2)$. Moreover, it is possible to choose $\ximax$ satisfying the bound
\begin{equation}\label{est:ximaxBoundRelative}
\ximax < \frac{a({\rmin}^2-r_+^2)}{({\rmin}^2+a^2)(r_+^2+a^2)},
\end{equation}
and that in addition, the inequality
 \begin{equation} \label{cond:rermin}
\re < \rmin
\end{equation}
holds for $m\omega \in (m\omega_+, m\omega_+ + \ximax m^2)$.
\end{proposition}
\begin{proof}

We first establish \eqref{est:LowerBoundSuperrad}. We write
\begin{equation}\label{id:fsharpprime}
\omega v\Rs\left(\frac{\omega-\omega_r}{v}\right) = \frac{2r\Delta \omega^2 - (r-M)(r^2+a^2)(\omega-\omega_r)\omega}{(r^2+a^2)^2}.
\end{equation}
We then expand
\begin{equation}\label{def:SuperDerExpansion}
(r^2+a^2)(\omega-\omega_r)\omega = \omega^2(r^2+a^2-\tfrac{am}{\omega}) = (r^2+a^2-2Mr_+)\omega^2 - (\tfrac{am}{\omega} - 2Mr_+)\omega^2.
\end{equation}
Consequently,
\begin{equation}
2r\Delta - (r-M)(r^2+a^2 - 2Mr_+) = (r-r_+)(2r(r_+-M)+(r-M)(r-r_+))\geq (r-r_+)^3 \, .
\end{equation}
For $\frac12m\omega_+ \leq m\omega \leq m\omega_+$ we ignore the $\tfrac{am}{\omega} - 2Mr_+$ term in \eqref{def:SuperDerExpansion}, as it is positive, and we use the inequality
\begin{equation}
\tfrac{r-r_+}{r} \geq \tfrac{\rmin-r_+}{\rmin} >0,
\end{equation}
as this function is increasing on $(r_+, \infty)$. The bound follows by noting $\frac{a}{2(r_+^2+a^2)} \leq \frac{\omega}{m} \leq \frac{a}{r_+^2+a^2}$ in this range. Likewise, in the region $m\omega_+ \leq m\omega \leq m\omega_+ + \ximax m^2$ we may write $|\tfrac{am}{\omega} - 2Mr_+|\leq C\ximax$, so for suitably small $\ximax$ depending on $\rmin$ we have
\begin{equation}
\left(2r\Delta - (r-M)(r^2+a^2 - \tfrac{am}{\omega})\right) \geq (r-r_+)^3 + (r-M)(\tfrac{am}{\omega} - 2Mr_+) \geq cr^3
\end{equation}
for some $c > 0$.
 For $0 \leq m\omega < \frac12m\omega_+$ we use the relation $\tfrac{am}{\omega} - 2Mr_+ \geq d \tfrac{am}{\omega}$ (for $d$ depending only on $M$, $a$) and carry out the analogous calculations. This proves \eqref{est:LowerBoundSuperrad}.
 
We may directly choose $\ximax$ to be small enough to satisfy \eqref{est:ximaxBoundRelative}. 
The bound \eqref{cond:rermin} follows directly from the statement that, for any $\delta > 0$, there exists an $\widetilde{\xi}_r$ such that
\begin{equation}
\re < r_+ + \delta
\end{equation}
holds whenever $m\omega \in (m\omega_+, m\omega_+ + \widetilde{\xi}_r m^2)$,
which in turn follows from Corollary \ref{cor:rnaturalasymptotic}.
\end{proof}
At this point we fix our values of $\rmin, \ximax$ as follows: By the results of Section \ref{sec:FrequencyConstantsAbsoluteBounds}, we take constants $\widetilde{\xi}_r$ satisfying \eqref{est:ximaxabsolutebound} and \eqref{est:xiboundsDRSR} and $\widetilde{\underline{r}}^\sharp$ satisfying \eqref{est:rminbound}, both of which depend only on $a, M$. Then, we choose $\rmin < \widetilde{\underline{r}}^\sharp$ arbitrary, and consequently fix $\ximax$ satisfying \eqref{est:LowerBoundSuperrad}, \eqref{est:ximaxBoundRelative} and \eqref{cond:rermin} by applying Proposition \ref{prop:LowerBoundSuperrad}. In particular, the $c_{aux}$ in (\ref{est:LowerBoundSuperrad}) is now also fixed.

\subsection{Fixing the quantities $\rmax$ and $\ximin$}\label{sec:SuperradiantEstimates2}
We first prove a preliminary estimate which is valid for all frequencies. Afterwards, we will use it in conjunction with Proposition \ref{prop:LowerBoundSuperrad} to obtain our main coercivity bound valid in all of $\mathcal{G}$ in Proposition \ref{prop:SuperradiantPositivity}.
\begin{proposition}\label{Cor:PositiveEELagrange}
There exists an explicit constant $C_1$ depending only on $M$ such that, for any $\tilde{c} \in (0,1)$,
\begin{equation}\label{est:PositiveEELagrange}
\omega v\Rs\left(\frac{\omega-\omega_r}{v}\right) - \frac{1}{4r}(\omega-\omega_r)^2 \geq \frac{1}{4r}\omega^2 - \frac{\tilde{c}}{4}\frac{v^2}{r}m^2,
\end{equation}
whenever
\begin{equation}\label{est:rCsharpcsmall}
 r > C_1 \tilde{c}^{-1/2}.
\end{equation}
\end{proposition}
\begin{proof} 
We write
\begin{equation}
\omega\Rs\left(\frac{\omega-\omega_r}{v}\right) = \omega(\omega-\omega_r)\Rs(v^{-1}) - v^{-1}\Rs(\omega_r).
\end{equation}
An expansion as in \eqref{id:fsharpprime} together with the asymptotic bounds $\Rs(v^{-1}) = 1 + O(r^{-1})$, $v^{-1}\Rs(\frac{a}{r^2+a^2}) = O(r^{-2})$ imply that for sufficiently large $r$, there exists a constant $C_0$ depending only on $M$ such that
\begin{equation}
\omega v\Rs\left(\frac{\omega-\omega_r}{v}\right) \geq  \frac{3}{4r}\omega^2  - \frac{C_0}{r^3}|am\omega|.
\end{equation}
It follows that
\begin{equation}
\omega v\Rs\left(\frac{\omega-\omega_r}{v}\right) - \frac{1}{4r}(\omega-\omega_r)^2 \geq\frac{1}{2r}\omega^2 - v\frac{C_0}{r^2}|am\omega| - \frac{1}{2r}|\omega\omega_r| - \frac{a^2m^2}{4r(r^2+a^2)^2}.
\end{equation}
Additionally, since $C_0$ depends only on $M$, there exists a $C_0'$ depending only on $M$ such that
\begin{equation}
v\frac{C_0}{r^2}|am\omega| + \frac{1}{2r}|\omega\omega_r| \leq \frac{1}{4r}\omega^2 + \frac{C_0'a^2m^2}{r^5},
\end{equation}
so, for a new constant $C'$ depending only on $M$,
\begin{equation}
\omega v\Rs\left(\frac{\omega-\omega_r}{v}\right) - \frac{1}{4r}(\omega-\omega_r)^2 \geq\frac{1}{4r}\omega^2  - \frac{C_0''a^2m^2}{r^5}.
\end{equation}
We remark that \eqref{est:rCsharpcsmall} directly implies $r > C_1$. We may therefore select $C_1$ satisfying the inequality $C_1^2 \geq 8C_0''a^2$ and such that the bound $r^{-5} \leq 2 v^2r^{-3}$ holds for $r > C_1$. Consequently, for $r$ satisfying \eqref{est:rCsharpcsmall} it follows that
\begin{equation}
\frac{C_0''a^2m^2}{r^5} \leq \frac{2C\tilde{c}}{C_1^2}\frac{2v^2}{r}m^2\leq \frac{\tilde{c}}{4}\frac{v^2}{r}m^2.
\end{equation}
The bound \eqref{est:PositiveEELagrange} directly follows.
\end{proof}
Proposition \ref{Cor:PositiveEELagrange} may be combined with Proposition \ref{prop:LowerBoundSuperrad} to prove a bound valid in the entirety of $\mathcal{G}$:
\begin{proposition}\label{prop:SuperradiantPositivity}
Recall that we fixed $\rmin, \ximax$, $c_{aux}$ at the end of Section \ref{sec:SuperradiantEstimates}. For any $c_{\mathcal{G}}$ with $0< c_{\mathcal{G}} < \min \left(\frac{1}{4}, \frac{c_{aux}}{2}\right)$, for all $\ximin \leq \widetilde{\xi}_\ell$ and for all $\rmax < \infty$ satisfying
\begin{equation}\label{def:rmaxconditions}
\rmax \geq 2C_1 c_{\mathcal{G}}^{-1/2},
\end{equation}
where $C_1$ is as defined in Proposition \ref{Cor:PositiveEELagrange}, the inequality
\begin{equation}\label{est:ELCombined}
\omega v \Rs\left(\frac{\omega-\omega_r}{v}\right) - \frac{c_{\mathcal{G}}}{r}(\omega-\omega_r)^2 + \frac{c_{\mathcal{G}} v^2}{2r}\left(\frac{m}{\sin\theta} - a\omega\sin\theta\right)^2 \geq \frac{c_{\mathcal{G}}}{4r}\omega^2 + \frac{c_{\mathcal{G}} v^2}{16r}\left(\frac{m}{\sin\theta}\right)^2
\end{equation}
holds for all $(\omega, m)$ satisfying $m \omega \in (-{\xi}_\ell m^2, m \omega_+ + {\xi}_r m^2)$ and all $r\in[R^\sharp, \infty)$.
\end{proposition}
\begin{proof}
We split the bound \eqref{est:ELCombined} into the three cases $m\omega \in (-\ximin m^2, \frac{am^2}{{\rmax}^2+a^2})$, $m\omega\in [\frac{am^2}{{\rmax}^2+a^2}, m\omega_+)$, and $m\omega \in [m\omega_+, m \omega_+ + {\xi}_r m^2)$.

For $m\omega \in (-\ximin m^2, \frac{am^2}{{\rmax}^2+a^2})$, we apply Proposition \ref{Cor:PositiveEELagrange}, with $\tilde{c} = c_{\mathcal{G}}/4$, to get
\begin{equation}
\omega v \Rs\left(\frac{\omega-\omega_r}{v}\right) - \frac{c_{\mathcal{G}}}{r}(\omega-\omega_r)^2 \geq \frac{c_{\mathcal{G}}}{r}\omega^2 - \frac{c_{\mathcal{G}} v^2}{16r}\left(am\right)^2 \, .
\end{equation}
 The inequality \eqref{est:ELCombined} follows from Proposition \ref{prop:AngularTermsLowerBound}. In particular, this holds on $[\rmax, \infty)$ for $\rmax$ satisfying \eqref{def:rmaxconditions}, using the same $C_1$ as in Proposition \ref{Cor:PositiveEELagrange}.

For $m\omega\in [\frac{am^2}{{\rmax}^2+a^2}, m\omega_+)$, the bound $|\omega-\omega_r| \leq |\omega|$ holds on $[\rsharp, \infty)$, which includes the region $[R^\sharp, \infty)$. The inequality \eqref{est:ELCombined} follows from Proposition \ref{prop:LowerBoundSuperrad} combined with Proposition \ref{prop:AngularTermsLowerBound}. In particular, \eqref{est:ELCombined} holds since we have $c_{\mathcal{G}} < \tfrac{c_{aux}}{2}$.

For $m\omega \in  [m\omega_+, m \omega_+ + {\xi}_r m^2)$, the inequality $|\omega-\omega_r| \leq |\omega|$ holds for all $r > r_+$. The inequality \eqref{est:ELCombined} follows again from Proposition \ref{prop:LowerBoundSuperrad} combined with Proposition \ref{prop:AngularTermsLowerBound}.
\end{proof}
We may establish an analogous estimate in $\mathcal{G}'$:
\begin{proposition}\label{prop:FirstLagrangianGPrime}
For any $\ximin >0$, and for any constant $c_{\mathcal{G}'}$ satisfying
\begin{equation}\label{est:cximinBound}
0 < c_{\mathcal{G}'} < \frac{1}{8}\min\left(\frac{r_+^2+a^2}{a}\ximin, 1\right), 
\end{equation}
the estimate
\begin{equation} \label{est:ELCombined2}
\omega(\omega-\omega_r) + c_{\mathcal{G}'}\left(-(\omega-\omega_r)^2 + \frac{v^2}{2}\left(\frac{m}{\sin\theta} - a\omega\sin\theta \right)^2\right) \geq c_{\mathcal{G}'}\left((\omega-\omega_r)^2 + \frac{v^2}{2}\left(\frac{m}{\sin\theta} - a\omega\sin\theta \right)^2\right)
\end{equation}
holds for all $(\omega, m)\in\mathcal{G}'$.
\end{proposition}
\begin{proof}
For $m\omega > m\omega_+$, this follows directly from the bounds
\begin{equation}
(\omega-\omega_r)^2 \leq \omega(\omega-\omega_r), c_{\mathcal{G}'}<\frac18.
\end{equation}
For $m\omega < -\frac12\ximin m^2$, we note that 
\begin{equation}
c_{\mathcal{G}'} \omega_r^2 < \frac18 \frac{r_+^2+a^2}{a}\frac{a^2}{(r^2+a^2)^2}\ximin m^2\leq -\frac14 \omega\omega_r.
\end{equation}
It follows from the inequality $-\omega\omega_r > 0$ that
\begin{equation}
\omega^2-\omega\omega_r + 2c_{\mathcal{G}'}(-\omega^2 + 2\omega\omega_r - \omega_r^2) \geq 0.
\end{equation}
\end{proof}

Before we fix our constants, we prove one last bound which will be used in Section \ref{sec:IntersectionCommutator} to prove bounds in the interpolated regime $\mathcal{G}\cap\mathcal{G}'$.

\begin{proposition}\label{prop:rInterpolationInequality}
Setting 
\begin{equation}\label{id:relationthree}
 \ximin = \frac{1}{8\rmax},
\end{equation}
then, for sufficiently large $\rmax$ depending only on $a, M$, the inequality
\begin{equation}\label{cond:rermax} 
\re < \rmax
\end{equation}
holds for $m\omega \in (-\ximin m^2, -\tfrac12\ximin m^2]$.
\end{proposition}
\begin{proof}
We may rewrite the condition $-\ximin m^2 < m\omega < -\frac12\ximin m^2$ as
\begin{equation}
8 a\rmax = \frac{a}{\ximin} < -\frac{am}{\omega} < \frac{2a}{\ximin} =  16a\rmax.
\end{equation}
Additionally, given $(\omega, m)$, with $0 < -\omega m <  \ximin m^2$, Corollary \ref{cor:rnaturalasymptotic} implies
\begin{equation}
r_+ < \re \leq 2 \left(-\tfrac{am}{\omega}\right)^{1/2}
\end{equation}
holds for sufficiently small $\ximin$.
Consequently, there exists a constant $C$ depending only on $a$ and $M$ such that
\begin{equation}
\re \leq C{\rmax}^{1/2}.
\end{equation}
For $\rmax \geq C^2$, the bound \eqref{cond:rermax} follows.
\end{proof}

We can now we fix $c_{\mathcal{G}}$, $c_{\mathcal{G}'}$, $\ximin$, $\rmax$ as described in Section \ref{sec:freqranges} below (\ref{se4}). We recall that $c_0, C_1$ are explicit constants depending only on $a$ and $M$. Following the argument of Section \ref{sec:freqranges} it is possible to choose $c_{\mathcal{G}}$, $c_{\mathcal{G}'}$, $\ximin$, $\rmax$, such that the relations \eqref{def:rmaxconditions}, \eqref{est:cximinBound}, and \eqref{id:relationthree}, as well as
\begin{equation}\label{id:relationfour}
c_{\mathcal{G}} = c_0 c_{\mathcal{G}'}
\end{equation}
are simultaneously satisfied as long as
\begin{equation}
 \frac{1}{\overline{r}^\sharp} < \frac{1}{2C_1} \sqrt{\frac{c_0}{80} \frac{r_+^2+a^2}{a} \frac{1}{\overline{r}^\sharp}} \, .
\end{equation}
 Additionally, if $c_{\mathcal{G}}$, $c_{\mathcal{G}'}$, $\ximin$, $\rmax$ satisfy simultaneously the relations \eqref{def:rmaxconditions}, \eqref{est:cximinBound}, \eqref{id:relationthree}, and \eqref{id:relationfour} for given values of $C_1$ and $c_0$, it follows that $\delta c_{\mathcal{G}}$, $\delta c_{\mathcal{G}'}$, $\delta \ximin$, $\frac{1}{\delta}\rmax$ also satisfy these relations for any $\delta \in (0,1)$, so we may in particular choose $\rmax$ to be arbitrarily large, or $c_{\mathcal{G}}$, $c_{\mathcal{G}'}$, $\ximin$ to be arbitrarily small. We therefore choose $c_{\mathcal{G}}$, $c_{\mathcal{G}'}$, $\ximin$, $\rmax$ such that the bounds \eqref{est:rnaturalasymptotic}, \eqref{est:rmaxbound}, \eqref{est:xiboundsDRSR}, and the bounds on $\rmax$ given in Proposition \ref{prop:rInterpolationInequality} simultaneously hold. Note that with these choices, in particular the bound
 \begin{equation} \label{rsharpnatural}
 \re \leq R^\sharp
 \end{equation}
 holds for $(\omega, m)\in\mathcal{G}\cap\mathcal{G}'$. 
 
 In summary, we have now fixed the parameters in such a way that:
 \begin{itemize}
\item For frequencies in $\mathcal{G}$ we have the coercivity estimate (\ref{est:ELCombined}), valid for $r \leq [R^\sharp, \infty)$.
\item For frequencies in $\mathcal{G}^\prime$ we have the coercivity estimate (\ref{est:ELCombined2}) valid globally in $r$.
 \end{itemize}
 
 We conclude this section by showing that for frequencies in $\mathcal{G}$ our choice of constants allows us to derive a coercive Lagrangian estimate in a neighborhood of $R^\sharp$.

 
\subsection{A coercivity estimate near $R^\sharp$ for superradiant frequencies}\label{sec:LocalLagrangian}
At this point we have fixed our regions $\mathcal{G}, \mathcal{G}'$ (i.e.~$\xi_\ell$ and $\xi_r$), as well as numbers $\rmin, \rmax$ and $c_\mathcal{G}$, $c_{\mathcal{G}^\prime}$ depending only on $a, M$. Propositions \ref{prop:FirstLagrangianGPrime} and \ref{prop:SuperradiantPositivity} will be central to obtaining a commuted spacetime energy estimate. However, in $\mathcal{G}$, we will have to contend with the fact that Proposition \ref{prop:SuperradiantPositivity} does not hold in a neighborhood of $r_+$. In this frequency regime, we will therefore prove an estimate via an appropriate current which is cut off away from the horizon. However, this will introduce higher order error terms in a small (frequency-dependent) interval in $r$. We prove here an estimate which will be used to bound these error terms, and which serves as an analogue of the $\Koppa^h$ current in \cite{DRSR}.

\begin{proposition}\label{prop:DominantLagrangian}
There exists a uniform constant $\delta^\sharp > 0$ depending on $\ximin, \ximax, \rmin, \rmax$, which we have fixed previously, such that, for some $c' > 0$ depending on $\ximin, \ximax, \rmin, \rmax$, and for all $(\omega, m)\in\mathcal{G}$, the estimate
\begin{equation}\label{est:DominantLagrangian}
 - \frac{1}{r}(\omega-\omega_r)^2 + \frac{v^2}{r}\left(\frac{m}{\sin\theta} - a\omega\sin\theta\right)^2 \geq c'\frac{m^2v^2}{r\sin^2\theta}.
\end{equation}
holds for $r\in [R^\sharp - \delta^\sharp, R^\sharp + \delta^\sharp]$.
\end{proposition}
\begin{proof}
We first assume
\begin{equation}
\delta^\sharp < \frac12\rmin, \qquad \delta < \min(|\rmin - \widetilde{\underline{r}}^\sharp|, |\rmax - \widetilde{\overline{r}}^\sharp|)
\end{equation}
so that in particular $r\in [R^\sharp - \delta^\sharp, R^\sharp + \delta^\sharp]$ implies
\begin{equation}
\frac12R^\sharp < r < 2R^\sharp.
\end{equation}
and that the estimates \eqref{est:rminbound} and \eqref{est:rmaxbound} hold in the intervals $(\rmin-\delta^\sharp, \rmin + \delta^\sharp)$ and $(\rmax-\delta^\sharp, \rmax+\delta^\sharp)$ respectively.

In order to prove the estimate \eqref{est:DominantLagrangian}, we divide $\mathcal{G}$ into three subregimes, depending on how $R^\sharp$ is defined in \eqref{def:Rsharp}, and prove the bound separately in each regime.

We first define the quantities $\omega_{\rsharp}, \omega_{\rmin}, \omega_{\rmax}$ to be the values of $\omega_r$ at $\rsharp, \rmin, \rmax$. We may rewrite the frequency subregimes in \eqref{def:Rsharp} as $m\omega \in (-\ximin m^2, m\omega_{\rmax}]$, $m\omega \in [m\omega_{\rmax}, m\omega_{\rmin}]$, and $m\omega \in [m\omega_{\rmin}, m_+ + \ximax m^2)$.

For $m\omega \in [m\omega_{\rmax}, m\omega_{\rmin}]$, the definition of $\rsharp$ implies $\omega-\omega_{\rsharp} = 0$, so
\begin{equation}
|\omega-\omega_r| =| (\omega-\omega_r)-(\omega-\omega_{\rsharp})| = \left|\frac{am(r^2-{\rsharp}^2)}{(r^2+a^2)({\rsharp}^2+a^2)}\right|\leq \frac{C\delta^\sharp}{r^3}|m|.
\end{equation}
For sufficiently small $\delta^\sharp > 0$, we may combine this with Proposition \ref{prop:AngularTermsLowerBound} to prove \eqref{est:DominantLagrangian} in this regime, noting that, as we have fixed $\rmin$ we may say $r^{-2}\leq Cv^2$ for some $C$ depending only on $M$ and $a$. 

Next, for $m\omega \in [-\ximin m^2, m\omega_{\rmax}]$ , we write
\begin{equation}
|m(\omega-\omega_{\rmax})| = \left|m\omega-\frac{am^2}{{\rmax}^2+a^2}\right|\leq \left|\ximin + \frac{a}{{\rmax}^2+a^2}\right|m^2.
\end{equation}
In this regime, $R^\sharp = \rmax$, so $|r-\rmax|< \delta$, which implies
\begin{equation}\label{est:deltasharpvariation}
\left|\left(\omega-\omega_r\right)-\left(\omega-\omega_{\rmax}\right)\right| = \left|\frac{am(r^2-{\rmax}^2)}{(r^2+a^2)({\rmax}^2+a^2)}\right|\leq \frac{C\delta^\sharp}{r^3}|m|.
\end{equation}
The triangle inequality along with the relation \eqref{id:relationthree} imply
\begin{equation}
|\omega-\omega_r|\leq \left(\frac{C\delta^\sharp}{{\rmax}^3} + \frac{1}{8\rmax} + \frac{a}{{\rmax}^2+a^2}\right)|m|.
\end{equation}
For $\delta^\sharp$ satisfying $C\delta^\sharp/{\rmax}^3 < c_1/\rmax$, the bound \eqref{est:DominantLagrangian} follows from Corollary \ref{cor:rminmaxbounds}.

Finally, in the regime $m\omega \in (m\omega_{\rmin}, m_+ + \ximax m^2)$, we  write
\begin{equation}
|m(\omega-\omega_{\rmin})| = \left|am\omega-\frac{am^2}{{\rmin}^2+a^2}\right|\leq \left|\frac{a}{r_+^2+a^2} - \frac{a}{{\rmin}^2+a^2} +\ximax\right|m^2.
\end{equation}
In this region, $R^\sharp =\rmin$, so $r\in [\rmin - \delta^\sharp, \rmin + \delta^\sharp]$. This implies the inequality
\begin{equation}
\left|\left(\omega-\omega_r\right)-\left(\omega-\omega_{\rmin}\right)\right| = \left|\frac{am(r^2-{\rmin}^2)}{(r^2+a^2)({\rmin}^2+a^2)}\right|\leq \frac{C\delta^\sharp}{r^3} |m|.
\end{equation}
Therefore,
\begin{equation}
|\omega-\omega_r|\leq \left(\frac{C\delta^\sharp}{{\rmin}^3} + \ximax +\frac{a({\rmin}^2-r_+^2)}{({\rmin}^2+a^2)(r_+^2+a^2)}\right)|m|.
\end{equation}
For sufficiently small $\delta^\sharp$, \eqref{est:DominantLagrangian} follows from \eqref{est:rminbound} along with the assumption \eqref{est:ximaxBoundRelative}.
\end{proof}
We fix $\delta^\sharp$ once and for all according to Proposition \ref{prop:DominantLagrangian}.

\section{Construction of the microlocal commutator} \label{sec:comc}
		In this section we define and analyse the microlocal commutator $W$.

Section \ref{sec:cutoffs} introduces various frequency cutoff functions to be used in the construction of $W$, which itself is carried out in detail in Section \ref{sec:constructionsummary}. Basic properties of the commutator, such as positive signs for certain terms and the structure of lower order terms, are collected in Section \ref{sec:GPrimeBounds}. The regularity properties of the commutator both in $r$ and in frequency $(\omega,m)$ are studied in Sections \ref{sec:fbasicregularity} and \ref{sec:addregest}. The regularity properties established here will play a key role in the convolution estimates later.
\subsection{The frequency cut-offs $\wchi_1$, $\wchi_2$ and the frequency dependent radial cut-off $\chi_\sharp$} \label{sec:cutoffs}
 We recall from (\ref{def:GCalDef})--(\ref{def:GCalPrimeDef}) the definition of $\mathcal{G}$ and $\mathcal{G}^\prime$ and define smooth frequency cutoffs $\wchi_1$, $\wchi_2$ such that
\begin{subequations}\label{def:wchi12}
\begin{equation}\index{chitildei@$\wchi_1, \wchi_2$}
\wchi_i := \begin{cases}
\chi_i\left(\frac{\omega}{m}\right) & m\neq 0, \\
1 & m = 0,
\end{cases}
\end{equation}
where
\begin{equation}  {\chi}_1(y) = \begin{cases}
1 & y \leq -\ximin,\\
0 & -\tfrac34\ximin \leq y \leq \tfrac{a}{r_+^2+a^2} + \tfrac{3}{4}\ximax \\
1 & y \geq \tfrac{a}{r_+^2+a^2}+\ximax,
\end{cases}, \quad {\chi}_2(y) = \begin{cases}
1 & y \leq -\tfrac34\ximin,\\
0 & -\tfrac12\ximin \leq y \leq \tfrac{a}{r_+^2+a^2} + \tfrac{1}{2}\ximax, \\
1 & y \geq \tfrac{a}{r_+^2+a^2}+\tfrac{3}{4}\ximax.
\end{cases}
\end{equation}
\end{subequations}
We recall from Section \ref{sec:choices} that with our choice of $\xi_\ell, \xi_r$, trapped frequencies exist only in $\mathcal{G}^c \subset\subset \mathcal{G}'$, for which $1-\wchi_1 = 1-\wchi_2 = 0$, and superradiant frequencies exist only in $(\mathcal{G}')^c\subset\subset \mathcal{G}$, for which $\wchi_1 = \wchi_2 = 0$. 

Recalling the function $R^\sharp$ defined in \eqref{def:Rsharp} and the $\delta^\sharp$ of Proposition \ref{prop:DominantLagrangian}, we may define a smooth function $R^{\musDoubleSharp}$, which is homogeneous in $(\omega, m)$ and which satisfies
\begin{equation}\index{RDoubleSharp@$R^{\musDoubleSharp}$}\label{def:RDoubleSharp}
R^\sharp \leq R^{\musDoubleSharp} < R^\sharp+\frac14\delta^\sharp.
\end{equation}
We remark that Proposition \ref{prop:SuperradiantPositivity} continues to hold if we replace $R^\sharp$ with $R^{\musDoubleSharp}$.
Consequently, at each frequency $(\omega, m) \in\mathcal{G}$ we define a smooth cutoff function $\chisharp$ which is homogeneous of degree 0 in $(\omega, m)$ satisfying
\begin{equation}\label{def:chisharp}\index{chisharp@$\chisharp$}
\chisharp = \begin{cases}
0 & r\leq R^{\musDoubleSharp}+\frac14\delta^\sharp\\
1 & r \geq R^{\musDoubleSharp} +  \frac{1}{2}\delta^\sharp
\end{cases}
\end{equation}
and $0 \leq \chisharp' \leq C$ for a constant $C$ independent of frequency. Note that since $\chisharp = 1 \text{ for }r > R^\sharp +\frac34\delta^\sharp$ and $\qquad\chisharp = 0 \text{ for } r < R^\sharp+\frac14\delta^\sharp$,
in particular the bound \eqref{est:DominantLagrangian} holds in the support of $\chisharp'$ (and hence good coercivity estimates will be available in the space region where $\chisharp^\prime$ is non-trivial and frequencies are in $\mathcal{G}$).

\subsection{Definition of the commutator vector field} \label{sec:constructionsummary}

Establishing global energy bounds relies on the construction of a global microlocal commutator of the form
\begin{equation}\label{def:CommutatorAbstract}
\Wsc = \wchi_1\Wsv + (1-\wchi_1)\chisharp\Wsl \index{W@$\Wsc$}.
\end{equation}
We define $\Wsv, \Wsl$ such that, for some functions $\f, \h$ defined in $\mathcal{G}'$ and functions $f_\sharp, h_\sharp$ defined in $\mathcal{G}$,
\begin{align}
\Wsv &= v^{-1}\Rs - i\omega\f = v^{-1}(\Rs - i(\omega-\omega_r)\h), \label{Wnaturalrel}\index{Wnat@$\Wsv$} \\
\Wsl &= v^{-1}\Rs - i\omega f_\sharp = v^{-1}(\Rs - i(\omega-\omega_r)h_\sharp). \label{Wsharprel}\index{Wsharp@$\Wsl$}
\end{align}
We will fix $\f$, $\fsharp$ (via $h_\sharp$) explicitly in \eqref{def:f2} and \eqref{def:hsharp} respectively.
For ease of calculation, we will often use $\f$ when working with $\Wsv$ and $\hsharp$ when working with $W^\sharp$. 

In Section \ref{sec:overviewf} we first provide an overview on how we choose $\f$ and $f_\sharp$. After deriving some general algebraic identities for the commutator in Section \ref{sec:WECommuted} we finally define $\f$ and $f_\sharp$ in Sections \ref{sec:CommDef} and \ref{sec:IntersectionCommutator}.

It will later be useful to look at $\Wsc$ acting on functions in physical space. Given a function $u(t, r, \theta, \phi)$ with Fourier transform $\ph(\omega, r, \theta, m)$, we define $(\Wsc u)(t, r, \theta, \phi)$ to be the function satisfying the relation
\begin{equation}
\mathcal{F}(\Wsc u) = \Wsc\ph
\end{equation}


\subsubsection{Overview} \label{sec:overviewf}
We will construct $\f, f_\sharp$ such that the commutator satisfies the following conditions:
\begin{enumerate}
\item\label{def:WscRegion1} For $m\omega < -\ximin m^2$ or $m\omega > m\omega_+ + \ximax m^2$, $\wchi_1 = 1$ and hence $W=\Wsv$. The $f_\natural$ defining $\Wsv$ is chosen such that $\Wsv$ commutes nicely through the microlocal wave operator for all $r$. In particular, general second derivatives in the commutator $[\Pam, \Wsv]$ vanish, with the exception of a term which will produce an additional positive definite energy quantity. This regime includes all trapped frequencies.
\item\label{def:WscRegion2} For $-\ximin m^2 < m\omega <-\frac34\ximin m^2 $ or $\frac34\ximax \leq m\omega \leq m\omega_+ + \ximax m^2$, $\Wsl = \Wsv$ and $W= \wchi_1\Wsv + (1-\wchi_1)\chi_\sharp\Wsv$. In this regime, the commutator $[\Pam, \Wsc]$ contains error terms containing second derivatives, which arise from differentiating $\wchi_1 + \chisharp(1-\wchi_1)$ in the commutation. By construction, however, these appear only in the support of $\chisharp'$, where we have coercivity by a Lagrangian estimate (see Section \ref{sec:LocalLagrangian}).
\item\label{def:WscRegion3} For $-\frac34\ximin m^2 < m\omega <-\frac12\ximin m^2 $ or $m\omega_+ +\frac12\ximax \leq m\omega \leq m\omega_+ + \frac34\ximax m^2$, $\wchi_1 = 0$ and hence $W=\chi_\sharp W^\sharp$. The $h_\sharp$ defining the operator $\Wsl$ is chosen such that $\Wsl$ interpolates between $\Wsv$ and a local operator $\Wp$ defined in \eqref{def:Wplus} \emph{in such a way that general second order derivatives are avoided in the commutator outside the support of $\chi_\sharp^\prime$}. 
\item\label{def:WscRegion4} For $-\frac12\ximin m^2 \leq m\omega \leq m\omega_+ +\frac12\ximax$, $\wchi_1 = 0$ and $W = \chi_\sharp\Wp$. This regime includes all superradiant frequencies, and may be handled directly using Proposition \ref{prop:SuperradiantPositivity}. Again errors only appear where $\chi_\sharp^\prime$ is non-trivial and a coercive Lagrangian estimate is available for general second derivatives.
\end{enumerate}
	\subsubsection{The structure of the commuted wave equation}\label{sec:WECommuted}
		For a generic function $\ph$, we write the commutation identity
\begin{equation}
\Pam\Wsc\ph = v^2\Wsc(v^{-2}\Pam\ph) + v^2\left([v^{-2}\Pam, \Wsc]\ph\right).
\end{equation}
If $[\Wsc, \Ltheta]= 0$, then the presence of $v^{-2}$ ensures that no problematic derivatives in $\theta$ appear in the commutator. For a function $\ph$ satisfying
\begin{equation}\label{eq:WaveFreqGeneric}
\Pam\ph = H,
\end{equation}
and an operator $\Ws$ of the form
\begin{equation}
\Ws = v^{-1}\left(\Rs - i\omega f\right),
\end{equation}
where $f$ is a function of $(\omega, m, \rs)$ only, we have the commutation identity
\begin{equation}\label{def:CommIdentityBasic}
\Pam \Ws\ph = -2i\omega v f' W^\star\ph + \Ph^\star\ph + \Pl^\star\ph + \PH^\star,
\end{equation}
where $f' := \Rs f$, and
\begin{subequations}
\begin{align}
\Ph^\star &= v\left(\omega^2 f^2 - \tfrac{(\omega-\omega_r)^2}{v^2}\right)', \\
\Pl^\star &= -i\omega v\Rs\left(v^{-2}\Rs(v f)\right) - v\left(\tfrac{v'}{v^2}\right)'\Ws + v\left(\tfrac{V_1}{v^2}\right)', \\
\PH^\star &= v^2 W^\star\left(\tfrac{H}{v^2}\right).
\end{align}
\end{subequations}
Multiplying $\Ws$ by a cutoff function $\chi = \chi(r)$ gives the cut-off commutator identity
\begin{equation}\label{def:CommIdentityBasic2}
\Pam (\chi W^\star\ph) = -2i\omega v \chi f' W^\star\ph + \chi \Ph^\star\ph + \chi \Pl^\star\ph + \chi \PH^\star + \PX\Ws\ph,
\end{equation}
where 
\begin{equation}
\PX = \chi'' + 2\chi' \Rs.
\end{equation}
Commuting $\Wsc$ as defined in \eqref{def:CommutatorAbstract} through the wave equation using \eqref{def:CommIdentityBasic}, \eqref{def:CommIdentityBasic2} gives the (global) expansion
\begin{align}\label{def:CommutedBoxSimp}
\Pam(\Wsc\ph) & = -2i\omega v\wchi_1 \f' \Wsv\ph - 2i\omega v (1-\wchi_1)\chisharp f_\sharp'\Wsl\ph + \wchi_1\Ph^\natural\ph + (1-\chi_1)\chisharp\Ph^\sharp\ph + \\
&\qquad+ \wchi_1\Pl^\natural\ph + (1-\wchi_1)\chisharp\Pl^\sharp\ph +\PH + (1-\wchi_1)\mathcal{P}_{\chisharp}\Wsl\ph\nonumber \, , 
\end{align}
where
\begin{subequations}\label{def:CommutedTerms}
\begin{align}
\index{PhNatural@$\Ph^\natural$}\Ph^\natural & = v(\omega^2\f^2-\tfrac{(\omega-\omega_r)^2}{v^2})',\label{def:Phnatural}\\
\index{PlNatural@$\Pl^\natural$}\Pl^\natural &= -i\omega \f'' - v(\tfrac{v'}{v^2})'(\Wsv + i\omega \f) + v(\tfrac{V_1}{v^2})',\label{def:Pl}\\
\index{PhSharp@$\Ph^\sharp$}\Ph^\sharp &= v\left(\omega^2\fsharp^2 - \tfrac{(\omega-\omega_r)^2}{v^2}\right)', \\
\index{PlSharp@$\Pl^\sharp$}\Pl^\sharp &=-i\omega f_\sharp'' - v\left(\tfrac{v'}{v^2}\right)(\Wsl + i\omega f_\sharp)+ v(\tfrac{V_1}{v^2})', \label{def:Plsharp} \\
\index{PchiSharp@$\mathcal{P}_{\chisharp}$}\mathcal{P}_{\chisharp} &= 2\chisharp'\Rs  + \chisharp'', \label{Pchisubsharp}\\
\index{PH@$\PH$}\PH & = v^2\Wsc\left(\tfrac{H}{v^2}\right). \label{PHdef}
\end{align}
\end{subequations}
In the following two subsections we will define $\f$ and $\fsharp$ respectively, so that their derived quantities appearing in \eqref{def:CommutedBoxSimp} have certain desired properties, which will be summarised in Section \ref{sec:GPrimeBounds}. 

We note already that it will follow from  \eqref{def:f2} and \eqref{def:hsharp} that $\f =\fsharp$ when $0< \wchi_1 < 1$, so that we can write
\begin{equation} \label{freli}
-2i\omega v\wchi_1 \f' \Wsv\ph - 2i\omega v (1-\wchi_1)\chisharp f_\sharp'\Wsl\ph = -2i\omega vf'\Wsc\ph.
\end{equation}
for the global function
\begin{equation}\label{def:fcombined}\index{f@$f$} 
f = \wchi_1\f + (1-\wchi_1)\fsharp \, ,
\end{equation}
which simplifies the expression in the first line of (\ref{def:CommutedBoxSimp}).

\subsubsection{Defining $\f$ (or equivalently $h_\natural$)}\label{sec:CommDef}
For $(\omega, m)$ such that both $\omega \neq 0$ and $m\omega \notin (0, m\omega_+]$ (i.e.~including the set $\mathcal{G}^\prime \setminus \{(0,0)\}$), we define
\begin{align}\label{def:f2}
\f\index{feta@$\f$} = \f(\omega,m,r) &:= \begin{cases}
-\left(\frac{(\omega-\omega_r)^2}{\omega^2 v^2} - \inf_{r\in (r_+, \infty)} \frac{(\omega-\omega_r)^2}{\omega^2 v^2}\right)^{1/2} & r< \re,\\
\left(\frac{(\omega-\omega_r)^2}{\omega^2 v^2} - \inf_{r\in (r_+, \infty)} \frac{(\omega-\omega_r)^2}{\omega^2 v^2}\right)^{1/2} & r\geq \re,
\end{cases} \\
\h\index{heta@$\h$} = \h(\omega,m,r) &:= \begin{cases}
-\left(1 - \frac{v^2}{(\omega-\omega_r)^2}\inf_{r\in (r_+, \infty)} \frac{(\omega-\omega_r)^2}{v^2}\right)^{1/2} & r < \re,\\
\left(1 - \frac{v^2}{(\omega-\omega_r)^2}\inf_{r\in (r_+, \infty)} \frac{(\omega-\omega_r)^2}{v^2}\right)^{1/2} & r \geq \re,
\end{cases}
\end{align}
where $\re$ was defined in (\ref{def:rnatural}). We define a frequency-space version of the field $\oW$ defined in (\ref{W0def}), overloading notation slightly, via the formula
\begin{equation}\label{def:oWfreq}
(\oW \ph)_m^\omega = v^{-1} (\Rs\ph_m^\omega  -i(\omega-\omega_r)\ph_m^\omega).
\end{equation}
One easily sees that $(\oW\ph)_m^\omega = (\mathcal{F}(\oW u))_m^\omega$, and that $(\Wsv\ph)_0^\omega = (\oW\ph)_0^\omega$. With the above definition one has 
\begin{equation}\label{def:fsquaredabstract}
\f^2 = \frac{(\omega-\omega_r)^2}{\omega^2v^2}-\inf_{r\in (r_+, \infty)}\left(\frac{(\omega-\omega_r)^2}{\omega^2v^2}\right) \, ,
\end{equation}
which is a smooth function in $r$.
Since the last term is a (frequency dependent) constant, this clearly implies $\Ph^\natural = 0$. Next note $\f^2$ goes to $\infty$ as $r\to r_+$ and $r\to\infty$. By Proposition \ref{est:rnaturalbounds}, for fixed frequencies the function $\f^2$ assumes a unique global minimum at $\re = \re(\omega,m) \in (r_+,\infty)$, which by definition is zero. It follows that $\f$ as defined above is smooth in $r \in (r_+,\infty)$, monotonically increasing ($\f^\prime >0$) and vanishing linearly at $\re$. 
%





\subsubsection{Defining $f_\sharp$ (or equivalently $h_\sharp$)}\label{sec:IntersectionCommutator}
We first define the function $h_+=1$
and the corresponding (physical space) operator
\begin{equation}\label{def:Wplus}
W_+ = v^{-1}\Rs - i\frac{(\omega-\omega_r)}{v}. \index{Wplus@$\Wp$} 
\end{equation}
Note that $W_+$ is not regular at the horizon.

We then define the function $h_\sharp=h_\sharp(r, \omega,m)$ for frequencies in $\mathcal{G}$ by:
\begin{align} \label{def:hsharp}\index{hsharp@$\hsharp$}
h_\sharp = 
\begin{cases}
h_+=1 & \textrm{for $(\omega,m) \in \mathcal{G}\cap (\mathcal{G}^\prime)^c$}, \\
\left(1 - \widetilde{\chi}_2 \frac{v^2}{(\omega-\omega_r)^2}\inf_{r\in (r_+, \infty)} \frac{(\omega-\omega_r)^2}{v^2}\right)^{1/2}& \textrm{for $(\omega,m) \in \mathcal{G} \cap \mathcal{G}^\prime$ where $\widetilde{\chi}_2(\frac{\omega}{m}) < 1$}, \\
h_\natural & \textrm{for $(\omega,m) \in \mathcal{G}\cap \mathcal{G}^\prime$ where $\widetilde{\chi}_2(\frac{\omega}{m}) =1$ }.
\end{cases}
\end{align}
Recall also that $\widetilde{\chi}_2(\frac{\omega}{m})=0$ in $\mathcal{G}\cap (\mathcal{G}^\prime)^c$.
By defining $h_\sharp$ in this way we have that $\left(\tfrac{\omega-\omega_r}{v}\right)^2 - \left(\tfrac{(\omega-\omega_r)h_\sharp}{v}\right)^2$ is constant in $r$ for fixed frequencies $(\omega,m)$, which ensures that $P^\sharp_h \hat{u}=0$ for the commutator and hence the avoidance of general second derivatives.

The function $f_\sharp$ is now also defined from $h_\sharp$ by the relation (\ref{Wsharprel}), at least if $\omega \neq 0$. (For $\omega = 0$, we must define $\Wsl$ using $\hsharp$.)
Note that the only way $W^\sharp$ enters the definition of $W$ in (\ref{def:CommutatorAbstract}) is with the cut-off $\chi_\sharp$ and that $\chi_\sharp W^\sharp$ vanishes in a neighbourhood $[r_+, \rmin)$ of the horizon.




	\subsection{Basic properties of the commutator}\label{sec:GPrimeBounds}

\subsubsection{Commutation properties of $W^\natural$}
We summarise the main properties of $\f$ and the corresponding vector field $W^\natural$ in the following proposition:
\begin{proposition}[Properties of the commutator $\Wsv$]\label{thm:WProperties}
Given spacetime parameters $a, M$, with $0 \leq a < M$, and frequency parameters $(\omega, m)\in \mathbb{R}\times\mathbb{Z}$, with
\begin{equation}\label{def:trappedext}
(\omega, m) \in \mathcal{G}',
\end{equation}
 the functions $\f(\omega, m, \rs)$ and $\h(\omega, m, \rs)$ defined in (\ref{def:f2}), \eqref{Wnaturalrel}, are homogeneous of degree 0 in $(\omega, m)$ and satisfy the following properties:
 \begin{enumerate}[(1)]
 \item\label{fNaturalCondition1} {\bf Lower bounds on commutator term}\label{Thm:flowerbound} Recalling
$
\Wsv = v^{-1}\Rs - i\omega\f,
$
we have that the bound
\begin{equation}\label{est:fprimebound}
\omega^2v \f' \geq \frac{c_0}{r}\omega(\omega-\omega_r)
\end{equation}
 holds for all $(\omega, m) \in \mathcal{G}'$, where
\[
c_0 = \tfrac{r_+^2}{r_+^2+a^2}\left(\tfrac{r_+-M}{r_+}\right)^3.
\]
\item\label{fNaturalCondition2} {\bf General second derivatives vanish in the commutator} The identity
\begin{equation}
\Ph^\natural \ph = 0
\end{equation}
holds for all $r$.
 \item\label{fNaturalCondition3} {\bf Bounds on lower order terms}\label{Thm:Plupperbound} Under the separation \eqref{def:CommutedBoxSimp}, one has the pointwise bound
 \begin{equation}\label{est:Plupperbound}
|\Pl^\natural\ph |\leq C\big( r^{-3}|\Wsv\ph| +  v^{3}(|\omega\ph| + |\ph|)\big) \, .
 \end{equation}
 \end{enumerate}
\end{proposition}
We prove only the first two statements for now and postpone the proof of Statement \eqref{fNaturalCondition3} to Section \ref{sec:fnatUpperBound}, as it requires the use of some technical results which in turn use the bound \eqref{est:fprimebound}. 
\begin{proof}[Proof of Statements \eqref{fNaturalCondition1} and \eqref{fNaturalCondition2}]
Condition \eqref{fNaturalCondition2} was already observed after \eqref{def:fsquaredabstract}. For the proof of Condition (\ref{fNaturalCondition1}) we recall the notation $\Gamma$ and the polynomial $P_\eta$ from (\ref{def:gamplus}) and (\ref{def:Peta}) respectively. From $P_\eta(\re) = 0$, we have writing $\Gz:=\Gamma(\re)$ and $\Dz:=\Delta (\re)$:
\begin{subequations}\label{id:difference}
\begin{align}
2\re\Dz &= (\re - M)\Gz \, , \label{MinimumId}\\
\Gamma - \Gz &=  (r-\re)^2 + 2\re(r-\re), \label{GamDif}\\
\Delta - \Dz &=  (r-\re)^2 + 2(\re - M)(r-\re)\label{DelDif}.
\end{align}
\end{subequations}
An expansion of \eqref{def:fsquaredabstract} gives
\begin{equation}\label{id:fsquared}
\f^2 = \frac{\Gamma^2}{\Delta} - \frac{\Gz^2}{\Dz}= \frac{(\Gamma - \Gz)(\Gamma\Dz + \Gz\Delta) - (\Delta-\Dz)\Gamma\Gz}{\Delta\Dz},
\end{equation}
and subsequently applying \eqref{GamDif}, \eqref{DelDif} gives
\begin{equation}
\f^2 = \frac{(r-\re)^2(\Gamma\Dz - \Gamma\Gz + \Delta\Gz)}{\Delta\Dz} + \frac{(r-\re)(2\re(\Gamma\Dz + \Gz\Delta) - 2(\re-M)\Gamma\Gz)}{\Delta\Dz}.
\end{equation}
Additional careful cancellations using \eqref{MinimumId} give
\begin{align}
2\re(\Gamma\Dz + \Gz\Delta) - 2(\re-M)\Gamma\Gz &= 2\re\Gz\Delta - (\re - M)\Gamma\Gz, \\
&= \Gz(P_\eta(\re) + 2\re(\Delta-\Dz) - (\re - M)(\Gamma-\Gz)).\nonumber
\end{align}
The identities \eqref{GamDif}, \eqref{DelDif} combined with $P_\eta(\re) = 0$ give
\begin{equation}
\f^2 = \frac{(r-\re)^2}{\Delta\Dz}\big(- \Gz^2 + 2\Dz\Gz + (\re+ r)^2\Dz\big).
\end{equation}
Then, \eqref{MinimumId} gives
\begin{equation}
- \Gz^2 + 2\Dz\Gz + (\re+ r)^2\Dz = \frac{(-4\re^2 + 4\re(\re-M))\Dz^2+ (\re-M)^2(\re+r)^2\Dz}{(\re - M)^2}.
\end{equation}
Consequently, we may write
\begin{equation}  \label{def:fetaTrue}
\f = \frac{r-\re}{(\re - M)\Delta^{1/2}}((\re-M)^2(\re+r)^2-4M\re\Dz)^{1/2} = \Delta^{-1/2}(r-\re)\ve \, , 
\end{equation}
where we have defined
\begin{equation}\label{def:veta}
\ve(r, \eta) = \left((r+\re)^2 - 4M\re\frac{\Dz}{(\re-M)^2}\right)^{1/2},
\end{equation}
where $\re, \Dz$ are to be taken as functions of $\eta$.
Since $\Dz \leq (\re - M)^2$, we may again use \eqref{MinimumId} to write
\begin{equation}
\ve^2 \geq (r+\re)^2 - 4M\re = (r-M)^2 + (\re - M)^2 + 2(r-M)(\re + M).
\end{equation}
Since $\frac{r-M}{r} \geq \frac{r_+-M}{r_+}$ for $r\in(r_+, \infty)$, it follows that
\begin{equation}\label{est:vebound}
\left(\tfrac{r_+-M}{r_+}\right)(r+\re) \leq \ve \leq r+\re.
\end{equation}
Additionally,  \eqref{basicineq} implies
\begin{equation} \label{est:vebound2}
\ve(r_+) \geq \tfrac{r_+-M}{r_+}\Gamma_+^{1/2}, \qquad (\re - r_+)\ve(r_+) \geq \tfrac{r_+-M}{r_+}(\re^2 - r_+^2)\geq \tfrac{r_+-M}{r_+} \Gamma_+.
\end{equation}
We now bound $v\Rs\f$, via bounds on $\partial_r \f$. We may expand
\begin{equation}\label{id:partialfe}
\partial_r\f = \frac{-(r-M)(r-\re)\ve^2 + \Delta\ve^2 + (r^2-\re^2)\Delta}{\Delta^{3/2}\ve}.
\end{equation}
In order to bound this below, we take the Taylor expansion at $r=r_+$. Writing
\begin{equation}
F_\natural =\Delta^{3/2}\ve(\partial_r\f) =  -(r-M)(r-\re)\ve^2 + \Delta(\ve^2 + r^2-\re^2),
\end{equation}
we see
\begin{subequations}
\begin{align}
\partial_r F_\natural &= (\re - M)\ve^2 + (4r+2\re)\Delta, \\
\partial_r^2 F_\natural &= 2(\re - M)(r+\re) + 4\Delta + (8r+4\re)(r-M), \\
\partial_r^3 F_\natural &= 2(\re - M)+ 16(r-M) + 8r+4\re, \\
\partial_r^4 F_\natural &= 24
\end{align}
\end{subequations}
and consequently
\begin{subequations}
\begin{align}
F_\natural(r_+) &= (r_+-M)(\re - r_+)(\ve^2(r_+)),\\
\partial_r F_\natural(r_+) &= (\re - M)\ve^2(r_+)\\
\partial_r^2 F_\natural(r_+) &=2(\re-M)(\re + r_+)+ (8r_+ + 4\re)(r_+-M), \\
\partial_r^3 F_\natural(r_+) &= 2(\re-M) + 16(r_+-M)+8r_++4\re, \\
\partial_r^4 F_\natural(r_+) &=24 \, .
\end{align}
\end{subequations}
In order to bound this below, we note that \eqref{est:vebound} implies
\begin{equation}
r\ve\Gamma \leq r(r+\re)(r^2-r_+^2+\Gp).
\end{equation}
Expanding in a polynomial in $(r-r_+)$ gives
\begin{align}
r(r+\re)(r^2-r_+^2+\Gp) &= (r-r_+)^4 + (\re + 4r_+)(r-r_+)^3 + (\Gp + 5r_+^2 + 3\re r_+)(r-r_+)^2 + \\
&\quad+ (2r_+^2\re + 2r_+^3 + \re\Gp + 2 r_+\Gp)(r-r_+) + \Gp r_+(\re+r_+).\nonumber
\end{align}
A term-by-term comparison along with the bounds \eqref{est:vebound},  \eqref{est:vebound2} and \eqref{basicineq}, and the inequality $r-M > \frac{r_+-M}{r_+}r$ implies
\begin{equation}
\partial_r\f \geq \left(\frac{r_+-M}{r_+}\right)^3\frac{r\Gamma}{\Delta^{3/2}} \, .
\end{equation}
The bound \eqref{est:fprimebound} directly follows.

\end{proof}

\subsubsection{Commutation properties of $W^\sharp$}
We summarise the main properties of $f_\sharp$ and the corresponding vector field $W^\sharp$ in the following proposition:

\begin{proposition}[Properties of the commutator $W^\sharp$]\label{thm:WnatProperties}
Given spacetime parameters $a, M$, with $0 \leq a < M$, and frequency parameters $(\omega, m)\in \mathbb{R}\times\mathbb{Z}$, with
\begin{equation}
(\omega, m) \in \mathcal{G},
\end{equation}
 the functions $f_\sharp(\omega, m, \rs)$ and $h_\sharp(\omega, m, \rs)$ defined in (\ref{def:hsharp}), \eqref{Wsharprel}, are homogeneous of degree 0 in $(\omega, m)$ and satisfy the following properties:
 \begin{enumerate}[(1)]
 \item\label{fSharpCondition1} {\bf  Lower bounds on commutator term} The inequality
\begin{equation}\label{est:hSharpDerivativeIneq}
\omega^2 v f_\sharp^\prime = \omega v \Rs\left(\tfrac{(\omega-\omega_r)\hsharp}{v}\right) \geq \omega^2 v \Rs\left(\tfrac{\omega-\omega_r}{v}\right)
\end{equation}
holds in the support of $\chisharp$.
 \item\label{fSharpCondition2}{\bf General second derivatives vanish in the commutator} The identity
\begin{equation} \label{sharpvanco}
\Ph^\sharp \ph = 0
\end{equation}
holds for all $r$.
 \item\label{fSharpCondition3} {\bf Bounds on lower order terms}\label{Thm:Plupperbound} Under the separation \eqref{def:CommutedBoxSimp}, one has the pointwise bound
\begin{equation}\label{est:PLSuperradBounds}
|\Pl^\sharp\ph|\lesssim r^{-3}|\Wsl\ph| + v^3(|m\ph|+|\ph|)
\end{equation}
in the support of $\chi_\sharp$.
 \end{enumerate}
\end{proposition}
As before, we prove Statements \eqref{fSharpCondition1} and \eqref{fSharpCondition2} here and postpone the proof of Statement \eqref{fSharpCondition3} until Section \ref{sec:fnatUpperBound}.
\begin{proof}[Proof of Statements \eqref{fSharpCondition1} and \eqref{fSharpCondition2}]
We start with (\ref{est:hSharpDerivativeIneq}). This is of course true for $\hsharp = 1$, so it suffices to prove this bound in the frequency regime $\mathcal{G}\cap\mathcal{G}'$, for which $\h$ and $\re$ are also defined. It follows from the definition of $h_\sharp$ and $\h$ that
\begin{equation}\label{id:fsquaredconstant}
\omega^2\Rs\left(\tfrac{((\omega-\omega_r)\h)^2}{v^2}\right) = 
\omega^2\Rs\left(\tfrac{((\omega-\omega_r)h_\sharp)^2}{v^2}\right) = 
\omega^2\Rs\left(\tfrac{(\omega-\omega_r)^2}{v^2}\right).
\end{equation}
Additionally, for $r \geq \re$ (and therefore in the support of $\chisharp$),
\begin{equation}
\Rs\left(\tfrac{(\omega-\omega_r)^2}{v^2}\right) \geq 0.
\end{equation}
From the definitions \eqref{def:f2} \eqref{def:hsharp}, one can see that $0 < \h \leq \hsharp \leq 1$ in the support of $\chisharp$. Additionally, $\omega(\omega-\omega_r) > 0$ for $(\omega, m)\in \mathcal{G}'\setminus(0,0)$. It follows that
\begin{equation}\label{est:fsquaredmonotone}
0 < \omega\left(\tfrac{(\omega-\omega_r)\h}{v}\right) \leq \omega\left(\tfrac{(\omega-\omega_r)h_\sharp}{v}\right) \leq \omega\left(\tfrac{\omega-\omega_r}{v}\right).
\end{equation}
Expanding \eqref{id:fsquaredconstant} and applying the inequality \eqref{est:fsquaredmonotone} gives \eqref{est:hSharpDerivativeIneq}. We remark that, for $(\omega, m)\in \mathcal{G}\cap\mathcal{G}'$, we also have the upper bound
\begin{equation}\label{est:hSharpSandwich}
\omega\Rs\left(\tfrac{(\omega-\omega_r)\hsharp}{v}\right) \leq \omega\Rs\left(\tfrac{(\omega-\omega_r)\h}{v}\right) \,
\end{equation}
in the support of $\chisharp$.

The condition (\ref{sharpvanco}) was shown immediately after the definition (\ref{def:hsharp}). 
\end{proof}

\subsection{The regularity of $\f$ and $f_\sharp$}\label{sec:fbasicregularity}
The following proposition summarises the basic regularity statements for $\f=\f(r, \eta)$ and $f_\sharp(r,\eta)$, both in $r$ and with respect to the frequency parameter $\eta=\frac{am}{\omega}$. We will only need these statements for $k\leq 2$.

\begin{proposition} \label{prop:vfDerUpperBound}
Recall $h_0$ from (\ref{def:h0}). For any $k\geq 0$ there exists a constant $C=C(a, M, \ximin, \ximax,k)$ such that we have 
\begin{enumerate}
\item For frequencies in $\mathcal{G}^\prime \setminus \{(0,0)\}$ the bounds
\begin{alignat}{4}
|\partial_r^k(v\f)| &\leq C r^{-2-k} \label{est:drf} \qquad&  \textrm{for } k& \geq 1, \\
\left|\partial_\eta^k \f\right| &\leq C \frac{1}{\Delta^{1/2}} \qquad&   \textrm{for } k& \geq 1,  \label{est:detafPrelim} \\
\left|\partial_\eta^{k} \left(\f - \frac{r^2+a^2-\eta}{\Delta^{1/2}}h_0\right)\right|\ &\leq C \frac{\Delta^{1/2}}{r^2+a^2} \qquad&   \textrm{for } k& \geq 0, \label{est:detafFinal} \\
\left|\partial_{\rs}^{k} \left(\f - \frac{r^2+a^2-\eta}{\Delta^{1/2}}h_0\right)\right|\ &\leq C \frac{\Delta^{1/2}}{r^{2+k}} \qquad&   \textrm{for } k& \geq 0. \label{est:drfFinal} 
\end{alignat}
\item In the support of $\chi_\sharp$ for frequencies in $\mathcal{G}$ the bounds\footnote{Note that weights in $\Delta$ are irrelevant in the support of $\chi_\sharp$, as $\chi_\sharp$ is supported uniformly away from the horizon.}
\begin{alignat}{4}
|\partial_r^k(vf_\sharp )| &\leq C r^{-2-k}|\eta| \qquad&   \textrm{for } k& \geq 1,\label{drf2}\\
\left|\partial_\eta^k f_\sharp \right| &\leq C r^{-1} \qquad&   \textrm{for } k& \geq 1,\label{est:detafPrelim2} \\
\left|\partial_\eta^{k} \left(f_\sharp - \frac{r^2+a^2-\eta}{\Delta^{1/2}}h_0\right)\right|\ &\leq C r^{-1} \qquad&   \textrm{for } k& \geq 0,\label{est:detafFinal2} \\
\left|\partial_{\rs}^{k} \left(f_\sharp - \frac{r^2+a^2-\eta}{\Delta^{1/2}}h_0\right)\right|\ &\leq C r^{-k-1}|\eta| \qquad&   \textrm{for } k& \geq 0.\label{est:drfFinal2} 
\end{alignat}
\end{enumerate}
\end{proposition}

We begin with the bounds for $\f$. We recall that in \eqref{def:fetaTrue} we already derived the algebraic relation 
\[
\f = \Delta^{-1/2}(r-\re)\ve \, ,
\]
with $\ve$ defined in (\ref{def:veta}). Establishing regularity for $\f$ therefore boils down to understanding regularity of $\re$ and $\ve$. We begin with $\re$ recalling again the notation $\Gamma$ and the polynomial $P_\eta$ from (\ref{def:gamplus}) and (\ref{def:Peta}).

\subsubsection{Regularity of $\re$}

We now show (Corollary \ref{prop:rebounds} below) that $\re$ is contained in a compact subset of $(r_+,\infty)$. In Proposition \ref{prop:rderiv1} we obtain the desired uniform bounds on derivatives of $\re$. We start with two preliminary bounds in the relevant frequency regime.

\begin{lemma}[Absolute bounds for $\eta$]\label{prop:GpBounds}
There exist quantities $\etamin, \etamax$ depending on $\ximin, \ximax$ such that
\begin{equation}\label{est:etaBounds}
-\infty < \etamin < \eta = \frac{am}{\omega} < \etamax < r_+^2 + a^2 \, 
\end{equation}
holds for $(\omega, m) \in \mathcal{G}^\prime \setminus \{(0,0)\}$.
Therefore, there exists a $C>0$ depending on $\ximin, \ximax$ such that for these frequencies
\begin{equation}\label{est:GpBounds}
C^{-1} \leq \Gamma_+ = r_+^2+a^2- \eta \leq C \, .
\end{equation}
\end{lemma}
\begin{proof} For $m\omega > m\omega_+ + \frac12\ximax m^2$ (with $a>0$),
\[
0 < \eta < \left(\frac{1}{ r_+^2 + a^2}+\frac{\ximax}{2a}\right)^{-1} <  r_+^2 + a^2.
\]
Likewise, when $\eta < 0$, similar calculations give
$
\eta > -\frac{2a}{\ximin}.
$
The bound \eqref{est:etaBounds} directly follows. The bound \eqref{est:GpBounds} then follows from the definition of $\Gamma_+$.
\end{proof}

\begin{corollary}[Absolute bounds for $\omega-\omega_r$]\label{prop:TrapInterior}
There exists a constant $C > 0$ depending on $\ximin, \ximax$ such that, for all $(\omega, m)\in \mathcal{G}' \setminus \{(0,0)\}$ and for all $r \in (r_+, \infty)$,
\begin{equation}
C^{-1}|\omega|\leq |\omega-\omega_r|\leq C|\omega|.
\end{equation}
\end{corollary}
\begin{proof}
We expand
\begin{equation}
\frac{\omega-\omega_r}{\omega} = 1-\frac{\eta}{r^2+a^2}.
\end{equation}
The result follows directly from Lemma \ref{prop:GpBounds}
\end{proof}

The bounds \eqref{basicineq} and Lemma \ref{prop:GpBounds} imply the following:
\begin{corollary}[Compact support of $\re$ for frequencies in $\mathcal{G}' \setminus \{(0,0)\}$]\label{prop:rebounds}
For all $\ximin, \ximax > 0$, there exists a compact subset $I^1_{trap} \subset (r_+, \infty)$ depending on $\ximin, \ximax$ such that $\re(\omega, m) \in I^1_{trap}$ for all $(\omega, m) \in\mathcal{G}' \setminus \{(0,0)\}$.
\end{corollary}

We now estimate derivatives of $\re$ with respect to $\eta$. We note that $(\omega, m)\in \mathcal{G}' \setminus \{(0,0)\}$, $\re$ is smooth and decreasing in $\eta$ and quantify this with the following bound:
\begin{proposition}[Regularity of $\re$]\label{prop:rderiv1}
Given $(\omega, m) \in \mathcal{G}' \setminus \{(0,0)\}$, 
there exists a constant $C > 0$ depending on $a, M, \ximin, \ximax$  such that, viewing $\re$ as a function of $\eta$, the bound
\begin{equation}\label{est:rderiv2}
-C \leq \partial_\eta\re \leq -C^{-1}
\end{equation}
holds.
Additionally, for any $k > 0$, there exists a constant $C_k$ such that
\begin{equation}\label{est:rderiv3}
|\partial^{k}_\eta \re| \leq C_k \, .
\end{equation}
\end{proposition}
\begin{proof}
We recall from (\ref{def:Peta}) that $P_\eta(\re) = 0$. Treating $\re$ as an implicit function of $\eta$, 
\begin{equation}\label{rderivimplicit}
(\partial_r P_\eta)(\re) \partial_\eta(\re) + (\re - M) = 0.
\end{equation}
It follows that
\begin{equation}\label{id:rimplicitderivative}
\partial_\eta(\re) = -\frac{\re - M}{(\partial_r P_\eta)(\re)} = \frac{\re - M}{3\re^2 - 6M\re + (a^2+\eta)\re}.
\end{equation}
Since $\partial_r^2 P_\eta > 0$, the mean value theorem combined with the identity $P_\eta(\re) = 0$ and Proposition \ref{est:rnaturalbounds} imply
\begin{equation}\label{est:dPboundlower}
(\partial_r P_\eta)(\re) \geq \frac{-P_\eta(r_+)}{\re - r_+} = \frac{(\re - M)\Gamma_+}{\re - r_+} \geq (\re - M)\Gamma_+^{1/2}.
\end{equation}
Additionally, \eqref{MinimumId} implies the relation $2\Dz < \Gz$, so
\begin{equation}\label{est:dPboundupper}
(\partial_r P_\eta)(\re) = 2\Gz + 2\re(\re-M) - \Gz \leq 2\re(\re-M).
\end{equation}
Therefore, \eqref{id:rimplicitderivative} implies
\begin{equation}
-\frac{1}{\Gp^{1/2}} \leq \partial_\eta(\re) \leq -\frac{1}{2\re}.
\end{equation}
The inequality \eqref{est:rderiv2} then follows from Proposition \ref{prop:GpBounds} and Corollary \ref{prop:rebounds}.

 The bound \eqref{est:rderiv3} follows inductively from the applying the chain rule to \eqref{id:rimplicitderivative} via the characterization
 \begin{equation}
 \partial_\eta^k\re = \frac{P_{k}(\re(\eta), \eta)}{[(\partial_r P_\eta)(\re)]^{2k+1}}
 \end{equation}
for some set of polynomials $\{P_k\}$. Then, $P_k$ may be bounded using Proposition \ref{prop:GpBounds} and Corollary \ref{prop:rebounds}, and $(\partial_r P_\eta)^{2k+1}$ may be bounded using \eqref{est:dPboundlower}, \eqref{est:dPboundupper}.
\end{proof}

\subsubsection{Regularity of $\ve$}
We next prove a regularity result for $\ve$ in both physical and frequency space. For $(r,s)\in (r_+, \infty)\times I^1_{trap}$, we define
\begin{equation}\label{def:gamma}\index{gamma@$\gamma$}
\gamma(r, s) = \left((r+s)^2 - 4Ms\frac{s^2 - 2Ms + a^2}{(s-M)^2}\right)^{1/2},
\end{equation}
so that, for $\ve$ defined in \eqref{def:veta},
$\ve(r, \eta) = \gamma(r, \re(\eta))$.

\begin{proposition}[Regularity of $\ve$]\label{prop:drGamma}
For $\gamma$ as defined in \eqref{def:gamma}, and $(r, s)\in (r_+, \infty)\times I^1_{trap}$, $\gamma$ satisfies the uniform bounds
\begin{subequations}
\begin{align}
\label{est:dreGamma}|\partial_{s}^k\gamma| &\leq\begin{cases}
C_k & k = 1, \\
C_k r^{-1} & k \geq 2,
\end{cases} \\
\label{est:drGamma}|\partial_r^k \gamma| &\leq C_k r^{1-k},
\end{align}
\end{subequations}
for $k \geq 1$, where $C_k$ is a constant depending on $a, M, k$, as well as bounds on $I^1_{trap}$.
\end{proposition}
\begin{proof}
In each case, this follows from writing
\begin{equation}
\partial_{s}\gamma = \frac{\partial_{s}(\gamma^2)}{2\gamma}= \frac{r+Q^1(s)}{\gamma}, \qquad \partial_r\gamma = \frac{\partial_r(\gamma^2)}{2\gamma} = \frac{r+Q^2(s)}{\gamma},
\end{equation}
where $Q_1, Q_2$ are rational functions which are defined everywhere in $[r_+, \infty)$. The result in each case follows from an inductive approach similar to the proof of the bound \eqref{est:rderiv3}. To prove \eqref{est:dreGamma} it suffices to show that derivatives in $s$ may be written in the form
\begin{equation}
\partial_{s}^k\gamma = \sum_{j=2}^k\sum_{i=2}^j\frac{Q^1_{i,j,k}(s)r^{j}}{\gamma^{2j-1}} + \sum_{j=1}^k\sum_{i=1}^j \frac{\widetilde{Q}^1_{i,j,k}(s) r^{i-1}}{\gamma^{2j-1}},
\end{equation}
for $k \geq 2$, where $Q^1_{i,j,k}(s)$, $\widetilde{Q}^1_{i,j,k}(s)$ are rational functions which are continuous for $s \in [r_+, \infty)$ and therefore bounded for $s \in I^1_{trap}$. Similarly, to prove \eqref{est:drGamma} we write
\begin{equation}
\partial_r^k\gamma = \sum_{0\leq i \leq j \leq k}\frac{Q^2_{i,j,k}(s)r^{i}}{\gamma^{k+j-1}},
\end{equation}
where $Q^2_{i,j,k}(s)$ are finite rational functions which are also bounded in magnitude for $s \in I^1_{trap}$. Both results then follow from a straightforward inductive argument.
\end{proof}
\begin{proposition}[Asymptotic bounds on $\ve$ for frequencies in $\mathcal{G}^\prime \setminus \{(0,0)\}$]\label{prop:vetabounds}
Writing $\ve = \ve(r, \eta)$ as defined in \eqref{def:veta}, $\ve$ satisfies the uniform bounds
\begin{equation}
|\partial_r^k(\ve - (r+\re))|\leq C_k r^{-1-k}
\end{equation}
and 
\begin{equation}
|\partial_\eta^k(\ve - (r+\re))|\leq C_k r^{-1}
\end{equation}
\end{proposition}
\begin{proof}
In each case this is a straightforward induction proof using the identity
\begin{equation}\label{id:gammaAsymptotic}
\ve - (r+\re) = \frac{-4M\re\tfrac{\Dz}{(\re - M)^2}}{\ve + (r+\re)},
\end{equation}
along with Propositions \ref{prop:rebounds}, \ref{prop:rderiv1}, and \ref{prop:drGamma}.
\end{proof}


\subsubsection{Concluding the proof of Proposition \ref{prop:vfDerUpperBound} for $\f$}\label{sec:fnatUpperBound}
From \eqref{def:fetaTrue} one easily obtains
\begin{equation}
v\f = \frac{r^2-\re^2}{r^2+a^2} + \frac{(r-\re)(\ve-(r+\re))}{r^2+a^2}.
\end{equation}
The estimates (\ref{est:drf}) and \eqref{est:detafPrelim} now follow directly from elementary calculus using the bounds from Corollary \ref{prop:rebounds}, Proposition \ref{prop:rderiv1} and Proposition \ref{prop:vetabounds}.

To prove \eqref{est:detafFinal}, we note that when $r \in I^1_{trap}$, the result follows from \ref{est:detafPrelim}, along with compactness of $I^1_{trap}$. For $r \notin I^1_{trap}$, we take the expansion
\begin{equation}\label{id:fnatminusf0}
\f - \frac{\Gamma}{\Delta^{1/2}}h_0 = \frac{(1-h_0^2)\tfrac{\Gamma^2}{\Delta} - \tfrac{\Gz^2}{\Dz}}{\f + \tfrac{\Gamma}{\Delta^{1/2}}h_0} = v\left(\frac{v_0^{-2}\left(1-\tfrac{\eta}{r^2+a^2}\right)^2 - \tfrac{\Gz^2}{\Dz}}{v\f + \left(1-\tfrac{\eta}{r^2+a^2}\right)h_0}\right).
\end{equation}
It follows from Corollary \ref{prop:rebounds} that $r_0$, the point at which $h_0$ vanishes, is in the interior of $I^1_{trap}$. Since $h_0' > 0$, it follows from Lemma \ref{prop:GpBounds} that there exists a constant $c_1>0$ depending on $a, M, \ximin, \ximax$ such that, for $r\notin I^1_{trap}$, 
\begin{equation}
\left|\left(1-\tfrac{\eta}{r^2+a^2}\right)h_0\right|\geq c_1
\end{equation}
Since $v\f$ and $h_0$ have the same sign outside $I^1_{trap}$,
\begin{equation}
\left|v\f + \left(1-\tfrac{\eta}{r^2+a^2}\right)h_0\right| \geq c_1 \, .
\end{equation}
Additionally, compact support of $I^1_{trap}$ implies that for some $c_2>0$ depending on $a, M, \ximin, \ximax$,
\begin{equation}
\Dz > c_2.
\end{equation}
An inductive argument using Corollary \ref{prop:rebounds}, Proposition \ref{prop:rderiv1}, Proposition \ref{prop:vetabounds}, and the bound \eqref{est:detafPrelim} implies
\begin{equation}
\left|\partial_\eta^k\left(\frac{v_0^{-2}\left(1-\tfrac{\eta}{r^2+a^2}\right)^2 - \tfrac{\Gz^2}{\Dz}}{v\f + \left(1-\tfrac{\eta}{r^2+a^2}\right)h_0}\right)\right|\leq C_k \, .
\end{equation}
To prove \eqref{est:drfFinal2}, we again restrict to the exterior of $I^1_{trap}$, expand \eqref{id:fnatminusf0} and apply the bounds
\begin{align}
|\partial_{\rs}^k v| &\lesssim \frac{\Delta^{1/2}}{r^{2+k}},\label{est:drsv} \\
\big|\partial_{\rs}^k \big(v_0^{-2}\big(1-\tfrac{\eta}{r^2+a^2}\big)^2 - \tfrac{\Gz^2}{\Dz}\big)\big| & \lesssim r^{-2- k}, \\
\big|\partial_{\rs}^k \big(v\f + \big(1-\tfrac{\eta}{r^2+a^2}\big)h_0\big)\big| &\lesssim r^{-2-k},
\end{align}
for $k \geq 1$, the latter two of which follow from \eqref{est:drf}.

\subsubsection{Concluding the proof of Proposition \ref{prop:vfDerUpperBound} for $f_\sharp$}

We first note that, if $\hsharp = 1$, the bounds \eqref{drf2}-\eqref{est:drfFinal2} directly follow from the identity
\begin{equation}
\fsharp = \frac{1}{v}\left(1-\frac{\eta}{r^2+a^2}\right) = \frac{r^2+a^2-\eta}{\Delta^{1/2}}
\end{equation}
and the inequality
\begin{equation}
|1-h_0|\lesssim Cr^{-2}
\end{equation}
away from the horizon. Additionally, for $\hsharp = \h$, the bounds \eqref{drf2}-\eqref{est:drfFinal2} follow from the corresponding bounds \eqref{est:drf}-\eqref{est:drfFinal}, as Lemma \ref{prop:GpBounds} implies that $|\eta|$ is bounded above and below in $\mathcal{G}\cap\mathcal{G}'$. It therefore suffices to prove the result for the case $0 < \wchi_2 < 1$. 

In the support of $\chisharp$ and for frequencies in $\mathcal{G}\cap\mathcal{G}'$, $\f$ and $\fsharp$ may be uniformly bounded below by a constant $c>0$ which is independent of frequency via the inequality
\begin{equation}
\fsharp \geq \f \geq \f\big|_{r = \re + \tfrac14\delta^\sharp} \geq \tfrac14\delta^\sharp\left(\inf_{r \leq \rmax+\frac14\delta^\sharp}\partial_r \f \right) \geq c,
\end{equation}
which we get by integrating $\f'$ from $\re$, noting that $\f$ vanishes at $\re$ and that \eqref{est:fprimebound} along with Lemma \ref{prop:GpBounds} gives an absolute lower bound on $\partial_r\f$. It follows that there exists a constant $C > 0$ such that
\begin{equation}
C^{-1} \leq v\f \leq C, \qquad C^{-1}\leq v\fsharp \leq C
\end{equation}
in the support of $\chisharp$ and for frequencies in $\mathcal{G}\cap\mathcal{G}'$.

 Consequently, \eqref{drf2} follows from an inductive argument using the expansion
\begin{equation}
(v\fsharp)' = \frac{v\f}{v\fsharp}(v\f)' + \frac12\frac{1}{v\fsharp}(v^2C_\eta)',
\end{equation}
and the estimate \eqref{est:drf}, where 
\begin{equation}
C_\eta = (1-\wchi_2)\inf_{r\in (r_+, \infty)}\left(\frac{(\omega-\omega_r)^2}{\omega^2v^2}\right), \qquad \fsharp^2 = \f^2 + C_\eta.
\end{equation}

The bound \eqref{est:detafPrelim2} follows from a similar argument using the expansion
\begin{equation}
\partial_\eta \fsharp = \frac{v\f}{v\fsharp}\partial_\eta\f + \frac{v}{2v\fsharp}\partial_\eta C_\eta,
\end{equation}
using the fact that $C_\eta$ may be written as a smooth function of $\eta$ and $\re$, and that Lemma \ref{prop:GpBounds} and Proposition \ref{basicineq} ensure that $\eta$ and therefore $\re$ are taken in a bounded domain, so in particular $C_\eta$ and its derivatives are globally bounded.

For $k \geq 1$, the estimate \eqref{est:detafFinal2} follows immediately from \eqref{est:detafPrelim2}. For $k=0$, \eqref{est:detafFinal2} follows from \eqref{est:detafFinal} and the bound
\begin{equation}\label{est:fsharpminusfexpansion}
\fsharp - \f = \frac{vC_\eta}{v(\fsharp+\f)} \lesssim v.
\end{equation}

Finally, the estimate \eqref{est:drfFinal2} follows from \eqref{est:drfFinal} and \eqref{est:fsharpminusfexpansion}, using the inequalities \eqref{est:drsv} as well as \eqref{est:drf} and \eqref{drf2}.
\subsection{Additional regularity results}\label{sec:addregest}

We now prove some auxiliary bounds, including Statement \eqref{fNaturalCondition3} in Proposition \ref{thm:WProperties} and in \ref{thm:WnatProperties}, as well as a bound which we will use in Section \ref{sec:boundaryterms} to establish boundary conditions on the separated equation.

\subsubsection{Completing the proof of Propositions \ref{thm:WProperties} and \ref{thm:WnatProperties}}

In order to establish suitable regularity at the horizon, we write \eqref{def:Phnatural} as
\begin{equation}\label{est:PLfnatTerms}
-i\omega\Big(v\big(\tfrac{v'}{v^2}\big)'\f + \f'' \Big) = -i\omega v \Rs\Big(v^{-2}\Rs(v\f)\Big).
\end{equation}
Observing $v^{-2}\Rs = (r^2+a^2)\partial_r$ and applying the bound \eqref{est:drf} gives the bound
\begin{equation}
\Big|v\omega\big(\tfrac{v'}{v^2}\big)'\f + \omega\f''\Big|\leq C |\omega|v^3 \, ,
\end{equation}
and similarly, \eqref{drf2} implies
\begin{equation}
\Big|v\omega\big(\tfrac{v'}{v^2}\big)'\fsharp + \omega\fsharp''\Big|\leq C |m|v^3 \, .
\end{equation}
Direct calculation gives the bounds
\begin{equation}\label{est:PlBoundnof}
\big|v\big(\tfrac{v'}{v^2}\big)'\big|\leq C r^{-3}, \qquad \big|v\big(\tfrac{V_1}{v^2}\big)'\big| \leq C v^3 \, .
\end{equation}
The estimates \eqref{est:Plupperbound} and \eqref{est:PLSuperradBounds} directly follow.
%
%
\subsubsection{The asymptotic behavior of $\Wsc$}
It will be necessary later to prove that $\Wsc$ approaches a first-order linear differential operator at spatial infinity, up to a pseudodifferential error term which decays in $r$. In order to do this, we now show pointwise error bounds on the difference between $\Wsc-\oW$ in frequency space.
\begin{corollary} \label{cor:asymptoticWinfinity}
For $r \geq \rmax+\delta^\sharp$, the operator $\Wsc-\oW$ satisfies for each $(\omega, m)$ the bounds
\begin{align}
|\Wsc\ph-\oW\ph|&\lesssim\frac{1}{r}(|m|+|\omega|)|\ph|, \\
|\Rs\Wsc\ph - \Rs\oW\ph|&\lesssim \frac{1}{r}(|m|+|\omega|)|\Rs\ph| + \frac{1}{r^2}(|m|+|\omega|)|\ph| \, .
\end{align}
\end{corollary}
\begin{proof}
It follows from the definition \eqref{def:CommutatorAbstract} that
\begin{equation}\label{id:WdifferenceBasic}
\Wsc\ph-\oW\ph = \wchi_1(\Wsv-\oW) + (1-\wchi_1)\chisharp\Wsl - (1-\wchi_1)\oW.
\end{equation}
For $r \geq \rmax+\delta^\sharp$, $\chisharp = 1$, so in this range the identity \eqref{id:WdifferenceBasic} becomes
\begin{equation}
\Wsc\ph-\oW\ph = -i\omega\wchi_1\left(\f -\frac{(\omega-\omega_r)h_0}{\omega v}\right)\ph  -i\omega(1-\wchi_1)\left(\fsharp -\frac{(\omega-\omega_r)h_0}{\omega v}\right)\ph.
\end{equation}
The result then follows directly from \eqref{est:drfFinal} and \eqref{est:drfFinal2} 
\end{proof}
We may also use the expansion \eqref{id:WdifferenceBasic} along with Proposition \ref{prop:vfDerUpperBound} to show a pointwise bound in frequency space for $r < \rmin$, noting that here $\chisharp = 0$ for all frequencies in $\mathcal{G}$. 
\begin{corollary}\label{prop:WMinusWZeroBounds}
For $r < \rmin$, the operator $\Wsc$ satisfies the pointwise bounds
\begin{align}\label{est:WMinusWZeroBounds}
|\Wsc\ph -\oW\ph|&\lesssim v(|\omega|+|m|)|\ph| + |\oW\ph|, \\
|\Rs\Wsc\ph -\Rs\oW\ph|&\lesssim |\Rs\oW\ph| + v(|\omega|+|m|)|\Rs\ph| +  \tfrac{v}{r}(|\omega|+|m|)|\ph|.
\end{align}
\end{corollary}

\section{The proof of the main theorem}\label{sec:CEE}
	In this section we prove Theorem \ref{thm:Main}. We first prove (\ref{pse}) with the additional assumptions (\ref{aux1}) and (\ref{aux2}) and remove these assumptions at the very end of Section \ref{sec:MainTheoremProof} below. We recall the cut-off version $\Psi_{\chi,\mathcal{T}}$ (defined in (\ref{def:WaveEqCutoff})) and that we established future integrability of $\Psi_{\chi,\mathcal{T}}$ in Proposition \ref{prop:suffi}. By the remark at the end of Section \ref{sec:comof}  it suffices to prove  (\ref{pse}) for $\Psi_{\chi,\mathcal{T}}$ instead of $\Psi$ with a constant independent of $\mathcal{T}$ and this is what we will do.

\subsection{The microlocal energies} \label{sec:energies}

In this subsection, we define the microlocal energies used in the proof. The energies introduced in Section \ref{sec:energy3} have simple physical space analogues in the Schwarzschild case. The auxiliary energy $S_1^\sharp[\widehat{u}]$, defined in Section \ref{sec:auxenergy}, is specific to the Kerr case and controls second order derivatives both localised in space and in (superradiant) frequency.

\subsubsection{The (non)-degenerate energies and the commuted energy} \label{sec:energy3}
 For $\widehat{u}=\widehat{{\bf u}}_\chi=\sqrt{r^2+a^2} \mathcal{F}(\Psi_{\chi,\mathcal{T}})^\omega_m$, we define the following microlocal energies. We first recall the $\alpha > 0$ fixed in Theorem \ref{thm:Main}. For fixed $(\omega, m)$ we then define
\begin{equation}\label{def:S0Local}\index{S@$\mathcal{S}[\ph]$}
\mathcal{S}[\ph]_m^\omega (r_1,r_2)= \iint_{\mcA(r_1,r_2)} \big(r^{-1-\alpha}\big(\big(\tfrac{r}{r-r_+}\big)^2|\Rs\ph +i(\omega-\omega_r)\ph|^2 + |\omega\ph|^2 +|\partial_r\ph|^2 + r^{-2}|\ph|^2\big) + r^{-3}|\widehat{\slashed\nabla}\ph|^2 \big),
\end{equation}
as well as the ``$W$-commuted” energy
\begin{equation}\label{def:S1WLocal}\index{SW@$\mathcal{S}_1^W[\ph]$}
\mathcal{S}^W_1[\ph]_m^\omega (r_1,r_2) = \iint_{\mcA(r_1,r_2)} \frac{r^2+a^2}{\Delta}\left( r^{-1}|\omega\Wsc\ph|^2 +  r^{-1-\alpha}|\Rs \Wsc\ph|^2 + gv^2r^{-1}|\widehat{\slashed\nabla}\Wsc\ph|^2 \right).
\end{equation}
We need a further energy, the degenerate microlocal energy, $(\mathcal{S}_{trap})[\ph]_{m\ell}^{(a\omega)} (r_1,r_2)$, defined for the fully decomposed pieces $\widehat{u}^{(a\omega)}_{m \ell}$. We recall (\ref{eq:WaveGeneric}) and from \cite{DRSR} the function $r_{trap}=r_{trap}(\omega, m, \Lambda=\Lambda_{m\ell})$, which is uniformly supported in a compact subset $[\bar{r}_1, \bar{r}_2] \subset (r_+, \infty)$ with $\bar{r}_1, \bar{r}_2$ depending only on the parameters $M$, $a$. Defining the trapping cutoff $
\chi_{trap} = (1-\tfrac{r_{trap}}{r})^2$, we define
\begin{equation}
(\mathcal{S}_{trap})[\ph]_{m\ell}^{(a\omega)}(r_1,r_2) = \int_{r_1}^{r_2}\frac{|\partial_r\ph|^2}{r^{1+\alpha}} + \frac{|\ph|^2}{r^{3+\alpha}} + \chi_{trap} \left(\frac{\omega^2}{r^{1+\alpha}} + \frac{\Lambda}{r^3}\right) |\ph|^2 + \mathbbm{1}_{r \leq \bar{r}_1} \Delta^{-2} | u^\prime - i (\omega-\omega_r) u |^2 \big)\,dr.\\
\end{equation}
We can extend the energies to spacetime energies. We define
\begin{equation}\label{def:SGlobal}
\mathcal{S}[\ph] (r_1,r_2)= \sum_m\int_{\mathbb{R}}\mathcal{S}[\ph]_m^\omega (r_1,r_2)\, d\omega \, , 
\end{equation}
\begin{equation}\label{def:S1WGlobal}
\mathcal{S}^W_1[\ph](r_1,r_2) = \sum_m\int_{\mathbb{R}}\mathcal{S}^W_1[\ph]_m^\omega (r_1,r_2)\, d\omega \, , 
\end{equation}
\begin{equation}\label{def:S1Strapped}
\mathcal{S}_{trap} [\ph](r_1,r_2) = \sum_{m, \ell} \int_{\mathbb{R}}\mathcal{S}_{trap}[\ph]^{(a\omega)}_{m \ell} (r_1,r_2)\, d\omega \, .
\end{equation}
Now by Plancherel  (\ref{id:Plancherel4}) we have for $r_+ \leq r_1 < r_2 \leq \infty$,
\begin{equation} \label{plai}
\mathbb{I}[\Psi_{\chi,\mathcal{T}}](r_1,r_2) \lesssim \mathcal{S}[\ph] (r_1,r_2) \, , 
\end{equation}
where
\begin{equation} \label{Imaindef}
\mathbb{I}[\Psi](r_1,r_2)  = \int_{-\infty}^\infty dt \int_{r_1}^{r_2} dr \int d\omega  \left\{ r^{1-\alpha}\left(|L \Psi|^2 + \Big|\frac{r^2+a^2}{\Delta} \underline{L} \Psi\Big|^2 \right)+ r |{\slashed{\nabla}} \Psi|_{\slashed{g}}^2 +r^{-1-\alpha} |\Psi|^2 \right\} \, .
\end{equation}
Note that, by the remark at the end of Section \ref{sec:comof}, proving the  main theorem requires controlling $\mathbb{I}[\Psi_{\chi,\mathcal{T}}]= \mathbb{I}[\Psi_{\chi,\mathcal{T}}](r_+,\infty)$ by initial data (independently of $\mathcal{T}$!) with a certain loss of derivatives. In view of (\ref{plai}), we will achieve this by controlling the microlocal energy $\mathcal{S}[\ph] (r_+,\infty)$.

\subsubsection{Definition of the auxiliary energy $S_1^\sharp[\widehat{u}]$} \label{sec:auxenergy}
Let us recall from (\ref{def:RDoubleSharp}) the definition of $R^{\musDoubleSharp}$ and $\delta^\sharp$.
%
We define $g_\sharp$ to be a smooth nonnegative function of $r-R^{\musDoubleSharp}$ satisfying
\begin{equation}\label{def:gsharpdef}
g_\sharp\index{gsharp@$g_\sharp$} =  \begin{cases}
1 &   \tfrac{1}{4}\delta^\sharp \leq r-R^{\musDoubleSharp}  \leq \tfrac{1}{2}\delta^\sharp ,\\
0 & r-R^{\musDoubleSharp}\leq \tfrac{1}{8}\delta^\sharp \text{ or }r-R^{\musDoubleSharp} \geq \tfrac{5}{8}\delta^\sharp.
\end{cases}
\end{equation}
It follows from the definition \eqref{def:RDoubleSharp} that the bound \eqref{est:DominantLagrangian} holds in the support of $g_\sharp$, and that -- independently of frequency -- $g_\sharp$ is always supported in the compact interval $[\rmin, \rmax+\delta^\sharp]$. Additionally, from (\ref{def:chisharp}),
\begin{equation} \label{gcutoff}
|\chisharp'| + |\chisharp''| \leq Cg_\sharp,
\end{equation}
where $C$ is a constant depending on $\chisharp$. We finally recall from Section \ref{sec:cutoffs} that $1-\widetilde{\chi}_1$ is supported in $\mathcal{G}$ and identically equal to $1$ in $(\mathcal{G}^\prime)^c$ and define the local energies (recall (\ref{def:Wplus}))
\begin{align}
\mathcal{S}^\sharp[\ph]_m^\omega &= \iint_{\mcA} (1-\wchi_1)^2g_\sharp^2\left(|\omega\ph|^2+|m\ph|^2 + |\partial_\theta\ph|^2 + |\widehat{u}^{\prime}|^2 \right) \, ,  \nonumber \\
\Ssharp[\ph]_m^\omega &= \mathcal{S}^\sharp[\omega \ph]_m^\omega +  \mathcal{S}^\sharp[m \ph]_m^\omega +  \mathcal{S}^\sharp[W_+ \ph]_m^\omega + \mathcal{S}^\sharp[\ph]_m^\omega . \label{def:Ssharp}
\end{align}
The energy $\Ssharp[\ph]_m^\omega$ should be viewed as controlling general second derivatives of frequencies in $(\mathcal{G}^\prime)^c$ near $R^{\musDoubleSharp}$. We make the familiar definitions
\begin{align}
\mathcal{S}^\sharp[\ph] &= \sum_m\int_{\mathbb{R}}\mathcal{S}^\sharp[\ph]_m^\omega\, d\omega \, , \\
\Ssharp[\ph] &= \sum_m\int_{\mathbb{R}}\Ssharp[\ph]_m^\omega\, d\omega \, .
\end{align}


	\subsection{The degenerate microlocal estimate}\label{sec:LOT}

We recall the separated equations (\ref{def:WavePerturbedFreq}) for $(\widehat{\mathbf{u}}_\chi)^\omega_m=\sqrt{r^2+a^2} \mathcal{F}(\Psi_{\chi,\mathcal{T}})^\omega_m$ and the fully separated (\ref{eq:WaveGeneric}) for $(\widehat{\mathbf{u}}_\chi)^{(a\omega)}_{m \ell}$. Repeating the proof of Theorem 8.1 of \cite{DRSR} produces the following estimate:

\begin{proposition} \label{thm:DRSRMain}
With $r_{trap}= r_{trap} (M,a,\omega,m, \Lambda)$ defined as in \cite{DRSR}, the following estimate holds\footnote{The estimate in \cite{DRSR} is stated for a compact range of integration from parameters $R^\star_-$ close to the horizon and $R^\star_+$ close to infinity and without weights in the integrand. Using standard local estimates near the horizon and infinity, it is straightforward to extend the integrals all the way to the boundary and to include the appropriate weights.} for the $\widehat{u}=(\widehat{\mathbf{u}}_\chi)^{(a\omega)}_{m \ell}$ associated with the solution $\Psi_{\chi,\mathcal{T}}$:
\begin{align} \label{est:DRSRMain} 
(\mathcal{S}_{trap})[\ph]_{m\ell}^{(a\omega)}(r_+,\infty)  
\lesssim \varepsilon \int_{r_+}^\infty dr   \left(\frac{\Lambda}{r^3} +\frac{\omega^2}{r^{1+\alpha}}+\frac{1}{r^{3+\alpha}}\right) \big| \widehat{u}^{(a\omega)}_{m \ell}\big|^2 + \mathfrak{H}_1  + \mathfrak{H}_2 \,,
\end{align}
where $ \mathfrak{H}_1$ and $ \mathfrak{H}_2$ are defined as follows. We have
\begin{align}
 \mathfrak{H}_1 := \int_{-\infty}^\infty d\rs \,  \Hchi\cdot (f, h, y, \chi)\cdot (\ph, \ph') \, , 
 \end{align}
the right hand side being defined by equation (71) in \cite{DRSR}, where $(\Hchi)^{(a\omega)}_{m\ell}$ is the decomposition of the right hand side $\Hchi$, as appearing in \eqref{def:HDecomp}, into spheroidal harmonics.
Furthermore,
\begin{align} \label{mosta1}
\mathfrak{H}_2 := \mathbbm{1}_{bfr}|\uchi(-\infty)|^2, 
\end{align}
where, for some parameters $\omega_{\text{low}}, \omega_{\text{high}}, \epsilon_{\text{width}}$ depending only on the black hole parameters $M$ and $a$
\begin{equation} \label{mosta2}
\mathbbm{1}_{bfr} = \begin{cases}
1 & \omega_{\text{low}} \leq |\omega| \leq \omega_{\text{high}} \text{ and } \Lambda \leq \epsilon_{\text{width}}^{-1}\omega_{\text{high}}^2, \\
0 & \text{otherwise}.
\end{cases}
\end{equation}
\end{proposition}

\begin{proof}
We repeat the proof of Theorem 8.1 of \cite{DRSR}, i.e.~we apply exactly the same microlocal multipliers to the fully separated inhomogeneous wave equation (\ref{def:WavePerturbedFreq}) for $(\widehat{\mathbf{u}}_\chi)^{(a\omega)}_{m\ell}$. Note that $\mathfrak{H}_1$ arises from the various multipliers multiplying the $H_\chi$-part of the inhomogeneity in the proof. The $H_{\BfB}$-part of the inhomogeneity produces the first term on the right hand side of (\ref{est:DRSRMain}).
\end{proof}
We next show how to estimate the terms $\mathfrak{H}_1$ and $\mathfrak{H}_2$ appearing in Proposition \ref{thm:DRSRMain}, which of course already appear in the same form in \cite{DRSR}. To appreciate the statements, recall that the estimate of Proposition \ref{thm:DRSRMain} will eventually be summed in $m, \ell$ and integrated in $\omega$ to transform it into a physical space estimate using Plancherel. 

\begin{proposition} \label{prop:estimateh}
For $\mathfrak{H}_1$ we have for any $\delta>0$ the estimate
\begin{equation}
\int_{-\infty}^\infty \sum_{m, \ell} \mathfrak{H}_1 \leq \delta \int_{-\infty}^\infty \sum_{m, \ell}  \mathcal{S}_{trap} [\uchi]^{(a\omega)}_{m \ell} + C\delta^{-1}E_0[\Psi] \, .
\end{equation}
For $\mathfrak{H}_2$ we have the estimate
\begin{equation} \label{h2estimate}
\int_{-\infty}^\infty\sum_{m, \ell} \mathfrak{H}_2\, d\omega \leq  C\left(E_0[\Psi]+\varepsilon^2 \mathcal{S}[\uchi]\right).
\end{equation}
\end{proposition}

\begin{proof}
The bound (188) in \cite{DRSR} directly implies the first estimate. The second estimate follows by repeating the argument of \cite{SR14}, which is invoked in \cite{DRSR}.  For convenience, we have included the relevant changes necessary to  adapt \cite{SR14} to our setting with additional first order term in Appendix \ref{appendix:yakov}. 
\end{proof}

\begin{corollary}
The estimate (\ref{est:DRSRMain}) implies
\begin{align}
\mathcal{S}_{trap} [\uchi] \lesssim \varepsilon \mathcal{S}[\uchi] + E_0[\Psi] \, .
\end{align}
\end{corollary}
\begin{proof}
Sum the estimate (\ref{est:DRSRMain}) in $m$, $\ell$ and integrate in $\omega$. Proposition \ref{prop:estimateh} and (\ref{def:S1Strapped}) imply the result.
\end{proof}

\subsection{A Hardy-type inequality from the commuted energy}\label{sec:NDLOT} We next show how to bound $\mathcal{S}[\ph]$ (the non-degenerate lower order energy) from $\mathcal{S}_{trap}[\ph]$ (the degenerate lower order energy) borrowing a small part of the commuted energy $\mathcal{S}_W[\ph]$. This takes the following form:
 \begin{lemma}\label{lem:DBoundtoNDBound}
For any $\delta' > 0$ there exists a constant $C_{\delta'}$ such that the following inequality holds:
\begin{equation}\label{est:lem:DBoundtoNDBound}
\mathcal{S}[\ph] \leq {\delta'} \, \mathcal{S}_1^W[\ph] + C_{\delta'} \mathcal{S}_{trap}[\ph] \, .
\end{equation}
 \end{lemma}
 
 \begin{proof}
 We consider first the fully separated $\uchi$, i.e.~$(\uchi)^{(a \omega)}_{m \ell}$ from Section \ref{sep:KTensor}.
 We will prove for any frequency tuple  $(\omega, m, \Lambda=\Lambda_{m \ell})$ the statement 
 \begin{align} \label{imspo}
 \mathcal{S}[\ph]^{(a\omega)}_{m \ell} (r_+,\infty) \leq {\delta'} \, \mathcal{S}_1^W[\ph]^{(a\omega)}_{m \ell}  (r_+,\infty) + C_{\delta'} \mathcal{S}_{trap}[\ph]^{(a\omega)}_{m \ell} (r_+,\infty) \, , 
 \end{align}
 where  now $\widehat{u} = (\uchi)^{(a \omega)}_{m \ell}$ and
 \begin{equation}
\mathcal{S}[\ph]_{m \ell}^{(a \omega)} (r^\star_1,r^\star_2)= \int_{r^\star_1}^{r^\star_2} dr^\star \big(r^{-1-\alpha}\big(\big(\tfrac{r}{r-r_+}\big)^2|\Rs\ph +i(\omega-\omega_r)\ph|^2 + |\omega\ph|^2 +|\partial_r\ph|^2 + r^{-2}|\ph|^2\big) + r^{-3}\Lambda_{m \ell} |\ph|^2 \big), \nonumber
\end{equation}
\begin{equation}
\mathcal{S}^W_1[\ph]_{m \ell}^{(a \omega)} (r^\star_1,r^\star_2) = \int_{r^\star_1}^{r^\star_2} dr^\star r^{-1}|\omega\Wsc\ph|^2 +  r^{-1-\alpha}|\Rs \Wsc\ph|^2 + gv^2r^{-1} \Lambda_{m \ell} |\Wsc\ph|^2. \nonumber
\end{equation}
The above energies should be compared with (\ref{def:S0Local}) and (\ref{def:S1WLocal}) for the $(\uchi)^\omega_m$. Using that for fixed $\omega, m$
\begin{align}
\int_0^\pi \sin \theta \left( |\partial_\theta \widehat{u}^\omega_m|^2 + (m^2 +\omega^2)|\widehat{u}^\omega_m|^2 \right) \lesssim \sum_{\ell} \left(\Lambda_{m \ell} + \omega^2 \right) |\widehat{u}^{(a\omega)}_{m \ell}|^2 \lesssim \int_0^\pi \sin \theta \left( |\partial_\theta \widehat{u}^\omega_m|^2 + (m^2 +\omega^2)|\widehat{u}^\omega_m|^2 \right) \, , \nonumber
\end{align}
and that $W$ does not depend on the $\Lambda_{m \ell}$ (but just $\omega$ and $m$), it is clear that summing the estimate (\ref{imspo}) in $m$ and $\ell$ and integrating in $\omega$ will imply (\ref{est:lem:DBoundtoNDBound}). 
 
 We next note that (\ref{imspo}) holds trivially (without the first term on the right) for $(\omega,m) \in \mathcal{G}$ as $r_{trap}=0$ for these frequencies by Proposition \ref{est:TrappingInclusion} and the sentence following it.
We can therefore restrict to prove (\ref{imspo}) for $(\omega,m) \in \mathcal{G}^c$, where $W=W^\natural$. 

Because non-trivial values of $r_{trap}$ are all contained in a uniformly (in frequency) compact subset of $(r_+,\infty)$, we first define a smooth (frequency independent) radial function $\chi_{I}$ which is compactly supported in $(r_+,\infty)$ and equal to 1 for all non-trivial values of $r_{trap}$. Then,
\begin{equation}
\int_{-\infty}^\infty \Rs\Big(\big(\chi_{I}\cdot(r-r_{trap})\big(|\omega\ph|^2 + m^2|\ph|^2 + \Lambda|\ph|^2 \big)\big)\Big) \,d\rs = 0 \, , 
\end{equation}
where $\widehat{u}=(\uchi)^{(a \omega)}_{m \ell}$. Using the the identity
$
\Rs(|\ph|^2) = 2v\mathfrak{R}(\Wsv\ph\overline\ph),
$
it follows that
\begin{align}
\int_{-\infty}^\infty \frac{\Delta}{r^2+a^2} \chi_{I}\big(|\omega\ph|^2 + m^2|\ph|^2 + \Lambda|\ph|^2 \big) \, d\rs &= -\int_{-\infty}^\infty(r-r_{trap})\Rs(\chi_{I})\big(|\omega\ph|^2 + m^2|\ph|^2 + \Lambda|\ph|^2 \big)\, d\rs \nonumber \\
&-\int_{-\infty}^\infty 2v(r-r_{trap})\chi_{I}\mathfrak{R}(\omega\Wsv\ph\overline{\omega\ph}+m\Wsv\ph\overline{m\ph}+\Lambda\Wsv\ph\overline{\ph}).\nonumber
\end{align}
The quantity $(r-r_{trap})\Rs(\chi_{I})$ is bounded, compactly supported, and vanishes on a set that strictly contains all potential non-trivial values of $r_{trap}$, so
\begin{equation}
\left|\int_{-\infty}^\infty(r-r_{trap})\Rs(\chi_{I})\big(|\omega\ph|^2 + m^2|\ph|^2 + \Lambda|\ph|^2 \big)\, d\rs\right|\lesssim \mathcal{S}_{trap}[\ph]^{(a \omega)}_{m \ell} \, .
\end{equation}
For any $\delta'' > 0$, the Cauchy-Schwarz inequality additionally implies
\begin{align}
v(r-r_{trap})\chi_{I}\mathfrak{R}(\omega\Wsv\ph\overline{\omega\ph}+m\Wsv\ph\overline{m\ph}+\Lambda\Wsv\ph\overline{\ph}) &\leq \frac{v^2}{\delta''}(r-r_{trap})^2\chi_{I}(|\omega\ph|^2 + |m\ph|^2 + \Lambda|\ph|^2) + \\
&\qquad+ \delta''\chi_{I}(|\omega\Wsv\ph|^2 + |m\Wsv\ph|^2 + \Lambda|\Wsv\ph|^2).\nonumber
\end{align}
Compact support of $\chi_{I}$ implies that there exists a constant $C$ (depending only on $M$ and $a$) such that
\begin{equation}
\chi_{I}(|\omega\Wsv\ph|^2 + |m\Wsv\ph|^2 + \Lambda|\Wsv\ph|^2)\leq C\mathcal{S}^W_1[\ph]^{(a\omega)}_{m \ell}.
\end{equation}
The inequality \eqref{imspo} follows from selecting $\delta''$ small enough that $C\delta'' < \delta'$.
\end{proof}
\subsection{A local Lagrangian estimate in $\mathcal{G}$ for general second derivatives} \label{sec:loc2d}

In order to deal with the $\mathcal{P}_{\chisharp}$ terms in the commutator (\ref{Pchisubsharp}), which will appear in our main energy and Lagrangian estimates, we will have to bound the quantity
\[
\sum_m \int d\omega \iint_{\mcAs} \left|(1-\wchi_1)\mathcal{P}_{\chisharp}\Wsl\ph\right|^2.
\]
This contains second-order derivatives which do not appear in the $W$-commuted energy and hence cannot be controlled. However, these terms are spatially localised in the region where the derivative of the cut-off is $\chi_\sharp$ non-trivial and they are frequency localised in $\mathcal{G}$. The use of both restrictions allows us to prove a coercive estimate. 
Recall the energy $S_1^\sharp[\widehat{u}]$ defined in (\ref{def:Ssharp}). We will prove an estimate for fixed $(\omega, m) \in \mathcal{G}$ but it can easily be integrated in $\omega$ and summed in $m$ in view of the presence of the cut-off $\widetilde{\chi}_1$ in all terms.


\begin{proposition} \label{prop:LE2}
For $(\omega, m)\in\mathcal{G}$, we have for $\hat{u}=\widehat{{\bf u}}_\chi$ the estimate
\begin{equation}\label{est:S1SharpBound}
\Ssharp[\ph]_m^\omega \lesssim  (1-\wchi_1)^2\mathcal{S}[\ph]_m^\omega +  \mathcal{N}_{\sharp}[H, H]_m^\omega
\end{equation}
holds, where
\begin{equation}\index{Nsharp@$(\mathcal{N}_{\sharp})_m^\omega$}
\mathcal{N}_{\sharp}[H, H]_m^\omega = \iint_{\mcAs}(1-\wchi_1)^2g_\sharp^2\left(|\omega H^\omega_m|^2+|mH^\omega_m|^2+\left|v^2\Wp \left(\tfrac{H^\omega_m}{v^2}\right)\right|^2\right) \, . \label{Nshp}
\end{equation}
\end{proposition}
\begin{proof}
We begin with a Lemma, which may be viewed as the ``uncommuted'' Lagrangian estimate.
\begin{lemma} \label{prop:LE}
Let $\ph$ be a $C^2$ solution of \eqref{def:SepWaveOperator}. Then, for $(\omega, m)\in\mathcal{G}$, the following estimate holds:
\begin{equation}
\mathcal{S}^\sharp[\ph]_m^\omega \lesssim \iint_{\mcAs} (1-\wchi_1)^2(g_\sharp^2 + (g_\sharp^2)'')|\ph^\omega_m|^2 + g_\sharp^2(1-\wchi_1)^2|H_m^\omega|^2.
\end{equation}
\end{lemma}
\begin{proof}
We take the Lagrangian current and angular Lagrangian current given in \eqref{def:LagrangianCurrent} for $g = g_\sharp^2(1-\wchi_1)^2$, and integrate the identity \eqref{Div:Lagrangian} in $(\rs, \theta)$, noting that compact support of $g_\sharp$ in $(r_+, \infty)$ implies that the boundary terms vanish. Then, Proposition \ref{prop:DominantLagrangian} gives our result.
\end{proof}

We now prove \eqref{est:S1SharpBound}. The bounds for the terms $\mathcal{S}^\sharp[\omega \ph]_m^\omega$ and  $ \mathcal{S}^\sharp[m \ph]_m^\omega$ appearing  in $\Ssharp[\ph]_m^\omega$ follow  immediately from Lemma \ref{prop:LE} after using the facts that $g_\sharp$ is uniformly compactly supported and that $\omega, m$ both commute through $\Pam$. Consequently, since bounding $\mathcal{S}^\sharp[\ph]_m^\omega$ in terms of $\mathcal{S}[\ph]_m^\omega$ is trivial,
\begin{equation} \label{poj}
\mathcal{S}^\sharp[\omega \ph]_m^\omega + \mathcal{S}^\sharp[m \ph]_m^\omega + \mathcal{S}^\sharp[\ph]_m^\omega \lesssim \mathcal{S}[\ph]^\omega_m + \iint_{\mcAs} (1-\wchi_1)^2g_\sharp^2(|\omega H^\omega_m|^2 + |m H^\omega_m|^2).
\end{equation}
To bound also the term containing $\Wp$, we commute $\Pam$ with $W_+$ and combine Lemma \ref{prop:LE} with the commutator identity \eqref{def:CommIdentityBasic} and (\ref{poj}) (in particular the estimate on $\mathcal{S}^\sharp[\omega \ph]_m^\omega$ to control the (only) second order term on the right hand side). 
\end{proof}

	\subsection{The main microlocal commutator estimate}\label{sec:GlobalEE}
		We now construct our main energy and Lagrangian estimates which are central to the proof of the main theorem. We recall the microlocal energies $\mathcal{S}[\ph]_m^\omega, \mathcal{S}_1^{\Wsc}[\ph]_m^\omega, \Ssharp[\ph]_m^\omega$ as defined in \eqref{def:S0Local}, \eqref{def:S1WLocal}, and \eqref{def:Ssharp}. We additionally recall the consequent global energies \eqref{def:SGlobal}, \eqref{def:S1WGlobal}.
\begin{theorem}\label{thm:CommutedEEMain}
Let the assumptions of Theorem \ref{thm:Main} hold and assume in addition (\ref{aux1}), (\ref{aux2}).
Consider $\Psi_{\chi,\mathcal{T}}$ as defined in Section \ref{sec:comof} and $\uchi, \HB, \Hchi$ as defined in Section \eqref{sec:RJHSDecomp}, with $H = \HB + \Hchi$. Recall (\ref{PHdef}) and define
\begin{equation}\label{def:N}\index{N@$\mathcal{N}$}
\mathcal{N}[\Wsc\uchi, \PH] = \mathcal{N}_T[\Wsc\uchi, \PH] + 2c_{\mathcal{G}} \mathcal{N}_{\alpha}[\Wsc\uchi, \PH],
\end{equation}
where $c_{\mathcal{G}}$ is the constant fixed in Section \ref{sec:SuperradiantEstimates2},
\begin{equation}\index{NT@$\mathcal{N}_T$}
\mathcal{N}_T[\Wsc\uchi, \PH] = \iiiint_{\FMs}\mathfrak{I}\left(\omega\Wsc\uchi\overline{\PH}\right), \label{imterm}
\end{equation}
\begin{equation}\index{Nalpha@$\mathcal{N}_\alpha$}
\mathcal{N}_{\alpha}[\Wsc\uchi, \PH] = \iiiint_\FMs  - g_\alpha\mathfrak{R}\left(\Wsc\uchi\overline{\PH}\right), \label{imterm2}
\end{equation}
for
\begin{equation}\label{def:galpha}\index{galpha@$g_\alpha$}
g_\alpha = \frac{1}{2r} + \frac{M^\alpha}{2r^{1+\alpha}},
\end{equation}
and from (\ref{Nshp})
\begin{equation}\label{def:Nsharp}\index{Nsharp@$\mathcal{N}_\sharp$}
\mathcal{N}_{\sharp}[H, H] = \sum_m \int_\mathbb{R}\mathcal{N}_{\sharp}[H, H]_m^\omega\, d\omega = \iiiint_{\FMs}(1-\wchi_1)^2g_\sharp^2\left(|\omega H|^2+|mH|^2+\left|v^2\Wp \left(\tfrac{H}{v^2}\right)\right|^2\right).
\end{equation}
 Then, there exist positive constants $C$ and $C'$ such that $\uchi$ satisfies the energy bound
\begin{equation}\label{est:EnergyBoundAllFrequencies}
\mathcal{S}_1^{\Wsc}[\uchi] + \Ssharp[\uchi] \leq C\mathcal{N}[\Wsc\uchi, \PH] + C'\left(\mathcal{S}[\uchi] + \mathcal{N}_{\sharp}[H, H]\right).
\end{equation}
\end{theorem}
\begin{remark}
Note that the quantities $\mathcal{N}_T$ and $\mathcal{N}_\alpha$ correspond to physical space multipliers  and, unlike $\mathcal{N}_\sharp$, do not contain absolute value terms. This will be exploited when we bound terms coming from the initial data via Parseval’s identity, see Section \ref{sec:InitialDataBounds}.
\end{remark}
The proof extends over the following subsections and will conclude in Section \ref{sec:ProofofMainEnergyEstimate}. It is based on applying the microlocal energy identities (\ref{microlocalE}) and (\ref{microlocalg}) with $\Wsc\uchi$ in place of $\ph$, for currents of the form
\begin{equation}
Q^T - 2c_{\mathcal{G}} Q^{g_\alpha}+ k(1-\wchi_1)Q^{g_\sharp^2}, \qquad A^T - 2c_{\mathcal{G}} A^{g_\alpha} + k(1-\wchi_1)A^{g_\sharp^2} \, .
\end{equation}
In particular, these currents are physical, up to a term which is compactly supported in $r$. We first show, in Section \ref{sec:boundaryterms}, that the boundary terms that appear  in (\ref{microlocalE}) and (\ref{microlocalg}) applied with $\Wsc\uchi$ in place of $\hat{u}$ will always vanish after summation in $m$ and integration over $\omega$.

\subsubsection{Controlling the boundary terms}  \label{sec:boundaryterms}
\begin{proposition} \label{prop:boundaryterms}
Under the assumptions of Theorem \ref{thm:CommutedEEMain}, the $\uchi$ satisfies:
 \begin{equation}\label{est:OBCderiv}
\lim_{r\to r_+} \sum_m \int_{-\infty}^\infty d\omega \int_{0}^{\pi} \sin \theta d\theta \left( |W\uchi|^2 + |\Rs W\uchi|^2 \right) =  0
 \end{equation}
and
 \begin{equation}\label{est:OBCderiv2}
\lim_{r \rightarrow \infty}  \sum_m \int_{-\infty}^\infty d\omega \int_{0}^{\pi} \sin \theta d\theta \left( |\Wsc\uchi(r)|^2 + |\Rs\Wsc\uchi(r)|^2 \right) = 0 \, .
 \end{equation}
\end{proposition}
\begin{proof}
We first prove (\ref{est:OBCderiv2}). For $W$ replaced by the physical space vector field $W_0$, the limit is a direct consequence of (\ref{sqi2}) and the Plancherel identity. Replacing $W$ by $W-W_0$, the result in turn follows from the estimates
$|(W-W_0)\psi|^2 \leq \frac{\omega^2}{r^2} |\psi|^2 +\frac{m^2}{r^2} |\psi|^2$ and $|R^\star (W-W_0)\psi|^2 \leq \frac{\omega^2}{r^2} |R^\star \psi|^2 +\frac{m^2}{r^2} |R^\star \psi|^2 +\frac{\omega^2}{r^2} |\psi|^2 +\frac{m^2}{r^2} |\psi|^2$, valid for $r \geq \overline{r}^\sharp + \delta^\sharp$ from Corollary \ref{cor:asymptoticWinfinity}. Converting the right hand sides to physical space with Plancherel and using again  (\ref{sqi2}) produces the desired result.

We next prove (\ref{est:OBCderiv}). It clearly suffices to prove this separately for $W^\natural$ and $\chi_\sharp W^\sharp$. Since $\chi_\sharp W^\sharp$ is identically zero near $r_+$, it remains to prove it for $W^\natural$. 
Again replacing $W=W-W_0+W_0$ and using the bounds of Corollary \ref{prop:WMinusWZeroBounds} as well as the regularity of $\Delta^{-\frac{1}{2}} W_0$ at the horizon, the result follows.
\end{proof}

\subsubsection{The commuted $T$-energy estimate}
We now prove the basic commuted $T$-energy estimate arising from applying (\ref{microlocalE}) with $\Wsc\uchi$:
\begin{proposition}\label{lem:CommutedEEBasic}
Under the assumptions of Theorem \ref{thm:CommutedEEMain}, for any constant $\delta > 0$, there exists a constant $C_\delta > 0$ such that
\begin{equation}\label{est:BasicEE}
\iiiint_{\FMs} 2vf'|\omega\Wsc\uchi|^2\leq \delta \mathcal{S}_1^{\Wsc}[\uchi] +\mathcal{N}_T[\Wsc\uchi, \PH]+ C_\delta \left( \mathcal{S}[\uchi] + \Ssharp[\uchi] + \mathcal{N}_{\sharp}[H, H]  \right) \, . 
\end{equation}
\end{proposition}
\begin{proof}
Taking the energy currents $Q^T[\Wsc\uchi], A^T[\Wsc\uchi]$ as defined in \eqref{def:EnergyCurrent} and applying the energy identity \eqref{id:StandardMLDiv} gives
\begin{equation}
\iiiint_{\FMs(r_1^\star,r_2^\star)} \Rs Q^T[\Wsc\uchi] + \frac{1}{\sin\theta}\partial_\theta(\sin\theta A^T[\Wsc\uchi]) =\iiiint_{\FMs(r_1^\star,r_2^\star)}\mathfrak{I}\left(\Pam\Wsc\uchi\overline{\omega\Wsc\uchi}\right).
\end{equation}
The left hand side produces two boundary terms (see (\ref{microlocalE})) which both go to zero at $r_1^\star \rightarrow -\infty$ and $r_2^\star \rightarrow 0$ by the boundary conditions \eqref{est:OBCderiv}, \eqref{est:OBCderiv2}. For the remaining expression we 
expand $\Pam\Wsc\uchi$ according to \eqref{def:CommutedBoxSimp} and treat each term separately. Using $\Ph^\natural = \Ph^\sharp = 0$ we have that at each frequency (for the rest of the proof we write $\int = \iiiint_{\FMs}$)
\begin{equation}\label{est:EEFullBasic}
\int 2v f'|\omega\Wsc\uchi|^2\leq \int \mathfrak{I}\left(\PH\overline{\omega\Wsc\uchi}\right) + \int |\omega\Wsc\uchi||\wchi_1\Pl^\natural\uchi + (1-\wchi_1)\chisharp\Pl^\sharp\uchi +  (1-\wchi_1)\mathcal{P}_{\chisharp}\Wsl\uchi|.
\end{equation}
The Cauchy-Schwarz inequality implies
\begin{equation}
\int |\omega\Wsc\uchi||\wchi_1\Pl^\natural\uchi + (1-\wchi_1)\chisharp\Pl^\sharp\uchi| \leq \frac{\delta}{2}\int \frac{1}{r}|\omega\Wsc\uchi|^2  + \frac{1}{2\delta}\int r|\wchi_1\Pl^\natural\uchi + (1-\wchi_1)\chisharp\Pl^\sharp\uchi|^2. \nonumber
\end{equation}
The first term on the right precisely appears in $\mathcal{S}_1^{\Wsc}[\uchi]$, and we may bound the second term by $\mathcal{S}[\uchi]$ using \eqref{est:Plupperbound} and \eqref{est:PLSuperradBounds}. Additionally, we have the bound
\begin{equation}
\int |\omega\Wsc\uchi||(1-\wchi_1)\mathcal{P}_{\chisharp}\Wsl\uchi| \leq \frac{\delta}{2}\int \frac{1}{r}|\omega\Wsc\uchi|^2  + \frac{1}{2\delta}\int r|(1-\wchi_1)\mathcal{P}_{\chisharp}\Wsl\uchi|^2.
\end{equation}
The first term again appears in $\mathcal{S}_1^{\Wsc}[\uchi]$ and for the second we claim
\begin{equation}\label{est:PChiSharpBound}
\iint_{\mcAs} |(1-\wchi_1)r \mathcal{P}_{\chisharp}\Wsl\ph|^2 \lesssim \mathcal{S}_1^\sharp[\ph]_m^\omega \lesssim \mathcal{S}[\ph]_m^\omega + \mathcal{N}_{\sharp}[H, H]_m^\omega \, , 
\end{equation}
from which the result directly follows after summation in $m$ and integration in $\omega$. To verify (\ref{est:PChiSharpBound}) note that the second inequality follows directly from (\ref{est:S1SharpBound}) while for the first, we note that the expansion \eqref{Pchisubsharp} along with the bounds \eqref{drf2} and $|\omega|\lesssim |m|$ in the support of $1-\widetilde{\chi}_1$, imply the pointwise bound
\begin{equation}
|\Rs\Wsl\ph| = \left|\Rs\left(\Wp\ph -i \tfrac{(\hsharp-1)(\omega-\omega_r)}{v}\ph\right)\right| \lesssim |\Rs\Wp\ph|+ |m\Rs\ph| + |m\ph|.
\end{equation}
\end{proof}
\subsubsection{The commuted Lagrangian estimate}
We now want to add to  \eqref{est:BasicEE} a Lagrangian estimate such that the sum on the left controls the energy $S_1^W[\uchi]$ which subsequently absorbs the first term on the right of  \eqref{est:BasicEE}. For most frequencies $f^\prime$ has a good sign and the Lagrangian estimate is relatively straightforward. When $f^\prime$ is negative, the $m^2$-terms in the Lagrangian estimate dominate the $\omega^2$-terms, leading to the desired estimate. We start with the standard Lagrangian estimate which has the time derivative term as the term with the opposite sign:
%
\begin{proposition}\label{lem:LagrangianDecayCommuted}
Under the assumptions of Theorem \ref{thm:CommutedEEMain} there exists a uniform constant $C$ such that $\Wsc\uchi$ satisfies the estimate
\begin{align}\label{est:LagrangianDecayCommuted}
\int-\frac{1}{r}|\omega-\omega_r|^2|\Wsc\uchi|^2 + \frac{M^\alpha}{2r^{1+\alpha}}|\Rs\Wsc\uchi|^2 + \frac{v^2}{2r}|\slashed\nabla\Wsc\uchi|^2 & \leq -\mathcal{N}_{\alpha}[\Wsc\uchi, \PH] \nonumber \\
& \ \ + C(\mathcal{S}[\uchi] + \Ssharp[\uchi] +  \mathcal{N}_{\sharp}[H, H] ) \, , 
\end{align}
where we use the shorthand $\int =  \iiiint_{\FMs}$.
\end{proposition}


\begin{proof}
We apply the identity \eqref{microlocalg} with $\Wsc\uchi$ in place of $\hat{u}$, integrate in $\omega$ and sum over $m$. The boundary terms on the right hand side vanish in the limit $r_1^\star \rightarrow -\infty$, $r_2^\star \rightarrow \infty$ by Proposition \ref{prop:boundaryterms}. Using this fact and $\frac{1}{2r} \leq g_\alpha \leq \frac{1}{r}$, we next claim that \eqref{microlocalg} (applied with $\Wsc\uchi$ in place of $\hat{u}$) implies the estimate
\begin{equation}\label{est:LagrangianDecay}
\int -\frac{1}{r}|\omega-\omega_r|^2|\Wsc\uchi|^2 + \frac{M^\alpha}{2r^{1+\alpha}}|R^\star \Wsc\uchi|^2 + \frac{v^2}{2r}|\slashed\nabla \Wsc\uchi|^2 \leq \int \frac{C}{r^{3+\alpha}} |\Wsc\uchi|^2 - g_\alpha\mathfrak{R}(\Pam \Wsc\uchi \overline{\Wsc\uchi}) \, ,
\end{equation}
Indeed, (\ref{est:LagrangianDecay}) is easily seen to follow if we can establish for some $C \geq 0$ the estimate
\begin{equation} \label{cozi}
\int -\frac{1}{r}|\ph'|^2 + \frac12\Rs\Rs\left(\frac1r\right)|\ph|^2 \leq C \int r^{-3-\alpha} |\ph|^2 \, .
\end{equation}
The estimate (\ref{cozi}) is in turn a direct consequence of the fundamental theorem of calculus identity (boundary terms vanishing by Proposition \ref{prop:boundaryterms})
\begin{equation}
0 \geq \int -|r^{-1/2}\ph'- r^{-3/2}\Wsc\uchi|^2- \Rs(r^{-2}|\Wsc\uchi|^2) = \int -\frac{1}{r}|\ph'|^2 - \frac{1}{r^3}|\Wsc\uchi|^2 - \frac{2\Rs(r)}{r^3}|\ph|^2 \,  ,
\end{equation}
and the easily verified pointwise inequality
\begin{equation}
\left|\frac{1}{r^3}|\Wsc\uchi|^2 - \frac{2\Rs(r)}{r^3}|\Wsc\uchi|^2 - \frac12\Rs\Rs\left(\frac1r\right)|\Wsc\uchi|^2\right| \lesssim  r^{-4}|\Wsc\uchi|^2.
\end{equation}
Having established (\ref{est:LagrangianDecay}), Proposition \ref{prop:WMinusWZeroBounds} implies
\begin{equation} \label{wbe}
\int C r^{-3-\alpha} |\Wsc\uchi|^2 \lesssim \mathcal{S}[\uchi] , 
\end{equation}
so that (\ref{est:LagrangianDecayCommuted}) is implied by the following Lemma, which deals with the last term on the right of \eqref{est:LagrangianDecay}.
\end{proof}
\begin{lemma}\label{prop:GlobalLagrangianBound1}
Under the assumptions of Theorem \ref{thm:CommutedEEMain} and using the shorthand $\int =  \iiiint_{\FMs}$, we have the bound
\begin{equation}
\left| \int g_\alpha\mathfrak{R}(\Pam\Wsc\uchi\overline{\Wsc\uchi}) - \int g_\alpha\mathfrak{R}(\PH \overline{\Wsc\uchi})\right|\lesssim \mathcal{S}[\uchi] + \Ssharp[\uchi] +  \mathcal{N}_{\sharp}[H, H]   \, .
\end{equation}
\end{lemma}
\begin{proof}
Expanding \eqref{def:CommutedBoxSimp} we see from the identity \eqref{freli}  that the terms with $\f'$, $\fsharp'$ contribute as purely imaginary, so in particular
\begin{equation}
\left| \int g_\alpha\mathfrak{R}(\Pam\Wsc\uchi\overline{\Wsc\uchi}) -\int g_\alpha \mathfrak{R}(\PH \overline{\Wsc\uchi})\right| \leq \int\frac{2}{r}|\Wsc\uchi||\wchi_1\Pl^\natural\uchi + (1-\wchi_1)\chisharp\Pl^\sharp\uchi + (1-\wchi_1)\mathcal{P}_{\chisharp}\Wsl\uchi| . \nonumber
\end{equation}
We now estimate 
\begin{equation}
\int \frac{2}{r}|\Wsc\uchi||\wchi_1\Pl^\natural\uchi + (1-\wchi_1)\chisharp\Pl^\sharp\uchi|\lesssim \mathcal{S}[\uchi] 
\end{equation}
from Cauchy-Schwarz and (\ref{est:PLSuperradBounds}), (\ref{est:Plupperbound}). Moreover, we have
\begin{equation}
\int \frac{2}{r}|\Wsc\uchi|\left|(1-\wchi_1)\mathcal{P}_{\chisharp}\Wsl\uchi\right|\lesssim \int r^{1+\alpha}\left|(1-\wchi_1)\mathcal{P}_{\chisharp}\Wsl\uchi\right|^2 + \int r^{-3-\alpha}\left|\Wsc\uchi\right|^2.
\end{equation}
Since $\mathcal{P}_{\chisharp}$ is uniformly supported in a compact set, we may ignore the weight in $r$ in the first integral, and bound this using \eqref{est:PChiSharpBound} integrated in $\omega$ and summed in $m$. The second integral is controlled from (\ref{wbe}).
\end{proof}

\subsubsection{Coercivity from summing the estimates}\label{sec:ProofofMainEnergyEstimate}
We next show that adding the estimates \eqref{est:BasicEE} and \eqref{est:LagrangianDecayCommuted} appropriately yields an expression on the left hand side which is coercive over all frequencies up to lower order terms:
\begin{proposition}\label{prop:fglobal}
Under the assumptions of Theorem \ref{thm:CommutedEEMain}, we have the coercivity bound
\begin{align}\label{est:fglobal}
c\mathcal{S}^{\Wsc}_1[\uchi] \leq  \iiiint_{\FMs} 2vf'|\omega\Wsc\uchi|^2  - \tfrac{2c_{\mathcal{G}}}{r}(\omega-\omega_r)^2|\Wsc\uchi|^2 + \tfrac{c_{\mathcal{G}} v^2}{r}|\widehat{\slashed\nabla}\Wsc\uchi|^2 + \tfrac{M^\alpha c_{\mathcal{G}} }{r^{1+\alpha}}|\Rs\Wsc\uchi|^2
\end{align}
for some uniform constant $c > 0$.
\end{proposition}
\begin{proof}
We will establish the bound pointwise for the integrands on the left and on the right of (\ref{est:fglobal}).
We split this into the two cases $\wchi_2 = 1$ and $0 \leq \wchi_2 < 1$. 
When $\wchi_2 = 1$, we have $f=f^\natural$ by construction and moreover frequencies are contained in $\mathcal{G}^\prime$. It follows that the bounds \eqref{est:fprimebound} and \eqref{est:ELCombined2} together give
\begin{equation}
v\omega^2\f' - \frac{c_0c_{\mathcal{G}'}}{r}\left(-(\omega-\omega_r)^2 + \frac{v^2}{2}\left(\frac{m}{\sin\theta} - a\omega\sin\theta \right)^2\right) \geq \frac{c_0c_{\mathcal{G}'}}{r}\left((\omega-\omega_r)^2 + \frac{v^2}{2}\left(\frac{m}{\sin\theta} - a\omega\sin\theta \right)^2\right). \nonumber
\end{equation}
Replacing $c_\mathcal{G} = c_0c_{\mathcal{G}'}$ (from \eqref{id:relationfour}) and adding $\frac{1}{2} \tfrac{c_{\mathcal{G}} v^2}{r}|\partial_\theta \Wsc\uchi|^2 + \frac{1}{2} \tfrac{M^\alpha c_{\mathcal{G}} }{r^{1+\alpha}}|\Rs\Wsc\uchi|^2$ to both sides the desired bound follows in this range.

When $0 \leq \wchi_2 < 1$, we note that $\wchi_1= 0$, so by the definition \eqref{def:CommutatorAbstract} it suffices to prove the bound for the integrands in the support of $\chisharp$.
In this regime $f=f_\sharp$ and  \eqref{est:hSharpDerivativeIneq} as well as the bound \eqref{est:ELCombined} imply
\begin{equation}\label{est:fsharplowerbound}
2v f^\prime_\sharp \omega^2 - \frac{2{c_{\mathcal{G}}}}{r}(\omega-\omega_r)^2 + \frac{{c_{\mathcal{G}}}v^2}{r}\left(\frac{m}{\sin\theta} - a\omega\sin\theta\right)^2 \geq \frac{{c_{\mathcal{G}}}}{4r}\omega^2 + \frac{{c_{\mathcal{G}}}v^2}{16r}\left(\frac{m}{\sin\theta}\right)^2.
\end{equation}
Adding $\frac{1}{2} \tfrac{c_{\mathcal{G}} v^2}{r}|\partial_\theta \Wsc\uchi|^2 + \frac{1}{2} \tfrac{M^\alpha c_{\mathcal{G}} }{r^{1+\alpha}}|\Rs\Wsc\uchi|^2$ to both sides, the desired bound follows in this range also.
\end{proof}

\subsubsection{Completing the proof of Theorem \ref{thm:CommutedEEMain}}
We are now ready to complete the proof of Theorem \ref{thm:CommutedEEMain}.
We bound the right hand side of \eqref{est:fglobal} using Proposition \ref{lem:CommutedEEBasic} (where we select $\delta = \frac{c}{2}$, where $c$ is the constant appearing on the left hand side of \eqref{est:fglobal}) and Proposition \ref{lem:LagrangianDecayCommuted}. Consequently,
\begin{equation}
\frac{c}{2}\mathcal{S}^{\Wsc}_1[\uchi] \leq  C\left(\mathcal{S}[\uchi] + \Ssharp[\uchi]\right) + \mathcal{N}[\Wsc\uchi, \PH]  \, . 
\end{equation}
Proposition \ref{prop:LE2} implies that, for some new value of $C$, the bound
\begin{equation}
\frac{c}{2}\left(\mathcal{S}^{\Wsc}_1[\uchi]+ \Ssharp[\uchi]\right) \leq  C\left(\mathcal{S}[\uchi] + \mathcal{N}_{\sharp}[H, H] \right) +\mathcal{N}[\Wsc\uchi, \PH]  \, .
\end{equation}
Multiplying by $\frac{2}{c}$ produces the desired bound \eqref{est:EnergyBoundAllFrequencies}. 

	\subsection{The error estimates} \label{sec:convolutionsub}
	 Having proven the estimate (\ref{est:EnergyBoundAllFrequencies}), we now seek to bound $\mathcal{N}[\Wsc\uchi, \PH]$ and $\mathcal{N}_\sharp[H,H]$ for $\uchi$, which we recall solves \eqref{def:WavePerturbedFreq}, with $H = \HB + \Hchi$. The main difficulty here is that we cannot in general do this at the level of frequency, as in Fourier space, $\mathcal{N}[\Wsc\uchi, \PH]$ and $\mathcal{N}_\sharp[H,H]$ cannot be expressed as a function of $\ph_m^\omega$ and its derivatives alone. We first introduce in Section \ref{sec:Bnormdef} a Fourier norm on the operator $\BT$ defined in (\ref{def:BT}) that appears in the frequency decomposed equation \eqref{def:WavePerturbedFreq} and subsequently show that if $\BT$ is derived from an operator $\BfB$ satisfying $\mathcal{L}_T\BfB = \mathcal{L}_\phi\BfB = 0$, then this norm is bounded by the constants $C$ appearing in the bounds \eqref{Bassumption1}, \eqref{Bassumption2} of Theorem \ref{thm:Main}. The main result of this section is then Theorem \ref{thm:RHSBoundMain} which controls 
$\mathcal{N}[\Wsc\uchi, \PH]$ and $\mathcal{N}_\sharp[H,H]$ as desired. As it will be straightforward to conclude the proof of the main theorem from Theorem \ref{thm:RHSBoundMain}, we will do so first in Section \ref{sec:MainTheoremProof} and postpone the lengthy proof of Theorem \ref{thm:RHSBoundMain} itself to Section \ref{sec:proofofconvolution}.

\subsubsection{The norm on $\BfB$ and estimating it from the assumptions in the main theorem} \label{sec:Bnormdef}
%

We first define the space $\Wbb$ to be the space of functions on $\mathcal{R}$ such that the norm
\begin{equation}\index{Ws@$\Wbb$}
\lVert u\rVert_{\Wbb} = \sup_{r, \theta}\left[ r^{s}\left(  \left\lVert \omega\widehat{u}\right\rVert_{\ell^1_mL^1_\omega} + \left\lVert m\widehat{u}\right\rVert_{\ell^1_mL^1_\omega} + \left\lVert  \oW \widehat{u}\right\rVert_{\ell^1_mL^1_\omega} + \left\lVert \widehat{u}\right\rVert_{\ell^1_mL^1_\omega}\right)\right]
\end{equation}
is finite, where we recall the weight $v$ from (\ref{vdef}) and where we have abused notation slightly in that $\oW \widehat{u} := v^{-1} ( \widehat{u}^\prime -i\omega \widehat{u} + \frac{im}{r^2+a^2}\widehat{u})$ denotes the Fourier version of the physical space vector field $W_0$ from (\ref{W0def}).

We next define the norm on the first order linear operator $\BT$ defined in (\ref{def:BT})  to be
\begin{equation} \label{Bnorm}
\lVert \BT\rVert_{[\alpha]} := \lVert \BT^{\tilde{\tau}} \rVert_{\Wbb[1+\alpha]}+\lVert \BT^{r}\rVert_{\Wbb[1+\alpha]}+\lVert \BT^\theta\rVert_{\Wbb[2+\alpha]}+\lVert \BT^{\tilde{\phi}}\rVert_{\Wbb[2+\alpha]}+\lVert \BT^0\rVert_{\Wbb[2+\alpha]}.
\end{equation}
The next proposition shows that $\lVert \BT\rVert_{[\alpha]}$ is indeed finite and controlled independently of $\mathcal{T}$ if $\BfB$ is $T$ and $\Phi$ symmetric and satisfies \eqref{Bassumption1}, \eqref{Bassumption2}:

\begin{proposition} \label{prop:Bfrommain}
If the operator $\BT$ defined in (\ref{def:BT}) arises from a vector field $\BfB$ which is $\mathcal{L}_T$ and $\mathcal{L}_{\phi}$ invariant, then 
\begin{equation}\label{est:PointwiseBComponents2}
\lVert \BfB_{\mathcal{T}}\rVert_{[\alpha]} \leq C \sup_{r,\theta} \sum_{\beta \in \{ \BfB^{\tilde{\tau}}, \BfB^r, \BfB^\theta, \BfB^{\tilde{\phi}}, \BfB^0\}}  r^{p[\beta] + \alpha} \left(|\beta| + |\oW(\beta)|\right),
\end{equation}
where
\[
p[\beta] = \begin{cases}
1 & \beta \in \{\BfB^{\tilde{\tau}}, \BfB^r\}, \\
2 & \beta \in \{\BfB^0, \BfB^{\tilde{\phi}}, \BfB^\theta\}
\end{cases}
\]
holds for a constant $C$ which is independent of $\mathcal{T}$. In particular, if $\BfB$ satisfies in addition the assumptions \eqref{Bassumption1}, \eqref{Bassumption2}, then the norm $\lVert \BfB_{\mathcal{T}}\rVert_{[\alpha]}$ is bounded by a constant depending only on $(M,a)$, the $C$ appearing in \eqref{Bassumption1}, \eqref{Bassumption2} and the choice of the cut-off function $\chi$ in (\ref{def:chiOneD}).
\end{proposition}

\begin{proof}
We will show
\begin{equation}\label{est:PointwiseBComponents}
\sum_{\beta \in \{ \BfB^{\tilde{\tau}}, \BfB^r, \BfB^\theta, \BfB^{\tilde{\phi}}, \BfB^0\}} \lVert \widehat{\chi_{\mathcal{T}} \beta} \rVert_{\ell^1_m L^1_\omega}   \leq C\sum_{\beta \in \{ \BfB^{\tilde{\tau}}, \BfB^r, \BfB^\theta, \BfB^{\tilde{\phi}}, \BfB^0\}}\left(|\beta| + |\oW(\beta)|\right)
\end{equation}
from which the desired bound is immediate by the definition of the norms. We recall that the components of $\BfB$ are independent of $t$ and $\phi$ and also that $\chi_{\mathcal{T}}$ is independent of $\phi$ and that hence $\widehat{\chi_{\mathcal{T}}}$ is supported on $m=0$ only. Moreover, from (\ref{def:timefoliation}) and  (\ref{defchiT}) we have $W_0 \chi_{\mathcal{T}}=0$ by construction hence $ W_0 \widehat{\chi_{\mathcal{T}}}:=v^{-1} ( \widehat{\chi_{\mathcal{T}}}^\prime -i\omega \widehat{\chi_{\mathcal{T}}} + \frac{im}{r^2+a^2}\widehat{\chi_{\mathcal{T}}})=\widehat{W_0\chi_{\mathcal{T}}}=0$. It then follows that the left hand side of (\ref{est:PointwiseBComponents}) is bounded by 
\begin{align}
 C \sum_{\beta \in \{ \BfB^{\tilde{\tau}}, \BfB^r, \BfB^\theta, \BfB^{\tilde{\phi}}, \BfB^0\}} |\beta| \left( \| \widehat{\chi_{\mathcal{T}}} \|_{\ell^1_m L^1_\omega} + \|  \omega \widehat{\chi_{\mathcal{T}}} \|_{\ell^1_m L^1_\omega} \right) +  |W_0(\beta)| \| \widehat{\chi_{\mathcal{T}}} \|_{\ell^1_m L^1_\omega} \, .
\end{align}
Note that the summation in $m$ always contains just one summand since $\widehat{\chi_{\mathcal{T}}}$ is supported on $m=0$ only. Letting the hat denote Fourier transformation in $t$ only for the moment, we conclude the proof by the definition (\ref{defchiT}) and the scaling property of the Fourier transform:
\begin{align}
\|\widehat{\chi_{\mathcal{T}}}\|_{L^1_\omega} = \int d\omega |\widehat{\chi_{\mathcal{T}}}(\omega) |  = \int d\omega \frac{1}{{\mathcal{T}}} \Big|\widehat{\chi_{\mathcal{T}}}\left(\frac{\omega}{{\mathcal{T}}}\right) \Big|   = \int d\omega  \Big|\widehat{\chi_1}(\omega) \Big| \leq C \, ,
\end{align}
since $\chi_1$ is smooth and compactly supported in $[-1,2]$, its Fourier-transform therefore a Schwartz function. The estimate for $\|\widehat{\chi_{\mathcal{T}}}\|_{L^1_\omega}$ follows analogously, now with an extra power of ${\mathcal{T}}^{-1}$ to spare.
\end{proof}


\begin{remark}
One can generalise the above argument to components of $\BfB$ which instead of being $T$ and $\Phi$ invariant are allowed to oscillate in time. For instance, for $\BfB$ of the form
\begin{equation}
\BfB = \sum_{k=1}^N e^{-i(\omega_k \wtau - m_k\phi)}\BfB_k(r, \theta),
\end{equation}
where $N \in \mathbb{N}$ and $(m_k, \omega_k)\in \mathbb{Z}\times\mathbb{R}$ are fixed constants and each $\BfB_k$ is a vector field commuting with $T, \Phi$, the above proof carries through (now using also the property of the Fourier transform under translations) with the constant now depending also on the $(m_k,\omega_k)$. We are not aiming to isolate the most general class of allowed $\BfB$ here.
\end{remark}


\subsubsection{The main error estimate}
 We may now bound the right hand side of \eqref{est:EnergyBoundAllFrequencies}, which will close our energy estimate to highest order.
\begin{theorem}\label{thm:RHSBoundMain}
Under the assumptions of Theorem \ref{thm:CommutedEEMain} for every $\delta > 0$ there exists a constant $C_\delta$ depending on $\delta$ such that
\begin{equation}\label{est:RHSBoundMain}
\mathcal{N}[\Wsc\uchi, \PH] + \mathcal{N}_\sharp[H, H] \leq (\delta + C\varepsilon\lVert \BT\rVert_{[\alpha]})(\mathcal{S}_1^{\Wsc}[\uchi] + \Ssharp[\uchi] + \mathcal{S}[\uchi]) + C_\delta E^1_0[\Psi].
\end{equation}
where $E^1_0[\Psi]$ is the energy (\ref{e10def}) appearing on the right hand side of (\ref{pse}).
\end{theorem}

	\subsection{Completing the proof of Theorem \ref{thm:Main}}\label{sec:MainTheoremProof}
		We first prove a theorem in frequency space which will easily imply Theorem \ref{thm:Main}.
\begin{theorem}\label{thm:MainFreqSpace}
Under the assumptions of Theorem \ref{thm:CommutedEEMain} and with $\BfB$ satisfying in addition
\begin{equation}\label{MainBFreqAssumption}
\sup_{\mathcal{T}\in[2, \infty)}\lVert\BfB\rVert_{[\alpha]} \leq C < \infty \, ,
\end{equation}
we have the estimate
\begin{equation}\label{est:MainFrequencyBound}
\mathcal{S}_1^{\Wsc}[\uchi] + \Ssharp[\uchi] +  \mathcal{S}[\uchi]  \leq C \, E_0^1[\Psi]
\end{equation}
with the right hand side as defined in (\ref{e10def}) and $C=C(M,a)$, i.e.~in particular independent of $\mathcal{T}$.
\end{theorem}
\begin{proof}
Theorem \ref{thm:CommutedEEMain} gives
\begin{equation}
\mathcal{S}_1^{\Wsc}[\uchi] + \Ssharp[\uchi] \leq C\left( \mathcal{S}[\uchi] + \mathcal{N}[\Wsc\uchi, H]+  \mathcal{N}_{\sharp}[H, H]\right).
\end{equation}
Lemma \ref{lem:DBoundtoNDBound} implies, for $\delta'$ chosen to  be sufficiently small,
\begin{equation}
\mathcal{S}_1^{\Wsc}[\uchi] + \Ssharp[\uchi] +  \mathcal{S}[\uchi]\leq  C\left(\mathcal{S}_{trap}[\uchi] + \mathcal{N}[\Wsc\uchi, H] + \mathcal{N}_{\sharp}[H, H]\right).
\end{equation}
Proposition \ref{thm:DRSRMain} implies the integrated estimate
\begin{equation}
\mathcal{S}_{trap}[\uchi] \leq C' {E}^1_0[\Psi] + \varepsilon C'\mathcal{S}[\uchi] \, , 
\end{equation}
so in particular for sufficiently small $\varepsilon$ (such that $CC'\varepsilon < \frac12$) we may write
\begin{equation}
\mathcal{S}_1^{\Wsc}[\uchi] + \Ssharp[\uchi] +  \mathcal{S}[\uchi] \leq C \left( E_0^1[\Psi] + \mathcal{N}[\Wsc\uchi, H] + \mathcal{N}_{\sharp}[H, H]\right).
\end{equation}
Finally  Theorem \ref{thm:RHSBoundMain} implies
\begin{equation}
\mathcal{S}_1^{\Wsc}[\uchi] + \Ssharp[\uchi] +  \mathcal{S}[\uchi]\leq C(\delta + C\varepsilon\lVert \BfB\rVert_{[\alpha]})\left(\mathcal{S}_1^{\Wsc}[\uchi] + \Ssharp[\uchi] +  \mathcal{S}[\uchi] \right) + C_\delta {E}_0^1[\Psi] \, .
\end{equation}
Choosing sufficiently small $\delta, \varepsilon$ gives \eqref{est:MainFrequencyBound}.
\end{proof}

We can now complete the proof of Theorem \ref{thm:Main}, first with the two additional assumptions (\ref{aux1}), (\ref{aux2}) appearing in Theorem \ref{thm:CommutedEEMain}. By Proposition \ref{prop:Bfrommain}, the assumptions on $\BfB$ in Theorem \ref{thm:Main} imply \eqref{MainBFreqAssumption}. Therefore, it suffices to show that the microlocal energy estimate \eqref{est:MainFrequencyBound} implies the physical-space energy estimate \eqref{pse}. This clearly follows from 
\begin{equation} \label{lastuy}
\text{(l.h.s.) }\eqref{pse} \lesssim \text{(l.h.s.) }\eqref{est:MainFrequencyBound} \, ,
\end{equation}
since the right hand sides of \eqref{est:MainFrequencyBound} and (\ref{pse}) agree.
To verify (\ref{lastuy}), we first note that by Plancherel, the left hand side of \eqref{est:MainFrequencyBound} controls the left hand side of (\ref{pse}) with $\Psi$ replaced by $\Psi_{\chi,\mathcal{T}}$ and the integration taken over all of $\mathcal{R}$. We restrict the integration region to $\tau \in [2,\mathcal{T}]$ where $\Psi_{\chi,\mathcal{T}}=\Psi$. In addition, for $\mathcal{T} \leq 2$ we can control the left hand side of (\ref{pse}) by the right hand side independently  by a standard Gronwall estimate. This shows (\ref{pse}) with $\mathsf{T}=\mathcal{T}$ and since $\mathcal{T}$ was arbitrary, (\ref{pse}) holds for all $\mathsf{T}$ as claimed.

We finally remove the additional assumptions (\ref{aux1}), (\ref{aux2}). The compact support assumption on the data can be removed by a standard density argument. To remove the assumption on $\BfB$, one constructs for a given $\BfB$ satisfying (\ref{Bassumption1}), (\ref{Bassumption2}) a sequence of smooth compactly supported (in $r$) $\BfB_n$  converging to $\BfB$ in $\|\BfB\|_{[\alpha]}$. Fixing a $\mathsf{T}$, the estimate (\ref{pse}) holds uniformly for all $\BfB_n$. Moreover, the solution $\Psi_n$ associated with $\BfB_n$ for $n \geq N$ agrees with the solution $\Psi_N$ associated with $\BfB_N$ on a domain that monotonically increases as $N \rightarrow \infty$ to $\cup_{\tilde{\tau}=0}^{\mathsf{T}}\Sigma_{\tilde{\tau}}$. Applying the monotone convergence theorem yields the result. 

\section{The proof of the error estimates} \label{sec:proofofconvolution}
	In this section, we prove Theorem \ref{thm:RHSBoundMain}, which requires bounding the expressions (\ref{def:N}) and (\ref{def:Nsharp}). We begin, in Section \ref{sec:convo}, by introducing the main workhorses in the proof, two convolution estimates involving pseudodifferential commutators. We will then prove the estimate  \eqref{est:RHSBoundMain} by proving it separately for the $\HB$ and $\Hchi$ quantities via the following argument: We write $\PH = \PHB + \PHchi$, where
\begin{equation}\label{def:PHDecomposition}\index{PHB@$\PHB$} \index{PHchi@$\PHchi$}
\PHB =v^2 \Wsc\left(\frac{\HB}{v^2}\right), \qquad \PHchi =v^2 \Wsc\left(\frac{\Hchi}{v^2}\right) \, , 
\end{equation}
with $\HB$ and $\Hchi$ defined in (\ref{def:WavePerturbedFreq}). The expression $\mathcal{N}[\Wsc\uchi, \PH]$ is linear in $\PH$, and $\mathcal{N}[H,H]$ is quadratic in $H$, so
\begin{equation}\label{est:NSum}
\mathcal{N}[\Wsc\uchi, \PH] =\mathcal{N}[\Wsc\uchi, \PHB] + \mathcal{N}[\Wsc\uchi, \PHchi], \qquad \mathcal{N}_\sharp[H,H]\leq 2\left(\mathcal{N}_\sharp[\HB, \HB] + \mathcal{N}_\sharp[\Hchi,\Hchi]\right) \, .
\end{equation}
The bounds on $\mathcal{N}[\Wsc\uchi, \PHB]$ and $\mathcal{N}_\sharp[\HB, \HB]$ will be given in Proposition \ref{thm:PHBBoundFirst}, and the proof will comprise the bulk of Section \ref{sec:PerturbativeBounds}. The corresponding bounds on $\mathcal{N}[\Wsc\uchi, \PHchi]$ and $\mathcal{N}_\sharp[\Hchi, \Hchi]$ will be given in Proposition \ref{thm:PHchiBound} with the proof in Section \ref{sec:InitialDataBounds}. To bound the final term on the right hand side of \eqref{est:PHBBoundFirst}, we will use Proposition \ref{prop:PSBoundIDWorst} and the Plancherel identity \eqref{id:Plancherel4}.


\subsection{The convolution estimates} \label{sec:convo}
In bounding $\PHB$ and $\PHchi$ we will have to look at how $\Wsc$ acts on products (equivalently, on convolutions in frequency space). For a function $g$ depending on $r$ and $\theta$ alone, the basic product rule
\begin{equation}\label{id:WBasicLeibniz}
\Wsc(g\ph) = (\Wsc g)\ph+ g(\Wsc \ph) = v^{-1}g'\ph+ g\Wsc \ph
\end{equation}
continues to hold. 

However, when bounding $\PH$ appearing in (\ref{imterm}), (\ref{imterm2}), we must contend with the fact that $\PHB$ contains convolved terms like $\Wsc(\rho^2(\BThat^{\wtau}*\omega\uchi))$, where $\BT^{\wtau}$ is by definition time-dependent. Since we must take a convolution in the frequency variables, for fixed $(\omega, m)$ we cannot write $(\BThat^{\wtau}*\omega\uchi)_m^\omega$ in terms of $(\uchi)_m^\omega$, and therefore cannot close our argument at the level of frequency.

We instead look to bound differences of the form
\[
\lVert\Wsc((\rho^2\BThat^{\wtau})*(\omega\uchi)) - (\rho^2\BThat^{\wtau})*(\omega\Wsc\uchi) - \omega\uchi * (\Wsc(\rho^2\BThat^{\wtau}))\rVert,
\]
as well as analogous quantities on other derivatives of $\uchi$, in a suitable (weighted $L^2$) norm.\footnote{Note that the expression vanishes if $W$ is a physical space vector field, e.g.~if $m=0$ and $W=W_0$.} From this we can establish bounds on the first term via bounds on the latter two terms. 
We bound this difference by proving a microlocal product rule in the spirit of the Coifman-Meyer inequality, see \cite{CM78}, which will allow us to avoid general second derivatives of $\uchi$ when estimating the above. 

We first define a smooth nonnegative frequency-independent function
\begin{equation}\index{chirmax@$\chi_{\rmax}$}
\chi_{\rmax} = \begin{cases}
1 & r < \rmax + \tfrac78\delta^\sharp,\\
0 & r > \rmax + \delta^\sharp.
\end{cases}
\end{equation}
We additionally define a smooth nonnegative frequency-independent function $\chi_{\rmin\!,\rmax}$ such that $\chi_{\rmin\!,\rmax} \leq \chi_{\rmax}$ and
\begin{equation}\index{chirminmax@$\chi_{\rmin,\rmax}$}
\chi_{\rmin\!,\rmax} = \begin{cases}
1 & \rmin+\tfrac18\delta^\sharp \leq r \leq \rmax+\tfrac78\delta^\sharp, \\
0 & r < \rmin \text{ or } r > \rmax+\delta^\sharp
\end{cases}
\end{equation}
One can easily see from a support argument that, for all $(\omega, m)\in\mathcal{G}$,
\begin{equation}\label{est:chisharprmax}
1-\chisharp + |\chisharp'| \lesssim \chi_{\rmax}
\end{equation}
and, recalling the definition \eqref{def:gsharpdef},
\begin{equation}\label{est:gsharpprimebound}
|g_\sharp| + |g_\sharp'| + |g_\sharp''|\lesssim \chi_{\rmin\!,\rmax} \lesssim \chi_{\rmax}.
\end{equation}

As we will typically estimate a convolution by 
a weighted version of the standard Young's inequality
$\|\widehat{g} * \widehat{u}\|_{\ell^2_m L^2_\omega} \leq \left(\sup_{r,\theta} \|\widehat{g}\|_{\ell^1_m L^1_\omega} \right) \|\widehat{u}\|_{\ell^2_m L^2_\omega} (r,\theta)$, we define the space $\Wdot$ to be the closure of the set $\{\ph : u \in C_0^\infty(\mathcal{R})\}$ under the weighted $W^{1,1}$ Sobolev norm
\begin{equation}\label{def:Wdotnorm}
\lVert \widehat{u}\rVert_{\Wdot} = \left\lVert \omega\widehat{u}\right\rVert_{\ell^1_mL^1_\omega} + \left\lVert m\widehat{u}\right\rVert_{\ell^1_mL^1_\omega} + \left\lVert v^{-q} \oW \widehat{u}\right\rVert_{\ell^1_mL^1_\omega},\index{Wdot@$\Wdot$}
\end{equation}
 recalling (\ref{vdef}) and \eqref{def:oWfreq}, and we define $\Hdot$ to be the same under the weighted $L^2$ norm
\begin{equation}\label{def:Hdotnorm}
\Vert \ph(r, \theta) \rVert_{\Hdot} = \left\lVert v^q \ph(r, \theta)\right\rVert_{\ell^2_m L^2_\omega} +\left\lVert \chi_{\rmax}^{1/2}\tfrac{\oW\ph(r, \theta)}{\max(|m|, |\omega|, 1)}\right\rVert_{\ell^2_m L^2_\omega}. \index{Hdot@$\Hdot$}
\end{equation}
We cannot in general avoid the derivative appearing in the second term of \eqref{def:Hdotnorm}, as we take both derivatives in frequency space and in physical space. However, the extra decay in frequency space means that this can still be treated like an undifferentiated term (after taking a suitable Lagrangian estimate). A precise treatment will be given in Lemma \ref{lem:ZeroOrderInverseBound}.

We have the following convolution estimates which play the key role in the proof of Theorem \ref{thm:RHSBoundMain}. 
\begin{proposition}\label{lem:CMBound}
Let $\ph\in\Hdot$ and $\widehat{w}\in\Wdot$ and let $a(r, \theta) = a_1(r, \theta)a_2(r,\theta)$ for some nonnegative functions $a_1, a_2$ (not necessarily continuous). Then, the inequality
\begin{equation}\label{est:CIGlobal}
\iiiint_{\FM}a^2\big|\Wsc(\widehat{w}*\ph) - \widehat{w}*\Wsc\ph\big|^2\lesssim \sup_{(r, \theta)}\left(a_1^2\left\lVert\widehat{w}\right\rVert_{\Wdot}^2 \right) \cdot\iint_{\mcA} a_2^2\left\lVert\ph\right\rVert_{\Hdot}^2
\end{equation}
holds for all $q\in[0,1]$.
\end{proposition}
This will be proven in Section \ref{sec:CMConvolution}; however, we outline our approach here. In order to show suitable decay in $r$, we first subtract off the vector field $\oW$ from $\Wsc$, as they display similar asymptotic behavior. Defining the nonlocal error
\begin{equation}\label{def:ferror}
\wfe\index{f@$\wfe$} := \f - \tfrac{\omega-\omega_r}{\omega v}h_0 , \text{ so that } \Wsv - \oW = -i\wfe,
\end{equation}
we write
\begin{equation}\label{id:WMinusWZero}
\Wsc - \oW = \left(-i\wchi_1\omega\wfe + (1-\wchi_1)\left(-i\chisharp\frac{\hsharp-h_0}{v}(\omega-\omega_r) - (1-\chisharp)\oW\right)\right).
\end{equation}
Then, the desired bound will follow from proving Lipschitz continuity for $\Wsc-\oW$ in $\omega$, which we will show via bounds on the relevant derivatives, and decay in $r$.

The second convolution estimate involves more general pseudodifferential operators which are however compactly supported both in space (near $R^{\musDoubleSharp}$, we recall (\ref{def:gsharpdef})) and in frequency (namely in $\mathcal{G}$, we recall (\ref{def:wchi12})), which makes their proof significantly easier than that of (\ref{est:CIGlobal}). For convenience, we will nevertheless employ the same weighted spaces used in Proposition \ref{lem:CMBound}. Finally, we note that the $X_\sharp$ we will prove the estimate for are specifically tailored to $\mathcal{N}_\sharp$ appearing in (\ref{def:Nsharp}).

\begin{proposition}\label{lem:CMBound2}
Take $\ph\in \Hdot$, $\widehat{w}\in\Wdot$, and define the operators
\begin{equation}\label{def:GenSecondOperators}
X_\sharp^T = -ig_\sharp(1-\wchi_1)\omega, \qquad X_\sharp^\phi = ig_\sharp(1-\wchi_1)m, \qquad X_\sharp^{\Wp} = g_\sharp(1-\wchi_1)\Wp,
\end{equation}
where $g_\sharp$ is as defined in \eqref{def:gsharpdef}. Then, for $X_\sharp \in \{X_\sharp^T, X_\sharp^\phi, X_\sharp^{\Wp}\}$, there exists a compactly supported function $\chi_{\rmin\!,\rmax}$ such that the bound
\begin{equation}\label{est:CISharp}
\iiiint_{\FM}a^2|X_\sharp(\widehat{w}*\ph) - \widehat{w}*(X_\sharp(\ph))|^2 \lesssim \sup_{(r, \theta)}\left(a_1^2\lVert \widehat{w}\rVert_{\Wdot}^2\chi_{\rmin\!,\rmax}\right)\iint_{\mcA}a_2^2\left\lVert\ph\right\rVert^2_{\Hdot}\chi_{\rmin\!,\rmax}
\end{equation}
holds for all $q\in[0,1]$, and for nonnegative functions $a(r, \theta), a_1(r, \theta), a_2(r, \theta)$ satisfying $a = a_1a_2$, as long as the right hand side is finite.
\end{proposition}




	\subsection{The perturbative terms $H_{\BfB}$}\label{sec:PerturbativeBounds}
			We prove (\ref{est:RHSBoundMain}) for $H_{\BfB}$:
\begin{proposition}\label{thm:PHBBoundFirst}
Under the assumptions of Theorem \ref{thm:CommutedEEMain} there exists a constant $C > 0$ such that
\begin{equation}\label{est:PHBBoundFirst}
|\mathcal{N}[\Wsc\uchi, \PHB]|+  |\mathcal{N}_\sharp[\HB, \HB]|\leq C\varepsilon\lVert \BT\rVert_{[\alpha]}^2\left(\mathcal{S}_1^{\Wsc}[\uchi]  + \Ssharp[\uchi] + \mathcal{S}[\uchi] + \int\chi_{\rmax}^2\left|v^2\oW\left(\tfrac{\Hchi}{v^2}\right)\right|^2\right) .
\end{equation}
\end{proposition}
We prove this by bounding each term separately, splitting $\mathcal{N}[\Wsc\uchi, \PHB]$ according to the decomposition \eqref{def:N}. Then, the bound on $\mathcal{N}$ follows from Lemmas \ref{lem:PHBBoundPrelim}, \ref{prop:WdrTildeCommutator}, and \ref{lem:ZeroOrderInverseBound}, and the bound on $\mathcal{N}_\sharp$ follows from Lemma \ref{lem:NSharpBfBBound}.


\subsubsection{Bounding $\mathcal{N}[\Wsc\uchi, \PHB]$}
We first note
\begin{equation}
\int r^{-1-\alpha}|\omega\Wsc\ph|^2 + r^{-3-\alpha}|\Wsc\ph|^2\lesssim \mathcal{S}_1^{\Wsc}[\uchi]+\mathcal{S}[\uchi],
\end{equation}
which via the Cauchy-Schwarz inequality implies that for any $\delta > 0$
\begin{equation}\label{est:CSPHB}
\mathcal{N}[\Wsc\uchi, \PHB] \lesssim \varepsilon \left(\mathcal{S}_1^{\Wsc}[\uchi]+\mathcal{S}[\uchi]\right) + \varepsilon^{-1} \int r^{1+\alpha}|\mathcal{P}_{\HB}|^2.
\end{equation}
Next, we expand
\begin{equation}
\int r^{1+\alpha}|\mathcal{P}_{\HB}|^2 = \int r^{1+\alpha}\left|v^2\Wsc\left(\tfrac{\HB}{v^2}\right)\right|^2.
\end{equation}
Recalling the regular coordinates $(\wtau, r, \theta, \wphi)$ from Section \ref{sec:hyperboloidal}, we decompose
\begin{equation}\label{def:BuchiExpansion}
\mathcal{F}\left(\BT\mathbf{u}_{\chi}\right) = \BThat^\wtau * (-i\omega\uchi) + \BThat^{r}*\wpa_r\uchi + \BThat^{\theta}*\wpa_\theta\uchi + \BThat^{\wphi}*(im\uchi)+\BThat^0*\uchi,
\end{equation}
and apply Proposition \ref{lem:CMBound} to obtain the following preliminary bound:
\begin{lemma}\label{lem:PHBBoundPrelim}
Under the assumptions of Theorem \ref{thm:CommutedEEMain}, $\uchi$ satisfies the bound
\begin{align}\label{est:PHBWeightedBound}
\int& r^{1+\alpha}|\mathcal{P}_{\HB}|^2 \lesssim \varepsilon^2\left\lVert\BfB\right\rVert_{[\alpha]}^2\! \left(\!\mathcal{S}_1^{\Wsc}[\uchi]\! + \!\mathcal{S}[\uchi] \!+\! \int\!\! \frac{\Delta^2}{r^{5+\alpha}}\left(\!\big|[\Wsc\!,\wpa_r]\uchi\big|^2\!\! +\! \chi_{\rmax}\frac{|\oW\wpa_r\uchi|^2\!+\!|\tfrac{1}{r}\oW\wpa_\theta\uchi|^2}{|\max(|\omega|, |m|, 1)|^2}\right)\!\!\right).
\end{align}
\end{lemma}
\begin{proof}
The identity \eqref{def:HDecompB} implies
\begin{equation}
\frac{\HB}{v^2} = c\varepsilon\rho^2\left(\mathcal{F}(\BT\bf{u}_\chi)\right),
\end{equation}
so consequently the identity \eqref{id:WBasicLeibniz} and the easily verified bounds
\begin{equation}
|v^2\Wsc(\rho^2)|\lesssim \frac{\Delta^{3/2}}{r^3} \text{ and }\qquad |v^2\rho^2|\lesssim \frac{\Delta}{r^2}
\end{equation}
together imply
\begin{equation}\label{est:PHBWeightedLeibniz}
\int r^{1+\alpha}|\mathcal{P}_{\HB}|^2 \lesssim \varepsilon^2 \int r^{1+\alpha}\left(\left|\frac{\Delta^{3/2}}{r^3}\left(\mathcal{F}(\BT\bf{u}_\chi)\right)\right|^2 + \left|\frac{\Delta}{r^2} \Wsc\!\left(\mathcal{F}(\BT\bf{u}_\chi)\right)\right|^2\right).
\end{equation}
We take the expansion \eqref{def:BuchiExpansion} and bound this term by term, starting with
\begin{equation}
\mathcal{F}\left(\BT^\wtau\partial_\wtau\mathbf{u}_\chi\right) = \BThat^\wtau * (-i\omega\uchi) \, .
\end{equation}
To bound the first term on the right hand side of \eqref{est:PHBWeightedLeibniz}, we apply Young's convolution inequality in $(\omega, m)$ and a weighted $L^2-L^\infty$ bound in $(r, \theta)$ to obtain
\begin{equation}\label{est:PHBBoundPart1}
 \int r^{1+\alpha}\left|\frac{\Delta^{3/2}}{r^3}(\BThat^\wtau * \omega\uchi)\right|^2\lesssim \sup_{r, \theta} \Big(r^{1+\alpha}\lVert \BThat^{\wtau}\rVert_{L^1_{\omega, m}}\Big)^2\int r^{-1-\alpha}\frac{\Delta^3}{r^6}|\omega\uchi|^2\lesssim  \left\lVert\BfB\right\rVert_{[\alpha]}^2\mathcal{S}[\uchi].
\end{equation}
In order to bound the second term on the right hand side of  \eqref{est:PHBWeightedLeibniz}, we first apply Proposition \ref{lem:CMBound} (with $q=0$) to obtain
\begin{equation}\label{est:PHBBoundPart2Prelim}
\int r^{1+\alpha}\frac{\Delta^2}{r^4}\left|\Wsc\!\left(\BThat^\wtau \!*\! (\omega\uchi)\right) - \BThat^\wtau \!*\! \Wsc(\omega\uchi)\right|^2 \lesssim  \sup_{r, \theta}\Big( r^{1+\alpha}\lVert\BThat^\wtau\rVert_{\Wdot[0]}\Big)^2\int\frac{\Delta^2}{r^{5+\alpha}}\lVert\omega\uchi\rVert_{\Hdot[0]}^2.
\end{equation}
From the definition \eqref{def:Hdotnorm} one easily obtains the bound
\begin{equation}\label{est:HdotBound}
\lVert\omega\uchi\rVert_{\Hdot[0]} \leq \lVert\omega\uchi\rVert_{\ell^2_m L^2_\omega}+ \lVert \chi_{\rmax}^{1/2}\oW\uchi\rVert_{\ell^2_m L^2_\omega}.
\end{equation}
Since $\chi_{\rmax}$ is supported in the region $r \in (r_+, \rmax+\delta^\sharp)$, it follows that
\begin{equation}\label{est:PHBBoundPart2}
\int r^{1+\alpha}\frac{\Delta^2}{r^4}\left|\Wsc\!\left(\BThat^\wtau \!*\! (\omega\uchi)\right) - \BThat^\wtau \!*\! \Wsc(\omega\uchi)\right|^2  \lesssim \lVert\BThat^\wtau\rVert_{\Wbb[1+\alpha]}^2 \mathcal{S}[\uchi].
\end{equation}
It follows from Young's convolution inequality in $(\omega, m)$, along with the fact that $[\Wsc, \omega] = 0$, that
\begin{align}\label{est:HighestOrderSecondConv}
\int r^{1+\alpha}\frac{\Delta^2}{r^4}\left| \BThat^\wtau \!*\! \Wsc(\omega\bf{u}_\chi)\right|^2 &\lesssim  \sup_{r, \theta}\Big( r^{1+\alpha}\lVert\BThat^\wtau\rVert_{L^1_{\omega, m}}\Big)^2\int\frac{\Delta^2}{r^{5+\alpha}}|\omega\Wsc\uchi|^2\big), \nonumber \\
&\lesssim \lVert\BThat^\wtau\rVert_{\Wbb[1+\alpha]}^2 \mathcal{S}_1^{\Wsc}[\uchi].
\end{align}
Consequently, it follows from the bounds \eqref{est:PHBBoundPart1}, \eqref{est:PHBBoundPart2}, and \eqref{est:HighestOrderSecondConv} that
\begin{equation}
\int r^{1+\alpha}\left|v^2\Wsc\left(c\varepsilon\rho^2\left(\BThat^\wtau * (-i\omega\uchi)\right)\right)\right|^2\lesssim \varepsilon^2\left\lVert\BT\right\rVert_{[\alpha]}^2 \left(\mathcal{S}_1^{\Wsc}[\uchi]\! + \mathcal{S}[\uchi]\right).
\end{equation}

The bounds for $\BThat^{\wphi}*(im\uchi)$ and $\BThat^0*\uchi$ follow from virtually identical arguments, noting that the additional power of $r$ appearing in the norm on $\BT$ in \eqref{Bnorm} compensates for the additional decay required in $\mathcal{S}[\mathbf{u}_\chi]$ and $\mathcal{S}_1^{\Wsc}[\mathbf{u}_\chi]$. 

In order to bound the terms containing $\BThat^{\theta}*\wpa_\theta\uchi$, we follow the same outline, noting that \eqref{est:PHBBoundPart1} and \eqref{est:PHBBoundPart2Prelim} hold with the appropriate modifications, but we must now replace \eqref{est:HdotBound} with the modified estimate
\begin{equation}\label{est:HdotBoundTheta}
\lVert\wpa_\theta\uchi\rVert_{\Hdot[0]} \leq \lVert\wpa_\theta\uchi\rVert_{\ell^2_m L^2_\omega}+ \left\lVert \chi_{\rmax}^{1/2}\frac{\oW\wpa_\theta\uchi}{\max(|m|, |\omega|, 1)}\right\rVert_{\ell^2_m L^2_\omega}.
\end{equation}
Carrying out the remainder of the argument as before, noting that $[\Wsc, \wpa_\theta] = 0$, gives the bound
\begin{equation}
\int r^{1+\alpha}\left|v^2\Wsc\left(c\varepsilon\rho^2\left(\BThat^{\theta}*\wpa_\theta\uchi\right)\!\right)\right|^2\!\!\lesssim \varepsilon^2\left\lVert\BT\right\rVert_{[\alpha]}^2  \left(\!\mathcal{S}_1^{\Wsc}[\uchi]\! + \mathcal{S}[\uchi] \!+\! \int\!\! \frac{\Delta^2}{r^{5+\alpha}}\chi_{\rmax}\frac{|\tfrac{1}{r}\oW\wpa_\theta\uchi|^2}{\max(|\omega|, |m|, 1)^2}\right).
\end{equation}
When bounding the $\wpa_r$ term, we again replace \eqref{est:HdotBound} with a modified estimate
\begin{equation}\label{est:HdotBoundr}
\lVert\wpa_r\uchi\rVert_{\Hdot[0]} \leq \lVert\wpa_r\uchi\rVert_{\ell^2_m L^2_\omega}+ \left\lVert \chi_{\rmax}^{1/2}\frac{\oW\wpa_r\uchi}{\max(|m|, |\omega|, 1)}\right\rVert_{\ell^2_m L^2_\omega}.
\end{equation}
Additionally, the commutator $[\Wsc, \wpa_r]$ does not vanish, so we must replace \eqref{est:HighestOrderSecondConv} with the bound
\begin{align}\label{est:HighestOrderSeconddr}
\int r^{1+\alpha}\frac{\Delta^2}{r^4}\left| \BThat^r \!*\! \Wsc(\wpa_r\uchi)\right|^2 &\lesssim  \sup_{r, \theta}\Big( r^{1+\alpha}\lVert\BThat^r\rVert_{L^1_{\omega, m}}\Big)^2\int\frac{\Delta^2}{r^{5+\alpha}}\big(|\wpa_r\Wsc\uchi|^2 + |[\Wsc, \wpa_r]\uchi|^2\big), \nonumber \\
&\lesssim \lVert\BThat^r\rVert_{\Wbb[1+\alpha]}^2 \left(\mathcal{S}_1^{\Wsc}[\uchi] + \int\frac{\Delta^2}{r^{5+\alpha}}|[\Wsc, \wpa_r]\uchi|^2\right).
\end{align}
Carrying out the remainder of the argument as in the $\widetilde{\partial}_{\theta}$ case concludes our result.
\end{proof}
We turn our attention to the integrals on the right hand side of \eqref{est:PHBWeightedBound}. In order to prove the estimate on $|[\Wsc, \wpa_r]\uchi|^2$, we use the following bound:
\begin{lemma}\label{prop:WdrTildeCommutator}
We have the following estimate pointwise in frequency space:
\begin{equation}\label{est:WdrCommutatorBound}
|\tfrac{\Delta}{r^2+a^2}[\Wsc,\wpa_r]\,\ph|\lesssim |r^{-1}\oW\ph| + \tfrac{\Delta^{1/2}}{r^3}\big(|\omega| +|m|\big) .
\end{equation}
\end{lemma}
\begin{proof}
We expand $\Wsc = \oW + (\Wsc-\oW)$. Then, the identity \eqref{def:wpar} implies the commutator identity
\begin{equation}\label{id:CommutatorWzeroRegular}
\oW\widetilde{\pa}_r = \widetilde{\pa}_r\oW - \frac{r-M}{\Delta^{1/2}}\wpa_r = \widetilde{\pa}_r\oW - \frac{r-M}{\Delta}\oW,
\end{equation}
from which the bound
\begin{equation}
|\tfrac{\Delta}{r^2+a^2}[\oW,\wpa_r]\,\ph|\lesssim |r^{-1}\oW\ph| 
\end{equation}
directly follows.
The expansion \eqref{id:WMinusWZero} implies
\begin{align}\label{id:WdrDecomposition}
[\Wsc - \oW, \widetilde{\pa}_r] &= i\wpa_r\left(\wchi_1\omega\wfe + (1-\wchi_1)\chisharp\frac{\hsharp - h_0}{v}(\omega-\omega_r)\right) - \\
&\qquad -(1-\wchi_1)\left(\wpa_r\chisharp\oW - (1-\chisharp)\frac{r-M}{\Delta}\oW\right).\nonumber
\end{align}
The bound for all derivatives containing $\oW$ is straightforward, as are bounds on terms where $\wpa_r$ falls on $\chisharp$.

Next, we deal with terms containing $\wfe$. It follows from the bound \eqref{est:drfFinal} that for all $r\in(r_+, \infty)$
\begin{equation}\label{est:drfNaturalBound}
|\partial_r(\f - f_0)| \lesssim \frac{1}{r\Delta^{1/2}}.
\end{equation}

Next, we bound terms containing $\hsharp$. First, outside $\mathcal{G}'$, direct calculation gives
\begin{equation}\label{est:drHZeroMax}
\chisharp\left|\wpa_r\left(\tfrac{(\omega-\omega_r)(1 - h_0)}{v}\right)\right|\lesssim \frac{1}{r\Delta^{1/2}}(|\omega|+|\omega_r|),
\end{equation}
so \eqref{est:WdrCommutatorBound} for corresponding terms in \eqref{id:WdrDecomposition} directly follows.
For $(\omega, m)\in\mathcal{G}\cap\mathcal{G}'$, the estimate \eqref{est:drfFinal2}  implies that, in the support of $\chisharp$,
\begin{equation}
\chisharp\left|\wpa_r\left(\tfrac{(\omega-\omega_r)(\hsharp - h_0)}{v}\right)\right|\lesssim \frac{1}{r\Delta^{1/2}}(|m|)
\end{equation}
The bounds \eqref{est:drfNaturalBound} and \eqref{est:drHZeroMax} give our result.
\end{proof}
In order to bound the last term on the right hand side of \eqref{est:PHBWeightedBound}, we prove a Lagrangian estimate. As this will also be useful when bounding $\mathcal{N}_\sharp$, we state it in a more general form.
\begin{lemma}\label{lem:ZeroOrderInverseBound}
Let $\uchi$ be as in Theorem \ref{thm:CommutedEEMain}, and let $\overline{g}$ be a smooth nonnegative function of $r$ depending on $(\omega, m)$ such that $\overline{g}$ is supported for $r\in (r_+, \rmax+\delta^\sharp]$ and $\overline{g}$, $\partial_r\overline{g}$, and $\partial_r^2 \overline{g}$ are uniformly bounded. The bound
\begin{equation}\label{est:ZeroOrderInverseBound}
\int \overline{g}\tfrac{\Delta^2}{(r^2+a^2)^2} \left(\left|\tfrac{\oW\wpa_\theta\uchi}{\max(|\omega|, |m|, 1)}\right|^2+\left|\tfrac{\oW\wpa_r\uchi}{\max(|\omega|, |m|, 1)}\right|^2\right) \leq C\left( \mathcal{S}[\uchi] +\int \overline{g} \left|v^2\oW\left(\tfrac{\Hchi}{v^2}\right)\right|^2\right)
\end{equation}
holds for sufficiently small $\varepsilon$, where $C$ is a constant depending on $\overline{g}$ and its derivatives.
\end{lemma}
\begin{proof} We commute $\wpa_\theta, \wpa_r$ through $\oW$, as the commutator identity \eqref{id:CommutatorWzeroRegular} implies
\begin{equation}
\int \overline{g}\tfrac{\Delta^2}{(r^2+a^2)^2}\left|\tfrac{[\oW,\wpa_r]\uchi}{\max(|\omega|, |m|, 1)}\right|^2\lesssim \int \overline{g}\frac{\Delta}{r^2+a^2}|\wpa_r\uchi|^2 \lesssim \mathcal{S}[\uchi].
\end{equation}

Next we apply the Lagrangian identity \eqref{Div:Lagrangian} at each frequency to $\ph = \oW\uchi$ and
\[
g = \tfrac{\overline{g}}{(\max(|\omega|, |m|, 1)^2},
\]
noting that
\begin{equation}
\int \overline{g} \left|\tfrac{(\omega-\omega_r)\oW\ph}{\max(|\omega|, |m|, 1)}\right|^2\lesssim \int \overline{g} \left|\oW\ph\right|^2\lesssim \mathcal{S}[\uchi],
\end{equation}
as well as the outgoing boundary conditions \eqref{est:OBCderiv}, to obtain
\begin{equation}\label{est:InverseCommutators1}
\int \overline{g}\left|\tfrac{\Delta}{r^2+a^2}\tfrac{\wpa_\theta \oW \uchi}{\max(|\omega|, |m|, 1)}\right|^2 +\overline{g}\left|\tfrac{\Rs\oW \uchi}{\max(|\omega|, |m|, 1)}\right|^2\lesssim \mathcal{S}[\uchi] + \int  \overline{g}\left|\tfrac{\mathcal{P}_{a, M}\oW\uchi}{\max(|\omega|, |m|, 1)}\right|^2.
\end{equation}
Carefully expanding $\Pam\oW\uchi$ using \eqref{def:CommIdentityBasic}, noting that
\begin{equation}
\Ph^\star \uchi = v\left((\omega-\omega_r)^2\frac{h_0^2-1}{v^2}\right)' = cv\Rs((\omega-\omega_r)^2) = -2cv(\omega-\omega_r)\omega_r'
\end{equation}
 and bounding $\Pl^\star\uchi$ as in \eqref{est:Plupperbound}, noting that
\begin{equation}
v\omega f = (\omega-\omega_r)h_0,
\end{equation}
gives the bound
\begin{equation}\label{est:InverseCommutators2}
\int \overline{g}\left|\frac{\mathcal{P}_{a, M}\oW\uchi}{\max(|\omega|, |m|, 1)}\right|^2 \lesssim \mathcal{S}[\uchi] + \int \overline{g}\frac{\left|v^2\oW\left(\tfrac{H}{v^2}\right)\right|^2}{\max(|\omega|, |m|, 1)^2}.
\end{equation}
Linearity of $\oW$ along with the decomposition \eqref{def:WavePerturbedFreq} implies
\begin{equation}\label{est:InverseCommutators3}
\int \overline{g}\frac{\left|v^2\oW\left(\tfrac{H}{v^2}\right)\right|^2}{\max(|\omega|, |m|, 1)^2} \lesssim \int \overline{g}\frac{\left|v^2\oW\left(\tfrac{\HB}{v^2}\right)\right|^2}{\max(|\omega|, |m|, 1)^2} + \int \overline{g}\left|v^2\oW\left(\tfrac{\Hchi}{v^2}\right)\right|^2
\end{equation}
We expand \eqref{def:HDecompB} and \eqref{def:BuchiExpansion}, and bound terms as in Lemma \ref{lem:PHBBoundPrelim} to obtain
\begin{equation}\label{est:InverseCommutators4}
\int \overline{g}\frac{\left|v^2\oW\left(\tfrac{\HB}{v^2}\right)\right|^2}{\max(|\omega|, |m|, 1)^2}\lesssim \varepsilon^2 \mathcal{S}[\ph] + \varepsilon^2 \int \overline{g}\frac{\Delta^2}{(r^2+a^2)^2} \left(\left|\frac{\oW\partial_\theta\uchi}{\max(|\omega|, |m|, 1)}\right|^2+\left|\frac{\oW\wpa_r\uchi}{\max(|\omega|, |m|, 1)}\right|^2\right) .
\end{equation}
The estimates \eqref{est:InverseCommutators1}, \eqref{est:InverseCommutators2}, \eqref{est:InverseCommutators3}, and \eqref{est:InverseCommutators4} imply that, for some $C$ independent of $\varepsilon$,
\[
(l.h.s.) \eqref{est:ZeroOrderInverseBound} \leq C (r.h.s.)\eqref{est:ZeroOrderInverseBound}  + C\varepsilon^2(l.h.s.)\eqref{est:ZeroOrderInverseBound} 
\]
For sufficiently small $\varepsilon$ depending on $C$, the bound \eqref{est:ZeroOrderInverseBound} directly follows.
\end{proof}
We combine \eqref{est:PHBWeightedBound} with Lemmas \ref{prop:WdrTildeCommutator} and \ref{lem:ZeroOrderInverseBound}, with $\overline{g} = \chi_{\rmax}$. When bounding the last term in \eqref{est:PHBWeightedBound}, the support of $\chi_{\rmax}$ allows us to ignore the weight in $r$. It then follows that
\begin{equation}
\varepsilon^{-1}\int r^{1+\alpha}|\mathcal{P}_{\HB}|^2\lesssim \varepsilon\lVert \BT\rVert_{[\alpha]}^2\left(\mathcal{S}_1^{\Wsc}[\uchi]  + \mathcal{S}[\uchi] + \int\chi_{\rmax}^2\left|v^2\oW\left(\tfrac{\Hchi}{v^2}\right)\right|^2\right).
\end{equation}
Consequently, \eqref{est:CSPHB} implies
\begin{equation}
|\mathcal{N}[\Wsc\uchi, \PHB]|\lesssim \varepsilon\lVert \BT\rVert_{[\alpha]}^2\left(\mathcal{S}_1^{\Wsc}[\uchi]  +  \mathcal{S}[\uchi] + \int\chi_{\rmax}^2\left|v^2\oW\left(\tfrac{\Hchi}{v^2}\right)\right|^2\right) \, ,
\end{equation}
which proves (\ref{est:PHBBoundFirst}) for $\mathcal{N}[\Wsc\uchi, \PHB]$.

\subsubsection{Bounding $\mathcal{N}_{\sharp}[\HB, \HB]$}
To complete the proof of Proposition \ref{thm:PHBBoundFirst}, it suffices to prove the following:
\begin{lemma}\label{lem:NSharpBfBBound}
Under the assumptions of Theorem \ref{thm:CommutedEEMain}, $\uchi$ satisfies the bound
\begin{equation}
|\mathcal{N}_{\sharp}[\HB, \HB]|\lesssim C\varepsilon^2\lVert \BT\rVert_{[\alpha]}^2\left(\Ssharp[\uchi] + \mathcal{S}[\uchi] + \int\chi_{\rmin\!,\rmax}\left|v^2\oW\left(\tfrac{\Hchi}{v^2}\right)\right|^2\right).
\end{equation}
\end{lemma}
\begin{proof}
The bound follows from expanding \eqref{def:Nsharp} and \eqref{def:BuchiExpansion} and bound these term-by-term via Proposition \ref{lem:CMBound2}. We prove this for the $\wpa_r$ component of $\BT$ and the $\Wp$ term in $\mathcal{N}_\sharp$, noting that all other terms are more straightforward. Proposition \ref{lem:CMBound2} implies
\begin{align}
\int\left|v^2(1-\wchi_1)g_\sharp \Wp (\rho^2 \BThat^r*\wpa_r\uchi)\right|^2&\lesssim \int \left|v^2\rho^2\BThat^r *\left((1-\wchi_1)g_\sharp \Wp\wpa_r\uchi\right)\right|^2\\
&\qquad + \sup_{r, \theta}\left(v^4\left\lVert\rho^2\BThat^r\right\rVert_{\Wdot[0]}^2\chi_{\rmin\!,\rmax}\right)\iint_{\mcA} \left\lVert\wpa_r\uchi\right\rVert_{\Hdot[0]}^2\chi_{\rmin\!,\rmax}.\nonumber
\end{align}
Young's convolution inequality and uniform compact support of $g_\sharp$ in $r$ together imply that we can commute $[\Wp, \wpa_r]$ via the bound
\begin{equation}
\int \left|v^2\rho^2\BThat^r *\left((1-\wchi_1)g_\sharp [\Wp, \wpa_r]\uchi\right)\right|^2\lesssim \lVert\BT\rVert_{[\alpha]}^2\mathcal{S}[\uchi],
\end{equation}
and consequently
\begin{equation}
\int \left|v^2\rho^2\BThat^r *\left((1-\wchi_1)g_\sharp \Wp\wpa_r\uchi\right)\right|^2 \lesssim \lVert\BT\rVert_{[\alpha]}^2\left(\mathcal{S}[\uchi] + \mathcal{S}_1^\sharp[\uchi] \right).
\end{equation}
Finally, Young's convolution inequality with Lemma \ref{lem:ZeroOrderInverseBound} implies
\begin{equation}
\sup_{r, \theta}\left(v^4\left\lVert\rho^2\BThat^r\right\rVert_{\Wdot[0]}^2\chi_{\rmin\!,\rmax}\right)\iint_{\mcA} \left\lVert\wpa_r\uchi\right\rVert_{\Hdot[0]}^2\chi_{\rmin\!,\rmax}\lesssim \lVert\BT\rVert_{[\alpha]}^2\left(\mathcal{S}[\uchi] + \int  \left|v^2\chi_{\rmin\!,\rmax}\oW\left(\tfrac{\Hchi}{v^2}\right)\right|^2\right).
\end{equation}
The result follows.
\end{proof}

	\subsection{The initial data terms $H_\chi$}\label{sec:InitialDataBounds}
		We now look to the $\PHchi$ terms, which include error terms coming from the cutoff on the initial data. A naive approach would be to first take the bound
\[
\mathcal{N}[\Wsc\uchi, \PHchi] \lesssim \iiiint_{\FMs}|\omega\Wsc\uchi|^2 +|r^{-1}\Wsc\uchi|^2 + |\PHchi|^2
\]
and then convert to physical space via the $L^2$ Plancherel identity \eqref{id:Plancherel2}. However, one would then run into the problem that compact support of $|\mathcal{F}^{-1}\Hchi|^2$ in $\wtau$ does not imply compact support of $|\partial_t\Wsc\mathbf{u}_\chi|^2$ (or even of $|\mathcal{F}^{-1}\PHchi|^2$), so the resulting term cannot be bounded directly by the energy at time 0. Additionally, the lack of decaying radial weights in the $|\Wsc\uchi|^2$ integral would stymie attempts to bound this using an ILED estimate. This difficulty is in some sense unavoidable for frequency-dependent currents.

The approach used in \cite{DRSR} used a current which may be written as the sum of a physical current, which preserves compact support in time, plus a frequency-dependent current with weights decaying rapidly in $r$, for which an ILED bound may be used. We take a similar approach, noting that $\mathcal{N}$ is a physical current, and that $\mathcal{N}_\sharp$ is strongly frequency-dependent but compactly supported in $r$, allowing us to reduce this problem to that of bounding derivatives of $\Wsc\mathbf{u}_\chi$ for $\wtau \in [0,1]$. This is still nonlocal in time; however, the weak Coifman-Meyer bound given in Proposition \eqref{lem:CMBound} then allows us to decompose $\Wsc\mathbf{u}_\chi$ into the sum of an ``almost local'' term, which may be bounded by the initial energy via local energy bounds, plus a lower-order term with better spatial decay. A more precise formulation of this is given as Proposition \ref{lem:IDLemmaLargeR}. For a given function $\Psi$ solving \eqref{weintro}, we recall the energy $E_0^1$ given in \eqref{e10def}.

Our main result for this section will be the following:
\begin{proposition}\label{thm:PHchiBound}
Under the assumptions of Theorem \ref{thm:CommutedEEMain}, for any $\delta > 0$ there exists a $C_\delta$ such that
\begin{equation}\label{est:PHchiBound}
\mathcal{N}[\Wsc\uchi, \Hchi] + \mathcal{N}_\sharp[\Hchi, \Hchi] \leq \delta \left( \mathcal{S}_1^{\Wsc}[\uchi] + \mathcal{S}[\uchi]\right) + C_\delta E_0^1[\Psi].
\end{equation}
\end{proposition}
As before, we split the proof into bounds for $\mathcal{N}$ and $\mathcal{N}_\sharp$, and prove them in Lemmas \ref{lem:IDBoundsFirst} and \ref{lem:PHchiBoundLocal} respectively.

By local energy bounds, it suffices to prove Proposition \ref{thm:PHchiBound} with $E_0^1[\Psi]$ replaced by
\begin{equation}
\mathcal{E}^1 [\Psi]= \int_0^2 E_\wtau[\Psi] + E_\wtau[W_0 \Psi] + E_\wtau[T\Psi]+ E_\wtau[\Phi\Psi]\, d\wtau
\end{equation}
on the right hand side. We define the physical space quantity
\begin{equation}\index{Hchicheck@$\Hchicheck$}
\Hchicheck = \mathcal{F}^{-1}\Hchi.
\end{equation}
We note that $\Hchicheck$ is supported when $\wtau\in[0,1]$, but $\mathcal{F}^{-1}(\PHchi)$ is not!

\subsubsection{Expanding $\Hchi$ in regular coordinates}
We give a useful expansion of $\Hchi$ in physical space using the regular coordinates defined in Section \ref{sec:hyperboloidal}. This will be used in the proof of Proposition \ref{prop:PSBoundIDWorst}.
\begin{proposition}\label{prop:Hchi}
Under the assumptions of Theorem  \ref{thm:CommutedEEMain} with $H_\chi$ as defined in \eqref{def:HDecomp}, the identity
\begin{equation}\label{def:HChiExpansion}
H_\chi = \mathcal{F}\left(\chi'\left(\frac{-2h_0}{r^2+a^2}\wpa_{r}\unochi + F^{\wtau}\wpa_{\wtau} \unochi + F^{\wphi}\wpa_{\wphi}\unochi  + G_1 \unochi\right) + \chi''\left(G_2\unochi\right) +\chi'G_{\BfB}\unochi\right)
\end{equation}
holds for functions $F^{\wtau}, F^{\wphi}, G_1, G_2$ of $(r, \theta)$ which satisfy the bounds
\begin{subequations}\label{est:FGBounds}
\begin{align}
|F^{\wtau}| + |F^{\wphi}| + |G_1| + |G_2|&\lesssim v^2, \\
|\Rs (v^{-2}F^{\wtau})|+ |\Rs(v^{-2}F^{\wphi})| + |\Rs(v^{-2} G_1)| + |\Rs(v^{-2} G_2)| &\lesssim v^2 
\end{align}
\end{subequations}
and 
\begin{equation}\label{est:GBfBBounds}
\left|\frac{G_{\BfB}}{v^2}\right| + \left|\oW\left(\frac{G_{\BfB}}{v^2}\right)\right|\lesssim \varepsilon r^{1-\alpha} \, .
\end{equation}
Here $\unochi=\sqrt{r^2+a^2}\Psi_{\mathcal{T}}$\index{ut@$\unochi$}.
\end{proposition}
\begin{proof}
For a generic smooth nonvanishing function $\kappa$, we write
\begin{equation}
[\Box_g, \chi]\Psi_{\mathcal{T}} = [\Box_g, \kappa^{-1}\chi]\kappa\Psi_{\mathcal{T}} + \kappa^{-1}\chi[\Box_g,\kappa]\Psi_{\mathcal{T}}.
\end{equation}
Additionally,
\begin{equation}
\kappa^{-1}[\Box_g,\kappa]\Psi_{\mathcal{T}} = -\Box_g(\kappa^{-1})(\kappa\Psi_{\mathcal{T}}) +2\frac{\nabla^\alpha\kappa}{\kappa}\frac{\nabla_{\alpha}(\kappa\Psi_{\mathcal{T}})}{\kappa}.
\end{equation}
Therefore, setting $\kappa = \rweight$,
\begin{equation}\label{id:HChiNew}
[\Box_g, \chi]\Psi_{\mathcal{T}}  = 2\nabla^\alpha\chi\frac{\nabla_\alpha(\unochi)}{\rweight}+ \big(\Box_g(\rweight^{-1}\chi) - \chi\Box_g(\rweight^{-1})\big)(\unochi).
\end{equation}
For spherically symmetric $\psi$, the identity \eqref{def:SepWaveOperator2} implies
\begin{equation}
\mathcal{F}(\Box_{g_{M,a}}\psi) = \frac{\rweight^3}{\Delta\rho^2}\big(\Rs\Rs\ph + \omega^2 \big(1-\tfrac{a^2\sin^2\theta \Delta}{\rweight^4}\big)\ph - V_1\ph \big).
\end{equation}
Therefore, one has in physical space
\begin{equation}
\Box_{g_{M,a}}(\rweight^{-1}) = -\tfrac{\rweight^3}{\Delta\rho^2}V_1, \qquad
\Box_{g_{M,a}}(\rweight^{-1}\chi ) = \tfrac{\rweight^3}{\Delta\rho^2}\left(\Rs\Rs\chi -  \left(1-\tfrac{a^2\sin^2\theta \Delta}{\rweight^4}\right)\partial_t^2\chi - V_1\chi \right),
\end{equation}
and consequently
\begin{equation}
 \left(\Box_{g_{M,a}} \big(\tfrac{\chi}{\rweight}\big) - \chi\Box_{g_{M,a}}\big(\tfrac{1}{\rweight}\big)\right)u = \frac{\rweight^3}{\Delta\rho^2}\big(\Rs\Rs\chi -  \big(1-\tfrac{a^2\sin^2\theta \Delta}{\rweight^4}\big)\partial_t^2\chi\big)u.
\end{equation}
If $\chi$ is a function of $\wtau$ only then the definition \eqref{def:timefoliation} implies
\begin{equation}
\Rs\Rs\chi -  \big(1-\tfrac{a^2\sin^2\theta \Delta}{\rweight^4}\big)\partial_t^2\chi = -h_0'\chi' + (h_0^2 - 1 + a^2\sin^2\theta v^2)\chi''.
\end{equation}
Writing $\chi = \chi(\wtau)$ and expanding $h_0^2-1$ using \eqref{def:h0} gives the identity
\begin{equation}\label{id:HChiLOT}
\Rs\Rs\chi -  \big(1-\tfrac{a^2\sin^2\theta \Delta}{\rweight^4}\big)\partial_t^2\chi = G_1(r, \theta)\chi' + G_2(r, \theta)\chi'',
\end{equation}
for some bounded functions $G_1, G_2$ satisfying
\begin{equation}
|G_1|+|G_2|\lesssim v^2.
\end{equation} 
Next, since $\chi$ is independent of $\phi, \theta$, we may expand
\begin{equation}
\frac{2\nabla^\alpha\chi}{\rweight}\nabla_\alpha (\unochi) = 2\chi'\frac{\rweight^3}{\Delta\rho^2}\left(\left(-1+\frac{a^2\Delta\sin^2\theta}{(r^2+a^2)^2}\right)\partial_t(\unochi) - \frac{2Mar}{(r^2+a^2)^2}\partial_\phi(\unochi) - h_0\Rs(\unochi)\right).
\end{equation}
We convert this to the regular hyperboloidal coordinates of Section \ref{sec:hyperboloidal} to obtain
\begin{align}
\left(-1+\tfrac{a^2\Delta\sin^2\theta}{(r^2+a^2)^2}\right)\partial_t\unochi - \tfrac{2Mar}{(r^2+a^2)^2}\partial_\phi\unochi - h_0\Rs\unochi  =  \tfrac12\left(F^{\wtau}\wpa_{\wtau}\unochi  + F^{\wphi}\wpa_\wphi\unochi -\tfrac{2h_0}{r^2+a^2}\wpa_{r}\unochi\right),
\end{align}
where
\begin{subequations}
\begin{align}
\tfrac12F^t &= -a_\theta^2v^2\\
\tfrac12F^{\wphi} &=av^2\tfrac{\rho^2-a_\theta^2}{r^2+a^2}
\end{align}
\end{subequations}
The corresponding bounds \eqref{est:FGBounds} then follow from the basic Leibniz rule \eqref{id:WBasicLeibniz}.
The bound on the $G_{\BfB}$ term follows from direct calculation, as in regular coordinates
\begin{equation}
\chi_{\mathcal{T}}\BfB(\chi) = \chi'\BT^{\wtau}, \qquad v^{-2}G_\BfB = \varepsilon\rho^2\BT^{\wtau}.
\end{equation}
\end{proof}

\subsubsection{Bounding $\mathcal{N}[\Wsc\uchi, \PHchi]$}
We first prove the following estimate, using the notation of the decomposition \eqref{def:N}:
\begin{lemma}\label{lem:IDBoundsFirst}
Under the assumptions of Theorem  \ref{thm:CommutedEEMain}, for every $\delta > 0$ there exists a constant $C_\delta$ depending on $\delta$ such that
\begin{equation}\label{est:NChiMainBound}
\mathcal{N}[\Wsc\uchi, \PHchi]\leq \delta\left( \mathcal{S}_1^{\Wsc}[\uchi] + \mathcal{S}[\uchi]\right) + C_\delta \mathcal{E}^1[\Psi] \, .
\end{equation}
\end{lemma}
In order to bound the quantities $\mathcal{N}[\Wsc\uchi, \Hchi]$, we must use the fact that the energy current is independent of frequency up to terms compactly supported in $r$. We first write
\begin{align}
|\mathcal{N}_T[\Wsc\uchi, \PHchi]| &\leq \int\left|\omega\Wsc\uchi\overline{v^2(\Wsc-\oW)\left(\tfrac{\Hchi}{v^2}\right)}\right| + \left|\int\mathfrak{I}\Big(\omega\Wsc\uchi\overline{v^2\oW\left(\tfrac{\Hchi}{v^2}\right)}\Big)\right|, \label{NTchiFirstBound}\\
|\mathcal{N}_\alpha[\Wsc\uchi, \PHchi]| &\leq \int \left|\frac{1}{r}\Wsc\uchi\overline{v^2(\Wsc-\oW)\left(\tfrac{\Hchi}{v^2}\right)}\right| + \left|\int \mathfrak{R}\Big(\frac{1}{r}\Wsc\uchi\overline{v^2\oW\left(\tfrac{\Hchi}{v^2}\right)}\Big)\right|.\label{NalphachiFirstBound}
\end{align}

The first term on the right hand sides of \eqref{NTchiFirstBound} and \eqref{NalphachiFirstBound} are straightforward to bound using the fact that $\Wsc$ decays to a physical space operator. More specifically, a term-by-term expansion using the expansion \eqref{id:WdifferenceBasic} along with Proposition \ref{prop:vfDerUpperBound} and Proposition \ref{prop:Hchi}, followed by the Plancherel identity \eqref{id:Plancherel2}, implies
\begin{equation}
\int r^{1+\alpha}\left|v^2(\Wsc-\oW)\left(\frac{\Hchi}{v^2}\right)\right|^2 \lesssim \mathcal{E}^1[\Psi].
\end{equation}
Consequently,
\begin{align}
\int\left|\omega\Wsc\uchi\overline{v^2(\Wsc-\oW)\left(\tfrac{\Hchi}{v^2}\right)}\right|&\leq \int \delta r^{-1-\alpha}|\omega\Wsc\uchi|^2 + \delta^{-1} r^{1+\alpha}\left|v^2(\Wsc-\oW)\left(\tfrac{\Hchi}{v^2}\right)\right|^2, \nonumber \\
&\leq \delta\mathcal{S}_1^{\Wsc}[\uchi] + C\delta^{-1}\mathcal{E}^1[\Psi], \label{est:IDDifferenceBound1}\\
\int\left|r^{-1}\Wsc\uchi\overline{v^2(\Wsc-\oW)\left(\tfrac{\Hchi}{v^2}\right)}\right|&\leq \int \delta r^{-3-\alpha}|\Wsc\uchi|^2 + \delta^{-1} r^{1+\alpha}\left|v^2(\Wsc-\oW)\left(\tfrac{\Hchi}{v^2}\right)\right|^2, \nonumber \\
&\leq \delta\mathcal{S}[\uchi] + C\delta^{-1}\mathcal{E}^1[\Psi].\label{est:IDDifferenceBound2}
\end{align}
Therefore, the proof of Lemma \ref{lem:IDBoundsFirst} may be reduced to bounding the second term on the right in (\ref{NTchiFirstBound}) and (\ref{NalphachiFirstBound}), for which we will need to exploit the support in time of $\Hchi$.

Recalling the frame $(\wtau, r, \theta, \widetilde{\phi})$ of Section \ref{sec:hyperboloidal}, we define the region
\begin{equation}\label{def:Dab}\index{Dab@$D_a^b$}
D_a^b = \{(\wtau, r, \theta, \widetilde{\phi}) : a \leq \wtau \leq b\},
\end{equation}
and correspondingly the integral
\begin{equation}
\int_{D_a^b} u = \int_a^b\int_{-\infty}^\infty \int_0^{2\pi}\int_0^\pi u \sin\theta\, d\theta\,d\widetilde{\phi}\,  d\rs\, d\wtau.
\end{equation}

In each case, the Plancherel identity \eqref{id:Plancherel1} and support for $\oW\left(\tfrac{\Hchicheck}{v^2}\right)$ in $\wtau\in[0,1]$ imply
\begin{align}
\left|\int \mathfrak{I}\left(\omega\Wsc\uchi\overline{v^2\oW\left(\tfrac{\Hchi}{v^2}\right)}\right)\right| \leq \delta \int_{D_0^1}|\partial_t\Wsc\mathbf{u}_\chi|^2 + \delta^{-1}\int_{D_0^1}\left|v^2\oW\left(\tfrac{\Hchicheck}{v^2}\right)\right|^2, \label{est:ImIDPrelim}\\
\left|\int \mathfrak{R}\left(r^{-1}\Wsc\uchi\overline{v^2\oW\left(\tfrac{\Hchi}{v^2}\right)}\right)\right| \leq \delta \int_{D_0^1}|r^{-1}\Wsc\mathbf{u}_\chi|^2 + \delta^{-1}\int_{D_0^1}\left|v^2\oW\left(\tfrac{\Hchicheck}{v^2}\right)\right|^2.\label{est:RealIDPrelim}
\end{align}

In order to bound the first term on the right hand side of \eqref{est:ImIDPrelim} and \eqref{est:RealIDPrelim} we first need the following proposition:
\begin{proposition}\label{lem:IDLemmaLargeR} Let $u$ be a function satisfying \eqref{sqi} for $N = 1$ which is supported for $\wtau > 0$, and let $a_\sharp$ be a nonnegative function supported in the region $r\in[\rmax+\delta^\sharp, \infty)$. Then, the inequality
\begin{equation}\label{est:IDLemmaLargeR}
\int_{D_0^1} a_\sharp^2|\Wsc u|^2 \lesssim \int_{D_0^2}a_\sharp^2(|\oW u|^2 + |vT u|^2+|v\Phi u|^2) + \iiiint_{\WMs}a_\sharp^2|vu|^2
\end{equation}
holds provided that the right hand side is finite.
\end{proposition}
\begin{proof}
We define a smooth nonnegative cutoff $\chi_1$ (consistent with $\chi_{\mathcal{T}}$ for ${\mathcal{T}} = 1$) satisfying
\[
\chi_1(y) = \begin{cases}
1 & 0 \leq y \leq 1, \\
0 & y \geq 2 \text{ or } y \leq -1.
\end{cases}
\]
Pointwise bounds in physical space and the Plancherel identity \eqref{id:Plancherel4} imply that
\begin{equation}
\int_{D_0^1} a_\sharp^2| \Wsc u|^2  \leq \iiiint_{\WMs} a_\sharp^2|\chi_1 \Wsc u|^2 = \int a_\sharp^2|\widehat{\chi}_1 * \Wsc(\ph)|^2 \, .
\end{equation}
We now apply Proposition \ref{lem:CMBound} with $q = 1$ to obtain
\begin{equation}
\int a_\sharp^2|\widehat{\chi}_1 * \Wsc(\ph)|^2 \lesssim \int a_\sharp^2|\Wsc(\widehat{\chi}_1 * \ph)|^2 + \sup_{(r, \theta)}\left(\left\lVert\widehat{\chi}_1\right\rVert_{\Wdot[-1]}^2 \right) \cdot\iint_{\mcAs}a_\sharp^2\left\lVert\ph\right\rVert_{\Hdot[1]}^2 \, .
\end{equation}
 Since $\chi_1$ is a function of $\wtau$ only, we note
\begin{equation}
\sup_{r, \theta} \left\lVert \widehat{\chi}_1\right\rVert_{\Wdot[-1]}^2 \leq C ,
\end{equation}
where $C$ depends on the definition of $\chi_1$. Additionally, since the supports of $\chi_{\rmax}$ and $a_\sharp$ are disjoint, the Plancherel identity \eqref{id:Plancherel4} implies 
\begin{equation}
\iint_{\mcAs}a_\sharp^2\left\lVert \ph\right\rVert_{\Hdot[1])} \lesssim \int a_\sharp^2|v\ph|^2 \lesssim \iiiint_{\WMs}a_\sharp^2|vu|^2.
\end{equation}
Finally, Corollary \ref{prop:WMinusWZeroBounds}, along with \eqref{id:Plancherel4} and the fact that $\oW\chi_1 = \partial_\phi\chi_1 = 0$, implies
\begin{align}
\int a_\sharp^2| \Wsc(\widehat{\chi}_1 *\ph)|^2 &\lesssim \int a_\sharp^2\left(| \oW(\widehat{\chi}_1 *\ph)|^2 + |\omega v(\widehat{\chi}_1 *\ph)|^2+|m v(\widehat{\chi}_1 *\ph)|^2\right), \\
&\lesssim \iiiint_{\WMs}a_\sharp^2\left(|\chi_1\oW  u|^2 + |(\partial_t \chi_1) v u|^2+|\chi_1 v \partial_t u|^2+|\chi_1 v \partial_{\phi} u|^2\right). \nonumber
\end{align}
 The estimate\eqref{est:IDLemmaLargeR} then directly follows from pointwise bounds in physical space.
\end{proof}

\begin{corollary}\label{prop:PSBoundCompact}
Under the assumptions of Theorem \ref{thm:CommutedEEMain}, for every $\delta > 0$ there exists a constant $C_\delta$ depending on $\delta$ such that
\begin{equation}\label{est:PSBoundCompact}
\int_{D_0^1}|\partial_t\Wsc\mathbf{u}_\chi|^2 + r^{-2}|\Wsc\mathbf{u}_\chi|^2 \lesssim \mathcal{E}^1[\Psi] + \mathcal{S}[\uchi]+ \mathcal{S}_1^W[\uchi]
\end{equation}
\end{corollary}
\begin{proof}
We split this into the regions $r \leq \rmax+\delta^\sharp$ and $r \leq \rmax + \delta^\sharp$. For $r \leq \rmax+\delta^\sharp$, this follows directly from the definitions of $\mathcal{S}[\uchi], \mathcal{S}_1^W[\uchi]$, and boundedness of $r$. For $r \geq \rmax + \delta^\sharp$, the result follows from Proposition \ref{lem:IDLemmaLargeR} applied for $\ph = \partial_t\uchi$, $a_\sharp = 1$ and $\ph = \uchi, a_\sharp = r^{-1}$ respectively (where $a_\sharp$ is supported in the appropriate region in each case).
\end{proof}
In order to complete the proof of Lemma \ref{lem:IDBoundsFirst} it therefore suffices to bound the second term on the right hand side of \eqref{est:ImIDPrelim} and \eqref{est:RealIDPrelim}.
\begin{proposition}\label{prop:PSBoundIDWorst}
Under the assumptions of Theorem \ref{thm:CommutedEEMain}, for every $\delta > 0$ there exists a constant $C_\delta$ depending on $\delta$ such that
\begin{equation}\label{est:PSBoundIDWorst}
\iiiint_{\WMs}\left|v^2\oW\left(\frac{\Hchicheck}{v^2}\right) \right|^2 \lesssim \mathcal{E}^1[\Psi].
\end{equation}
\end{proposition}
\begin{proof}
We take the expansion \eqref{def:HChiExpansion} and commute $\oW$ through $\chi'$, $\chi''$ to obtain
\begin{equation}
\oW\left(\frac{\Hchicheck}{v^2}\right) = \chi'\oW\left(-\frac{2}{v}h_0\oW\bfuT + \frac{F^{\wtau}}{v^2}\wpa_{\wtau}\bfuT + \frac{F^{\wphi}}{v^2}\wpa_{\wphi} \bfuT + \frac{G_1}{v^2}\bfuT + \frac{G_{\BfB}}{v^2}\bfuT\right) + \chi''\oW\left(\frac{G_2}{v^2}\bfuT\right) \, .
\end{equation}
Physical space calculations give the bound
\begin{equation}\label{est:oWchioneBounds}
\left|v^2\oW\left(\frac{h_0}{v}\oW\unochi\right)\right| \lesssim |r^{-1}\oW\unochi|+|v\oW\oW\unochi|.
\end{equation}
Compact support of $\chi'$ in $\wtau$ implies
\begin{equation}
\iiiint_{\WMs}\left|v^2\chi'\oW\left(\frac{h_0\oW(\unochi)}{v}\right) \right|^2 \lesssim \mathcal{E}^1[\Psi].
\end{equation}
Similar calculations using the bounds \ref{est:FGBounds} imply
\begin{equation}
\iiiint_{\WMs}\left|v^2\chi'\oW\left(\frac{F^{\wtau}}{v^2}\partial_t\bfuT + \frac{F^{\wphi}}{v^2}\partial_\phi \bfuT + \frac{G_1}{v^2}\bfuT\right) + v^2\chi''\oW\left(\frac{G_2}{v^2}\bfuT\right)\right| \lesssim \mathcal{E}^1[\Psi],
\end{equation}
and \ref{est:GBfBBounds} implies
\begin{equation}
\iiiint_{\WMs}\left|v^2\oW\left(\frac{G_{\BfB}}{v^2}\bfuT\right)\right|^2\lesssim \mathcal{E}^1[\Psi] \, ,
\end{equation}
which completes our proof.
\end{proof}

To complete the proof of Lemma \ref{lem:IDBoundsFirst}, we split $\mathcal{N} = \mathcal{N}_T + \mathcal{N}_\alpha$ as in (\ref{def:N}) and take the preliminary bounds \eqref{NTchiFirstBound} and \eqref{NalphachiFirstBound}, bound the first term on the right hand side of each inequality using the bounds \eqref{est:IDDifferenceBound1} and \eqref{est:IDDifferenceBound2}, and bound the second term by taking the preliminary expansions \eqref{est:ImIDPrelim} and \eqref{est:RealIDPrelim} and bounding each term in the expansions via Corollary \ref{prop:PSBoundCompact} and Proposition \ref{prop:PSBoundIDWorst}. The bound \eqref{est:NChiMainBound} follows (up to a constant multiple of $\delta$).

\subsubsection{Bounding $\mathcal{N}_\sharp[\Hchi, \Hchi]$} To conclude the proof of Proposition \ref{thm:PHchiBound} it suffices to prove the following:
\begin{lemma}\label{lem:PHchiBoundLocal} Under the assumptions of Theorem \ref{thm:CommutedEEMain} we have the bound
\begin{equation}\label{PHchiBoundLocal}
\mathcal{N}_\sharp[\Hchi, \Hchi]  \leq  C\mathcal{E}^1[\Psi] \, . 
\end{equation}
\end{lemma}
\begin{proof}The bound \eqref{est:gsharpprimebound} and the definition of $\mathcal{N}_\sharp$ imply
\begin{align}
\mathcal{N}_\sharp[\Hchi, \Hchi] \lesssim \int \chi_{\rmin\!,\rmax}\left(|\omega \Hchi|^2+|m\Hchi|^2+\left|v^2\Wp \left(\tfrac{\Hchi}{v^2}\right)\right|^2 \right).
\end{align}
This is local in time in physical space, so Plancherel's theorem combined with the decomposition of Proposition \ref{prop:Hchi} proves our result.
\end{proof}
\section{The proof of the convolution estimates}\label{sec:CMConvolution}
		Here we prove Propositions \ref{lem:CMBound} and \ref{lem:CMBound2}. As the proof of Proposition \ref{lem:CMBound} in particular is rather lengthy, we give a brief model estimate for functions in $\mathbb{R}$.

 Suppose $P = P(\xi) \in S_{1,0}^1$ is the symbol of a first-order pseudodifferential operator. Then, for functions $u$ and $w$ with Fourier transforms $\ph, \widehat{w}$, Lipschitz continuity of $P$ (via global bounds for $|\partial_\xi P|$) implies
\begin{align}
\big|\big(P(\ph * \widehat{w}) - \widehat{w}*(P\ph)\big)(\xi) \big|&= \left|\int P(\xi)\widehat{u}(\xi')\widehat{w}(\xi-\xi') - P(\xi)(\widehat{u}(\xi')\widehat{w}(\xi-\xi')\, d\xi'\right|, \\
&\leq \sup_{\xi, \xi'}\left|\frac{P(\xi)-P(\xi')}{\xi-\xi'}\right| \int \big|\widehat{u}(\xi')(\xi-\xi')\widehat{w}(\xi-\xi')\big|d\xi',\nonumber\\
&\leq  \Big(\sup|\partial_\xi P|\Big)(|\ph|*|\widehat{\pa w}|),\nonumber
\end{align}
 i.e.~the first-order terms with $\ph$ vanish. Young's convolution inequality then implies
\begin{equation}\label{est:CIToy}
\lVert P(\ph * \widehat{w}) - \widehat{w}*(P\ph)\rVert_{L^2}\lesssim \lVert\widehat{\partial w}\rVert_{L^1}\lVert\ph\rVert_{L^2},
\end{equation}
so for sufficiently regular $w$ we can bound this difference without derivatives of $u$.\footnote{The estimate in \cite{CM78} in fact gives a stronger version of \eqref{est:CIToy}, replacing $\lVert\widehat{\partial w}\rVert_{L^1}$ with the weaker norm $\lVert\partial w\rVert_{L^\infty}$.}

In order to prove Proposition \ref{lem:CMBound}, it suffices to prove the similar bound
\begin{equation}\label{est:CIMain}
\lVert\Wsc(\widehat{w}*\ph) - \widehat{w}*(\Wsc(\ph))\rVert_{\ell^2_m L^2_\omega} \lesssim \lVert \widehat{w}\rVert_{\Wdot}\Vert \ph \rVert_{\Hdot}
\end{equation}
for each $(r, \theta)$, as the inequality \eqref{est:CIGlobal} then follows from squaring both sides of \eqref{est:CIMain}, multiplying by $a^2$, and applying H\"older's inequality in $(r, \theta)$. We remark that here we must handle physical weights and derivatives, as well as a Fourier series decomposition in $m$, and so the proof is somewhat more involved, and does not follow directly from the classical estimate even after establishing regularity for $\Wsc$.

We recall the decomposition \eqref{id:WMinusWZero}, which will guide our approach. In Sections \ref{PDObounds} and \ref{sec:CutoffBounds}, we prove regularity bounds for terms appearing in $\mathcal{G}'$ and $\mathcal{G}$ respectively, including estimates which will allow us to interpolate between these regimes. We put these together in Section \ref{sec:WGlobalBounds} to obtain an estimate which works across all frequencies.

The proof of Proposition \ref{lem:CMBound2} is similar but substantially simpler, as uniform compact support in $r$ means that we do not have to worry about the asymptotic behavior as $r\to r_+$ or $r\to\infty$. We prove this result in Section \ref{sec:AppendixLocalizedEnergy}.

Throughout the section, given two pairs $(\omega,m)$ and $(\omega^\prime, m^\prime)$ we will use  the shorthands
\begin{equation}\label{def:FreqDoublePrime}
\omega'' = \omega - \omega', \qquad m'' = m-m',
\end{equation}
as well as
\begin{equation}\label{def:OmegarDoublePrime}
\omega_r' = \frac{am'}{r^2+a^2}, \qquad \omega_r'' = \frac{am''}{r^2+a^2}.
\end{equation}
\subsection{Frequency regularity estimates in in $\mathcal{G}'$}\label{PDObounds}
We recall the nonlocal error $\wfe$ from \eqref{def:ferror}. When the frequency coefficients $(\omega, m)$ for $\wfe$ (or for other frequency-dependent quantities) are not implicit, including estimates taken over a range of $\omega, m$, we will specify
\begin{equation}
\widetilde{f}_m^\omega\index{f@$\widetilde{f}_m^\omega$} := \wfe, \text{ etc.}
\end{equation}
For $m = 0$, $\wchi_1 = 1$ and $\Wsv = \oW$, so consequently
\begin{equation}\label{id:wfezero}
\widetilde{f}_0^\omega = 0 \, . 
\end{equation}
In order to capture the decay of $\wfe$ for $m \neq 0$, we divide the interval $(r_+, \infty)$ into three regions, recalling that, for fixed $\mathcal{G}'$, the set of potentially trapping radii $\{\re\}$ is contained within a compact interval $I^1_{trap} = [R_-, R_+]$. We therefore choose $R_-' \in (r_+, R_-)$ and $R_+' \in (R_+. \infty)$, depending only on $a, M, \delta$ and consequently the intervals
\begin{equation}\label{def:Intervals}
I_- = (r_+, R_-) , \qquad I_0 = (R_-', R_+'), \qquad I_+ = (R_+, \infty).
\end{equation}
Thus, trapping takes place only in $I_0$, within which we may use the inequality $v^{-1} \leq C v$.

We first use Proposition \ref{prop:vfDerUpperBound} to show decay for $\wfe$ at $r=r_+$ and $r=\infty$, as well as $\omega\to\infty$.
\begin{corollary}\label{prop:absbounds}
In $\mathcal{G}'$ the quantity $\wfe$ satisfies
\begin{equation}\label{est:wfebound}
|\wfe|\leq C|\eta|v,
\end{equation}
where $C$ is independent of frequency.
\end{corollary}
\begin{proof}
This follows from $\wfe(r,\eta=0)=0$ and the bound \eqref{est:detafFinal} along with Taylor's theorem.
\end{proof}
Next, we characterize the regularity of $\wfe$. 
\begin{lemma}\label{prop:regineta}
For $(\omega, m), ({\omega'}, m) \in \mathcal{G}'\setminus (0,0)$ the quantity $\wfe$ satisfies the following estimate:
\begin{equation}\label{prop:CommutatorPointwise}
|\omega\widetilde{f}_m^\omega - {\omega'}\widetilde{f}_m^{\omega'}| \lesssim a^2m^2\frac{|\omega''|}{|\omega||{\omega'}|}v \, .
\end{equation}
\end{lemma}
\begin{proof}
For $m=0$ the estimate follows trivially from \eqref{id:wfezero}.

For $m\neq 0$, we split the problem into two cases, depending on whether $\omega, {\omega'}$ have the same sign.
First, we prove this when $\omega, {\omega'}$ have the same sign. In this case we can integrate $|\eta|^2$ in $\omega$ directly and hence (\ref{prop:CommutatorPointwise}) will follow directly from  
\begin{equation}\label{est:domegaf}
|\partial_\omega(\omega \widetilde{f}_m^\omega)|\lesssim v|\eta|^2 \, .
\end{equation}
Writing $\omega=\frac{am}{\eta}$, we first have
\begin{equation}
\partial_\omega(\omega \widetilde{\f}) = \wfe - \eta\partial_\eta\wfe \, .  
\end{equation}
The estimate (\ref{est:domegaf}) now follows from Taylor expanding the right hand side (noting that the linear term cancels) and bounding the remainder term using \eqref{est:detafFinal} with $k=2$. When $\omega, {\omega'}$ have different signs, we expand
\begin{equation}\label{est:CommutatorPointwiseOpposite}
|\omega\widetilde{f}_m^\omega - {\omega'}\widetilde{f}_m^{\omega'}| \leq |\omega\widetilde{f}_m^\omega - am(\partial_\eta\wfe)(r, 0)| + | {\omega'}\widetilde{f}_m^{\omega'}-am(\partial_\eta\wfe)(r, 0)|,
\end{equation}
writing $\wfe = \wfe(r, \eta)$. Taylor's theorem with \eqref{est:detafFinal} again gives
\begin{equation}
|\wfe - \eta(\partial_\eta\wfe)(r, 0)|\leq \frac{|\eta|^2}{2}\sup_\eta |\partial_\eta^2\wfe|\leq Cv|\eta|^2.
\end{equation}
Multiplying this bound by $|\omega|$ and applying it to \eqref{est:CommutatorPointwiseOpposite} gives
\begin{equation}
|\omega\widetilde{f}_m^\omega - {\omega'}\widetilde{f}_m^{\omega'}| \lesssim \frac{a^2m^2}{|\omega|}v+\frac{a^2m^2}{|\omega'|}v \lesssim \frac{a^2m^2}{\min(|\omega|, |\omega'|)}v \, .
\end{equation}
The bound $|\omega''| \geq \max(|\omega|,| {\omega'}|)$ yields \eqref{prop:CommutatorPointwise}.
\end{proof}
Next, we use this for a more general result. 
\begin{lemma}\label{prop:reginetanew}
Take two frequency pairs $(\omega, m)$ and $(\omega', m')$ such that $(\omega, m) \in \mathcal{G}'$.
Then, if $({\omega'}, m')\in\mathcal{G}'$ we have
\begin{equation}\label{est:reginetafirst}
\left| {\omega'} \widetilde{f}_{m'}^{\omega'} - \frac{m'\omega}{m}\widetilde{f}_m^\omega\right|\lesssim (|\omega''| + |m''|)v.
\end{equation}
If $(\omega'', m'')\in\mathcal{G}'$ we have
\begin{equation}\label{est:reginetasecond}
\left| \omega''\widetilde{f}_{m''}^{\omega''} - \frac{m''\omega}{m}\widetilde{f}_m^\omega\right|\lesssim (|\omega''| + |m''|)v.
\end{equation}
In the case $m = 0$, given $({\omega'}, m')\in\mathcal{G}'$ we have
\begin{equation}\label{est:reginetazerofirst}
\left| \omega \widetilde{f}_0^\omega - {\omega'}\widetilde{f}_{m'}^{{\omega'}}\right|\lesssim |m''|v,
\end{equation}
and given $(\omega'', {m''})\in\mathcal{G}'$ we have
\begin{equation}\label{est:reginetazerosecond}
\left| \omega \widetilde{f}_0^\omega - {\omega''}\widetilde{f}_{m''}^{{\omega''}}\right|\lesssim |m''|v.
\end{equation}
\end{lemma}
\begin{proof}

The estimate \eqref{est:reginetasecond} follows directly from Corollary \ref{prop:absbounds} and boundedness of $\eta$ in $\mathcal{G}'$. The estimates \eqref{est:reginetazerofirst} and \eqref{est:reginetazerosecond} are similarly straightforward, using the identity $\widetilde{f}_0^\omega = 0$, Corollary \ref{prop:absbounds}, boundedness of $\eta$, and the fact that $|m'| = |m''|$ for $m = 0$.

To prove \eqref{est:reginetafirst} we first assume $m' \neq 0$, as otherwise the result is trivial. Homogeneity of $\wfe$ in frequency implies
\begin{equation}
\widetilde{f}_m^\omega = \widetilde{f}_{m'}^{\frac{m'\omega}{m}},
\end{equation}
so we may apply Lemma \ref{prop:regineta} to obtain
\begin{equation}
\left| {\omega'} \widetilde{f}_{m'}^{\omega'} - \frac{m'\omega}{m}\widetilde{f}_{m'}^{\frac{m'\omega}{m}}\right|\lesssim |m'|\frac{|m'\omega-m{\omega'}|}{|\omega||{\omega'}|}v.
\end{equation}
Then, \eqref{est:reginetafirst} follows from the inequalities
\begin{equation}\label{est:reginetaauxiliary}
|m'\omega-m{\omega'}|\leq |m||\omega''|+|\omega||{m''}|, \qquad \frac{|m'|}{|{\omega'}|} \leq C, \qquad \frac{|m|}{|\omega|} \leq C.
\end{equation}
\end{proof}

\subsection{Frequency regularity estimates in $\mathcal{G}$}\label{sec:CutoffBounds}
We now bound terms coming from the frequency cutoffs $\wchi_1$, $\wchi_2$ as defined in \eqref{def:wchi12}. We prove this for a more general class of functions, which will be of use in Section \ref{sec:WGlobalBounds}. We first recall the shorthand $\wchi(\omega, m) = \wchi_m^\omega$.
\begin{lemma}\label{prop:wchireg}
Let $\wchi(\omega, m)$ be a function which is homogeneous of degree zero in $(\omega, m)$ and constant for $(\omega, m)\notin\! \mathcal{G}$. For any $m \in\mathbb{Z}\setminus \{0\}$, and for $\omega, \sigma \in\mathbb{R}$, the bound
\begin{equation}\label{est:wchibound}
|\wchi_m^\omega - \wchi_m^{\sigma}| \leq C\Big|\frac{|\omega-\sigma|}{\max(|\omega|, |\sigma|, |m|)}\Big|\cdot\max(\sup|\wchi|, \sup_\omega |\partial_\omega(\wchi(\omega, 1))|)
\end{equation}
holds for a constant $C$ depending only on $(M, a, \ximin, \ximax)$. If $m=0$ the left hand side of \eqref{est:wchibound} vanishes.
\end{lemma}
\begin{proof} The case $m=0$ is trivial. Otherwise, by symmetry, it suffices to show \eqref{est:wchibound} with $\max(|\omega|, |\sigma|, |m|)$ replaced by $\max(|\omega|, |m|)$. We select a constant $C$ large enough that if $|\omega|\geq C|m|$, then $(\frac12\omega, m)\notin\mathcal{G}$. Since $\wchi$ is constant outside $\mathcal{G}$, then for all $\omega$ satisfying $|\omega|\geq C|m|$, and for all $\sigma$, $\wchi_m^\omega - \wchi_m^\sigma \neq 0$ implies $|\omega-\sigma| > \frac12|\omega|$. In this regime,
\begin{equation}\label{est:CEPreliminary}
|\wchi_m^\omega - \wchi_m^{\sigma}| \leq 2\sup|\wchi| \leq 4\frac{|\omega-\sigma|}{|\omega|}\sup|\wchi| \leq \frac{4}{C}\frac{|\omega-\sigma|}{|m|}\sup|\wchi|.
\end{equation}
For $|\omega| \leq C|m|$, \eqref{est:wchibound} follows from the inequality $|\wchi_m^\omega - \wchi_m^\sigma| \leq |\omega-\sigma|\sup_\omega|\partial_\omega\wchi(\omega, m)|$. 
\end{proof}
Homogeneity of $\wchi$ in $(\omega, m)$ implies the following proposition:
\begin{proposition}\label{cor:CECutoff}
For any frequency pairs $(\omega, m)$ and $(\omega', m')$, and for $\wchi$ as in Proposition \ref{prop:wchireg}, the inequalities
\begin{subequations}
\begin{align}
|\wchi_m^\omega - \wchi_{m'}^{\omega'}|&\lesssim \frac{|m''|+|\omega''|}{\max(|m'|, |\omega'|, 1)}\cdot\max(\sup|\wchi|, \sup_\omega |\partial_\omega(\wchi(\omega, 1))|), \label{est:CECutoffBounds}\\
|\wchi_m^\omega - \wchi_{m'}^{\omega'}|&\lesssim \frac{|m''|+|\omega''|}{\max(|m|, |\omega|, 1)}\cdot\max(\sup|\wchi|, \sup_\omega |\partial_\omega(\wchi(\omega, 1))|). \label{est:CECutoffBounds2}
\end{align}
\end{subequations}
holds if $(\omega, m)\in\mathcal{G}$ or $(\omega', m') \in\mathcal{G}$. If neither is in $\mathcal{G}$, then the left hand side of \eqref{est:CECutoffBounds} and \eqref{est:CECutoffBounds2} vanish.
\end{proposition}
\begin{proof} 
The case $(\omega, m)\notin\mathcal{G}, (\omega', m')\notin\mathcal{G}$ follows from the fact that $\wchi$ is constant outside $\mathcal{G}$. We prove \eqref{est:CECutoffBounds}, noting that \eqref{est:CECutoffBounds2} then follows directly from a symmetry argument. 

 We may reduce this to proving \eqref{est:CECutoffBounds} for $|m|\neq 0$ and $\max(|m'|, |\omega'|, 1) = \max(|m'|, |\omega'|)$ as follows: if $|m| = 0$, then $(\omega, m)\notin\mathcal{G}$, so $(\omega', m')\in\mathcal{G}$. The estimates $|m''| = |m'| \neq 0$, $\max(|m'|, |\omega'|, 1) \lesssim |m'|$ imply \eqref{est:CECutoffBounds}. Next, if $\max(|m'|, |\omega'|) <1$, the bound \eqref{est:CECutoffBounds} follows from the identity $ |m''| = |m|  \geq 1$. 

Homogeneity of $\wchi$ in frequency combined with Lemma \ref{prop:wchireg} implies
\begin{equation}\label{est:CECutoffBoundPrelim}
\wchi_m^\omega - \wchi_{m'}^{\omega'} = \wchi_{m'\!m}^{m'\!\omega} - \wchi_{m'\!m}^{m\omega'}\lesssim \frac{\left|m'\omega - m\omega'\right|}{\max\left(\left|m'\omega\right|, |m'm|,|m\omega'|\right)}\cdot\max(\sup|\wchi|, \sup_\omega |\partial_\omega(\wchi(\omega, 1))|)
\end{equation}
To prove \eqref{est:CECutoffBounds}, we first take the expansion
\begin{equation}\label{id:momegaDiff}
m'\omega - m\omega' = m'\omega'' - \omega'm'' = m\omega'' - \omega m''.
\end{equation}
Then, if $(\omega, m)\in\mathcal{G}$, the bound \eqref{est:CECutoffBounds} follows from the inequality
\begin{equation}
\frac{\left|m'\omega - m\omega'\right|}{\max\left(\left|m'\omega\right|, |m'm|,|m\omega'|\right)} \leq \frac{|m''||\omega|+|m||\omega''|}{|m|\max(|m'|, |\omega'|)}
\end{equation}
and the estimate $|\omega|\lesssim |m|\in\mathcal{G}$. If $(\omega', m')\in\mathcal{G}$,  \eqref{est:CECutoffBounds} instead follows from 
\begin{equation}
\frac{\left|m'\omega - m\omega'\right|}{\max\left(\left|m'\omega\right|, |m'm|,|m\omega'|\right)} \leq \frac{|m''||\omega|+|m||\omega''|}{|m'|\max(|m|, |\omega|)}
\end{equation}
and the estimate $\max(|m'|, |\omega'|) \lesssim |m'| \in\mathcal{G}$.
\end{proof}
We state a corollary for first order terms:
\begin{corollary}\label{cor:CECutoff2}
For any frequency pairs $(\omega, m)$ and $(\omega', m')$, and for $\wchi$ as in Proposition \ref{prop:wchireg}, the inequalities
\begin{subequations}
\begin{align}
|\omega\wchi_m^\omega - \omega'\wchi_{m'}^{\omega'}| & \lesssim (|m''|+|\omega''|)\cdot\max(\sup|\wchi|, \sup_\omega |\partial_\omega(\wchi(\omega, 1))|), \label{est:CECutoffBoundomega}\\
|m\wchi_m^\omega - m'\wchi_{m'}^{\omega'}| &\lesssim (|m''|+|\omega''|)\cdot\max(\sup|\wchi|, \sup_\omega |\partial_\omega(\wchi(\omega, 1))|),\label{est:CECutoffBoundm}
\end{align}
\end{subequations}
hold if $(\omega, m)\in\mathcal{G}$ or $(\omega', m') \in\mathcal{G}$. If neither is in $\mathcal{G}$, \eqref{est:CECutoffBoundomega} and \eqref{est:CECutoffBoundm} both vanish.
\end{corollary}
\begin{proof}
We expand
\begin{equation}
|\omega\wchi_m^\omega - \omega'\wchi_{m'}^{\omega'}|  \leq |\omega-\omega'||\wchi_m^\omega| + |\omega'||\wchi_m^\omega - \wchi_{m'}^{\omega'}|.
\end{equation}
The result follows directly from Proposition \eqref{cor:CECutoff}. The same argument, replacing $\omega$ with $m$, gives \eqref{est:CECutoffBoundm}.
\end{proof}
We end this section by remarking  that Proposition \ref{cor:CECutoff} applies directly to $\wchi = 1-\wchi_1$, $\wchi=\chisharp$ and $\wchi=\hsharp-h_0$, as all of them are constant outside $\mathcal{G}$. We prove a brief bound on the $\hsharp-h_0$ term.
\begin{lemma}\label{lem:hsharph0aux}
The quantity $\hsharp - h_0$ satisfies the bound
\begin{equation}\label{est:hsharph0aux}
\sup_\omega \Bigg|\partial_\omega\left(\frac{\hsharp(\omega, 1) - h_0}{v}\right)\Bigg|\lesssim v 
\end{equation}
in the support of $\chisharp$.
\end{lemma}
\begin{proof}
This follows almost directly from the bound   \eqref{est:detafFinal2}  and the chain rule. We first note that $\hsharp(\omega, 1) = 1$ outside $\mathcal{G}\cap\mathcal{G}'$, so in particular $\partial_\omega \hsharp$ is supported in $\mathcal{G}\cap\mathcal{G}'$. In this region, we define $\lambda = \frac{\omega}{m}$. By the chain rule and Lemma \ref{prop:GpBounds},
\begin{equation}
\sup_\omega\partial_\omega\Bigg|\left(\frac{\hsharp(\omega, 1) - h_0}{v}\right)\Bigg| = \left|\partial_\lambda\left(\frac{\hsharp - h_0}{v}\right)\right| = \left|\partial_\lambda\eta \partial_\eta \left(\frac{\hsharp - h_0}{v}\right)\right| \lesssim \left|\partial_\eta \left(\frac{\hsharp - h_0}{v}\right)\right|
\end{equation}
in $\mathcal{G}\cap\mathcal{G}'$. Additionally, \eqref{est:detafFinal2}  and Lemma \ref{prop:GpBounds} imply
\begin{equation}
\left|\partial_\eta \left(\frac{\hsharp - h_0}{v}\right)\right| = \left|\partial_\eta \left(\left(f_\sharp - \frac{r^2+a^2-\eta}{\Delta^{1/2}}h_0\right)\left(\frac{r^2+a^2}{r^2+a^2-\eta}\right)\right)\right|\lesssim v.
\end{equation}
The bound \eqref{est:hsharph0aux} follows.
\end{proof}

\subsection{Completing the proof of Proposition \ref{lem:CMBound}}\label{sec:WGlobalBounds}
We now prove the bound \eqref{est:CIMain}, and therefore Proposition \ref{lem:CMBound}. We expand \eqref{est:CIMain} pointwise in $(\omega, m)$ and convolve over $(\omega', m')$, recalling $\omega'', m'', \omega_r'$, and $\omega_r''$ as defined in \eqref{def:FreqDoublePrime} and \eqref{def:OmegarDoublePrime}. We additionally define
\begin{align}
\Wsc(\widehat{w}_{{m''}}^{\omega''}\ph_{m'}^{\omega'}) &= v^{-1}\Rs(\widehat{w}_{{m''}}^{\omega''}\ph_{m'}^{\omega'}) - i\omega (\wchi_1\f + (1-\wchi_1)\chisharp f_\sharp)_m^\omega\widehat{w}_{{m''}}^{\omega''}\ph_{m'}^{\omega'},\label{def:WProductConv} \\
\oW(\widehat{w}_{{m''}}^{\omega''}\ph_{m'}^{\omega'})&= v^{-1}\Rs(\widehat{w}_{{m''}}^{\omega''}\ph_{m'}^{\omega'}) - i\frac{h_0(\omega-\omega_r)}{v}(\widehat{w}_{{m''}}^{\omega''}\ph_{m'}^{\omega'}),\label{def:oWProductConv}
\end{align}
so that for square-integrable functions $\ph$ and $\widehat{w}$,
\begin{equation}
\Wsc(\ph * \widehat{w}) = \sum_{m'}\int \Wsc(\ph_{m'}^{\omega'}\widehat{w}_{m''}^{\omega''})\, d\omega'.
\end{equation}

We claim that to prove \eqref{est:CIMain} it suffices to show the following bound in frequency space. Indeed, integrating the estimate (\ref{est:PointwiseConvolutionBound}) in $\omega^\prime$ and writing the right hand side as a convolution, (\ref{est:CIMain}) follows after taking $L^2_{\omega}$ on both sides and applying Young's convolution inequality on the right hand side.

\begin{proposition}
For any $q\in[0,1]$, the bound
\begin{align}\label{est:PointwiseConvolutionBound}
\big|\Wsc(\widehat{w}_{{m''}}^{\omega''}\ph_{m'}^{\omega'}) - &\widehat{w}_{{m''}}^{\omega''}(\Wsc\ph)_{m'}^{\omega'}\big| \lesssim \\
&\qquad\lesssim \left(|\omega''\widehat{w}_{m''}^{\omega''}| + |m''\widehat{w}_{m''}^{\omega''}| + v^{-q}\left|(\oW\widehat{w})_{m''}^{\omega''}\right|\right)\Bigg(v^q|\ph_{m'}^{\omega'}| + \chi_{\rmax}\frac{\big|(W_0\ph)_{m'}^{\omega'}\big|}{\max(|\omega'|, |m'|, 1)}\Bigg) \nonumber
\end{align}
holds whenever $\ph$ and $\widehat{w}$ are square integrable in the sense of \eqref{sqi}.
\end{proposition}
\begin{proof}
It suffices to prove this with $\Wsc$ replaced by $\Wsc-\oW$, as the difference is controlled by
\begin{equation}
\big|\oW(\widehat{w}_{{m''}}^{\omega''}\ph_{m'}^{\omega'}) - \widehat{w}_{{m''}}^{\omega''}(\oW\ph)_{m'}^{\omega'}\big| = \big|(\oW\widehat{w})_{{m''}}^{\omega''}\ph_{m'}^{\omega'}\big| = \big|v^{-q}(\oW\widehat{w})_{{m''}}^{\omega''}\big|\big|v^q\ph_{m'}^{\omega'}\big|.
\end{equation}
Fixing frequency pairs $(\omega, m)$ and $(\omega', m')$, the decomposition \eqref{id:WMinusWZero} implies
\begin{align}\label{est:PointwiseConvolution}
(\Wsc\!-\!\oW)(\widehat{w}_{{m''}}^{\omega''}\ph_{m'}^{\omega'}) \!-\! \widehat{w}_{{m''}}^{\omega''}(\Wsc\ph\!-\!\oW\ph)_{m'}^{\omega'}\! &=\!-i\left((\wchi_1)_m^\omega \omega \widetilde{f}_m^\omega - (\wchi_1)_{m'}^{\omega'} {\omega'} \widetilde{f}_{m'}^{\omega'}\right)\widehat{w}_{{m''}}^{\omega''}\ph_{m'}^{\omega'}\\
&\,+ \!(1\!-\!\wchi_1)_m^\omega\Big(\!\!-\!i(\chisharp)_m^\omega\frac{(\hsharp)_m^\omega\!-\!h_0}{v}(\omega\!-\!\omega_r) \!-\! (1\!-\!\chisharp)_m^\omega\oW\Big)(\widehat{w}_{{m''}}^{\omega''}\ph_{m'}^{\omega'}) \nonumber\\
&\, - \! \widehat{w}_{{m''}}^{\omega''}(1\!-\!\wchi_1)_{m'}^{\omega'}\Big(\!\!-\!i(\chisharp)_{m'}^{\omega'}\frac{(\hsharp)_{m'}^{\omega'}\!-\!h_0}{v}({\omega'}\!-\!\omega_r') \!-\! (1\!-\!\chisharp)_{m'}^{\omega'}\oW\!\Big)\ph_{m'}^{\omega'} \nonumber.
\end{align}
We split this up and seek to bound each set of terms by the right hand side of \eqref{est:PointwiseConvolutionBound}.

We first consider the case $m = 0$, so $|m'| = |m''|$ (we may assume $m'\neq 0$, as otherwise the right hand side of \eqref{est:PointwiseConvolution} vanishes). Since $\wchi_1$ is supported in $\mathcal{G}'$, and $\widetilde{f}_0^\omega=0$, the estimate \eqref{est:wfebound} implies
\begin{equation}
\left|(\wchi_1)_0^\omega \omega \widetilde{f}_0^\omega - (\wchi_1)_{m'}^{\omega'} {\omega'} \widetilde{f}_{m'}^{\omega'}\right| = \left| (\wchi_1)_{m'}^{\omega'} {\omega'} \widetilde{f}_{m'}^{\omega'}\right| \lesssim |m''|v,
\end{equation}
so consequently
\begin{equation}
\left|\left((\wchi_1)_0^\omega \omega \widetilde{f}_0^\omega - (\wchi_1)_{m'}^{\omega'} {\omega'} \widetilde{f}_{m'}^{\omega'}\right)\widehat{w}_{{m''}}^{\omega''}\ph_{m'}^{\omega'}\right| \lesssim |m''\widehat{w}_{m''}^{\omega''}|\big|v\ph_{m'}^{\omega'}\big| \, . 
\end{equation}
The second line of \eqref{est:PointwiseConvolution} vanishes, and the third may be bounded using the pointwise estimate
\begin{equation}
\left|\left((\chisharp)_{m'}^{\omega'}\frac{(\hsharp)_{m'}^{\omega'}-h_0}{v}({\omega'}-\omega_r') - (1-(\chisharp)_{m'}^{\omega'})\oW\right)(\ph_{m'}^{\omega'})\right|\lesssim|m''|\left(v|\ph_{m'}^{\omega'}| + |m'|^{-1}|\chi_{\rmax}(\oW\ph)_{m'}^{\omega'}|\right),
\end{equation}
as $|{\omega'}|\leq  C|m'|$ in the support of $(1-(\wchi_1)_{m'}^{\omega'})$. The bound \eqref{est:PointwiseConvolutionBound} directly follows.

When $m \neq 0$, it suffices to prove the following three bounds separately:
\begin{subequations}
\begin{align}\label{est:TrappedTermBound}
\left|\left((\wchi_1)_m^\omega \omega \widetilde{f}_m^\omega - (\wchi_1)_{m'}^{\omega'} {\omega'} \widetilde{f}_{m'}^{\omega'}\right)\right|\left|\widehat{w}_{{m''}}^{\omega''}\ph_{m'}^{\omega'}\right|&\lesssim \text{(r.h.s.)}\eqref{est:PointwiseConvolutionBound},\\
\label{est:SuperradSecondTermBound2}
\Big|(1\!-\!\wchi_1)_m^\omega(\chisharp)_m^\omega\!\frac{(\hsharp)_m^\omega\!-\!h_0}{v}(\omega\!-\!\omega_r\!)\!-\!(1\!-\!\wchi_1)_{m'}^{\omega'}(\chisharp)_{m'}^{\omega'}\frac{(\hsharp)_{m'}^{\omega'}\!-\!h_0}{v}({\omega'}\!-\!\omega_r')\Big|\!\left|\widehat{w}_{{m''}}^{\omega''}\ph_{m'}^{\omega'}\right| &\lesssim \text{(r.h.s.)}\eqref{est:PointwiseConvolutionBound}, \\
\label{est:CEThirdTermBound}
 \Big|\!(1\!-\!\wchi_1)_m^\omega(1\!-\!\chisharp)_m^\omega\oW(\widehat{w}_{{m''}}^{\omega''}\ph_{m'}^{\omega'})-\widehat{w}_{{m''}}^{\omega''}\!(1\!-\!\wchi_1)_{m'}^{\omega'}(1\!-\!\chisharp)_{m'}^{\omega'}\oW\ph_{m'}^{\omega'}\Big|&\lesssim \text{(r.h.s.)}\eqref{est:PointwiseConvolutionBound}.
\end{align}
\end{subequations}
To prove \eqref{est:TrappedTermBound}, we first note that the identity $m^{\prime \prime}=m-m^\prime$ implies
\begin{equation}\label{id:fchihomogeneous}
(\wchi_1)_m^\omega \omega\widetilde{f}_m^\omega = (\wchi_1)_{m}^{\omega}\tfrac{m'\omega}{m}\widetilde{f}_m^\omega + (\wchi_1)_m^\omega\tfrac{m''\!\omega}{m}\widetilde{f}_m^\omega,
\end{equation}
so we write
\begin{equation}\label{est:TrappedTermBoundPrelim}
\left|(\wchi_1)_m^\omega \omega \widetilde{f}_m^\omega - (\wchi_1)_{m'}^{\omega'} {\omega'} \widetilde{f}_{m'}^{\omega'}\right|\lesssim\left|(\wchi_1)_m^\omega \tfrac{m'\omega}{m} \widetilde{f}_m^\omega - (\wchi_1)_{m'}^{\omega'} {\omega'} \widetilde{f}_{m'}^{\omega'}\right| + \left|(\wchi_1)_m^\omega \tfrac{m''\omega}{m} \widetilde{f}_m^\omega \right| \, .
\end{equation}
Corollary \ref{prop:absbounds} implies
\begin{equation}
\big|(\wchi_1)_m^\omega\tfrac{m''\!\omega}{m}\widetilde{f}_m^\omega\big||\widehat{w}_{{m''}}^{\omega''}\ph_{m'}^{\omega'}|\lesssim v|m''||\widehat{w}_{{m''}}^{\omega''}\ph_{m'}^{\omega'}|.
\end{equation}
We now bound the first term on the right hand side of \eqref{est:TrappedTermBoundPrelim}, which in particular vanishes if $m' = 0$ or if neither $(\omega, m)$ nor $(\omega', m')$ is in $\mathcal{G}'$. If $(\omega, m)\in\mathcal{G}'$, we may expand
\begin{equation}\label{est:ConvTermGPrimeFirst}
\big|(\wchi_1)_{m}^{\omega}\tfrac{{m'}\omega}{m}\widetilde{f}_m^\omega - (\wchi_1)_{m'}^{\omega'} {\omega'} \widetilde{f}_{m'}^{\omega'}\big|\leq \big|{\tfrac{{m'}\omega}{m}}\widetilde{f}_m^\omega\big|\big|(\wchi_1)_{m}^{\omega}-(\wchi_1)_{m'}^{\omega'}\big| + \big|(\wchi_1)_{m'}^{\omega'}\big|\big|{\tfrac{{m'}\omega}{m}} \widetilde{f}_{m}^{\omega} - {\omega'}\widetilde{f}_{m'}^{\omega'}\big|,
\end{equation}
The estimates \eqref{est:wfebound} and \eqref{est:CECutoffBounds} along with the expansion \eqref{id:momegaDiff} and the bound $|m|\lesssim |\omega|$, imply
\begin{equation}
\big|{\tfrac{{m'}\omega}{m}}\widetilde{f}_m^\omega\big|\big|(\wchi_1)_m^\omega-(\wchi_1)_{m'}^{\omega'}\big| \lesssim |m'|\tfrac{|m\omega''-\omega m''|}{\max(|m'\omega|, |m'm|)}v\lesssim (|\omega''| + |{m''}|)v \, . 
\end{equation}
If $(\wchi_1)_{m'}^{\omega'} \neq 0$, $(\omega', m')\in\mathcal{G}'$, so we bound the second term on the right side of \eqref{est:ConvTermGPrimeFirst} using \eqref{est:reginetafirst}. Therefore,
\begin{equation}\label{est:CEFirstTerm}
\big|(\wchi_1)_{m}^{\omega}\tfrac{{m'}\omega}{m}\widetilde{f}_m^\omega - (\wchi_1)_{m'}^{\omega'} {\omega'} \widetilde{f}_{m'}^{\omega'}\big|\lesssim v(|\omega''|+|{m''}|)|\widehat{w}_{{m''}}^{\omega''}||\ph_{m'}^{\omega'}|.
\end{equation}
If $(\omega, m)\notin \mathcal{G}'$ and $(\omega', m')\in\mathcal{G}'$, the proof of \eqref{est:TrappedTermBound} follows similarly from the expansion
\begin{equation}\label{est:ConvTermGPrimeSecond}
\big|(\wchi_1)_{m}^{\omega}\tfrac{{m'}\omega}{m}\widetilde{f}_m^\omega - (\wchi_1)_{m'}^{\omega'} {\omega'} \widetilde{f}_{m'}^{\omega'}\big|\leq \big|\omega'\widetilde{f}_{m'}^{\omega'}\big|\big|(\wchi_1)_{m}^{\omega}-(\wchi_1)_{m'}^{\omega'}\big| + \big|(\wchi_1)_{m}^{\omega}\big|\big|{\tfrac{{m'}\omega}{m}} \widetilde{f}_{m}^{\omega} - {\omega'}\widetilde{f}_{m'}^{\omega'}\big|.
\end{equation}

Next, the bound \ref{est:SuperradSecondTermBound2} follows from Corollary \ref{cor:CECutoff2} with $\wchi=(1-\wchi_1)\chi_\sharp\frac{h_\sharp-h_0}{v}$ and $\wchi=\frac{a}{r^2+a^2}(1-\wchi_1)\chi_\sharp\frac{h_\sharp-h_0}{v}$ and Lemma \ref{lem:hsharph0aux}.

We finally prove \eqref{est:CEThirdTermBound}. From \eqref{def:oWProductConv} we immediately have
\begin{equation}
\oW(\widehat{w}_{{m''}}^{\omega''}\ph_{m'}^{\omega'}) = \widehat{w}_{{m''}}^{\omega''}(\oW\ph)_{m'}^{\omega'} + \ph_{m'}^{\omega'}(\oW\widehat{w})_{{m''}}^{\omega''},
\end{equation}
so we can reduce \eqref{est:CEThirdTermBound} to proving the bounds
\begin{align}\label{est:SuperradFourthTermBound}
\left|\!\big(\!(1\!-\!\wchi_1)_m^\omega(1\!-\!\chisharp)_m^\omega \!-\! (1\!-\!\wchi_1)_{m'}^{\omega'}(1\!-\!\chisharp)_{m'}^{\omega'}\big)\widehat{w}_{{m''}}^{\omega''}(\oW\ph)_{m'}^{\omega'}\right| &\lesssim \frac{|\omega''|+|{m''}|}{\max(|{\omega'}|, |m'|, 1)}|(\oW\ph)_{m'}^{\omega'}||\widehat{w}_{{m''}}^{\omega''}|\chi_{\rmax}, \\
\label{est:SuperradFifthTermBound}\left|\!(1\!-\!\wchi_1)_m^\omega(1\!-\!\chisharp)_m^\omega(\oW\widehat{w}_{{m''}}^{\omega''})\ph_{m'}^{\omega'}\right| &\lesssim |\ph_{m'}^{\omega'}||(\oW\widehat{w})_{{m''}}^{\omega''}|.
\end{align}
The bound \eqref{est:SuperradFifthTermBound} follows directly from the identity $|(1\!-\!\wchi_1)(1\!-\!\chisharp)|\leq 1$. To prove \eqref{est:SuperradFourthTermBound}, we apply Corollary \ref{cor:CECutoff} with $\wchi=(1\!-\!\wchi_1)(1\!-\!\chisharp)$.
\end{proof}

\subsection{Completing the proof of Proposition \ref{lem:CMBound2}}\label{sec:AppendixLocalizedEnergy}
We now prove an analogous inequality which will allow us to bound $\mathcal{N}_\sharp[H, H]$, which will again follow from pointwise bounds on terms appearing in the convolution. These are fortunately easier, as quantities which appear are uniformly compactly supported in $r$. We restate the definition \eqref{def:GenSecondOperators}:
\begin{equation}
X_\sharp^T = -ig_\sharp(1-\wchi_1)\omega, \qquad X_\sharp^\phi = ig_\sharp(1-\wchi_1)m, \qquad X_\sharp^{\Wp} = g_\sharp(1-\wchi_1)\Wp.
\end{equation}

\begin{proof}[Proof of Proposition \ref{lem:CMBound2}] As in the proof of Proposition \ref{lem:CMBound}, it first suffices to prove the pointwise bound
\begin{equation}\label{est:CIMain2}
\lVert X_\sharp(\widehat{w}*\ph) - \widehat{w}*(X_\sharp(\ph))\rVert_{\ell^2_m L^2_\omega} \lesssim \lVert \widehat{w}\rVert_{\Wdot}\left\lVert \ph\right\rVert_{\Hdot}\chi_{\rmin\!,\rmax}
\end{equation}
for each $X_\sharp$, for all $(r, \theta)$ and for all $q \in [0,1]$. We first state the following pointwise bounds, which follow directly from Proposition \ref{cor:CECutoff} (with $\wchi=g_\sharp(1-\wchi_1)$), Corollary \ref{cor:CECutoff2}, using smoothness and support of $g_\sharp$:
\begin{subequations}\label{est:Pointwiseboundsgsharpterms}
\begin{align}\label{est:Pointwiseboundsgsharpterms3}|(g_\sharp(1-\wchi_1))_m^\omega - (g_\sharp(1-\wchi_1))_{m'}^{\omega'}| & \lesssim \tfrac{|m''|+|\omega''|}{\max( |\omega'|,|m'|, 1)}\chi_{\rmin\!,\rmax}, \\
\label{est:Pointwiseboundsgsharpterms1}|(g_\sharp(1-\wchi_1))_m^\omega \omega - (g_\sharp(1-\wchi_1))_{m'}^{\omega'}\omega'| & \lesssim (|\omega''| + |m''|)\chi_{\rmin\!,\rmax}, \\
\label{est:Pointwiseboundsgsharpterms2}|(g_\sharp(1-\wchi_1))_m^\omega m - (g_\sharp(1-\wchi_1))_{m'}^{\omega'}m'| & \lesssim (|\omega''| + |m''|)\chi_{\rmin\!,\rmax}.
\end{align}
\end{subequations}
In order to see that \eqref{est:CIMain2} follows from \eqref{est:Pointwiseboundsgsharpterms}, we note that \eqref{est:gsharpprimebound} implies
\begin{equation}
|g_\sharp(1-\wchi_1)| \leq \chi_{\rmin\!,\rmax}
\end{equation}
for all $(\omega, m)$. Then, by Young's convolution inequality, \eqref{est:Pointwiseboundsgsharpterms1} and \eqref{est:Pointwiseboundsgsharpterms2} directly imply \eqref{est:CIMain2} for $X_\sharp^T$ and $X_\sharp^\phi$ respectively.

In order to prove \eqref{est:CIMain2} for $X_\sharp^{\Wp}$, we first define
$
 X_\sharp^{\oW}= g_\sharp(1-\wchi_1)\oW 
$
and expand the convolution in \eqref{est:CIMain2}, reducing the problem to proving a pointwise bound as before. The bound \eqref{est:CIMain2} for $X_\sharp = X_\sharp^{\oW}$ then follows a similar argument to the proof of the bound \eqref{est:CEThirdTermBound}.
The result for $X_\sharp^{\Wp}$ then follows from the identity
\begin{equation}
g_\sharp(1-\wchi_1) (\Wp\ph - \oW\ph)_m^\omega = g_\sharp(1-\wchi_1)\left(-i(\omega-\omega_r)\frac{1-h_0}{v}\right)\ph^\omega_m,
\end{equation}
and bounding the relevant terms appearing from this in (\ref{est:CIMain2}) using \eqref{est:Pointwiseboundsgsharpterms1} and \eqref{est:Pointwiseboundsgsharpterms2}.
\end{proof}

\appendix
\section{Physical space estimates}\label{App:Geometry}
		
\subsection{Energy estimates on asymptotically hyperboloidal foliations}\label{sec:FoliationCalculations}
Calculations shown here are standard and can be found in e.g. Appendix D of \cite{SR14}. 

We recall the definitions of Section \ref{sec:hyperboloidal}.
\begin{proposition}
The vector $\nabla\wtau$ normal to the hypersurfaces $\wSigtau$ satisfies the following bounds for some $C>0$:
\begin{equation}
C^{-1}r^{-2} \leq -g(\nabla\wtau, \nabla\wtau) \leq Cr^{-2} \, .
\end{equation}
\end{proposition}
\begin{proof}
We recall that in the $(t, \rs, \phi, \theta)$ coordinates the inverse metric has the components
\begin{equation}
g^{tt} = \frac{a^2\sin^2\theta\Delta - (r^2+a^2)^2}{\Delta\rho^2}, \qquad g^{t\phi} = \frac{-2Mar}{\Delta\rho^2}, \qquad g^{\rs\rs} = \frac{(r^2+a^2)^2}{\Delta\rho^2}.
\end{equation}
The definition \eqref{def:timefoliation} gives
\begin{equation}
-\nabla\wtau = \frac{(r^2+a^2)^2}{\Delta\rho^2}\left(h_0 v\oW + \big(1-h_0^2-a^2v^2\sin^2\theta\big)K^\star - \frac{a}{r^2+a^2}\frac{\Delta\rho^2}{(r^2+a^2)^2}\partial_\phi\right).
\end{equation}
Then, the definition of $h_0$ gives
\begin{equation}
g(\nabla\wtau, \nabla\wtau) = \rho^{-2}(-v_0^{-2} + a^2\sin^2\theta).
\end{equation}
We may define
\begin{equation}
a_\theta\index{atheta@$a_\theta$} = (v_0^{-2} - a^2\sin^2\theta)^{1/2}.
\end{equation}
We know
\[
v_0^{-2} = \inf(v^{-2}) \geq\inf(r^2+a^2) > r_+^2+a^2,
\]
so consequently $r_+< a_\theta < v_0^{-1}$ and the result follows.
\end{proof}
We immediately conclude that
\begin{equation}
n_{\wSigtau} = a_\theta^{-1}\rho \nabla \wtau
\end{equation}
and after some further computation 
%
%
%
%
%
%
the following expression for the induced volume form:
\begin{equation}\label{id:VolumeForm}
dV_{\wSigtau} = \rho a_\theta\sin\theta\frac{\Delta}{r^2+a^2}\, d\rs\, d\phi\, d\theta = \rho a_\theta \sin\theta\, dr\, d\phi\, d\theta \, .
\end{equation}
Recalling the definition (\ref{def:HField}) a careful computation yields:
\begin{proposition}
The energy momentum tensor for $\psi$,
\begin{equation} \label{def:EMTensor}
\mathbf{T}_{\mu\nu}[\psi] = \nabla_\mu \psi\nabla_\nu \psi- \tfrac12 g_{\mu\nu}\nabla^\gamma\psi\nabla_\gamma\psi,
\end{equation}
can be decomposed as
\begin{equation}
\mathbf{T}[\psi](\Ks, -\nabla\wtau) = \frac{1}{2\rho^2}|\oW\psi|^2 + \frac{a_\theta^2}{2\rho^2}|\Ks\psi|^2 + \frac{1}{2\rho^2}|\partial_\theta\psi|^2 + \frac{\rho^2}{2(r^2+a^2)^2}\frac{1}{\sin^2\theta}|\partial_\phi\psi|^2 \, .
\end{equation}
\end{proposition}
%
%
%
%
%
%
As a corollary, we obtain for ${\bf J}_\mu^X[\psi]:=\mathbf{T}_{\mu \nu}[\psi] X^\nu$ the equivalence
\begin{equation}\label{est:TimeSliceIF}
\int_{\wSigtau}\mathbf{J}^{\Ks}_\mu [\psi] n_{\wSigtau}^\mu\, dV_{\wSigtau} \approx \int \, \big(r^{-2}|\oW\psi|^2 +r^{-2} |\Ks\psi|^2 + r^{-2}|\partial_\theta\psi|^2 + r^{-2}\tfrac{1}{\sin^2\theta}|\partial_\phi\psi|^2\big)r^2\sin\theta\, dr\, d\theta\, d\phi \, .
\end{equation}
We may improve the weight on $\oW$ for $r$ close to $r_+$ using the standard redshift field $N$ (see for instance \cite{DRSR}):
\begin{equation}\label{def:Energy0}
\int_{\wSigtau}\mathbf{J}^{N}_\mu [\psi] n_{\wSigtau}^\mu\, dV_{\wSigtau} \approx \int \, \big(|\wpa_r\psi|^2 +r^{-2} |\Ks\psi|^2 + r^{-2}|\partial_\theta\psi|^2 + r^{-2}\tfrac{1}{\sin^2\theta}|\partial_\phi\psi|^2\big)r^2\sin\theta\, dr\, d\theta\, d\phi \, .
\end{equation}
Finally, given $\psi$ satisfying $\Box \psi = \varepsilon \BfB\psi$, the divergence identity $\nabla^\mu \mathbf{J}^X_\mu [\psi] = \mathbf{K}^X [\psi]+ \varepsilon \BfB\psi X\psi$ holds with $\mathbf{K}^X [\psi]:= \frac{1}{2} \mathbf{T}[\psi] \mathcal{L}_{X} g$). Using the easily verified bounds
\begin{equation}
\int_{\wSigtau}|\mathbf{K}^{\Ks} [\psi]|\, dV_{\wSigtau} \lesssim \int_{\wSigtau}\mathbf{J}^{\Ks}_\mu  [\psi] n_{\wSigtau}^\mu\, dV_{\wSigtau}  \ \ \ \textrm{and} \ \ \ 
\int_{\wSigtau}|\mathbf{K}^{N} [\psi]|\, dV_{\wSigtau} \lesssim \int_{\wSigtau}\mathbf{J}^{N}_\mu [\psi]  n_{\wSigtau}^\mu\, dV_{\wSigtau} \, ,
\end{equation}
we deduce from Gronwall's inequality the well-known estimate
\begin{align}
\int_{\wSigtau}\mathbf{J}^{N}_\mu  [\psi]  n_{\wSigtau}^\mu\, dV_{\wSigtau} \lesssim C_{\tilde{\tau}} \int_{\widetilde{\Sigma}_0}\mathbf{J}^{N}_\mu  [\psi]  n_{\widetilde{\Sigma}_0}^\mu\, dV_{\widetilde{\Sigma}_0} \, .
\end{align}
Using elementary Hardy inequalities on $\widetilde{\Sigma}_\tau$ one easily deduces (\ref{basiclocal}) from this and by commutation also the higher order estimate (\ref{higherlocal}).

\subsection{Commutator estimates near $\mathcal{H}^+$ and near infinity}\label{sec:Prop121Proof}
We prove here Proposition \ref{prop:ps}. We consider a generic real function $\Psi$ solving the equation
\begin{equation}
\Box_{g_{a,M}}\Psi = F.
\end{equation}
 Defining $\mathbf{u} = \rweight\Psi$, and recalling from (\ref{vdef}) the definition of $v$, 
 as well as other definitions from Section \ref{sec:FoliationCalculations}, the wave equation in physical space takes the form
\begin{equation}\label{id:WEPhysicalSpaceAppendix}
\rweight\rho^2\Box_{g_{a,M}}\Psi =v^{-2}\Rs\Rs\mathbf{u} - v^{-2}\Ks\Ks\mathbf{u} + \tfrac{1}{\sin\theta}\partial_\theta(\sin\theta\partial_\theta\mathbf{u}) + \Phi^\star\Phi^\star\mathbf{u} - v^{-2}V_1\mathbf{u}.
\end{equation}
We define $\Wpm := v^{-1}(\Rs \pm \Ks)$, or equivalently
\[
\Wp = W_{\Ical^+}, \qquad \Wm = W_{\mathcal{H}^+}.
\]
The commutator identities
\begin{equation}
[\Wpm, \tfrac{1}{v}\Ks] = -\tfrac{v'}{v^2}\tfrac{1}{v}\Ks - \tfrac{2ar}{r^2+a^2}\partial_\phi, \qquad[\Wpm, \Phi^\star] = [\Wpm, \partial_\theta] = 0
\end{equation}
together imply
\begin{align}
[\Wpm, \tfrac{1}{v^2}(\Rs\Rs - \Ks\Ks)]\mathbf{u} &= [\Wpm, (\tfrac{1}{v}\Rs)^2 + \tfrac{v'}{v^2}(\tfrac{1}{v}\Rs)- \big(\tfrac{1}{v}\Ks\big)^2]\mathbf{u} \\
&= [\Wpm, \big(\Wpm \mp \tfrac{1}{v}\Ks\big)^2 + \tfrac{v'}{v^2}\big(\Wpm \mp \tfrac{1}{v}\Ks\big) - \big(\tfrac{1}{v}\Ks\big)^2]\mathbf{u}\nonumber\\
&= [\Wpm, \mp [\Wpm, \tfrac{1}{v}\Ks] \mp \tfrac{2}{v}\Ks\Wpm + \tfrac{v'}{v^2}\big(\Wpm \mp \tfrac{1}{v}\Ks\big)]\mathbf{u} \nonumber\\
&= [\Wpm, \tfrac{2ar}{r^2+a^2}\partial_\phi \mp \tfrac{2}{v}\Ks\Wpm + \tfrac{v'}{v^2}\Wpm ]\mathbf{u}\nonumber\\
&= \Wpm\left(\tfrac{2ar}{r^2+a^2}\right)\partial_\phi\mathbf{u} \pm \tfrac{2v'}{v^2}\tfrac{1}{v}\Ks \Wpm\mathbf{u} \pm \tfrac{2ar}{r^2+a^2}\partial_\phi\Wpm\mathbf{u} + \Wpm\left(\tfrac{v'}{v^2}\right)\Wpm\mathbf{u}.\nonumber
\end{align}
Therefore, defining
 \begin{equation}
\mathbf{U}_\pm = \frac{\Wpm\mathbf{u}}{\rweight},
\end{equation}
if $\Psi$ solves \eqref{id:WEPhysicalSpaceAppendix}, then $\mathbf{U}_\pm$ solves
\begin{equation}
\Box_{g_{a,M}}\mathbf{U}_\pm = \rho^{-2}\rweight^{-1}\left(\Wpm\left(\rweight\rho^2F\right)- [\Wpm, \tfrac{1}{v^2}(\Rs\Rs - \Ks\Ks)]\mathbf{u}+ \Wpm\left(\tfrac{V_1}{v^2}\right)\mathbf{u}\right).
\end{equation}
We may now set up our energy estimates. 
An asymptotic analysis for $r \to \infty$ gives
\begin{equation}
\tfrac{2v'}{v^3} = -2r + O(1), \quad \Wpm\left(\tfrac{2ar}{r^2+a^2}\right) = O(r^{-1}), \quad \Wpm\left(\tfrac{v'}{v^2}\right) = O(r^{-1}), \quad \Wpm\left(\tfrac{V_1}{v^2}\right) = O(r^{-1}), 
\end{equation}
and as $r\to r_+$,
\begin{equation}
\tfrac{2v'}{v^3} = 2(\Delta^{-1}) + O(1), \quad \Wpm\left(\tfrac{2ar}{r^2+a^2}\right) = O(\Delta^{1/2}), \quad \Wpm\left(\tfrac{v'}{v^2}\right) = O(\Delta^{-1}), \quad v^2\Wpm\left(\tfrac{V_1}{v^2}\right) = O(\Delta^{3/2}).
\end{equation}
Then, for large $r$, selecting $\Wp$, taking the Cauchy-Schwartz inequality and commuting through gives
\begin{align}
\left|\partial_t\mathbf{U}_+\Box_{g_{a,M}}\mathbf{U}_+ - \tfrac{2}{\rweight}\left|\partial_t\mathbf{U}_+\right|^2\right|&\lesssim \rweight^{-2}\left|\partial_t\mathbf{U}_+\right|^2 + \rweight^{-4}\left(\left|\partial_\phi\mathbf{U}_+\right|^2 + \left|\mathbf{U}_+\right|^2 + |\partial_\phi\mathbf{u}|^2 + |\mathbf{u}|^2\right) + \\
&\qquad + \left|\partial_t\mathbf{U}_+\right|\left(|F| + \left|\tfrac{\Wp(\rweight F)}{\rweight}\right|\right).\nonumber
\end{align}
Additionally, we have the divergence identity
\begin{equation}
\nabla^\mu \left(\tfrac{M^\alpha}{r^{1+\alpha}}\nabla_\mu \left|\mathbf{U}_+\right|^2 - \nabla_\mu\left(\tfrac{M^\alpha}{r^{1+\alpha}}\right) \left|\mathbf{U}_+\right|^2\right) = \tfrac{M^\alpha}{r^{1+\alpha}}\left(\Uplus\Box_{g_{a,M}}\Uplus + \nabla_\mu \Uplus\nabla^\mu\Uplus\right) - \Box_{g_{a,M}}\left(\tfrac{M^\alpha}{r^{1+\alpha}}\right)|\Uplus|^2 \, .
\end{equation}
There exists a constant $C$ depending only on $a, M$ such that, for $r > 3M$, and for any $c \in (0,1]$,
\begin{equation}
\tfrac{2}{\rweight}\left|\partial_t\mathbf{U}_+\right|^2 + \tfrac{cM^\alpha}{r^{1+\alpha}}\nabla_\mu \Uplus\nabla^\mu\Uplus \geq \tfrac{cM^\alpha}{r^{1+\alpha}}|D\Uplus|^2 - \tfrac{C}{r^2}|D\Uplus|^2 \, .
\end{equation}
Consequently, defining the current
\begin{equation}
{\bf J}^{\mathcal{I^+}}_\mu [\mathbf{U}_+] = \chi_R\left(\mathbf{T}_{\mu\nu}[\mathbf{U}_+]T^\nu + \tfrac{M^\alpha}{r^{1+\alpha}}\nabla_\mu \left|\mathbf{U}_+\right|^2 - \nabla_\mu\left(\tfrac{M^\alpha}{r^{1+\alpha}}\right) \left|\mathbf{U}_+\right|^2\right),
\end{equation}
where $\chi_R$ is as defined in Proposition \ref{prop:ps} and $\mathbf{T}$ is the energy-momentum tensor (\ref{def:EMTensor}), the inequality
\begin{equation}
\left(\tfrac{M^\alpha}{r^{1+\alpha}}|\Uplus| + \left|\partial_t\mathbf{U}_+\right|\right)\left(|F| + \left|\tfrac{\Wp(\rweight F)}{\rweight}\right|\right) \leq \tfrac{\delta}{r^{1+\alpha}}\left(\tfrac{1}{r^2}|\Uplus|^2 + |\partial_t\Uplus|^2\right) + \tfrac{r^{1+\alpha}}{\delta}\left(|F|^2 + \left|\tfrac{\Wp(\rweight F)}{\rweight}\right|^2\right), \nonumber
\end{equation}
as well as the bound
\begin{equation}
r^{1+\alpha}\left(|F|^2 + \left|\tfrac{\Wp(\rweight F)}{\rweight}\right|^2\right) \leq C\tfrac{\varepsilon^2}{r^{1+\alpha}}\left(|D\Uplus|^2 + |D\Psi|^2\right)
\end{equation}
implies that for sufficiently small $\delta$ depending on $a, M$, and $\varepsilon$ depending on $a, M, \delta$, there exist positive constants $c, C$ such that
\begin{equation}
\nabla^\mu {\bf J}_\mu^{\mathcal{I}_+} [\mathbf{U}_+]\geq \chi_R\tfrac{cM^\alpha}{r^{1+\alpha}}|D\Uplus|^2 - C\chi_R\left(\tfrac{M}{r^2}|D\Uplus|^2 + \tfrac{M^\alpha}{r^{1+\alpha}}\left(|D\Psi|^2 + \left|\!\tfrac{\Psi}{r}\!\right|^2\right)\right) - C\chi_R' \left(|D\Uplus|^2 + |D\Psi|^2 + \left|\!\tfrac{\Psi}{r}\!\right|^2\right). \nonumber
\end{equation}
For sufficiently large $R$, the divergence theorem combined with an elementary Hardy inequality implies Proposition \ref{prop:ps}.

For $r$ close to the horizon, we set $\Wm$ and consequently $\Uminus$ and multiply by $K\Uminus$, where $K$ is the standard Hawking vector field
\begin{equation}\index{K@$K$}
K = \partial_t + \frac{a}{r_+^2+a^2}\partial_\phi = \Ks + \frac{a(r^2-r_+^2)}{(r^2+a^2)(r_+^2+a^2)}\partial_\phi.
\end{equation}
For $r-r_+$ sufficiently small, this implies, for any $\delta > 0$, that
\begin{align}
\left|K\mathbf{U}_-\Box_{g_{a,M}}\mathbf{U}_- - \tfrac{2}{\Delta}\left|K\mathbf{U}_-\right|^2\right|&\leq \frac{\delta}{\Delta}|K\Uminus|^2 + C_\delta\Delta\left(|\partial_\phi\Uminus|^2 + |\partial_\phi\mathbf{u}|^2 + \Delta|\mathbf{u}|^2\right) + \tfrac{C_\delta}{\Delta}|\Uminus|^2.
\end{align}
For some small $c$, we define the current
\begin{equation}
{\bf J}^{\mathcal{H}^+}_\mu [\mathbf{U}_-]= \chi_{\widetilde{\delta}}\left(\mathbf{T}_{\mu\nu}[\Uminus]T^\nu + \tfrac{cM^\alpha}{r^{1+\alpha}}\nabla_\mu \left|\Uminus \right|^2 - \nabla_\mu\left(\tfrac{cM^\alpha}{r^{1+\alpha}}\right) \left|\Uminus\right|^2\right),
\end{equation}
where $\chi_{\widetilde{\delta}}$ is again as defined in Proposition \ref{prop:ps}. In order to handle boundary terms, we note that for any sufficiently regular solution of the wave equation, as $r\to r_+$,
\begin{equation}
|\Uminus| + |\Rs\Uminus| + |T\Uminus| + |\slashed{\nabla}\Uminus| = O(\Delta^{1/2}).
\end{equation}
Applying the divergence theorem and bounding terms as before gives the remaining statement of Proposition \ref{prop:ps}.
\section{Quantitative mode stability for the perturbed equation}\label{appendix:yakov}
	
We explain the changes required in Section 3.1 of \cite{SR14} to prove the estimate (\ref{h2estimate}) of our paper. The reader is advised to read this appendix with the paper \cite{SR14} at hand.  As in \cite{SR14} we let $\mathcal{B}\subset\mathbb{R}$ and let $\mathcal{C} \subset \{(m, \ell) \in \mathbb{Z}\times\mathbb{Z} : \ell \geq |m|\}$ such that
\begin{align}
\sup_{\omega \in \mathcal{B}},  |\omega|+|\omega|^{-1} < \infty, \\
\sup_{(m, \ell)\in\mathcal{C}}|m|+|\ell| < \infty.
\end{align}
Then, for a future integrable function $\psi$, the projection $\mathcal{P}_{\mathcal{B}, \mathcal{C}}$ onto $\mathcal{B}\times\mathcal{C}$ is defined by
\begin{equation}
\mathcal{P}_{\mathcal{B}, \mathcal{C}}\psi = \sum_{m, \ell \in\mathcal{C}}\int_{\mathcal{B}}\, \psi_{m\ell}^{(a\omega)}e^{-i(\omega t - m\phi)}S_{m\ell}(a\omega, \cos\theta)d\omega,
\end{equation}
where $S_{m\ell}$ are defined in \eqref{def:spheroidalharmonics}. A careful selection of $\mathcal{B}$ and $\mathcal{C}$ along with standard asymptotic bounds for spheroidal harmonics implies that for any values of $\omega_{\text{low}}, \omega_{\text{high}}$, and $\epsilon_{\text{width}}$, there exist sets $\mathcal{B}$ and $\mathcal{C}$ such that
\begin{equation}
\int_{-\infty}^\infty \sum_{m, \ell} \mathfrak{H}_2 \leq \int_{\mathcal{H}^+}|\mathcal{P}_{\mathcal{B}, \mathcal{C}}\Psi_{\chi,\mathcal{T}}|^2.
\end{equation}
The right hand side is now bounded by the following Proposition, which is the exact analogue of Theorem 1.9 of \cite{SR14}. Note, however, the extra spacetime term on the right arising from the perturbative term.
\begin{proposition} \label{prop:yakov+}
Consider $\Psi_{\mathcal{T}}$ solving the equation \eqref{def:PhysCutoffEq} and the $\Psi_{\chi,\mathcal{T}}$ defined in (\ref{def:Psichi}) satisfying (\ref{def:WaveEqCutoff}), with Fourier transform $\uchi$ defined in (\ref{def:uchi}).
Then the following estimate holds:
\begin{align}
\int_{\mathcal{H}^+}&|\mathcal{P}_{\mathcal{B}, \mathcal{C}}\Psi_{\chi,\mathcal{T}}|^2 + \int_{\mathcal{I}^+}|\partial\mathcal{P}_{\mathcal{B}, \mathcal{C}}\Psi_{\chi,\mathcal{T}}|^2 + \int_{\mathbb{R}\times (r_0, r_1)\times \mathbb{S}^2}|\mathcal{P}_{\mathcal{B}, \mathcal{C}}\Psi_{\chi,\mathcal{T}}|^2\nonumber  \\
&\leq B(r_0, r_1, C_{\mathcal{B}}, C_{\mathcal{C}}, a, M)\left(\int_{\widetilde{\Sigma}_0}|\partial\Psi_{\mathcal{T}}|^2+ \varepsilon^2\mathcal{S}[\uchi]\right)  \, .
\end{align}
\end{proposition}

The proof of Proposition \ref{prop:yakov+} closely follows Section 3 in \cite{SR14}, so we will outline the proof here, but we only go into detail for additional considerations which arise from the perturbative term. We first define, for $\epsilon \geq 0$ (do not confuse $\epsilon$ and $\varepsilon$, the latter being the size of the perturbation term in the wave equation),
\begin{equation}
\Psi_{\chi,\mathcal{T}, \epsilon} = e^{-\epsilon t} \chi \Psi_\mathcal{T}, 
\end{equation}
remarking that this serves the same purpose as $\psi_{\epsilon, \text{\LeftScissors}}$ in the aforementioned work. Since $\chi$ is supported for $\wtau > 0$, for which we have the relation $t \geq |\rs| - C$ for some constant $C$, Proposition \ref{prop:suffi} implies exponential decay of $\Psi_{\chi, \mathcal{T}, \epsilon}$ as $t\to\infty$, and therefore also as $\rs\to\pm\infty$. 

We next set up the decomposed wave equation as in \cite{SR14}. Taking the spheroidal harmonic projections
\begin{equation}
\widehat{\mathbf{u}}_\epsilon = \big((r^2+a^2)^{1/2}\Psi_{\chi,\mathcal{T}, \epsilon}\big)_{m\ell}^{(a\omega)}, \qquad 
\end{equation}
If we define
\begin{align}
E_\epsilon = e^{-\epsilon t}\big(\varepsilon\chi_\mathcal{T}(\BfB \Psi_{\chi , \mathcal{T}}) + 2\nabla^\alpha \chi \nabla_\alpha \Psi_{\mathcal{T}} + (\Box_{g_{a, M}}\chi)\Psi_{\mathcal{T}} - \varepsilon \chi_{\mathcal{T}}\Psi_{\mathcal{T}}\left(\BfB \chi-\BfB^0 \chi\right)\big),
\end{align}
along with
\begin{equation}
F_\epsilon = (E_\epsilon)_{m\ell}^{(a\omega)}, \qquad H_\epsilon = \frac{\Delta}{(r^2+a^2)^{1/2}}F_\epsilon,
\end{equation}
and $\omega_\epsilon$, $V_\epsilon$ as in \cite{SR14} we obtain the decomposed equation 
$
\widehat{\mathbf{u}}_\epsilon^{\prime \prime} + (\omega_\epsilon^2 -V_\epsilon) \widehat{\mathbf{u}}_\epsilon = H_\epsilon.
$
Finally, we define $\ph_{hor}$ and $\ph_{out}$ to be solutions to the homogeneous equation 
$
\ph^{\prime \prime} + (\omega^2 -V) \ph = 0,
$
with the boundary conditions $\ph_{hor} \sim (r-r_+)^{\frac{i(am-2Mr_+\omega)}{r_+-r_-}}$ as $\rs \to -\infty$ and $\ph_{out}\sim e^{i\omega\rs}$ as $\rs\to\infty$, and 
$\ph_{hor, \epsilon}$ and $\ph_{out, \epsilon}$ to be solutions to the homogeneous equation
$
\ph^{\prime \prime} + (\omega_\epsilon^2 -V_\epsilon) \ph = 0
$
with analogous boundary conditions. Then, $\widehat{\mathbf{u}}_\epsilon(\rs)$ is given by the representation formula
\begin{equation}
\widehat{\mathbf{u}}_\epsilon(\rs) = W_\epsilon^{-1}\left(\ph_{out, \epsilon}(\rs)\int_{-\infty}^{\rs} \ph_{hor, \epsilon}H_\epsilon(x^\star)\, dx^\star + \ph_{hor, \epsilon}(\rs)\int_{\rs}^\infty \ph_{out, \epsilon}H_\epsilon(x^\star)\, dx^\star\right)
\end{equation}
proven in Proposition 3.1 of \cite{SR14}, where $W_\epsilon = \ph'_{out, \epsilon} \ph_{hor, \epsilon} -  \ph'_{hor, \epsilon}\ph_{out, \epsilon}$. If we can show that this representation formula continues to hold in $L^2$ for $\epsilon = 0$ (as given in Lemmas 3.4, 3.5 of \cite{SR14}), the analysis of Section 3.2 of that paper implies Proposition \ref{prop:yakov+}. It therefore suffices to prove that the representation formula converges in $L^2$. $L^\infty$ convergence of $\ph_{hor, \epsilon}$ and $\ph_{out, \epsilon}$ in particular allows us to reduce this problem to the $L^2$ convergence results which are similar to Lemmas 3.2 and 3.3 in the aforementioned work. We adapt Lemma 3.2 from \cite{SR14}, noting that we need to add a spacetime term on the right hand side.
\begin{lemma}
Lemma 3.2 of \cite{SR14} holds under the following modification:
\begin{equation}
\limsup_{\epsilon\to 0^+}\int_{\mathcal{B}}\sum_{m, l}\int_{r_+}^\infty |F_\epsilon|^2 r^2\, dr\, d\omega = \int_{\mathcal{B}}\sum_{m, l}\int_{r_+}^\infty |F|^2 r^2\, dr\, d\omega \lesssim \int_{\Sigma_0}|\partial\Psi_{\mathcal{T}}|^2 + \varepsilon^2\mathcal{S}[\uchi] \, .
\end{equation}
\end{lemma}
\begin{proof}
As in \cite{SR14}, one first proves the inequality, the equality then following from the  dominated convergence theorem.
To prove the inequality, one applies Plancherel's theorem (note $F_\epsilon$ is future integrable following in turn from $\Psi_{\mathcal{T}}$ and $\Psi_{\chi,\mathcal{T}}$ being future integrable) and the decomposition \eqref{def:WavePerturbedFreq} for $H_\epsilon$. Then, terms appearing in $\Hchi$ may be be bounded as in \cite{SR14}, noting that $\BfB\chi$ generally decays faster than $\Box_{g_{a, M}}\chi$ in $r$. The result for $\HB$ follows from direct integration along with the Cauchy-Schwarz inequality.
\end{proof}
We finally show how to adapt Lemma 3.3 of \cite{SR14}, which implies convergence of $\widehat{\mathbf{u}}_\epsilon$ to $\widehat{\mathbf{u}}$ as $
\epsilon \rightarrow 0$ for large $r$. There it turns out that the extra terms coming from $\BfB$ in the decomposition of $H$ have enough radial decay to estimate them with a direct application of Plancherel and the Cauchy-Schwarz inequality, if we allow for an ($\varepsilon$ small) additional spacetime term on the right. 
\begin{lemma}
Lemma 3.3 of \cite{SR14} holds with the following modification:
\begin{equation}
\left\lVert\int_{\rs}^\infty \ph_{out}(x^\star)H(x^\star)\, dx^\star\right\rVert_{L^2_{\omega\in\mathcal{B}, (m, \ell)\in\mathcal{C}}}^2\leq B(\rs, C_{\mathcal{B}}, C_{\mathcal{C}})\left(\int_{\Sigma_0}|\partial\Psi_{\mathcal{T}}|^2 + \varepsilon^2\mathcal{S}[\uchi]\right)
\end{equation}
and 
\begin{equation}
\lim_{\epsilon\to 0}\int_{\rs}^\infty \ph_{out, \epsilon}(x^\star)H_\epsilon(x^\star) \, dx^\star = \int_{\rs}^\infty \ph_{out}(x^\star)H(x^\star)\, dx^\star
\end{equation}
in $L^2_{\omega \in \mathcal{B}, (m, \ell)\in\mathcal{C}}$.
\end{lemma}
\begin{proof}
As in \cite{SR14} we may replace $\ph_{out}$ with its highest order part, and reduce this to Parseval's inequality in $x^\star$. We write
\begin{equation}
|\ph_{out}(\rs) - e^{i\omega \rs}|\leq Cr^{-1}
\end{equation}
and bound the terms with $r^{-1}$ decay as before. It suffices to prove the bound
\begin{equation}
\int_{\mathcal{B}}\sum_{m, \ell\in\mathcal{C}}\left|\int_{\rs}^\infty e^{i\omega x^\star}H(x^\star)\, dx^\star\right|^2 \, d\omega\leq B(\rs, C_{\mathcal{B}}, C_{\mathcal{C}})\left(\int_{\Sigma_0}|\partial\Psi_{\mathcal{T}}|^2 + \varepsilon^2\mathcal{S}[\uchi]\right).
\end{equation}
We again split $H$ as in \eqref{def:WavePerturbedFreq}, additionally decomposing in the spheroidal harmonics, and note that the theorem for $H_\chi$ follows as in \cite{SR14}, noting again that $\BfB\chi$ satisfies better decay bounds than $\Box_{g_{a, M}}\chi$ in $r$. To bound $H_\BfB$ terms, we note that the Cauchy-Schwarz inequality implies
\begin{equation}
\left|\int_{\rs}^\infty e^{i\omega x^\star}H_\BfB(x^\star)\, dx^\star\right| \leq \left\lVert \Delta^{-1/2} r^{3/2+\alpha} H_\BfB \right\rVert_{L^2_{\rs}}\left\lVert \Delta^{1/2}r^{-3/2-\alpha}\right\rVert_{L^2_{\rs}}.
\end{equation}
Plancherel's theorem and direct integration imply
\begin{equation}
\int_{\mathcal{B}}\sum_{m, \ell\in\mathcal{C}}\left\lVert \Delta^{-1/2} r^{3/2+\alpha} H_\BfB \right\rVert^2_{L^2_{\rs}} \lesssim \varepsilon^2\mathcal{S}[\uchi], \qquad \left\lVert \Delta^{1/2}r^{-3/2-\alpha}\right\rVert_{L^2_{\rs}} \lesssim C_\alpha.
\end{equation}
The remainder of the proof follows as in \cite{SR14}.
\end{proof}

\printindex
\newpage
\bibliographystyle{plain}
\bibliography{Bibliography}
\end{document}